\documentclass[english, 11pt, a4paper]{article}
\pdfoutput=1
\usepackage[top=2.5cm, bottom=2.5cm, left=2.5cm, right=2.5cm]{geometry}

\usepackage{titlesec}
\usepackage{etoolbox}
\makeatletter
\patchcmd{\ttlh@hang}{\parindent\z@}{\parindent\z@\leavevmode}{}{}
\patchcmd{\ttlh@hang}{\noindent}{}{}{}
\makeatother

\titlespacing\section{0pt}{12pt plus 4pt minus 2pt}{3pt plus 2pt minus 2pt}
\titlespacing\subsection{0pt}{12pt plus 4pt minus 2pt}{2pt plus 1pt minus 2pt}
\titlespacing\subsubsection{0pt}{12pt plus 4pt minus 2pt}{2pt plus 2pt minus 2pt}

\RequirePackage{amsthm}
\theoremstyle{plain}
\newtheorem{theorem}{Theorem}
\newtheorem{lemma}[theorem]{Lemma}

\theoremstyle{definition}

\theoremstyle{remark}

\usepackage{url}

\usepackage[dvinames]{xcolor}
\definecolor{ACMGreen}{rgb}{0,255,0}

\usepackage[english]{babel}
\usepackage{amsmath,amsfonts,amssymb}

\usepackage{graphicx}

\usepackage{tikz}
\usetikzlibrary{calc}
\usetikzlibrary{decorations.pathmorphing}
\usetikzlibrary{positioning}
\usetikzlibrary{backgrounds}
\usepackage{ifthen}

\usepackage[separate-uncertainty,group-separator={,},binary-units]{siunitx}

\usepackage[font={small,it}]{caption}

\usepackage{hyperref}

\usepackage{enumitem}
\newenvironment{itemize*}%
{\begin{itemize}[leftmargin=1.5em]%
   \setlength{\itemsep}{0pt}%
   \setlength{\parskip}{0pt}%
   \setlength{\topsep}{1pt}%
   \setlength{\partopsep}{1pt}%
}{\end{itemize}}

\newenvironment{enumerate*}%
{\begin{enumerate}[leftmargin=2.5em]%
		\setlength{\itemsep}{0pt}%
		\setlength{\parskip}{0pt}%
		\setlength{\topsep}{1pt}%
		\setlength{\partopsep}{1pt}%
}{\end{enumerate}}

\usepackage{todonotes}

\newcommand{\Oh}{\ensuremath{\mathcal O}}
\newcommand{\pld}[3]{\ensuremath{\textsc{Pld}\,([#1, #2), #3)}}

\newcommand*{\diff}{\mathop{}\!\mathrm{d}}
\newcommand{\emhh}{\textsl{EM-HH}}
\newcommand{\emes}{\textsl{EM-ES}}
\newcommand{\vles}{\textsl{VL-ES}}
\newcommand{\emca}{\textsl{EM-CA}}

\newcommand{\emcmes}{\textsl{EM-CM/ES}}

\newcommand{\emlfr}{\textsl{EM-LFR}}
\newcommand{\invar}[1]{(\texttt{I#1})}

\newcommand{\degs}{\ensuremath{\mathcal D}}
\newcommand{\defrel}{:=}

\newcommand{\eswapi}[4]{\ensuremath{\sigma_{#1}(\langle #2, #3\rangle,\: #4)}}
\newcommand{\eswap}[3]{\eswapi{}{#1}{#2}{#3}}
\newcommand{\mult}[1]{\#(#1)}
\newcommand{\moment}[1]{\langle #1 \rangle}

\DeclareMathOperator{\scan}{scan}
\DeclareMathOperator{\sort}{sort}
\DeclareMathOperator{\polylog}{polylog}

\usepackage{authblk}
\title{%
	I/O-Efficient Generation of Massive Graphs\\Following the LFR Benchmark%
	\footnote{%
		This work was partially supported 
		by the DFG under grants ME 2088/3-2, %
		WA 654/22-2. %
		Parts of this paper were published as \cite{hmpw-iogmg-17}.
	}
}

\author[$\dag$]{Michael Hamann}
\author[$\ddag$]{Ulrich Meyer}
\author[$\ddag$]{Manuel Penschuck}
\author[$\ddag$]{Hung Tran}
\author[$\dag$]{\\Dorothea Wagner}

\affil[$\dag$]{%
	Institute of Theoretical Informatics,
	Karlsruhe Institute of Technology,\par
	Am Fasanengarten~5, 76131 Karlsruhe, Germany\par
	\vspace{0.5em}
	\{michael.hamann, dorothea.wagner\}@kit.edu\par
	\vspace{0.5em}
	\ 
}
\affil[$\ddag$]{%
	Institute for Computer Science,
	Goethe-University Frankfurt,\par
	Robert-Mayer-Stra{\ss}e 11-15,
	60325 Frankfurt am Main, Germany\par
	\vspace{0.5em}
	\{umeyer, htran, mpenschuck\}@ae.cs.uni-frankfurt.de\par
}

\newcommand{\variantScale}[1]{\scalebox{0.9}{#1}}
\newcommand{\clearpageA}{\clearpage}
\newcommand{\clearpageB}{}

\newcommand{\inA}[1]{#1}
\newcommand{\inB}[1]{}

\makeatletter%
\begin{document}
\maketitle
\begin{abstract}%
LFR is a popular benchmark graph generator used to evaluate community detection algorithms.
We present \emlfr{}, the first external memory algorithm able to generate massive complex networks following the LFR benchmark.
Its most expensive component is the generation of random graphs with prescribed degree sequences which can be divided into two steps:
the graphs are first materialized deterministically using the Havel-Hakimi algorithm, and then randomized.
Our main contributions are \emhh{} and \emes{}, two I/O-efficient external memory algorithms for these two steps.
We also propose \emcmes{}, an alternative sampling scheme using the Configuration Model and rewiring steps to obtain a random simple graph.
In an experimental evaluation we demonstrate their performance;
our implementation is able to handle graphs with more than 37 billion edges on a single machine, is competitive with a massive parallel distributed algorithm, and is faster than a state-of-the-art internal memory implementation even on instances fitting in main memory.
\emlfr's implementation is capable of generating large graph instances orders of magnitude faster than the original implementation.
We give evidence that both implementations yield graphs with matching properties by applying clustering algorithms to generated instances.
Similarly, we analyse the evolution of graph properties as \emes{} is executed on networks obtained with \emcmes{} and find that the alternative approach can accelerate the sampling process.
\end{abstract}
 \clearpage
\section{Introduction}
Complex networks, such as  web graphs or social networks, are usually composed of communities, also called clusters, that are internally dense but externally sparsely connected.
Finding these clusters, which can be disjoint or overlapping, is a common task in network analysis.
A large number of algorithms trying to find meaningful clusters have been proposed (see~\cite{Fortunato201075,Harenberg2014,fh-ca-16} for an overview).
Commonly synthetic benchmarks are used to evaluate and compare these clustering algorithms, since for most real-world networks it is unknown which communities they contain and which of them are actually detectable through structure~\cite{Bader2014,fh-ca-16}.
In the last years, the LFR benchmark~\cite{Lancichinetti2008, Lancichinetti2009} has become a standard benchmark for such experimental studies, both for disjoint and for overlapping communities~\cite{ekgb-ancac-16}.

With the emergence of massive networks that cannot be handled in the main memory of a single computer, new clustering schemes have been proposed for advanced models of computation~\cite{Buzun2014,zy-a-16}.
Since such algorithms typically use hierarchical input representations, quality results of small benchmarks may not be generalizable to larger instances.
To produce such large instances exceeding main memory, we propose a generator in the external memory (EM) model of computation that follows the LFR benchmark.

The distributed CKB benchmark~\cite{Chykhradze2014} is a step in a similar direction, however, it considers only overlapping clusters and uses a different model of communities. %
In contrast, our approach is a direct realization of the established LFR benchmark and supports both disjoint and overlapping clusters.

\subsection{Random Graphs from a prescribed Degree Sequence}\label{subsec:empiricalconv}
In preliminary experiments, we identified the generation of random graphs with prescribed degree sequence as the main issue when transferring the LFR benchmark into an EM setting --- both in terms of algorithmic complexity and runtime.
To do so, the LFR benchmark uses the \emph{fixed degree sequence model} (FDSM), also known as edge-switching Markov-chain algorithm (e.g.\ \cite{Milo2003}).
It consists of a) generating a deterministic graph from a prescribed degree sequence and b) randomizing this graph using random edge switches.
For each edge switch, two edges are chosen uniformly at random and two of the endpoints are swapped if the resulting graph is still simple (for details, see section~\ref{sec:io-efficient-edge-swaps}).
Each edge switch can be seen as a transition in a Markov chain.
This Markov chain is irreducible~\cite{eh-s-80}, symmetric and aperiodic~\cite{gmz-tmcsm-03} and therefore converges to the uniform distribution.
It also has been shown to converge in polynomial time if the maximum degree is not too large compared to the number of edges~\cite{Greenhill2017}.
However, the analytical bound of the mixing time is impractically high even for comparably small graphs as it contains the sum of all degrees to the power of nine.

Experimental results on the occurrence of certain motifs in networks~\cite{Milo2003} suggest that $100m$ steps should be more than enough where $m$ is the number of edges.
Further results for random connected graphs~\cite{gmz-tmcsm-03} suggest that the average and maximum path length and link load converge between $2m$ and $8m$ swaps.
More recently, further theoretical arguments and experiments showed that $10m$ to $30m$ steps are enough~\cite{rps-awtyw-12}.

A faster way to realize a given degree sequence is the \emph{Configuration Model}.
The problem here is that multi-edges and loops may be generated.
In the \emph{Erased Configuration Model} these illegal edges are deleted.
However, doing so alters the graph properties since skewed degree distributions, necessary for the LFR benchmark, are not properly realized~\cite{Schlauch2015}.
In this context the question arises whether edge switches starting from the Configuration Model can be used to uniformly sample simple graphs at random.

\clearpageA

\subsection{Our Contribution}
Our main contributions are the first external memory versions of the LFR benchmark and the FDSM.
After defining our notation, we introduce the LFR benchmark in more detail and then focus on the FDSM.
We describe the realization of the two steps of the classic FDSM, namely a) generating a deterministic graph from a prescribed degree sequence [\emhh, section~\ref{sec:mat-degree-sequence}] and b) randomizing this graph using random edge switches [\emes, section~\ref{sec:io-efficient-edge-swaps}].
These steps form a pipeline moving data from one algorithm to the next.
In section~\ref{sec:cmes}, we describe the alternative approach \emcmes{} generating uniform random simple graphs using the Configuration Model and edge rewiring.

Sections~\ref{sec:community-assignment}, \ref{sec:lfr-rewiring} and \ref{sec:implementation} describe algorithms for the remaining steps of the external memory LFR benchmark, \emlfr.
We conclude with an experimental evaluation of our algorithms and demonstrate that our EM version of the FDSM is faster than an existing internal memory implementation, scales well to large instances, and can compete with a  distributed parallel algorithm. %
Further, we compare \emlfr{} to the original LFR implementation and show that \emlfr{} is significantly faster while producing equivalent networks in terms of community detection algorithm performance and graph properties.
We also investigate the mixing time of \emes{} and \emcmes{}, give evidence that our alternative sampling scheme quickly yields uniform samples and that the number of swaps suggested by the original LFR implementation can be kept for \emlfr.

\section{Preliminaries and Notation}\label{sec:notation}
Define $[k] \defrel \{1, \ldots, k\}$ for $k \in \mathbb N_{>0}$.
A graph $G = (V, E)$ has $n=|V|$ sequentially numbered nodes $V = \{v_1, \ldots, v_n\}$ and $m=|E|$ edges.
Unless stated differently, graphs are assumed to be undirected and unweighted.
It is called \emph{simple} if it contains neither multi-edges nor self-loops.
To obtain a unique representation of an \emph{undirected} edge $\{u, v\} \in E$, we write $[u, v] \in E$ where $u \le v$; in contrast to a directed edge, the ordering shall be used algorithmically but does not carry any meaning for the application.
$\mathcal D = (d_1, \ldots, d_n)$ is a degree sequence of graph $G$ iff $\forall v_i \{\in\} V: \deg(v_i) = d_i$.

We denote an integer powerlaw distribution with exponent $-\gamma \in \mathbb R$ for $\gamma \ge 1$ and values between the limits $a, b \in \mathbb N_{>0}$ with $a < b$ as $\pld ab\gamma$.
Let $X$ be an integer random variable drawn from $\pld ab\gamma$ then $\mathbb P[X{=}k] \propto k^{-\gamma}$ (proportional to) if $a \le k < b$ and $\mathbb P[X{=}k] = 0$ otherwise.
A statement depending on some number $x > 0$ is said to hold \emph{with high probability} if it is satisfied with probability at least $1 - 1/x^c$ for some constant $c \ge 1$.

\textsl{Also refer to Table~\ref{table:def-summary} (Appendix) which contains a summary of commonly used definitions}.

\subsection{External-Memory Model}\label{ssec:emm}
We use the commonly accepted external memory model by Aggarwal and Vitter~\cite{AggarwalVitter88}.
It features a two-level memory hierarchy with fast internal memory (IM) which may hold up to $M$ data items, and a slow disk of unbounded size.
The input and output of an algorithm are stored in EM while computation is only possible on values in IM.
The measure of an algorithm's performance is the number of I/Os required.
Each I/O transfers a block of $B$ consecutive items between memory levels.
Reading or writing $n$ contiguous items from or to disk requires $\scan(n) = \Theta(n/B)$~I/Os.
Sorting $n$ consecutive items triggers $\sort(n)=\Theta((n/B) \cdot \log_{M/B}(n/B))$~I/Os.
For all realistic values of $n$, $B$ and $M$, $\scan(n)<\sort(n)\ll n$.
Sorting complexity constitutes a lower bound for most intuitively non-trivial EM tasks~\cite{AggarwalVitter88, DBLP:conf/dagstuhl/2002amh}.

\clearpageA

\subsection{Time Forward Processing}
Let $\mathcal{A}$ be an algorithm performing discrete events over time (e.g., iterations of a loop) that produce values which are reused by following events.
The data dependencies of $\mathcal{A}$ can be modeled using a directed acyclic graph $G{=}(V,E)$ where every node $v\in V$ corresponds to an event~\cite{MahZeh-survey}.
The edge $(u,v)\in E$ indicates that the value produced by $u$ will be required by $v$.
When computing a solution, the algorithm traverses $G$ in some topological order.
For simplicity, we assume $G$ to be already ordered, i.e. $\forall\, (u,v) \in E\colon\ u<v$.
Then, the Time Forward Processing (TFP) technique uses a minimum priority queue (PQ) to provide the means to transport data as implied by $G$:
iterate over the events in increasing order and receive for each $u$ the messages sent to it by claiming and removing all items with priority $u$ from the PQ.
Inductively, these messages have minimal priority amongst all items stored in the PQ.
The event then computes its result $x_u$ and sends it to every successor $v$ by inserting $x_u$ into the PQ with priority $v$.
Using a suited EM PQ~\cite{DBLP:conf/wads/Arge95, Sanders00}, TFP incurs $\Oh(\sort(k))$ I/Os, where $k$ is the number of messages sent.
\section{The LFR Benchmark}\label{sec:lfr}
\begin{figure*}
	\begin{center}
	\variantScale{%
		\scalebox{0.75}{%
{\setlength{\baselineskip}{0.9em}
\begin{tikzpicture}[
	node distance=12.5em and 1.5em,
	description/.style={anchor=center, font=\small},
	headline/.style={anchor=center, font=\small, align=left},
	gnode/.style={draw, circle, inner sep=0, minimum width=.45em, fill=black},
	legend/.style={rectangle, inner sep=0, anchor=north west, align=left},
	intra/.style={},
	inter/.style={color=red, thick, dotted}
]
	\newcommand\sz{0.5}
	\begin{scope}[scale = \sz]
		\node[gnode] (s2c1) {};
		\node[gnode] at ($(s2c1) + (3em, 0)$) (s2c2) {};
		\node[gnode] at ($(s2c1) + (1em, -2.5em)$) (s2c3) {};
		\node[gnode] at ($(s2c3) + (2.5em, -0.25em)$) (s2c4) {};
		\node[gnode] at ($(s2c4) + (2.5em, -2em)$) (s2c5) {};
		\node[gnode] at ($(s2c5) + (1em, 3.5em)$) (s2c6) {};
		\node[gnode] at ($(s2c6) + (2.5em, 1em)$) (s2c7) {};
		\draw[intra] (s2c1) -- ($(s2c1) + (1em, -1em)$);
		\draw[intra] (s2c1) -- ($(s2c1) + (1.5em, -0.5em)$);
		\draw[intra] (s2c1) -- ($(s2c1) + (0.5em, -1.5em)$);
		\draw[intra] (s2c1) -- ($(s2c1) + (-0.5em, -1.5em)$);
		\draw[intra] (s2c2) -- ($(s2c2) + (0.5em, -1.5em)$);
		\draw[intra] (s2c2) -- ($(s2c2) + (1.5em, -0.5em)$);
		\draw[intra] (s2c3) -- ($(s2c3) + (0.5em, -1.5em)$);
		\draw[intra] (s2c3) -- ($(s2c3) + (-0.5em, -1.5em)$);
		\draw[intra] (s2c3) -- ($(s2c3) + (-1.5em, -0.5em)$);
		\draw[intra] (s2c4) -- ($(s2c4) + (0.7em, 1.2em)$);
		\draw[intra] (s2c4) -- ($(s2c4) + (1.5em, -0.5em)$);
		\draw[intra] (s2c5) -- ($(s2c5) + (1em, 1em)$);
		\draw[intra] (s2c5) -- ($(s2c5) + (-1.4em, 0.2em)$);
		\draw[intra] (s2c6) -- ($(s2c6) + (-0.5em, 1.5em)$);
		\draw[intra] (s2c6) -- ($(s2c6) + (-0.5em, -1.5em)$);
		\draw[intra] (s2c6) -- ($(s2c6) + (-1.5em, -0.5em)$);
		\draw[intra] (s2c7) -- ($(s2c7) + (0.5em, -1.5em)$);
		\draw[intra] (s2c7) -- ($(s2c7) + (-0.5em, -1.5em)$);
		\draw[intra] (s2c7) -- ($(s2c7) + (1.5em, -0.5em)$);
		\draw[intra] (s2c7) -- ($(s2c7) + (-1em, -1em)$);
		
		\draw[dashed] ($0.5*(s2c4) + 0.5*(s2c6) + (-0.5em, 0)$) ellipse[x radius=7.5em, y radius=4em];

		\node[gnode] at ($(s2c5) + (-9em, -3em)$) (s2d1) {};
		\node[gnode] at ($(s2d1) + (-2em, -1.5em)$) (s2d2) {};
		\node[gnode] at ($(s2d1) + (-0.75em, -2.5em)$) (s2d3) {};
		\node[gnode] at ($(s2d1) + (2em, -2em)$) (s2d4) {};
		
		\draw[intra] (s2d1) -- ($(s2d1) + (1em, -1em)$);
		\draw[intra] (s2d1) -- ($(s2d1) + (-1em, -1em)$);
		\draw[intra] (s2d2) -- ($(s2d2) + (0.7em, 1.2em)$);
		\draw[intra] (s2d2) -- ($(s2d2) + (0.1em, -1.5em)$);
		\draw[intra] (s2d3) -- ($(s2d3) + (0.7em, 1.2em)$);
		\draw[intra] (s2d3) -- ($(s2d3) + (1em, -1em)$);
		\draw[intra] (s2d4) -- ($(s2d4) + (-1.2em, -0.5em)$);
		\draw[intra] (s2d4) -- ($(s2d4) + (-0.7em, 1.2em)$);
		
		\draw[dashed] ($0.25*(s2d1) + 0.25*(s2d2) + 0.25*(s2d3) + 0.25*(s2d4)$) circle[y radius=2.5em, x radius=3em];

		\node[gnode] at ($(s2d1) + (18em, -0.5em)$) (s2e1) {};
		\node[gnode] at ($(s2e1) + (-2em, -1.5em)$) (s2e2) {};
		\node[gnode] at ($(s2e1) + (0.5em, -2em)$) (s2e3) {};
		
		\draw[intra] (s2e1) -- ($(s2e1) + (-1em, 1em)$);
		\draw[intra] (s2e1) -- ($(s2e1) + (-1em, -1em)$);
		\draw[intra] (s2e2) -- ($(s2e2) + (0.7em, 1.2em)$);
		\draw[intra] (s2e2) -- ($(s2e2) + (1em, -1em)$);
		\draw[intra] (s2e3) -- ($(s2e3) + (-1.2em, 0.5em)$);
		\draw[intra] (s2e3) -- ($(s2e3) + (0.7em, 1.2em)$);
		
		\draw[dashed] ($0.33*(s2e1) + 0.33*(s2e2) + 0.33*(s2e3)$) circle[radius=2.5em];

		\begin{scope}[overlay]
			\draw[inter] (s2e3) -- ($(s2e3) + (3em, -1em)$);
			\draw[inter] (s2e2) -- ($(s2e2) + (-3em, -1em)$);
			\draw[inter] (s2d3) -- ($(s2d3) + (-1em, -3em)$);
			\draw[inter] (s2d2) -- ($(s2d2) + (-3em, 1em)$);
			\draw[inter] (s2d1) -- ($(s2d1) + (-1em, 3em)$);
			\draw[inter] (s2d4) -- ($(s2d4) + (3em, -1em)$);
			\draw[inter] (s2e1) -- ($(s2e1) + (1em, 3em)$);
			\draw[inter] (s2c7) -- ($(s2c7) + (3em, 1em)$);
			\draw[inter] (s2c6) -- ($(s2c6) + (1em, 3.5em)$);
			\draw[inter] (s2c1) -- ($(s2c1) + (-3em, 1em)$);
			\draw[inter] (s2c1) -- ($(s2c1) + (-4.5em, -1em)$);
			\draw[inter] (s2c5) -- ($(s2c5) + (0.3em, -3.3em)$);
			\draw[inter] (s2c3) -- ($(s2c3) + (-3em, -3em)$);
			\draw[inter] (s2c4) -- ($(s2c4) + (-1.5em, -4.2em)$);
		\end{scope}
	\end{scope}

	\newcommand\dz{20em} %
	\begin{scope}[scale = \sz]
		\node[gnode] at ($(s2c1) + (1/\sz*\dz, 0em)$) (c1) {};
		\node[gnode] at ($(c1) + (3em, 0)$) (c2) {};
		\node[gnode] at ($(c1) + (1em, -2.5em)$) (c3) {};
		\node[gnode] at ($(c3) + (2.5em, -0.25em)$) (c4) {};
		\node[gnode] at ($(c4) + (2.5em, -2em)$) (c5) {};
		\node[gnode] at ($(c5) + (1em, 3.5em)$) (c6) {};
		\node[gnode] at ($(c6) + (2.5em, 1em)$) (c7) {};
		
		\draw[intra] (c1) -- (c3);
		\draw[intra] (c1) -- (c4);
		\draw[intra] (c2) -- (c3);
		\draw[intra] (c1) -- (c6);
		\draw[intra] (c2) -- (c7);
		\draw[intra] (c6) -- (c7);
		\draw[intra] (c5) -- (c7);
		\draw[intra] (c3) -- (c6);
		\draw[intra] (c4) to[out=0,in=-135] (c7);
		\draw[intra] (c1) to[out=-30,in=120] (c5);
		
		\draw[dashed] ($0.5*(c4) + 0.5*(c6) + (-0.5em, 0)$) ellipse[x radius=7.5em, y radius=4em];

		\node[gnode] at ($(c5) + (-9em, -3em)$) (d1) {};
		\node[gnode] at ($(d1) + (-2em, -1.5em)$) (d2) {};
		\node[gnode] at ($(d1) + (-0.75em, -2.5em)$) (d3) {};
		\node[gnode] at ($(d1) + (2em, -2em)$) (d4) {};
		
		\draw[intra] (d1) -- (d2) -- (d3) -- (d4) -- (d1);
		
		\draw[dashed] ($0.25*(d1) + 0.25*(d2) + 0.25*(d3) + 0.25*(d4)$) circle[y radius=2.5em, x radius=3em];

		\node[gnode] at ($(d1) + (18em, -0.5em)$) (e1) {};
		\node[gnode] at ($(e1) + (-2em, -1.5em)$) (e2) {};
		\node[gnode] at ($(e1) + (0.5em, -2em)$) (e3) {};
		
		\draw[intra] (e1) -- (e2) -- (e3) -- (e1);
		
		\draw[dashed] ($0.33*(e1) + 0.33*(e2) + 0.33*(e3)$) circle[radius=2.5em];

		\begin{scope}[overlay]
			\draw[inter] (e3) to[out=0,in=-45, looseness=1.75] (c7);
			\path[inter] (e3) to[out=0,in=-45, looseness=1.75] (c7);
			\draw[inter] (e2) to[out=180,in=-80, looseness=1.75] (c6);
			\draw[inter] (d3) to[out=-135,in=170, looseness=2.5] (c1);
			\draw[inter] (c3) to[out=-135,in=180, looseness=5] (c1);
			\draw[inter] (d2) to[out=120,in=120, looseness=4] (d1);
			\draw[inter] (d4) to[out=70,in=-90, looseness=1.5] (c5);
			\draw[inter] (e1) to[out=180,in=-90, looseness=1.5] (c4);
		\end{scope}
	\end{scope}

	\begin{scope}[scale = \sz]
		\node[gnode] at ($(c1) + (1/\sz*\dz, 0em)$) (s3c1) {};
		\node[gnode] at ($(s3c1) + (3em, 0)$) (s3c2) {};
		\node[gnode] at ($(s3c1) + (1em, -2.5em)$) (s3c3) {};
		\node[gnode] at ($(s3c3) + (2.5em, -0.25em)$) (s3c4) {};
		\node[gnode] at ($(s3c4) + (2.5em, -2em)$) (s3c5) {};
		\node[gnode] at ($(s3c5) + (1em, 3.5em)$) (s3c6) {};
		\node[gnode] at ($(s3c6) + (2.5em, 1em)$) (s3c7) {};
		
		\draw[intra] (s3c1) -- (s3c3);
		\draw[intra] (s3c1) -- (s3c4);
		\draw[intra] (s3c2) -- (s3c3);
		\draw[intra] (s3c1) -- (s3c6);
		\draw[intra] (s3c2) -- (s3c7);
		\draw[intra] (s3c6) -- (s3c7);
		\draw[intra] (s3c5) -- (s3c7);
		\draw[intra] (s3c3) -- (s3c6);
		\draw[intra] (s3c4) to[out=0,in=-135] (s3c7);
		\draw[intra] (s3c1) to[out=-30,in=120] (s3c5);
		
		\draw[dashed] ($0.5*(s3c4) + 0.5*(s3c6) + (-0.5em, 0)$) ellipse[x radius=7.5em, y radius=4em];

		\node[gnode] at ($(s3c5) + (-9em, -3em)$) (s3d1) {};
		\node[gnode] at ($(s3d1) + (-2em, -1.5em)$) (s3d2) {};
		\node[gnode] at ($(s3d1) + (-0.75em, -2.5em)$) (s3d3) {};
		\node[gnode] at ($(s3d1) + (2em, -2em)$) (s3d4) {};
		
		\draw[intra] (s3d1) -- (s3d2) -- (s3d3) -- (s3d4) -- (s3d1);
		
		\draw[dashed] ($0.25*(s3d1) + 0.25*(s3d2) + 0.25*(s3d3) + 0.25*(s3d4)$) circle[y radius=2.5em, x radius=3em];

		\node[gnode] at ($(s3d1) + (18em, -0.5em)$) (s3e1) {};
		\node[gnode] at ($(s3e1) + (-2em, -1.5em)$) (s3e2) {};
		\node[gnode] at ($(s3e1) + (0.5em, -2em)$) (s3e3) {};
		
		\draw[intra] (s3e1) -- (s3e2) -- (s3e3) -- (s3e1);
		
		\draw[dashed] ($0.33*(s3e1) + 0.33*(s3e2) + 0.33*(s3e3)$) circle[radius=2.5em];

		\begin{scope}[overlay]
			\draw[inter] (s3e3) to[out=0,in=-45, looseness=1.75] (s3c7);
			\draw[inter] (s3e2) to[out=180,in=-80, looseness=1.75] (s3c6);
			\draw[inter] (s3d3) to[out=-135,in=170, looseness=2.5] (s3c1);
			\draw[inter] (s3d1) to[out=90,in=180, looseness=2] (s3c3);
			\draw[inter] (s3d2) to[out=120,in=150, looseness=2] (s3c1);
			\draw[inter] (s3d4) to[out=70,in=-90, looseness=1.5] (s3c5);
			\draw[inter] (s3e1) to[out=180,in=-90, looseness=1.5] (s3c4);
		\end{scope}
	\end{scope}

	\node[headline] at ($(s2c1) + (\dz + \sz*5em, 2em)$) (s3head) {\textbf{Sample intra- and inter-community edges}};
	\node[headline] at ($(s2c1) + (\sz*5em, 2em)$) (s2head) {\textbf{Degrees, community sizes and memberships}};
	\node[headline] at ($(c1) + (\dz + \sz*5em, 2em)$) (s4head) {\textbf{Remove (rewire) illegal edges}};

	\begin{scope}
		\node[legend] at ($(c1) + (-15em, -10em)$) (leg_intra) {Intra-community edge};
		\path[draw, intra] ($(leg_intra) + (-1.5em, 1em)$) to ($(leg_intra) + (1.5em, 1em)$);
		\node[legend, right of=leg_intra] (leg_inter) {Inter-community edge};
		\path[draw, inter] ($(leg_inter) + (-1.5em, 1em)$) to ($(leg_inter) + (1.5em, 1em)$);
		\node[legend, right of=leg_inter] (comm) {Community};
		\path[draw, dashed] ($(comm) + (-1.5em, 1em)$) to ($(comm) + (1.5em, 1em)$);
	\end{scope}
	
\end{tikzpicture}} %
}%
	}
	\end{center}
	\caption{
		\textbf{Left:} Sample node degrees and community sizes from two powerlaw distributions.
		The mixing parameter $\mu$ determines the fraction of the inter-community edges.
		Then, assign each node to sufficiently large communities. 
		\textbf{Center:} Sample intra-community graphs and inter-community edges.
		\textbf{Right:} Lastly, remove illegal inter-community edges respective to the global graph.
	}
	\label{fig:lfr_overview}
\end{figure*}
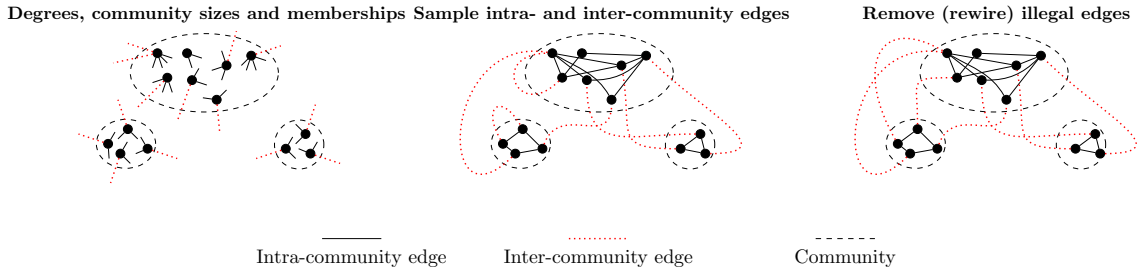
The LFR benchmark~\cite{Lancichinetti2008} describes a generator for random graphs featuring a planted community structure, a powerlaw degree distribution, and a powerlaw community size distribution.
A revised version~\cite{Lancichinetti2009} also introduces weighted and directed graphs with overlapping communities.
We consider the most commonly used versions with unweighted, undirected graphs and possibly overlapping communities.
All its parameters are listed in Table~\ref{tab:lfr-param} and are fully supported by \emlfr.

The revised generator \cite{Lancichinetti2009} changes the original algorithm \cite{Lancichinetti2008} even for the initial scenario of unweighted, undirected graphs and non-overlapping communities.
Here, we describe the more recent approach which is also used in the author's implementation:
initially, the degrees $\degs{=}(d_1, \ldots, d_n)$, the number of memberships $(\nu_1, \ldots, \nu_n)$ of each node, and community sizes $S=(s_1, \ldots, s_C)$ with $\sum_{\xi=1}^C s_\xi = \sum_{i=1}^n \nu_i$ are randomly sampled according to the supplied parameters.
Observe that the number of communities $C$ follows endogenously.
For our analysis, we assume that nodes are members in $\nu{=}\Oh(1)$ communities which implies $C{=}\Oh(n)$.\footnote{%
	If the maximal community size grows with $\Theta(n^\epsilon)$ for $\epsilon>0$, the number of communities is governed by $C {=} o(n)$.
}

Depending on the mixing parameter $0 < \mu < 1$, every node $v_i \in V$ features $d^\text{ext}_i = \mu d_i$ inter-community edges and $d^\text{in}_i = (1{-}\mu) \cdot d_i$ edges within its communities.
The algorithm assigns every node $v_i$ to either $\nu_i{=}1$ or $\nu_i = \nu$ communities at random such that the requested community sizes and number of communities per node are realized.
Further, the desired internal degree $d^\text{in}_i$ has to be strictly smaller than the size $s_\xi$ of its community~$\xi$.

In the case of overlapping communities, the internal degree is evenly split among all communities the node is part of.
Both the computation of $d^\text{in}_i$ and the splitting into several communities use non-deterministic rounding to avoid biases.

As illustrated in Fig. \ref{fig:lfr_overview}, the LFR benchmark then generates the inter-community graph using FDSM on the degree sequence $(d^\text{ext}_{1}, \ldots, d^\text{ext}_{n})$.
In order not to violate the mixing parameter~$\mu$, rewiring steps are applied to the global inter-community graph to replace edges between two nodes sharing a community.
Analogously, an intra-community network is sampled for each community.
In the overlapping case, rewiring steps are used to remove edges that exist in multiple communities and would result in duplicate edges in the final graph.

While for realistic parameters most intra-community graphs fit into main memory, we cannot assume the same for the global graph.
For the global graph and large communities, an EM variant of the FDSM is applicable, which we implement using \emhh{} and \emes{} described in sections~\ref{sec:mat-degree-sequence} and~\ref{sec:io-efficient-edge-swaps}.

\begin{table}
	\centering
	\scalebox{0.9}{\small
		\begin{tabular}{|p{0.25\columnwidth}|p{0.73\columnwidth}|}
			\hline
			\textsl{Parameter} & \textsl{Definition} \\\hline
			$n$ & Number of nodes to be produced \\
			$\pld{d_\text{min}}{d_\text{max}}{\gamma}$ & Degree distribution of nodes, typically $\gamma=2$ \\\hline
			$0 {\le} O {\le} n$, $\nu {\ge} 1$ & $O$ random nodes belong to $\nu$ communities; remainder has one membership\\\hline
			$\pld{s_\text{min}}{s_\text{max}}{\beta}$ & Size distribution of communities, typically $\beta{=}1$ \\\hline
			$0 < \mu < 1$ & Mixing parameter: fraction of neighbors of every node $u$ that shall not share a community with $u$\\\hline
		\end{tabular}
	}
	
	\caption{Parameters of overlapping LFR. The typical values follow suggestions by \cite{Lancichinetti2009}.}
	\label{tab:lfr-param}
\end{table} %
\section{\emhh: Deterministic Edges from a Degree Sequence}\label{sec:mat-degree-sequence}
In this section, we introduce an EM-variant of the well-known Havel-Hakimi scheme that takes a positive non-decreasing degree sequence%
\footnote{%
	Within our pipeline, we generate a monotonic degree sequence by first sampling a monotonic uniform sequence online based on the ideas of \cite{DBLP:journals/toms/BentleyS80, DBLP:journals/toms/Vitter87}.
	Applying the inverse sampling technique (carrying over the monotonicity) yields the required distribution.
	Thus, no additional sorting steps are necessary for the inter-community graph.
}
$\degs{=}(d_1, \ldots, d_n)$ and, if possible, outputs a graph $G_\degs$ which realizes these degrees.
A sequence $\degs$ is called \emph{graphical} if a matching simple graph $G_\degs$ exists.
Havel~\cite{Havel1955} and Hakimi~\cite{doi:10.1137/0110037} gave inductive characterizations of graphical sequences which directly lead to a graph generator:
given $\degs$, connect the first node $v_1$ with degree $d_1$ (minimal among all nodes) to $d_1$-many high-degree vertices by emitting edges to nodes $v_{n\, -\, (d_1-1)}, \ldots, v_n$.
Then remove $d_1$ from $\mathcal D$ and decrement the remaining degree of every new neighbor which yields an updated sequence $\degs'$.%
\footnote{This variant is due to \cite{doi:10.1137/0110037}; in \cite{Havel1955}, the node of maximal degree is picked and connected.}
Subsequently, remove zero-entries and sort $\degs'$ while keeping track of the original positions to be able to output the correct node indices.
Finally, recurse until no positive entries remain.
After every iteration, the size of $\degs$ is reduced by at least one resulting in $\Oh(n)$ rounds.

For an implementation, it is non-trivial to keep the sequence ordered after decrementing the neighbors' degrees.
Internal memory solutions typically employ priority queues optimized for integer keys, e.g., bucket-lists~\cite{DBLP:journals/corr/StaudtSM14, DBLP:journals/corr/abs-cs-0502085}.
This approach incurs $\Theta(\sort(n+m))$~I/Os using a na\"ive EM PQ since every edge triggers an update to the pending degree of at least one endpoint.

We propose the Havel-Hakimi variant \emhh{} which emits a stream of edges in lexicographical order.
It can be fed to any single-pass streaming algorithm without a round-trip to disk.
Additionally, \emhh{} may be used in time $\Oh(n)$ to test whether a degree sequence $\degs$ is graphical or to drop problematic edges yielding a graphical sequence (cf. section~\ref{sec:cmes}).
Thus, we consider only internal I/Os and emphasize that storing the output -- if necessary by the application -- requires $\Oh(m)$ time and $\Oh(\scan(m))$ I/Os where $m$ is the number of edges produced.

\clearpageA

\subsection{Data structure} 
Instead of maintaining the degree of every node in $\degs$ individually, \emhh{} compacts nodes with equal degrees into a group, yielding $D_\degs \defrel \big| \{d_i: 1 {\le} i {\le} n\} \big|$ groups.
Since $\degs$ is monotonic, such nodes have consecutive ids and the compaction can be performed in a streaming fashion%
\footnote{While direct sampling of the group's multinomial distribution is not beneficial in LFR, it may be used to omit the compaction phase for other applications.}.

The sequence is then stored as a doubly linked list $L = [g_j]_{1 \le j \le D_\degs}$ where group $g_j=(b_j, n_j, \delta_j)$ assigns degree $\delta_j$ to nodes $v_{b_j}, \ldots, v_{b_j + (n_j - 1)}$.
The algorithm is built around the following invariants holding at the begin of every iteration:
\begin{itemize}
	\item[\invar1] $\delta_j < \delta_{j+1}\ \ \forall 1 \le j < D_\degs$, i.e. the groups represent strictly monotonic degrees
	\item[\invar2] $b_j + n_j = b_{j+1} \ \ \forall 1 \le j < D_\degs$, i.e. there are no gaps in the node ids
\end{itemize}

These invariants allow us to bound the memory footprint in two steps:
first observe that a list $L$ of size $D_\degs$ describes a graph with at least $\sum_{i = 1}^{D_\degs} i / 2$ edges due to \invar1.
Thus, graphs from an arbitrary $L$ filling the whole IM have $\Omega(M^2)$ edges.
Even under pessimistic assumptions this amounts to an edge list of more than \SI{1}{\peta\byte} of size on realistic machines.%
\footnote{%
	A single item of $L$ can be represented by its three values and two pointers, i.e. a total of $5 {\cdot} 8 {=} 40$ bytes per item (assuming 64 bit integers and pointers).
	Just \SI{2}{\giga\byte} of IM suffice for storing \num{5e7} items, which result in at least \num{6.25e14} edges, i.e. storing just two bytes per edge would require more than one Petabyte. 
	Further, standard tricks (e.g., exploiting the redundancy due to \invar2) can be used to reduce the memory footprint of $L$.%
}
Therefore, even in the worst case the whole data structure can be kept in IM for all practical scenarios.
On top of this, a probabilistic argument applies:
While there exist graphs with $D_\degs {=} \Theta(n)$ (cf. Fig.~\ref{fig:degree_worstcase}), Lemma~\ref{lem:pwl_bound_different_degrees} gives a sub-linear bound on $D_\degs$ if $\degs$ is sampled from a powerlaw distribution (refer also to section~\ref{subsec:statesize} for experimental results).

\begin{lemma}\label{lem:pwl_bound_different_degrees}
	Let $\degs$ be a degree sequence with $n$ nodes sampled from $\pld 1n\gamma$.
	Then, the number of unique degrees $D_\degs = \big| \{d_i: 1\le i \le n\} \big|$ is bounded by $\Oh(n ^{1/\gamma})$ with high probability.%
\end{lemma}
\begin{proof}
	Consider random variables $(X_1, \ldots, X_n)$ sampled i.i.d. from  $\pld 1n\gamma$ as an unordered degree sequence.
	Fix an index $1 {\le} j {\le} n$. 
	Due to the powerlaw distribution, $X_j$ is likely to have a small degree.
	Even if all degrees $1, \ldots, {n^{1/\gamma}}$ were realized, their occurrences would be covered by the claim.
	Thus, it suffices to bound the number of realized degrees larger than~${n^{1/\gamma}}$.

	We first show that their total probability mass is small.
	Then we can argue that $D_\degs$ is asymptotically  unaffected by their rare occurrences:
	\begin{align*}
	\mathbb P [X_j {>} {n^{1/\gamma}}]
	&= \sum_{i={n^{1/\gamma}}+1}^{n-1} \mathbb P [X_j {=} i]
	=  \frac{\sum_{i={n^{1/\gamma}}+1}^{n-1} i^{-\gamma}}{\sum_{i=1}^{n-1} i^{-\gamma}}
	\stackrel{(i)}{=} \frac{\sum_{i={n^{1/\gamma}}+1}^{n-1} i^{-\gamma}}{\zeta(\gamma) - \sum_{i=n}^\infty i^{-\gamma}}
	\stackrel{(ii)}{\le} \frac{\int_{n^{1/\gamma}}^{n-1} x^{-\gamma} \diff x}{\zeta(\gamma) - \int_{n}^\infty x^{-\gamma} \diff x}\\
	&= \frac{\frac{1}{1-\gamma}\left[(n{-}1)^{1-\gamma} - {n^{1/\gamma}}/n \right]}{\zeta(\gamma) + \frac{1}{1-\gamma}n^{1-\gamma}}
	= \frac{{n^{1/\gamma}}/n - (n-1)^{1-\gamma}}{(\gamma - 1)\zeta(\gamma) - n^{1-\gamma}}
	= \Oh({n^{1/\gamma}}/n),
	\end{align*}
	\noindent where (i) $\zeta(\gamma) = \sum_{i=1}^\infty i^{-\gamma}$ is the Riemann zeta function which satisfies $\zeta(\gamma) \ge 1$ for all $\gamma{\in}\mathbb R,\ \gamma{\ge}1$.
	In (ii) we exploit the sum's monotonicity to bound it between the two integrals
	$\int_a^{b+1} x^{-\gamma} \diff x \le \sum_{i=a}^b i^{-\gamma} \le \int_{a-1}^{b} x^{-\gamma} \diff x$.

	In order to bound the number of occurrences, define Boolean indicator variables $Y_i$ with $Y_i {=} 1$ iff $X_i {>} {n^{1/\gamma}}$.
	Observe that they model Bernoulli trials $Y_i {\in} B(p)$ with $p {=} \Oh({n^{1/\gamma}}/n)$.
	Thus, the expected number of high degrees is $\mathbb E[\sum_{i=1}^n Y_i] = \sum_{i=1}^n \mathbb P [X_i {>} {n^{1/\gamma}}] = \Oh({n^{1/\gamma}})$.
	Chernoff's inequality gives an exponentially decreasing bound on the tail distribution of the sum which thus holds with high probability.
\end{proof}

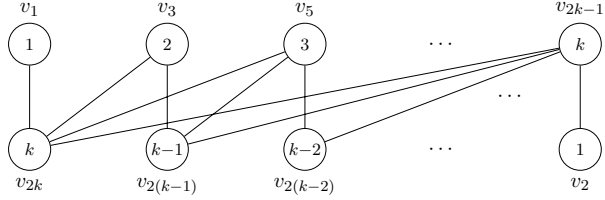
\begin{figure}
	\begin{center}
		\variantScale{
			\scalebox{0.8}{%
{\setlength{\baselineskip}{0.9em}%
\newcommand{\looptype}[1]{\\[0.5em] (\texttt{#1})}
\begin{tikzpicture}[
	node distance=3em and 4.5em,
	vertex/.style={draw, circle, minimum width=2em, minimum height=2em,font=\footnotesize, inner sep=0},
	keylabel/.style={node distance=0em,align=left, xshift=-1.2em},
	message/.style={draw, ->, bend left, shorten >=5pt},
	message label/.style={above,, align=center}
]
	\node[vertex,label=above:{$v_1$}]             (a1) {$1$};
	\node[vertex,label=above:{$v_3$},right=of a1] (a2) {$2$};
	\node[vertex,label=above:{$v_5$},right=of a2] (a3) {$3$};
	\node[                          right=of a3] (ap) {$\cdots$};
	\node[vertex,label=above:{$v_{2k-1}$},right=of ap] (ak) {$k$};
	
	\node[vertex,label=below:{$v_{2k}$},below=of a1] (b1) {$k$};
	\node[vertex,label=below:{$v_{2(k-1)}$},right=of b1] (b2) {$k{-}1$};
	\node[vertex,label=below:{$v_{2(k-2)}$},right=of b2] (b3) {$k{-}2$};
	\node[                          right=of b3] (bp) {$\cdots$};
	\node[vertex,label=below:{$v_2$},right=of bp] (bk) {$1$};
	
	\path[draw] (a1) to (b1)
	(a2) to (b1) (a2) to (b2)
	(a3) to (b1) (a3) to (b2) (a3) to (b3)
	(ak) to (b1) (ak) to (b2) (ak) to (b3) (ak) to node[left, xshift=-2.25em] {$\cdots$} (bk);
\end{tikzpicture}} %
}
		}
	\end{center}
	
	\caption{%
		Materialization of the degree sequence $\degs_k = (1,1,\,2,2,\,\ldots,\,k,k)$ with $D_{\degs_k}  =  k  =  \Theta(n)$
		which maximizes \emhh's memory consumption asymptotically.
		A node's label corresponds to the vertex' degree.
	}
	\label{fig:degree_worstcase}
\end{figure}

Due to Lemma~\ref{lem:pwl_bound_different_degrees}, a graph sampled from a powerlaw distribution with $m = \Oh(M^{2\gamma})$ can be computed in IM with high probability.

\subsection{Algorithm} %
\emhh{} works in $n$ rounds, where every iteration corresponds to a recursion step of the original formulation.
Each time it extracts vertex $v_{b_1}$ with the smallest available id and with minimal degree $\delta_1$.
The extraction is achieved by incrementing the lowest node id ($b_1' \gets b_1{+}1$) of group $g_1$ and decreasing its size ($n_1' \gets n_1{-}1$).
If the group becomes empty ($n_1' = 0$), it is removed from $L$ at the end of the iteration.
We now connect node $v_{b_1}$ to $\delta_1$ nodes from the end of $L$.
Let $g_j$ be the group of smallest index to which $v_{b_1}$ connects to.

\noindent Then there are two cases:

\begin{itemize}
	\item[(C1)]
		If node $v_{b_1}$ connects to all nodes in $g_j$, we directly emit the edges $\big\{[u,x] \,|\, n{-}\delta_1 < x \le n\big\}$ and decrement the degrees of all groups $g_j, \ldots, g_{|L|}$ accordingly.
		Since degree $\delta_{j-1}$ remains unchanged, it may now match the decremented $\delta_j$.
		This violation of \invar1 is resolved by \emph{merging} both groups.
		Due to \invar2, the union of $g_{j-1}$ and $g_j$ contains consecutive ids and it suffices to grow $n_{j-1} \gets n_{j-1}{+}n_j$ and to delete group~$g_j$.

	\item[(C2)]
		If $v_{b_1}$ connects only to a number $a < n_j$ of nodes in group $g_j$, we \emph{split} $g_j$ into two groups $g_j'$ and $g_j''$ of sizes $a$ and $n_j{-}a$ respectively.
		We then connect vertex $u$ to all $a$ nodes in the first fragment $g_j'$ and hence need to decrease its degree.
		Thus, a merge analogous to (C1) may be required (see Fig.~\ref{fig:hh_list}).
		Now, groups $g_{j{+}1}, \ldots, g_{|L|}$ are consumed wholly as in (C1).
\end{itemize}

If the requested degree $\delta_1$ cannot be met (i.e., $\delta_1 > \sum_{k=1}^{|L|}n_k$), the input is not graphical~\cite{doi:10.1137/0110037}.
Since the vast majority of nodes have low degrees, a sufficiently large random powerlaw degree sequence $\degs$ contains at most very few nodes that cannot be materialized as requested.
Therefore, we do not explicitly ensure that the sampled degree sequence is graphical and rather correct the negligible inconsistencies later on by ignoring the unsatisfiable requests.

\subsection{Improving the I/O-complexity}
In the current formulation of \emhh{} we perform constant work per edge which is already optimal.
However, we introduce a simple optimization which improves constant factors and gives I/O-efficient accesses.
It also allows \emhh{} to test whether $\degs$ is graphical in time $\Oh(n)$.
Observe that only groups in the vicinity of $g_j$ can be split or merge; we call these the \emph{active} frontier.
In contrast, the so-called \emph{stable} groups $g_{j{+}1},\ldots, g_{D_\degs}$ keep their relative degree differences as the pending degrees of all their nodes are decremented by one in each iteration.
Further, they will become neighbors to all subsequently extracted nodes until group $g_{j{+}1}$ eventually becomes an active merge candidate.
Thus, we do not have to update the stable degrees in every round, but rather maintain a single global iteration counter $I$ and count how many iterations a group remained stable:
when a group $g_k$ becomes stable in iteration $I_0$, we annotate it with $I_0$ by adding $\delta_k \gets \delta_k{+}I_0$.
If $g_k$ has to be activated again in iteration $I > I_0$, its updated degree follows as $\delta_k \gets \delta_k{-}I$. 
The degree $\delta_k$ remains positive since \invar1 enforces a timely activation.

\goodbreak

\begin{lemma}\label{lem:hphh-io}
	The optimized variant of \emhh{} requires $\Oh(\scan(D_\degs))$ I/Os if $L$ is stored in an external memory list.
\end{lemma}
\begin{proof}
	An external-memory list requires $\Oh(\scan(k))$ I/Os to execute any sequence of $k$ sequential read, insertion, and deletion requests to adjacent positions (i.e. if no seeking is necessary)~\cite{pagh2003basic}.
	We will argue that \emhh{} scans $L$ roughly twice, starting simultaneously from the front and back.
	
	Every iteration starts by extracting a node of minimal degree.
	Doing so corresponds to accessing and eventually deleting the list's first element $g_i$.
	If the list's head block is cached, we only incur an I/O after deleting $\Theta(B)$ head groups, yielding $\Oh(\scan(D_\degs))$ I/Os during the whole execution.
	The same is true for accesses to the back of the list:
	the minimal degree increases monotonically during the algorithm's execution until the extracted node has to be connected to all remaining vertices.
	In a graphical sequence, this implies that only one group remains and we can ignore the simple base case asymptotically.
	Neglecting splitting and merging, the distance between the list's head and the \emph{active} frontier decreases monotonically triggering $\Oh(\scan(D_\degs))$ I/Os.
	
	\textbf{Merging.} As described before, it may be necessary to reactivate \emph{stable} groups, i.e. to reload the group behind the active frontier (towards $L$'s end).
	Thus, we not only keep the block $F$ containing the frontier cached, but also block $G$ behind it.
	It does not incur additional I/O, since we are scanning backwards through $L$ and already read $G$ before $F$.
	The reactivation of \emph{stable} groups hence only incurs an I/O when the whole block $G$ is consumed and deleted.
	Since this does not happen before $\Omega(B)$ merges take place, reactivations may trigger $\Oh(\scan(D_\degs))$ I/Os in total.
	
	\textbf{Splitting.}
	Observe that $L$ at most doubles in size as splitting a group with degree $d$, which has a neighbor of degree $d {\pm} 1$, directly triggers another merge (cf. Fig.~\ref{fig:hh-split-merge}).
	Since a split replaces one group by two adjacent fragments which differ in their degree by exactly one, a second split to one of the fragments does not increase the size of the list.
\end{proof}

\begin{figure}
	\begin{center}
		\variantScale{\scalebox{0.78}{%
{
\def\sx{3em}
\def\sy{2.5em}
\begin{tikzpicture}[
	llabel/.style={anchor=west},
	group/.style={anchor=west, fill=black!20, minimum height=0.9em},
	hcell/.style={minimum width=(0.75*\sx), inner sep=0},
	graphn/.style={circle, draw, minimum width=2em, anchor=west},
	graphe/.style={draw, ->, thick},
]
	\node[llabel] at (16em, -0.5em) {Uncompressed Degree Sequence $\degs$};

	\newcommand{\hcell}[3]{\node[hcell] (i#1-#2) at ($(#2*\sx, -#1*\sy) + (12em, 0)$) {\phantom{$#3$}};}
	\hcell111 \hcell121 \hcell132 \hcell142 \hcell153 \hcell163 
	          \hcell221 \hcell232 \hcell242 \hcell252 \hcell263 
	                    \hcell332 \hcell342 \hcell352 \hcell362 
	                              \hcell441 \hcell451 \hcell462 
	                                        \hcell551 \hcell561 

	\newcommand{\hgroup}[3]{
		\node[group, minimum width=(#3-1)*\sx + (0.75*\sx)] at ($(i#1-#2.west)$) {};
	}
	
	\hgroup121 \hgroup132 \hgroup152
	           \hgroup233 \hgroup261
	\hgroup343
	\hgroup451 \hgroup461
	\hgroup561
	
	\foreach \i in {1, ..., 5} {
		\node[circle, red, draw, inner sep=0.4em] (ext-circle-\i) at ($(i\i-\i)$) {};
		\node[red, at=(ext-circle-\i), yshift=-0.8em] {\footnotesize \phantom{g}extract\phantom{g}};
	}

	\foreach \i/\n/\d in {1/5/2, 2/6/2, 3/4/1, 3/5/1, 4/6/1, 5/6/0} {
		\node[circle, blue, draw, inner sep=0.4em] (ext-circle-\i) at ($(i\i-\n)$) {};
		\node[blue, at=(ext-circle-\i), yshift=-0.8em] {\footnotesize edge-to};
		\node[blue, at=(ext-circle-\i), xshift=1em] {\footnotesize $\mathbf{\d}$};
		\path[draw, blue] (ext-circle-\i.south west) to (ext-circle-\i.north east);
	}

	\foreach \i/\n in {1/5, 3/5} {
		\pgfmathtruncatemacro{\nn}{\n+1}
		\node[ACMGreen, at=(i\i-\n), anchor=west, inner sep=0] (split-\i) at ($0.5*(i\i-\n) + 0.5*(i\i-\nn) + (0, 0.8em)$)	{\footnotesize\,split\phantom{g}};
		\path[ACMGreen, draw, thick] ($0.5*(i\i-\n) + 0.5*(i\i-\nn) + (0, 0.4*\sy)$) to +(0, -0.8*\sy);
	}
	
	\foreach \i/\n in {1/4, 2/5, 4/5} {
		\pgfmathtruncatemacro{\nn}{\n+1}
		\node[ACMGreen, at=(i\i-\n), inner sep=0] (merge-\i) at ($0.5*(i\i-\n) + 0.5*(i\i-\nn) + (0, 0.8em)$)	{\footnotesize \phantom{l}merge\phantom{l}};
		\node[ACMGreen, at=(i\i-\n), inner sep=0] (merge-\i) at ($0.5*(i\i-\n) + 0.5*(i\i-\nn)$)	{\footnotesize $\Leftrightarrow$};
				
	}

	\node[llabel] at (0, -0.5em) {List $[(b_i, n_i, \delta_i)]_i$};
	\node[llabel] at (0, -1*\sy) {$\big[(1, 2, 1),\ (3, 2, 2),\ (5, 2, 3)\big]$};
	\node[llabel] at (0, -2*\sy) {$\big[(2, 1, 1),\ (3, 3, 2),\ (6, 1, 3)\big]$};
	\node[llabel] at (0, -3*\sy) {$\big[(3, 4, 2)\big]$};
	\node[llabel] at (0, -4*\sy) {$\big[(4, 2, 1),\ (6, 1, 2)\big]$};
	\node[llabel] at (0, -5*\sy) {$\big[(5, 2, 1)\big]$};

	\renewcommand{\hcell}[3]{\node[] at ($(#2*\sx, -#1*\sy) + (12em, 0)$) {$#3$};}
	\hcell111 \hcell121 \hcell132 \hcell142 \hcell153 \hcell163 
	\hcell221 \hcell232 \hcell242 \hcell252 \hcell263 
	\hcell332 \hcell342 \hcell352 \hcell362 
	\hcell441 \hcell451 \hcell462 
	\hcell551 \hcell561

\end{tikzpicture}
} %
}}
		\hfill
		\variantScale{\scalebox{0.58}{%
\begin{tikzpicture}[
	grp/.style={draw, minimum width=6em, minimum height=2em},
	grpa/.style={grp},
	grpb/.style={grp, fill=black!30}
]

	\node[grpa, label={group $g_i$}] (a1) {$d{-}1$};
	\node[grpb, label={group $g_j$}] (b1) at ($(a1) + (8em, 0)$) {$d$};

	\node[grpa, label={group $g_i$}] (a2) at ($(a1) - (9em, 6.5em)$) {$d{-}1$};
	\node[grpb, label={group $g_j$}] (b2) at ($(a2) + (8em, 0)$) {$\mathbf{d{-1}}$};

	\node[align=center, anchor=north] at ($0.5*(a2) + 0.5*(b2) - (0, 1.2em)$) 
		{Before $g_i$ can be split, the degrees of\\
		 groups $g_j$ with $j > i$ are decreased};
	
	\node[grpa] (a21) at ($(a2) + (0, -7em)$) {};
	\node[grpb, label={group $g_j$}] (b21) at ($(a21) + (8em, 0)$) {$d{-}1$};
	
	\path[draw] ($(a21) + (0, 1.5em)$) to ($(a21) - (0, 1.5em)$);
	\node at ($(a21) - (1.5em, 0)$) {$d{-}2$};
	\node at ($(a21) + (1.5em, 0)$) {$d{-}1$};	
	
	\path[draw, decorate, decoration={snake}] (a21) to (b21);

	\node[grpa, label={group $g_i$}] (a3) at ($(b1) + (1em, -6.5em)$) {$d{-}1$};
	\node[grpb] (b3) at ($(a3) + ( 8em, 0)$) {};

	\path[draw] ($(b3) + (0, 1.5em)$) to ($(b3) - (0, 1.5em)$);
	\node at ($(b3) - (1.5em, 0)$) {$d{-}1$};
	\node at ($(b3) + (1.5em, 0)$) {$d$};	

	\path[draw, decorate, decoration={snake}] (a3) to (b3);

	\path[draw, dashed] (4em,-3em) to (4em, -14em);

	\node[align=center, anchor=north] at ($0.5*(a1) + 0.5*(b1) + (0, 4em)$) 
	{\textbf{Initial situation}};
	
	\node[align=center, anchor=north] at ($0.5*(a2) + 0.5*(b2) + (0, 4em)$) 
	{\textbf{Splitting at $g_i$ (front)}};
	
	\node[align=center, anchor=north] at ($0.5*(a3) + 0.5*(b3) + (0, 4em)$) 
	{\textbf{Splitting at $g_j$ (back)}};

\end{tikzpicture} %
}}
	\end{center}

	\caption{
		\textbf{Left:} 
		\emhh{} on $\degs{=}(1,1,2,2,3,3)$.
		Values of $L$ and $\degs$ in row $i$ correspond to the state at the beginning of the $i$-th iteration.
		Groups are visualized directly after extraction of the head vertex.
		The number next to an \texttt{edge-to} symbol indicates the new degree.
		After these updates, splitting and merging takes place.
		\textbf{Right:}
		Consider two adjacent groups $g_i$, $g_j$ with degrees $d{-}1$ and $d$.
		A split of $g_i$ (left) or $g_j$ (right) directly triggers a merge,	so the number of groups remains the same.
	}
	\label{fig:hh_list}
	\label{fig:hh-split-merge}	
\end{figure}
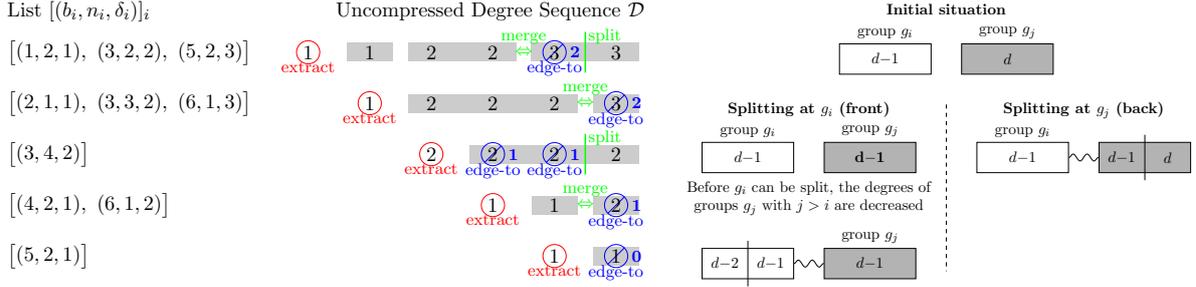

\section{\emes: I/O-efficient Edge Switching}\label{sec:io-efficient-edge-swaps}
\emes{} is a central building block of our pipeline and is used to randomize and rewire existing graphs.
It applies a sequence $S{=}[\sigma_s]_{1 \le s \le k}$ of edge swaps $\sigma_s$ to a simple graph $G{=}(V,E)$, where typically $k = c|E|$ for a constant $c \in [1, 100]$.
The graph is represented by a lexicographically ordered edge list $E_L{=}[e_i]_{1 \le i \le m}$ which only stores the pair $(u, v)$ and omits $(v, u)$ for every ordered edge $[u, v] \in E$.
As illustrated in Fig.~\ref{fig:es-swap-descriptor}, a swap \eswap abd is encoded by a direction bit $d$ and the edge ids $a$ and $b$ (i.e. the position in $E_L$) of the edges supposed to be swapped.
The switched edges are denoted as $e_a^\sigma$ and $e_b^\sigma$ and are given by $(e_a^\sigma,\ e_b^\sigma) \defrel \eswap abd$ as defined in Fig.~\ref{fig:es-swap-descriptor}.

We assume that the swap's constituents are drawn independently and uniformly at random.
Thus, the sequence can contain illegal swaps that would introduce multi-edges or self-loops if executed.
Such illegal swaps are simply skipped.
In order to do so, the following tasks have to be addressed for each \eswap abd:

\newcommand{\estask}[1]{(\texttt{T#1})}

\begin{itemize}
	  \setlength\itemsep{-0.2em}
    \item[\estask1] Gather the nodes incident to edges $e_a$ and $e_b$.
    \item[\estask2] Compute  $e_a^\sigma$ and $e_b^\sigma$ and skip if a self-loop arises.
    \item[\estask3] Verify that the graph remains simple, i.e. skip if edge $e_a^\sigma$ or $e_b^\sigma$ already exist in $E_L$.
    \item[\estask4] Update the graph representation~$E_L$.
\end{itemize}

\begin{figure}
	\begin{center}
		\variantScale{%
			\scalebox{0.8}{%
\begin{tikzpicture}[
	gnode/.style={draw, circle, inner sep=0, minimum width=1.5em},
	inactive edge/.style={draw, black!30}
]
	\def\spx{12em};

	\foreach \i in {0, 1, 2} {
		\node[gnode] (a\i) at (\i*\spx, 0) {$a$};
		\node[gnode] (b\i) at (\i*\spx+3em, 0) {$b$};
		\node[gnode] (c\i) at (\i*\spx,     -3em) {$c$};
		\node[gnode] (d\i) at (\i*\spx+3em, -3em){$d$};

		\foreach \x in {-1, -0.5, 0, 0.5, 1} {
			\path[inactive edge] (a\i) to ++(-1.5em, \x em);
			\path[inactive edge] (c\i) to ++(-1.5em, \x em);
			\path[inactive edge] (b\i) to ++( 1.5em, \x em);
			\path[inactive edge] (d\i) to ++( 1.5em, \x em);
		}

	}
	
	\path[draw, ->, thick] (6em, -1.5em) to (9em, -1.5em);
	\node at (19.5em, -1.5em) {or};
	
	\path[draw] (a0) to node[above] {1} (b0) 
				(c0) to node[above] {2} (d0);
	\path[draw, red] (a1) to node[right] {1} (c1); 
	\path[draw]	(b1) to node[right] {2} (d1);
	\path[draw] (a2) to node[right, pos=0.7] {1} (d2) 
				(b2) to node[left, pos=0.7] {2} (c2);

	\path[draw, dotted, bend right] (a0) to (c0);
	\path[draw, dotted, bend right, red] (a1) to (c1);
	\path[draw, dotted, bend right] (a2) to (c2);

	\node[anchor=south] at (1.5em+0*\spx, -5.5em) {input};
	\node[anchor=south] at (1.5em+1*\spx, -5.5em) {\eswap{1}{2}{\texttt{false}}};
	\node[anchor=south] at (1.5em+2*\spx, -5.5em) {\eswap{1}{2}{\texttt{true}}};

	\node[anchor=south] at (1.5em+1*\spx, -6.5em) {\small illegal (creates multi-edge)};
	
	\node[anchor=west, align=left] at (32em, -2em) {$\eswap abd = 
\begin{cases}
(\{\alpha_1, \beta_1\}, \{\alpha_2, \beta_2\}) & \text{if } d = \texttt{false} \\
(\{\alpha_1, \beta_2\}, \{\alpha_2, \beta_1\}) & \text{if } d = \texttt{true}
\end{cases}$,\\[0.5em]
where $[\alpha_1, \alpha_2] = e_a$ and $[\beta_1, \beta_2] = e_b$
are the edges\\ at positions $a$ and $b$ in the edge list $E_L$.
};
\end{tikzpicture} %
}
		}
	\end{center}
	
	\caption{
		A swap consists of two edge ids and a direction flag.
		The edge ids describe an induced subgraph (left);
		the flag indicates how the incident nodes are shuffled.
		\vspace*{-1em}
	}
	\label{fig:es-swap-descriptor}
\end{figure}
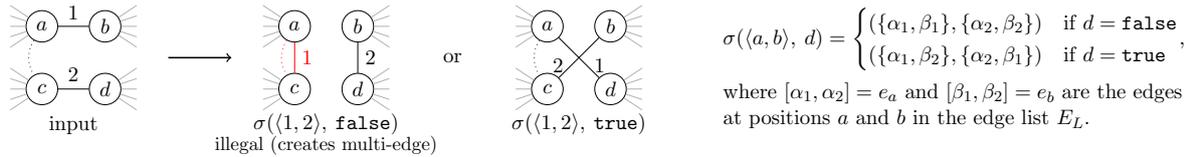

If the whole graph fits in IM, a hash set per node storing all neighbors can be used for adjacency queries and updates in expected constant time (e.g., \vles~\cite{DBLP:journals/corr/abs-cs-0502085}). 
Then, \estask3 and \estask4 can be executed for each swap in expected time $\Oh(1)$.
However, in the EM model this approach incurs $\Omega(1)$ I/Os per swap with high probability for a graph with $m \ge cM$ and any constant $c > 1$.

To improve the situation, we split $S$ into smaller runs of $r = \Theta(|E|)$ swaps which are then batchwise processed.
Note that two swaps within a run can depend on each other.
If an edge is contained in more than one swap, the nodes incident to the edge may change after the first swap has been executed.
We call this a \emph{source edge} dependency.
Since the resulting graph has to remain simple, there further is a dependency between two swaps $\sigma_i$, $\sigma_j$ through \emph{target edges} if executing $\sigma_i$ creates or removes an edge created by $\sigma_j$.
We model both types of dependencies explicitly and forward information between dependent swaps using Time Forward Processing.

As illustrated in Fig.~\ref{fig:es_tfp_overview}, \emes{} executes several phases for each run.
They roughly correspond to the four tasks outlined above.
For simplicity's sake, we first assume that all swaps are independent, i.e. that there are no two swaps that share any source edge id or target edge.
We will then explain how dependencies are handled.

\subsection{\emes{} for Independent Swaps}\label{subsec:estfp-independent}

During the \textsl{request nodes} phase, \emes{} requests for each swap $\eswap ab\cdot$ the endpoints of the two edges at positions $a$ and $b$ in $E_L$.
These requests are then executed in the \textsl{load nodes} phase.
In combination, both implement task \estask1.
Subsequently, the step \textsl{simulate swaps} computes $e_a^\sigma$ and $e_b^\sigma$ which corresponds to \estask2.
In the fourth step, \textsl{load existence}, we check for each of these target edges whether it already exists.
Then, the step \textsl{perform swaps} executes swaps iff the graph remains simple; this corresponds to \estask3.
To implement \estask4, the state of all involved edges is materialized in the \textsl{update graph} phase.

The communication between the different phases is mostly realized via external memory sorters%
\footnote{%
    The term \emph{sorter} refers to a container with two modes of operation:
    in the first phase, items are pushed into the write-only sorter in an arbitrary order by some algorithm.
	After an explicit switch, the filled data structure becomes read-only and the elements are provided as a lexicographically non-decreasing stream which can be rewound at any time.
    While a sorter is functionally equivalent to filling, sorting and reading back an EM vector, the restricted access model reduces constant factors in the implementation's runtime and I/O-complexity~\cite{DBLP:conf/ipps/BeckmannDS09}.
}.
Independent swaps require only the communication shown at the top of Fig.~\ref{fig:es_tfp_overview}.

\begin{figure*}
	\begin{center}\variantScale{\scalebox{0.7}{%
{\setlength{\baselineskip}{0.9em}%
\newcommand{\looptype}[1]{\\[0.5em] (\texttt{#1})}
\begin{tikzpicture}[
	node distance=1.5em and 5.5em,
	task/.style={draw, minimum width=4em, minimum height=4.6em, align=center, anchor=north west, font=\small, fill=white},
	keylabel/.style={node distance=0em,align=left, xshift=-1.2em},
	message/.style={draw, ->, bend left, shorten >=5pt},
	message label/.style={above,font=\footnotesize, align=center},
	sorter/.style={red},
	prioq/.style={blue, dashed},
	stream/.style={olive, thick, dotted},
	legend/.style={inner sep=0, anchor=north west, align=left}
]
	\path[draw, black!50, dotted] (-4, -2.3em) to (50em, -2.3em);

	\node[task] (gather_nodes) {request\\nodes\looptype{swap\_id}};
	\node[task, right=of gather_nodes, xshift=-1em] (dep_chain) {load\\nodes\looptype{edge\_id}};
	\node[task, right=of dep_chain] (sim_swaps) {simulate\\swaps\looptype{swap\_id}};
	\node[task, right=of sim_swaps] (load_exist) {load\\existence\looptype{edge}};
	\node[task, right=of load_exist] (perf_swaps) {perform\\swaps\looptype{swap\_id}};
	\node[task, right=of perf_swaps, xshift=-1em] (update_graph) {update\\graph\looptype{edge}};
	
	\node[anchor=west, align=left, inner sep = 0, black!50] (DependLimitLabel) at (-4, -2.3em) {basic edge switching \\[0.5em] dependency handling};	
	
	\path[message, sorter] ($(gather_nodes.north)+(1em, 0)$) to 
		node[message label] {request nodes\\ incident to \\ edge id (\texttt{EdgeReq})} (dep_chain.north west);
	
	\path[message, sorter] (dep_chain.north) to
		node[message label] {edge state to\\ first swap\\ (\texttt{EdgeMsg})} (sim_swaps.north west);
		
	\path[message, sorter, bend right] (dep_chain.south) to
		node[message label, below] {inform about \\ successor swap \\ (\texttt{IdSucc})} (sim_swaps.south west);

	\path[message, sorter] (sim_swaps.north) to
		node[message label] {edge existence\\request\\ (\texttt{ExistReq})} (load_exist.north west);
	
	\path[message, sorter] (load_exist.north) to 
		node[message label] {edge exist. info\\ to first swap\\ (\texttt{ExistMsg})} (perf_swaps.north west);
		
	\path[message, bend right, sorter] (load_exist.south) to 
		node[message label, below] {inform about \\ successor swap \\ (\texttt{ExistSucc})} (perf_swaps.south west);
	
	\path[message, sorter] (perf_swaps.north) to
		node[message label] {edges after \\ processed swaps \\ (\texttt{EdgeUpdates})} (update_graph.north west);
	
	\path[draw, ->, in=330, out=310, looseness=3, prioq] (sim_swaps.south) to 
		node[message label, below, yshift=-0.5em] {edge state\\ updates to\\ successor} (sim_swaps.east);
	
	\path[draw, ->, in=330, out=310, looseness=3, prioq] (perf_swaps.south) to
		node[message label, below, yshift=-0.5em] {edge state and \\ existence updates\\ to successor} (perf_swaps.east);
	
	\path[draw, ->, shorten >=5pt, stream] ($(gather_nodes.north)-(1em, 0)$) to ++(0, 5.5em) -| ($(update_graph.north)-(1em, 0)$);
	\node[font=\footnotesize, stream] at (20.25em, 6em) {markers for edges that receive updates (\texttt{InvalidEdge})};
	
	\node[draw, circle, fill=black, minimum width=0.35em, inner sep=0] at (dep_chain.north) {};
	\path[draw, ->, shorten >=5pt, sorter] (dep_chain.north) to ++(0, 4.5em) -| ($(perf_swaps.north)-(1em, 0)$);

	\node[draw, circle, fill=black, minimum width=0.35em, inner sep=0] at (dep_chain.south) {};
	\path[draw, ->, shorten >=5pt, sorter] (dep_chain.south) to ++(0, -5em) -| ($(perf_swaps.south)-(1em, 0)$);
	
	\path[draw, ->, stream] ($(dep_chain.north west) - (2em, 1.2em)$) to node[message label] {$E_L$} ($(dep_chain.north west) - (0, 1.2em)$);
	\path[draw, ->, stream] ($(load_exist.north west) - (2em, 1.2em)$) to node[message label] {$E_L$} ($(load_exist.north west) - (0, 1.2em)$);
	\path[draw, ->, stream] ($(update_graph.north west) - (2em, 1.2em)$) to node[message label] {$E_L$} ($(update_graph.north west) - (0, 1.2em)$);
	\path[draw, ->, stream] ($(update_graph.north east) - (0, 1.2em)$) to node[message label] {$E_L$} ($(update_graph.north east) - (-2em, 1.2em)$);

	\node[legend] (legend-sorter) at ($(DependLimitLabel.south west) + (1em, -5em)$) {Sorter};
	\path[draw, ->, sorter] ($(legend-sorter.north) + (-1.5em, 1em)$) to ($(legend-sorter.north) + (1.5em, 1em)$);

	\node[legend, right=3em of legend-sorter] (legend-stream) {Stream};
	\path[draw, ->, stream] ($(legend-stream.north) + (-1.5em, 1em)$) to ($(legend-stream.north) + (1.5em, 1em)$);

	\node[legend, right=3em of legend-stream] (legend-pq) {Priority Queue};
	\path[draw, ->, prioq, in=80, out=100, looseness=3] ($(legend-pq.north) + (-1em, 1em)$) to ($(legend-pq.north) + (1em, 1em)$);
	
\end{tikzpicture}} %
}}\end{center}
	
	\caption{
		Data flow during an \emes{} run.
		Communication between phases is implemented via EM sorters, self-loops use a PQ-based TFP.
		Brackets within a phase represent the type of elements iterated over.
		If multiple input streams are used, they are joined with this key.
	}
	\label{fig:es_tfp_overview}
\end{figure*}
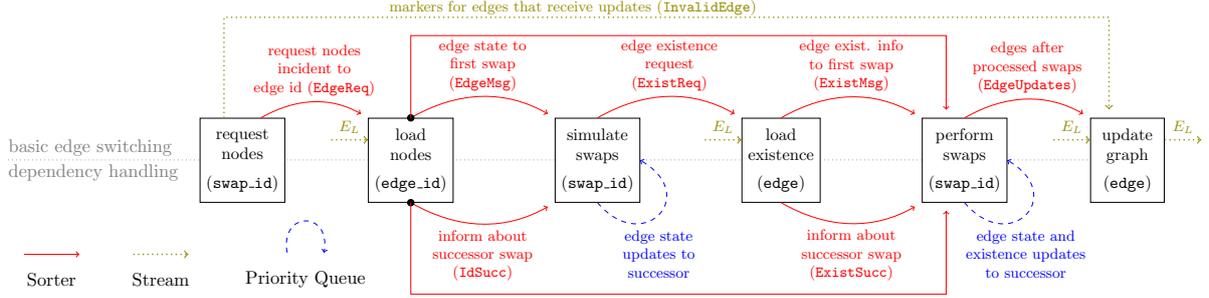

\subsubsection{\textsl{Request nodes} and \textsl{load nodes}}
	The goal of these two phases is to load every referenced edge.
	We iterate over the sequence $S$ of swaps.
	For the $s$-th swap \eswap abd, we push the two messages $\texttt{edge\_req}(a, s, 0)$ and $\texttt{edge\_req}(b, s, 1)$ into the sorter \texttt{EdgeReq}.
	A message's third entry encodes whether the request is issued for the first or second edge of a swap.
	This information only becomes relevant when we allow dependencies.
	\emes{} then scans in parallel through the edge list $E_L$ and the requests \texttt{EdgeReq}, which are now sorted by edge ids.
	If there is a request $\texttt{edge\_req}(i, s, p)$ for an edge $e_i {=} [u, v]$, the edge's node pair is sent to the requesting swap by pushing a message $\texttt{edge\_msg}(s, p, (u, v))$ into the sorter \texttt{EdgeMsg}.\goodbreak
	Additionally, for every edge we push a bit into the sequence \texttt{InvalidEdge}, which is asserted iff an edge received a request.
	These edges are considered invalid and will be deleted when updating the graph.
	Since both phases produce only a constant amount of data per input element, we obtain an I/O complexity of $\Oh(\sort(r) + \scan(m))$.

\subsubsection{\textsl{Simulate swaps} and \textsl{load existence}}
	The two phases gather all information required to decide whether a swap is legal.
	\emes{} scans through the sequence  $S$ of swaps and \texttt{EdgeMsg} in parallel:
	For the $s$-th swap \eswap abd, there are exactly two messages $\texttt{edge\_msg}(s, 0, e_a)$ and $\texttt{edge\_msg}(s, 1, e_b)$ in \texttt{EdgeMsg}.
	This information suffices to compute the switched edges $e_a^\sigma$ and $e_b^\sigma$, but not to test for multi-edges.

	To avoid these, it remains to check whether the switched edges already exist.
	Thus, we push existence requests $\texttt{exist\_req}(e_a^\sigma, s)$ and $\texttt{exist\_req}(e_b^\sigma, s)$ in the sorter \texttt{ExistReq}
	(in contrast to \textsl{request nodes} we use the node pairs rather than edge ids, which are not well defined here).
	Afterwards, a parallel scan through the edge list $E_L$ and \texttt{ExistReq} is performed to answer the requests.
	Only if an edge $e$ requested by swap id $s$ is found, the message $\texttt{exist\_msg}(s, e)$ is pushed into the sorter \texttt{ExistMsg}.
	Both phases hence incur a total of $\Oh(\sort(r) + \scan(m))$ I/Os.

\subsubsection{Perform swaps}
	We rewind the \texttt{EdgeMsg} sorter and jointly scan through the sequence of swaps $S$ and the sorters \texttt{EdgeMsg} and \texttt{ExistMsg}.
	As described in the simulation phase, \emes{} computes the switched edges $e_a^\sigma$ and $e_b^\sigma$ from the original state $e_a$ and $e_b$.
	The swap is marked illegal if a switched edge is a self-loop or if an existence info is received via \texttt{ExistMsg}.
	If $\sigma$ is legal we push the switched edges $e_a^\sigma$ and $e_b^\sigma$ into the sorter \texttt{EdgeUpdates}, otherwise we propagate the unaltered source edges $e_a$ and $e_b$.
	This phase requires $\Oh(\sort(r))$ I/Os.

\subsubsection{Update edge list}
	The new edge list $E'_L$ is obtained by merging the original edge list $E_L$ and the updated edges \texttt{EdgeUpdates}, triggering $\Oh(\scan(m))$ I/Os.
	During this process, we skip all edges in $E_L$ that are flagged invalid in the bit stream \texttt{InvalidEdge}.

\subsection{Inter-Swap Dependencies}\label{subsec:estfp-dependencies}
In contrast to the earlier over-simplification, swaps may share source ids or target edges.
In this case, \emes{} produces  the same result as a sequential processing of~$S$.
Two swaps containing the same source edges are detected during the \textsl{load nodes} phase.
In such a case there arrive multiple requests for the same edge id.
We record these dependencies as an explicit dependency chain (see below for details).

In the simulation phase we do not know yet whether a swap can be executed.
Therefore we need to consider both cases (i.e. a swap has been executed or an existing edge prevented its execution) and dynamically forward all possible edge states using a priority queue.

In the \textsl{load existence} phase, we detect whether several swaps might produce the same outcome; 
in this case both issue an existence request for the same edge during simulation.
Again, an explicit dependency chain is computed.
During the \textsl{perform swaps} phase, \emes{} forwards the source edge states and existence updates to successor swaps using information from both dependency chains.

\subsubsection{Target edge dependencies}
		Consider the case where a swap $\eswapi{s_1}{a}{b}{d}$ changes the state of edges $e_{a}$ and $e_{b}$ to $e_{a}^{\sigma_1}$ and $e_{b}^{\sigma_1}$ respectively.
		Later, a second swap $\sigma_2$ inquires about the existence of either of the four edges which has obviously changed compared to the initial state.
		We extend the simulation phase in order to track such edge modifications.
		Here, we not only push messages $\texttt{exist\_req}(e_{a}^{\sigma_1}, s_1)$ and $\texttt{exist\_req}(e_{b}^{\sigma_1}, s_1)$ into sorter \texttt{EdgeReq}, but also report that the original edges may change.
		This is achieved by using messages $\texttt{exist\_req}(e_{a}, s_1,$ $\texttt{may\_change})$ and $\texttt{exist\_req}(e_{b}, s_1,$ $\texttt{may\_change})$ that are pushed into the same sorter.
		If there are dependencies, multiple messages are received for the same edge $e$ during the \textsl{load existence} phase.
		In this case, only the request of the first swap involved is answered as before.
		Also, every swap $\sigma_{s_1}$ is informed about its direct successor $\sigma_{s_2}$ (if any) by pushing the message $\mathtt{exist\_succ}(s_1, e, s_2)$ into the sorter \texttt{ExistSucc}, yielding the aforementioned dependency chain.
		As an optimization, \texttt{may\_change} requests at the end of a chain are discarded since no recipient exists.
				
		During the \textsl{perform swaps} phase, \emes{} executes the same steps as described earlier.
		The swap may receive a successor for every edge it sent an existence request to, and it informs each successor about the state of the appropriate edge after the swap is processed.

\subsubsection{Source edge dependencies}
		Consider two swaps $\eswapi{s_1}{a_1}{b_1}{d_1}$ and $\eswapi{s_2}{a_2}{b_2}{d_2}$ with $s_1 {<} s_2$ which share a source edge id, i.e. $\{a_1, b_1\} \cap \{a_2, b_2\}$ is non-empty.
		This dependency is detected during the \textsl{load nodes} phase, where requests $\texttt{edge\_req}(e_i, s_1, p_1)$ and $\texttt{edge\_req}(e_i, s_2, p_2)$  arrive for the same edge id $e_i$.
		In this case, we answer only the request of $s_1$ and build a dependency chain as described before using messages $\texttt{id\_succ}(s_1, p_1, s_2, p_2)$ pushed into the sorter \texttt{IdSucc}.

		During the simulation phase, \emes{} cannot decide whether a swap is legal.
		Therefore, $s_1$ sends for every conflicting edge its original state $e$ as well as the updated state $e^{\sigma_1}$ to the $p_2$-th slot of ${s_2}$ using a PQ.
		If a swap receives multiple edge states per slot, it simulates the swap for every possible combination.
		
		During the \textsl{perform swaps} phase, \emes{} operates as described in the independent case:
		it computes the swapped edges and determines whether the swap has to be skipped.
		If a successor exists, the new state is not pushed into the \texttt{EdgeUpdates} sorter but rather forwarded to the successor in a TFP fashion.
		This way, every invalidated edge id receives exactly one update in \texttt{EdgeUpdates} and the merging remains correct.

Due to the second modification, \emes{}'s complexity increases with the number of swaps that target the same edge id.
This number is quite low in case $r = \Oh(m)$.
Let $X_i$ be a random variable expressing the number of swaps that reference edge $e_i$.
Since every swap constitutes two independent Bernoulli trials towards $e_i$, the indicator $X_i$ is binomially distributed with $p=1/m$, yielding an expected chain length of $2r/m$.\clearpageA
Also, for $r=m/2$ swaps, $\max_{1 \le i \le n}( X_i ) = \Oh(\ln(m) / \ln\ln(m))$ holds with high probability based on a balls-into-bins argument~\cite{motwani2010randomized}.
Thus, we can bound the largest number of edge states simulated with high probability by $\Oh(\polylog(m))$, assuming non-overlapping dependency chains.
Further observe that $X_i$ converges towards an independent Poisson distribution for large $m$.
Then the expected state space per edge is $\Oh(1)$.
Experiments suggest that this bound also holds for overlapping dependency chains (cf. section~\ref{subsec:exp-dep-infos}).

In order to keep the dependency chains short, \emes{} splits the sequence of swaps~$S$ into runs of equal size.
Our experimental results  show that a run size of $r = m/8$ is a suitable choice.
For every run, the algorithm executes the six phases as described before.
Each time the graph is updated, the mapping between an edge and its id may change.
The switching probabilities, however, remain unaltered due to the initial assumption of uniformly distributed swaps.
Thus \emes{} triggers $\Oh(k/m \sort(m))$ I/Os in total with high probability.
\section{\emcmes: Sampling of random graphs from prescribed degree sequence}\label{sec:cmes}
In this section, we propose an alternative approach to generate a graph from a prescribed degree sequence.
In contrast to \emhh{} which generates a highly biased but simple graph, we use the Configuration Model to sample a random but non-simple graph.
Thus, the resulting graph may contain self-loops and multi-edges which we then remove to obtain a simple graph.

\subsection{Configuration Model}
Let $\degs = [d_i]_{1 \le i \le n}$ be a degree sequence with $n$ nodes.
The Configuration Model builds a multiset of node ids which can be thought of as \textsl{half-edges} (or stubs).
It produces a total of $d_i$ \textsl{half-edges} labelled $v_i$ for each node $v_i$.
The algorithm then chooses two half-edges uniformly at random and creates an edge according to their labels.
It repeats the last step with the remaining half-edges until all are matched.
In this na\"ive implementation, the procedure requires $\Omega(m)$ I/Os with high probability for $m \ge cM$ and any constant $c > 1$.
It is therefore impractical in the fully external setting.

As illustrated in Fig.~\ref{fig:cm_example} and similar to \cite{kumar2015mathematical}, we instead materialize the multiset as a sequence in which each node appears $d_i$ times.
Subsequently, the sequence is shuffled to obtain a random permutation with $\Oh(\sort (m))$ I/Os by sorting the sequence by a uniform variate drawn for each half-edge\footnote{%
	The random permutation can be obtained with $\Oh(\scan(m))$~I/Os in case $M > \sqrt{mB}(1+o(1)) + \Oh(B)$ \cite{DBLP:journals/ipl/Sanders98} which does not affect the complexity of our total pipeline.
}.
Finally, we scan over the shuffled sequence and match pairs of adjacent half-edges.

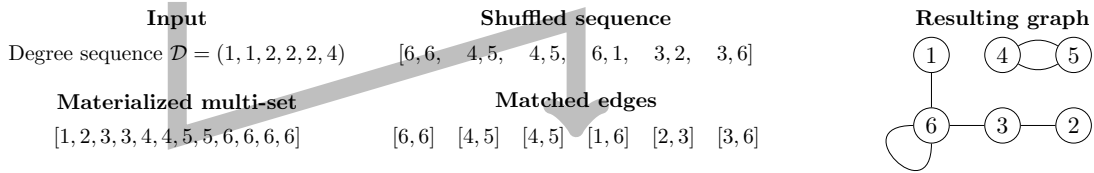
\begin{figure}
	\begin{center}
		\variantScale{\scalebox{0.9}{%
\begin{tikzpicture}[
	node distance=3.5em and 1.5em,
	description/.style={anchor=center, font=\small},
	headline/.style={anchor=center, font=\small},
	gnode/.style={draw, circle, inner sep=0, minimum width=1.5em}
]
	\node[description] (deg_seq) {Degree sequence $\degs = (1,1,2,2,2,4)$};
	\node[description, below of=deg_seq] (materialzed_vec) {$[1, 2, 3, 3, 4, 4, 5, 5, 6, 6, 6, 6]$};
	\node[description, right =of deg_seq] (shuffled_vec) {$[6, 6,\quad 4, 5,\quad 4, 5,\quad 6, 1,\quad 3, 2,\quad 3, 6]$};
	\node[description, below of=shuffled_vec] (matched_edges) {$[6,6]\quad [4,5]\quad [4,5]\quad [1,6]\quad [2,3]\quad [3,6]$};
	
	\node[gnode] at ($(shuffled_vec) + (15em, 0)$) (gn1) {1};
	\node[gnode] at ($(gn1) + (0em, -3em)$) (gn6) {6};
	\node[gnode] at ($(gn6) + (3em, 0)$) (gn3) {3};
	\node[gnode] at ($(gn3) + (3em, 0)$) (gn2) {2};
	\node[gnode] at ($(gn1) + (3em, 0)$) (gn4) {4};
	\node[gnode] at ($(gn4) + (3em, 0)$) (gn5) {5};
	
	\node[headline] at ($(deg_seq) + (0, 1.5em)$) (input) {\textbf{Input}};
	\node[headline] at ($(materialzed_vec) + (0, 1.5em)$) (multiset) {\textbf{Materialized multi-set}};
	\node[headline] at ($(shuffled_vec) + (0, 1.5em)$) (shuff_seq) {\textbf{Shuffled sequence}};
	\node[headline] at ($(shuff_seq) + (18em, 0)$) (graph) {\textbf{Resulting graph}};
	\node[headline] at ($(matched_edges) + (0, 1.5em)$) (edges) {\textbf{Matched edges}};
	
	\draw (gn6) to  [out=270, in=180, looseness=7] (gn6);
	\draw (gn1) to (gn6);
	\draw (gn6) to (gn3);
	\draw (gn3) to (gn2);
	\draw[bend left] (gn5) to (gn4);
	\draw[bend left] (gn4) to (gn5);
	
	\begin{scope}[on background layer]				
		\draw[->, line width=0.8em, gray!50] (input.north) to (materialzed_vec.center) to (shuff_seq.center) to  (matched_edges.center);
	\end{scope}
\end{tikzpicture} %
}}
	\end{center}
	
	\vspace{-2em}
	
	\caption{
		A Configuration Model run on degree sequence $\degs = (1,1,2,2,2,4)$.
	}
	\label{fig:cm_example}
\end{figure}

We now give upper bounds for the number of self-loops and multi-edges introduced by the Configuration Model.
Define  $\moment{\degs} \defrel \sum_v d_v / n$ and $\moment{\degs^2} \defrel \sum_v d_v^2 / n$ to be the mean and the second moment of the sequence $\degs$.

\clearpageB

In \cite{angel2016limit} and \cite{Newman:2010:NI:1809753} the expected number of self-loops and multi-edges has already been studied. 
The results are stated in the following two lemmata.

\begin{lemma}\label{lem:ub-sl}
	Let $\degs$ be a degree sequence with $n$ nodes.
	The expected number of self-loops is given by
	\[ \frac{\moment{\degs^2} - \moment{\degs}}{2(\moment{\degs} - 1/n)} \longrightarrow \frac{\moment{\degs^2} - \moment{\degs}}{2\moment{\degs}} \text{ for } n \to \infty. \]
\end{lemma}

\begin{lemma}\label{lem:ub-me}
	Let $\degs$ be a degree sequence with $n$ nodes.
	The expected number of multi-edges is bounded by
	\[ \frac{1}{2}\left( \frac{(\moment{\degs^2} - \moment{\degs})^2}{(\degs - 1/n)(\degs - 3/n)} \right) \longrightarrow \frac{1}{2}\left( \frac{\moment{\degs^2}-\moment{\degs}}{\moment{\degs}} \right)^2 \text{ for } n \to \infty. \]
\end{lemma}

Let $\degs$ be a degree sequence drawn from the powerlaw distribution $\pld{a}{b}{\gamma}$.
For fixed $\gamma = 2$, we bound the number of illegal edges as a function of $a$ and $b$.
Since each entry in $\degs$ is independently drawn, it suffices to give bounds for the expected value and the second moment of the underlying distribution.

For general $\gamma$, the expected value and second moment are given by $\sum_{i = a}^{b} i^{-\gamma + 1}$ and $\sum_{i = a}^{b} i^{-\gamma + 2}$ respectively. 
Both expressions are bound between the two integrals $\int_{a}^{b+1}x^{-p}\diff x \le \sum_{i=a}^{b}i^{-p} \le \int_{a-1}^{b}x^{-p}\diff x$ where $p = -\gamma + 1$ or $p = -\gamma + 2$ respectively.
In the case of $\gamma = 2$, the second moment is $\moment{\degs^2} = \sum_{i = a}^{b}1 = b {-} a {+} 1$.
Then, by using this identity and the lower bound $\int_{a}^{b+1}x^{-1}\diff x \le \moment{\degs}$, we obtain the two following lemmata:
\begin{lemma}\label{lem:pld-sl}
	Let $\degs$ be drawn from $\pld{a}{b}{2}$.
	The expected number of self-loops is bounded by
	\[ \frac{1}{2}\left( \frac{b - a + 1}{\ln(b + 1) - \ln(a)} \right). \]
\end{lemma}

\begin{lemma}\label{lem:pld-me}
	Let $\degs$ be drawn from $\pld{a}{b}{2}$.
	The expected number of multi-edges is bounded by
	\[ \frac{1}{2}\left( \frac{b - a + 1}{\ln(b + 1) - \ln(a)} \right)^2. \]
\end{lemma}

\subsection{Edge rewiring for non-simple graphs}\label{subsec:emcmes-rewiring}
As a consequence of lemmata~\ref{lem:pld-sl} and \ref{lem:pld-me}, graphs generated using the Configuration Model may contain multi-edges and self-loops.
In order to detect them, we first sort the edge list lexicographically.
Then both types of illegal edges can be detected in a single scan.
For each self-loop we issue a swap with a randomly selected partner edge.
Similarly, for each group of $f > 1$ of parallel edges, we generate $f{-}1$ swaps with random partner edges.
Subsequently, we execute the provisioned swaps using a variant of \emes{} (see below).
The process is repeated until all illegal edges have been removed.
To accelerate the endgame, we double the number of swaps for each remaining illegal edge in every iteration.

Since \emes{} is employed to remove parallel edges based on targeted swaps, it needs to process non-simple graphs.
Analogous to the initial formulation, we forbid swaps that introduce multi-edges even if they would reduce the multiplicity of another edge (cf. \cite{DBLP:journals/corr/Zhao13b}).
Nevertheless, \emes{} requires slight modifications for non-simple graphs. 

Consider the case where the existence of a multi-edge is inquired several times.
Since $E_L$ is sorted, the initial edge multiplicities can be counted while scanning $E_L$ during the \emph{load existence} phase.
In order to correctly process the dependency chain, we have to forward the (possibly updated) multiplicity information to successor swaps.
We annotate the existence tokens $\mathtt{exist\_msg}(s, e, \mult{e})$ with these counters where $\mult{e}$ is the multiplicity of edge $e$.

More precisely, during the \textsl{perform swaps} phase, swap $\sigma_1 = \eswap{a}{b}{d}$ is informed (amongst others) of multiplicities of edges $e_a, e_b, e_a^{\sigma_1}$ and $e_b^{\sigma_1}$ by incoming existence messages.
If $\sigma_1$ is legal, we send requested edges and multiplicities of the swapped state to any successor $\sigma_2$ of $\sigma_1$ provided in \texttt{ExistSucc}. 

The swapped state consists of the same edges where multiplicities for $e_a^{\sigma_1}$ and $e_b^{\sigma_1}$ are incremented and decremented for $e_a$ and $e_b$.
Otherwise, we forward the edges and multiplicities of the unchanged initial state.
As an optimization, edges which have been removed (i.e. have multiplicity zero) are omitted.
\section{\emca: Community Assignment}\label{sec:community-assignment}
For the sake of simplicity, we first restrict the EM community assignment \emca{} to the non-overlapping case, in which every node belongs to exactly one community.
Consider a sequence of community sizes $S {=} [s_\xi]_{1 \le \xi \le C}$ with $n{=}\sum_{\xi=1}^C s_\xi$ and a sequence of intra-community degrees $\degs {=} [d^\text{in}_v]_{1 \le v \le n}$.
Let $S$ and $\degs$ be non-increasing and positive.
The task is to find a random surjective assignment $\chi\colon V {\rightarrow} [C]$ with: 
\begin{enumerate*}
	\item[(R1)] Every community $\xi$ is assigned $s_\xi$ nodes as requested, with
	$s_\xi \defrel \big|\{v \,|\, v\in V \land \chi(v){=}\xi\}\big|$.
	\item[(R2)] Every node $v \in V$ becomes member of a sufficiently large community $\chi(v)$ with $s_{\chi(v)} > d^\text{in}_v$. 
\end{enumerate*}

\subsection{Ignoring constraint on community size (R2)}
Without constraint (R2), the bipartite assignment graph\footnote{%
	Consider a bipartite graph where the partition classes are given by the communities $[C]$ and nodes $[n]$ respectively.
	Each edge then corresponds to an assignment.
} $\chi$ can be sampled in the spirit of the Configuration Model (cf. section~\ref{sec:cmes}):
Draw a permutation $\pi$ of nodes uniformly at random and assign nodes $\{v_{\pi(x_\xi{+}1)}, \ldots, v_{\pi(x_\xi{+}s_\xi)}\}$ to community $\xi$ where $x_\xi \defrel \sum_{i=1}^{\xi-1} s_i$.
To ease later modifications, we prefer an equivalent iterative formulation:
while there exists a yet unassigned node $u$, draw a community $X$ with probability proportional to the number of its remaining free slots (i.e. $\mathbb P[X{=}\xi] \propto s_\xi$).
Assign $u$ to $X$, reduce the community's probability mass by updating $s_X \gets s_X - 1$ and repeat.
By construction, the first scheme is unbiased and the equivalence of both approaches follows as a special case of Lemma~\ref{lem:emca-uniformity}.

We implement the random selection process efficiently based on a binary tree where each community corresponds to a leaf with a weight equal to the number of free slots in the community.
Inner nodes store the total weight of their left subtree. 
In order to draw a community, we sample an integer $Y \in [0, W_C)$ uniformly at random where $W_C \defrel \sum_{\xi=1}^C s_\xi$ is the tree's total weight.
Following the tree according to $Y$ yields the leaf corresponding to community $X$.
An I/O-efficient data structure~\cite{doi:10.1137/1.9781611974317.4} based on lazy evaluation for such dynamic probability distributions  enables a fully external algorithm with $\Oh(n/B \cdot \log_{M/B}(C/B)) = \Oh(\sort(n))$ I/Os.
However, since $C < M$, we can store the tree in IM, allowing a semi-external algorithm which only needs to scan  through $\degs$, triggering $\Oh(\scan(n))$~I/Os.

\subsection{Enforcing constraint on community size (R2)}
To enforce (R2), we exploit the monotonicity of $S$ and $\degs$.
Define $p_v \defrel \max\{\xi\, |\, s_\xi > d^\text{in}_v \}$ as the index of the smallest community node $v$ may be assigned to.
Since $[p_v]_v$ is therefore monotonic itself, it can be computed online with $\Oh(1)$ additional IM and $\Oh(\scan(n))$ I/Os in the fully external setting by scanning through $S$ and $\degs$ in parallel.
In order to restrict the random sampling to the communities $\{1, \ldots, p_v\}$, we reduce the aforementioned random interval to $[0, W_v)$ where the partial sum $W_v \defrel \sum_{\xi=1}^{p_v - 1} s_\xi$ is available while computing $p_v$.
We generalize the notation of uniformity to assignments subject to (R2) as follows:

\clearpageB

\begin{lemma}\label{lem:emca-uniformity}
	Given $S{=}\{s_1, \ldots, s_C\}$ and $\degs$, let $u, v \in V$ be two nodes with the same constraints ($p_u = p_v$) and let $c$ be an arbitrary community.
	Further, let $\chi$ be an assignment generated by \emca.
	Then, $\mathbb P[\chi(u){=}c] = \mathbb P[\chi(v){=}c]$.
\end{lemma}

\begin{proof}
	Without loss of generality, assume that $p_u = p_1$, i.e. $u$ is one of the nodes with the tightest constraints.
	If this is not the case, we just execute \emca{} until we reach a node $u'$ which has the same constraints as $u$ does (i.e. $p_{u'} = p_u$), and apply the Lemma inductively.\clearpageA
	This is legal since \emca{} streams through $\degs$ in a single pass and is oblivious to any future values.
	In case $c > p_1$, neither $u$ nor $v$ can become a member of $c$.
	Therefore, $\mathbb P[\chi(u){=}c] = \mathbb P[\chi(v){=}c] = 0$ and the claim follows trivially.
	
	Now consider the case $c \le p_1$.
	Let $s^{(i)}_c$ be the number of free slots in community $c$ at the beginning of round $i \ge 1$ and $W^{(i)} = \sum_{j=1}^C s^{(i)}_j$ their sum at that time.
	By definition, \emca{} assigns node $u$ to community $c$ with probability $\mathbb P[\chi(u){=}c] = s^{(u)}_c / W^{(u)}$.
	Further, the algorithm has to update the number of free slots.
	Thus, for iteration $i \in [n]$ it holds that
	\begin{align*}
		s^{(i)}_c & = \begin{cases}
		s_c & \text{if $i = 1$} \\
		s^{(i-1)}_c - 1 & \text{if $i{-}1$ was assigned to $c$} \\
		s^{(i-1)}_c & \text{otherwise}
		\end{cases}.
	\end{align*}	

	\noindent The number of free slots is reduced by one in each step:
	\begin{align*}
		W^{(i)} & = \left(\sum_{j=1}^C S_j\right ) - (i-1)
	\end{align*}
	
	\noindent It remains to show $\mathbb P[\chi(u){=}c] = s^{(u)}_c / W^{(u)} = s^{(1)}_c / W^{(1)}$ and the claim follows by transitivity.
 	For $u {=} 1$ it is true by definition.
	Now, consider the induction step for $u {>} 1$:
	{\small
		\begin{align*}
		&\mathbb P[\chi(u){=}c] = s^{(u)}_c / W^{(u)} \\
		&\phantom{abc} = \mathbb P[\chi(u{-}1){=}c] \frac{s^{(u-1)}_c - 1}{W^{(u)}} + \mathbb P[\chi(u{-}1){\ne}c] \frac{s^{(u-1)}_c}{W^{(u)}} 
		= \frac{s^{(u-1)}_c}{W^{(u-1)}}\frac{s^{(u-1)}_c - 1}{W^{(u)}} +
		\left(1 -  \frac{s^{(u-1)}_c}{W^{(u-1)}}\right)\frac{s^{(u-1)}_c}{W^{(u)}} \\
		&\phantom{abc} = \frac{s^{(u-1)}_c \cdot W^{(u-1)} - s^{(u-1)}_c}{W^{(u-1)} \cdot W^{(u)}} = \frac{s^{(u-1)}_c (W^{(u-1)} - 1) }{W^{(u-1)} \cdot (W^{(u-1)}-1)}
		= \frac{s^{(u-1)}_c }{W^{(u-1)}}  \stackrel{\text{Ind. Hyp.}}{=} \frac{s^{(1)}_c}{W^{(1)}} 
		\end{align*}	
	}
\end{proof}

\subsection{Assignment with overlapping communities}
In the overlapping case, the weight of $S$ increases to account for nodes with multiple memberships.
There is further an additional input sequence $[\nu_v]_{1\le v \le n}$ corresponding to the number of memberships a node $v$ shall have, each of which has $d^\text{in}_v$ intra-community neighbors.
We then sample not only one community per node $v$, but $\nu_v$ different ones.

Since the number of memberships $\nu_v \ll M$ is small, a duplication check during the repeated sampling is easy in the semi-external case and does not change the I/O complexity.
However, it is possible that near the end of the execution there are less free communities than memberships requested.
We address this issue by switching to an offline strategy for the last $\Theta(M)$ assignments and keep them in IM.
As $\nu = \Oh(1)$, there are $\Omega(\nu)$ communities with free slots for the last $\Theta(M)$ vertices and a legal assignment exists with high probability.
The offline strategy proceeds as before until it is unable to find $\nu$ different communities for a node.
In that case, it randomly picks earlier assignments until swapping the communities is possible.

In the fully external setting, the I/O complexity grows linearly in the number of samples taken and is thus bounded by $\Oh(\nu \sort(n))$.
However, the community memberships are obtained lazily and out-of-order which may assign a node several times to the same community.
This corresponds to a multi-edge in the bipartite assignment graph.
It can be removed using the rewiring technique detailed in section~\ref{subsec:emcmes-rewiring}.
 \clearpage
\section{Merging and repairing the intra- and inter-community graphs}
\label{sec:lfr-rewiring}
\subsection{Global Edge Rewiring}
\label{sec:edge-rewiring}
The global graph is materialized without taking the community structure into account.
It therefore can contain edges between nodes that share a community.
Those edges have to be removed as they increase the mixing parameter $\mu$.

In accordance with LFR, we use rewiring steps to do so and perform an edge swap for each forbidden edge with a randomly selected partner.
Since it is unlikely that such a random swap introduces another illegal edge (if sufficiently many communities exist), the probabilistic approach effectively removes forbidden edges. 
We apply this idea iteratively and perform multiple rounds until no forbidden edges remain.

The community assignment step outputs a lexicographically ordered sequence $\chi$ of $(v, \xi)$-pairs containing the community $\xi$ for each node $v$.
For nodes that join multiple communities several such pairs exist.
Based on this, we annotate every edge with the communities of both incident vertices by scanning through the edge list twice:
once sorted by source nodes, once by target nodes.
For each forbidden edge, a swap is generated by drawing a random partner edge id and a swap direction.
Subsequently, all swaps are executed using \emes{} which now also emits the set of edges involved.
It suffices to restrict the scan for illegal edges to this set since all edges not contained have to be legal.

\textbf{Complexity. }
Each round needs $\Oh(\sort(m))$ I/Os for selecting the edges and executing the swaps.
The number of rounds is usually small but depends on the community size distribution: 
the smaller the communities, the less likely are edges inside them.

\subsection{Community Edge Rewiring}
\label{sec:community-edge-rewiring}
In the case of overlapping communities, an edge can be generated as part of multiple clusters.
Similarly to section~\ref{subsec:emcmes-rewiring}, we iteratively apply semi-random swaps to remove those parallel edges.
Here, however, the selection of random partners is more involved.
In order not to violate the community size distribution, both edges of a swap have to belong to the same community.
While it is easy to achieve this by considering all communities independently, we need to consider the whole merged graph to detect forbidden edges.

In order to do so, we annotate each edge with its community id and merge all communities together into one graph that possibly contains parallel edges.
During a scan through all sorted edges, we select from each set of parallel edges all but one as candidates for rewiring.
Then, a random partner from the same community is drawn for each of these edges.
For this, we sort all edges and the selected candidates by community.
By counting the edges per community, we can sample random partners, load them in a second scan, randomize their order, and assign them to the candidates of their community.
For the execution, multi-edges need to be considered, i.e. we do not only need to know if an edge exists, but also how often, and update that information accordingly.
Together with all loaded edges we also need to store community ids such that we can uniquely identify them and update the correct information.

While these steps are possible in external memory, we exploit the fact that there are significantly less communities than nodes.
Hence, storing some information per community in internal memory is possible.
We assume further that all candidates can be stored in internal memory.
If there were too many candidates, we would simply consider only some of them in each round.

We can avoid the expensive step of sorting all edges by community for every round using the following observation:
When scanning the edges, we can keep track of how many edges of each community we have seen so far.
We sort the edge ids to be loaded for every community and keep a pointer on the current position in the list for every community.
This allows us to load specific edges of all communities without the need to sort all edges by community.

\textbf{Complexity. }
The fully external rewiring requires $\Oh(\sort(m))$ I/Os for the initial step and each following round.
The semi-external variant triggers only $\Oh(\scan(m))$ I/Os per round.
The number of rounds is usually small and the overall runtime spend on this step is insignificant.
Nevertheless, the described scheme is a Las Vegas algorithm and there exist (unlikely) instances on which it will fail.%
\footnote{
	Consider a node which is a member of two communities in which it is connected to all other nodes.
	If only one of its neighbors also appears in both communities, the multi-edge cannot be rewired.
}
To mitigate this issue, we allow a small fraction of edges (e.g., $10^{-3}$) to be removed if we detect a slow convergence.
To speed up the endgame, we also draw additional swaps uniformly at random from communities which contain a multi-edge.

\section{Implementation}\label{sec:implementation}
We implemented the proposed algorithms in C++ based on the STXXL library~\cite{STXXL}, providing implementations of EM data structures, a parallel EM sorter, and an EM priority queue.
Among others, we applied the following optimizations for \emes{}:
\begin{itemize*}
	\item
	Most message types contain both a swap id and a flag indicating which of the swap's edges is targeted.
	We encoded both of them in a single integer by using all but the least significant bit for the swap id and store the flag in there.
	This significantly reduces the memory volume and yields a simpler comparison operator since the standard integer comparison already ensures the correct lexicographic order.
	
	\item
	Instead of storing and reading the sequence of swaps several times, we exploit the implementation's pipeline structure and directly issue edge id requests for every arriving swap.
	Since this is the only time edge ids are read from a swap, only the remaining direction flag is stored in an efficient EM vector, which uses one bit per flag and supports I/O-efficient writing and reading.
	Both steps can be overlapped with an ongoing \emes{}~run.	

	\item
	Instead of storing each edge in the sorted external edge list as a pair of nodes, we only store each source node once and then list all targets of that vertex.
	This still supports sequential scan and merge operations which are the only operations we need.
	This almost halves the I/O volume of scanning or updating the edge list.

	\item
	During the execution of several runs we can delay the updating of the edge list and combine it with the \emph{load nodes} phase of the next run.
    This reduces the number of scans per additional run from three to two.

	\item
	We use asynchronous stream adapters for tasks such as streaming from sorters or the generation of random numbers.
	These adapters run in parallel in the background to preprocess and buffer portions of the stream in advance and hand them over to the main thread.
\end{itemize*}

Besides parallel sorting and asynchronous pipeline stage, the current \emlfr{} implementation facilitates parallelism only during the generation and randomization of intra-community graphs which can be computed pleasingly parallel.

\newcommand{\settingsLin}{\textbf{(lin)}}
\newcommand{\settingsConst}{\textbf{(const)}}

\section{Experimental Results}\label{sec:exp}
The number of repetitions per data point (with different random seeds) is denoted with~$S$.
Errorbars correspond to the unbiased estimation of the standard deviation.
For LFR we perform experiments based on two different scenarios:
\begin{itemize}
	\inA{\item \settingsLin{} ---}\inB{\item[\settingsLin]}
		In one setting, the maximal degrees and community sizes scale linearly as a function of $n$.
		For a $n$ and $\nu$ the parameters are chosen as: $\mu \in \{0.2, 0.4, 0.6\}$, 
		$d_\text{min}{=}10\nu$, $d_\text{max}{=}n\nu/20$, $\gamma{=}2$,
		$s_\text{min}{=}20$, $s_\text{max}{=}n/10$, $\beta{=}1$, $O{=}n$.

	\inA{\item \settingsConst{} ---}\inB{\item[\settingsConst]}
		In the second setting, we keep the community sizes and the degrees constant and consider only non-overlapping communities.
		The parameters are chosen as:
		$d_\text{min}{=}50$, $d_\text{max}{=}\num{10000}$, $\gamma{=}2$,
		$s_\text{min}{=}50$, $s_\text{max}{=}\num{12000}$, $\beta{=}1$, $O{=}n$.
\end{itemize}

Real-world networks have been shown to have increasing average degrees as they become larger~\cite{lkf-gotdl-05}.
Increasing the maximum degree as in our first setting \settingsLin{} increases the average degree.
Having a maximum community size of $n/10$ means, however, that a significant proportion of the nodes belongs to such huge communities which are not very tightly knit due to the large number of nodes of low degree.
While a more limited growth is probably more realistic, the exact parameters depend on the network model.

Our second parameter set \settingsConst{} shows an example of much smaller maximum degrees and community sizes.
We chose the parameters such that they approximate the degree distribution of the Facebook network in May 2011 when it consisted of 721 million active users as reported in~\cite{DBLP:journals/corr/abs-1111-4503}.
Note however that strict powerlaw models are unable to accurately mimic Facebook's degree distribution~\cite{DBLP:journals/corr/abs-1111-4503}.
Further, they show that the degree distribution of the U.S. users (removing connections to non-U.S. users) is very similar to the one of the Facebook users of the whole world, supporting our use of just one parameter set for different graph sizes.

The actual minimum degree of the Facebook network is 1, but the smaller degrees are significantly less prevalent than a power law degree sequence would suggest, which is why we chose a larger value of \num{50}.
Our maximum degree of \num{10000} is larger than the one reported for Facebook (\num{5000}), but the latter was also an arbitrarily enforced limit by Facebook.
The expected average degree of this degree sequence is 264, which is slightly higher than the reported 190 (world) or 214 (U.S. only).
Our parameters are chosen such that the median degree is approximately 99, which matches the worldwide Facebook network.
Similar to the first parameter set, we chose the maximum community size slightly larger than the maximum degree of \num{12000} nodes.

\subsection{\emhh's state size}\label{subsec:statesize}
\begin{figure}
	\centering
	\scalebox{0.6}{%
\begingroup
  \makeatletter
  \providecommand\color[2][]{%
    \GenericError{(gnuplot) \space\space\space\@spaces}{%
      Package color not loaded in conjunction with
      terminal option `colourtext'%
    }{See the gnuplot documentation for explanation.%
    }{Either use 'blacktext' in gnuplot or load the package
      color.sty in LaTeX.}%
    \renewcommand\color[2][]{}%
  }%
  \providecommand\includegraphics[2][]{%
    \GenericError{(gnuplot) \space\space\space\@spaces}{%
      Package graphicx or graphics not loaded%
    }{See the gnuplot documentation for explanation.%
    }{The gnuplot epslatex terminal needs graphicx.sty or graphics.sty.}%
    \renewcommand\includegraphics[2][]{}%
  }%
  \providecommand\rotatebox[2]{#2}%
  \@ifundefined{ifGPcolor}{%
    \newif\ifGPcolor
    \GPcolortrue
  }{}%
  \@ifundefined{ifGPblacktext}{%
    \newif\ifGPblacktext
    \GPblacktextfalse
  }{}%
  \let\gplgaddtomacro\g@addto@macro
  \gdef\gplbacktext{}%
  \gdef\gplfronttext{}%
  \makeatother
  \ifGPblacktext
    \def\colorrgb#1{}%
    \def\colorgray#1{}%
  \else
    \ifGPcolor
      \def\colorrgb#1{\color[rgb]{#1}}%
      \def\colorgray#1{\color[gray]{#1}}%
      \expandafter\def\csname LTw\endcsname{\color{white}}%
      \expandafter\def\csname LTb\endcsname{\color{black}}%
      \expandafter\def\csname LTa\endcsname{\color{black}}%
      \expandafter\def\csname LT0\endcsname{\color[rgb]{1,0,0}}%
      \expandafter\def\csname LT1\endcsname{\color[rgb]{0,1,0}}%
      \expandafter\def\csname LT2\endcsname{\color[rgb]{0,0,1}}%
      \expandafter\def\csname LT3\endcsname{\color[rgb]{1,0,1}}%
      \expandafter\def\csname LT4\endcsname{\color[rgb]{0,1,1}}%
      \expandafter\def\csname LT5\endcsname{\color[rgb]{1,1,0}}%
      \expandafter\def\csname LT6\endcsname{\color[rgb]{0,0,0}}%
      \expandafter\def\csname LT7\endcsname{\color[rgb]{1,0.3,0}}%
      \expandafter\def\csname LT8\endcsname{\color[rgb]{0.5,0.5,0.5}}%
    \else
      \def\colorrgb#1{\color{black}}%
      \def\colorgray#1{\color[gray]{#1}}%
      \expandafter\def\csname LTw\endcsname{\color{white}}%
      \expandafter\def\csname LTb\endcsname{\color{black}}%
      \expandafter\def\csname LTa\endcsname{\color{black}}%
      \expandafter\def\csname LT0\endcsname{\color{black}}%
      \expandafter\def\csname LT1\endcsname{\color{black}}%
      \expandafter\def\csname LT2\endcsname{\color{black}}%
      \expandafter\def\csname LT3\endcsname{\color{black}}%
      \expandafter\def\csname LT4\endcsname{\color{black}}%
      \expandafter\def\csname LT5\endcsname{\color{black}}%
      \expandafter\def\csname LT6\endcsname{\color{black}}%
      \expandafter\def\csname LT7\endcsname{\color{black}}%
      \expandafter\def\csname LT8\endcsname{\color{black}}%
    \fi
  \fi
    \setlength{\unitlength}{0.0500bp}%
    \ifx\gptboxheight\undefined%
      \newlength{\gptboxheight}%
      \newlength{\gptboxwidth}%
      \newsavebox{\gptboxtext}%
    \fi%
    \setlength{\fboxrule}{0.5pt}%
    \setlength{\fboxsep}{1pt}%
\begin{picture}(6802.00,3684.00)%
    \gplgaddtomacro\gplbacktext{%
      \csname LTb\endcsname%
      \put(814,756){\makebox(0,0)[r]{\strut{}$10^{1}$}}%
      \csname LTb\endcsname%
      \put(814,1288){\makebox(0,0)[r]{\strut{}$10^{2}$}}%
      \csname LTb\endcsname%
      \put(814,1821){\makebox(0,0)[r]{\strut{}$10^{3}$}}%
      \csname LTb\endcsname%
      \put(814,2354){\makebox(0,0)[r]{\strut{}$10^{4}$}}%
      \csname LTb\endcsname%
      \put(814,2886){\makebox(0,0)[r]{\strut{}$10^{5}$}}%
      \csname LTb\endcsname%
      \put(814,3419){\makebox(0,0)[r]{\strut{}$10^{6}$}}%
      \csname LTb\endcsname%
      \put(1178,484){\makebox(0,0){\strut{}$10^{3}$}}%
      \csname LTb\endcsname%
      \put(2223,484){\makebox(0,0){\strut{}$10^{4}$}}%
      \csname LTb\endcsname%
      \put(3269,484){\makebox(0,0){\strut{}$10^{5}$}}%
      \csname LTb\endcsname%
      \put(4314,484){\makebox(0,0){\strut{}$10^{6}$}}%
      \csname LTb\endcsname%
      \put(5360,484){\makebox(0,0){\strut{}$10^{7}$}}%
      \csname LTb\endcsname%
      \put(6405,484){\makebox(0,0){\strut{}$10^{8}$}}%
    }%
    \gplgaddtomacro\gplfronttext{%
      \csname LTb\endcsname%
      \put(176,2061){\rotatebox{-270}{\makebox(0,0){\strut{}Number of unique elements}}}%
      \put(3675,154){\makebox(0,0){\strut{}Number $n$ of samples}}%
      \csname LTb\endcsname%
      \put(2662,3246){\makebox(0,0)[r]{\strut{}$\gamma=1\colon 0.379 \cdot x^{1/1}$}}%
      \csname LTb\endcsname%
      \put(2662,3026){\makebox(0,0)[r]{\strut{}$\gamma=2\colon 1.380 \cdot x^{1/2}$}}%
      \csname LTb\endcsname%
      \put(2662,2806){\makebox(0,0)[r]{\strut{}$\gamma=3\colon 1.270 \cdot x^{1/3}$}}%
    }%
    \gplbacktext
    \put(0,0){\includegraphics{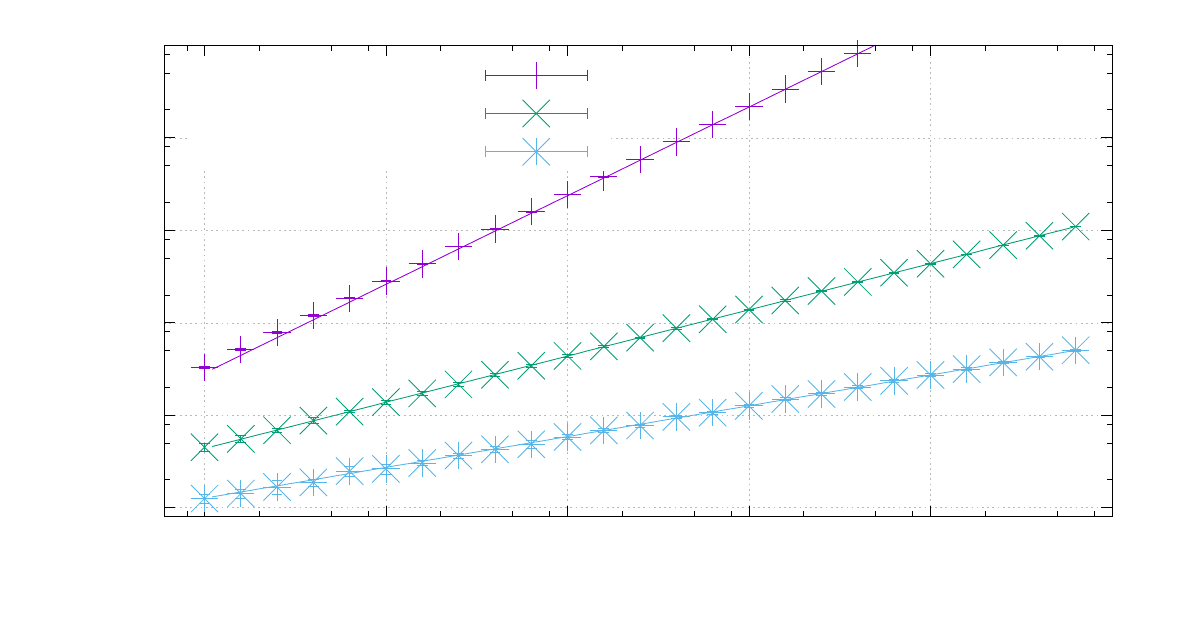}}%
    \gplfronttext
  \end{picture}%
\endgroup
}
	\hfill
	\scalebox{0.6}{%
\begingroup
  \makeatletter
  \providecommand\color[2][]{%
    \GenericError{(gnuplot) \space\space\space\@spaces}{%
      Package color not loaded in conjunction with
      terminal option `colourtext'%
    }{See the gnuplot documentation for explanation.%
    }{Either use 'blacktext' in gnuplot or load the package
      color.sty in LaTeX.}%
    \renewcommand\color[2][]{}%
  }%
  \providecommand\includegraphics[2][]{%
    \GenericError{(gnuplot) \space\space\space\@spaces}{%
      Package graphicx or graphics not loaded%
    }{See the gnuplot documentation for explanation.%
    }{The gnuplot epslatex terminal needs graphicx.sty or graphics.sty.}%
    \renewcommand\includegraphics[2][]{}%
  }%
  \providecommand\rotatebox[2]{#2}%
  \@ifundefined{ifGPcolor}{%
    \newif\ifGPcolor
    \GPcolortrue
  }{}%
  \@ifundefined{ifGPblacktext}{%
    \newif\ifGPblacktext
    \GPblacktextfalse
  }{}%
  \let\gplgaddtomacro\g@addto@macro
  \gdef\gplbacktext{}%
  \gdef\gplfronttext{}%
  \makeatother
  \ifGPblacktext
    \def\colorrgb#1{}%
    \def\colorgray#1{}%
  \else
    \ifGPcolor
      \def\colorrgb#1{\color[rgb]{#1}}%
      \def\colorgray#1{\color[gray]{#1}}%
      \expandafter\def\csname LTw\endcsname{\color{white}}%
      \expandafter\def\csname LTb\endcsname{\color{black}}%
      \expandafter\def\csname LTa\endcsname{\color{black}}%
      \expandafter\def\csname LT0\endcsname{\color[rgb]{1,0,0}}%
      \expandafter\def\csname LT1\endcsname{\color[rgb]{0,1,0}}%
      \expandafter\def\csname LT2\endcsname{\color[rgb]{0,0,1}}%
      \expandafter\def\csname LT3\endcsname{\color[rgb]{1,0,1}}%
      \expandafter\def\csname LT4\endcsname{\color[rgb]{0,1,1}}%
      \expandafter\def\csname LT5\endcsname{\color[rgb]{1,1,0}}%
      \expandafter\def\csname LT6\endcsname{\color[rgb]{0,0,0}}%
      \expandafter\def\csname LT7\endcsname{\color[rgb]{1,0.3,0}}%
      \expandafter\def\csname LT8\endcsname{\color[rgb]{0.5,0.5,0.5}}%
    \else
      \def\colorrgb#1{\color{black}}%
      \def\colorgray#1{\color[gray]{#1}}%
      \expandafter\def\csname LTw\endcsname{\color{white}}%
      \expandafter\def\csname LTb\endcsname{\color{black}}%
      \expandafter\def\csname LTa\endcsname{\color{black}}%
      \expandafter\def\csname LT0\endcsname{\color{black}}%
      \expandafter\def\csname LT1\endcsname{\color{black}}%
      \expandafter\def\csname LT2\endcsname{\color{black}}%
      \expandafter\def\csname LT3\endcsname{\color{black}}%
      \expandafter\def\csname LT4\endcsname{\color{black}}%
      \expandafter\def\csname LT5\endcsname{\color{black}}%
      \expandafter\def\csname LT6\endcsname{\color{black}}%
      \expandafter\def\csname LT7\endcsname{\color{black}}%
      \expandafter\def\csname LT8\endcsname{\color{black}}%
    \fi
  \fi
    \setlength{\unitlength}{0.0500bp}%
    \ifx\gptboxheight\undefined%
      \newlength{\gptboxheight}%
      \newlength{\gptboxwidth}%
      \newsavebox{\gptboxtext}%
    \fi%
    \setlength{\fboxrule}{0.5pt}%
    \setlength{\fboxsep}{1pt}%
\begin{picture}(6802.00,3684.00)%
    \gplgaddtomacro\gplbacktext{%
      \csname LTb\endcsname%
      \put(946,834){\makebox(0,0)[r]{\strut{}$10^{-6}$}}%
      \csname LTb\endcsname%
      \put(946,1265){\makebox(0,0)[r]{\strut{}$10^{-5}$}}%
      \csname LTb\endcsname%
      \put(946,1695){\makebox(0,0)[r]{\strut{}$10^{-4}$}}%
      \csname LTb\endcsname%
      \put(946,2126){\makebox(0,0)[r]{\strut{}$10^{-3}$}}%
      \csname LTb\endcsname%
      \put(946,2557){\makebox(0,0)[r]{\strut{}$10^{-2}$}}%
      \csname LTb\endcsname%
      \put(946,2988){\makebox(0,0)[r]{\strut{}$10^{-1}$}}%
      \csname LTb\endcsname%
      \put(946,3419){\makebox(0,0)[r]{\strut{}$10^{0}$}}%
      \csname LTb\endcsname%
      \put(3310,484){\makebox(0,0){\strut{}$10$}}%
      \csname LTb\endcsname%
      \put(6405,484){\makebox(0,0){\strut{}$100$}}%
    }%
    \gplgaddtomacro\gplfronttext{%
      \csname LTb\endcsname%
      \put(176,2061){\rotatebox{-270}{\makebox(0,0){\strut{}Relative frequency}}}%
      \put(3741,154){\makebox(0,0){\strut{}Number of edge configuration received by swap}}%
      \csname LTb\endcsname%
      \put(5418,3246){\makebox(0,0)[r]{\strut{}Run size $0.05 n$}}%
      \csname LTb\endcsname%
      \put(5418,3026){\makebox(0,0)[r]{\strut{}Run size $0.125 n$}}%
      \csname LTb\endcsname%
      \put(5418,2806){\makebox(0,0)[r]{\strut{}Run size $0.5 n$}}%
    }%
    \gplbacktext
    \put(0,0){\includegraphics{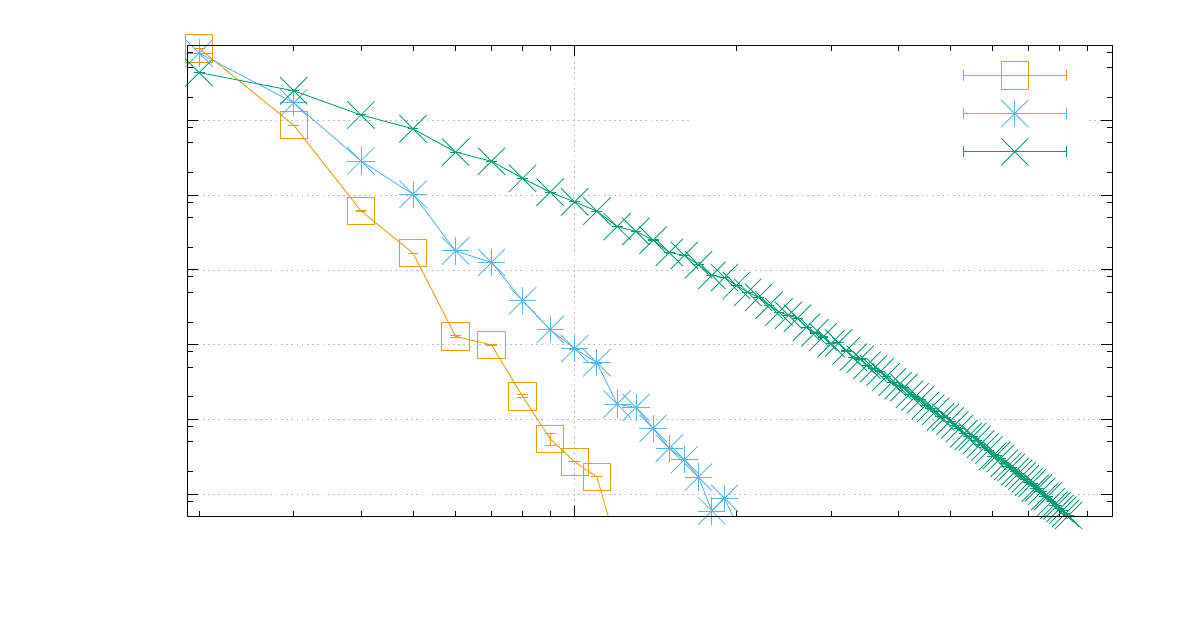}}%
    \gplfronttext
  \end{picture}%
\endgroup
}%
	
	\caption{
		\textbf{Left:} Number of distinct elements in $n$ samples (i.e. node degrees in a degree sequence) taken from $\pld 1n\gamma$.\inB{\\}
		\textbf{Right:} Overhead induced by tracing inter-swap dependencies.  Fraction of swaps as function of the number of edge configurations they receive during the simulation phase.
	}
	\label{fig:powerlaw_pop_count}
	\label{fig:inter_swap_dependencies}	
\end{figure}
In Lemma~\ref{lem:pwl_bound_different_degrees}, we bound the internal memory consumption of \emhh{} by showing that a sequence of $n$ numbers randomly sampled from $\pld 1n\gamma$ contains only $\Oh(n ^{1/\gamma})$ distinct values with high probability.\goodbreak
In order to support Lemma~\ref{lem:pwl_bound_different_degrees} and to estimate the hidden constants, samples of varying size between $10^3$ and $10^8$ are taken from distributions with exponents $\gamma \in \{1,2,3\}$.
Each time, the number of unique elements is computed and averaged over $S={9}$ runs with identical configurations but different random seeds.
The results illustrated in Fig.~\ref{fig:powerlaw_pop_count} support the predictions with small constants.
For the commonly used exponent~$2$, we find $1.38 \sqrt n$ distinct elements in a sequence of length $n$.

\subsection{Inter-Swap Dependencies}\label{subsec:exp-dep-infos}
Whenever multiple swaps target the same edge, \emes{} simulates all possible states to be able to retrieve conflicting edges.
We argued that the number of dependencies (and thus the state size) remains manageable if the sequence of swaps is split into sufficiently short runs.
We found that for $m$ edges and $k$ swaps, $8 k/m$ runs minimize the runtime for large instances of \settingsLin.
As indicated in Fig.~\ref{fig:inter_swap_dependencies}, in this setting \SI{78.7}{\percent} of swaps do not receive additional edge configurations during the simulation phase and less than \SI{0.4}{\percent} have to consider more than four additional states.
Similarly, \SI{78.6}{\percent} of existence requests remain without dependencies.

\subsection{Performance benchmarks}
\noindent Runtime measurements were conducted on the following~systems:
\begin{itemize}
	\inA{\item \textbf{(SysA)} ---}\inB{\item[\textbf{(SysA)}]}
	\textit{inexpensive compute server}:\\
	Intel Xeon E5-2630~v3 (8~cores/16 threads, 2.40GHz), \SI{64}{\giga\byte}~RAM, 3$\times$~Samsung~850~PRO SATA SSD (1~TB).
	
	\inA{\item \textbf{(SysB)} ---}\inB{\item[\textbf{(SysB)}]}
	\textit{commodity hardware}:\\
	Intel Core i7 CPU 970 (6~cores/12 threads, 3.2GHz), \SI{12}{\giga\byte} RAM, 1$\times$~Samsung~850~PRO SATA SSD (1~TB).
\end{itemize}

Since edge switching scales linearly in the number of swaps (in case of \emes{} in the number of runs),
some of the measurements beyond \SI{3}{\hour} runtime are extrapolated from the progress until then.
We verified that errors stay within the indicated margin using reference measurements without extrapolation.

\begin{figure}
	\centering
    \scalebox{0.6}{%
\begingroup
  \makeatletter
  \providecommand\color[2][]{%
    \GenericError{(gnuplot) \space\space\space\@spaces}{%
      Package color not loaded in conjunction with
      terminal option `colourtext'%
    }{See the gnuplot documentation for explanation.%
    }{Either use 'blacktext' in gnuplot or load the package
      color.sty in LaTeX.}%
    \renewcommand\color[2][]{}%
  }%
  \providecommand\includegraphics[2][]{%
    \GenericError{(gnuplot) \space\space\space\@spaces}{%
      Package graphicx or graphics not loaded%
    }{See the gnuplot documentation for explanation.%
    }{The gnuplot epslatex terminal needs graphicx.sty or graphics.sty.}%
    \renewcommand\includegraphics[2][]{}%
  }%
  \providecommand\rotatebox[2]{#2}%
  \@ifundefined{ifGPcolor}{%
    \newif\ifGPcolor
    \GPcolortrue
  }{}%
  \@ifundefined{ifGPblacktext}{%
    \newif\ifGPblacktext
    \GPblacktextfalse
  }{}%
  \let\gplgaddtomacro\g@addto@macro
  \gdef\gplbacktext{}%
  \gdef\gplfronttext{}%
  \makeatother
  \ifGPblacktext
    \def\colorrgb#1{}%
    \def\colorgray#1{}%
  \else
    \ifGPcolor
      \def\colorrgb#1{\color[rgb]{#1}}%
      \def\colorgray#1{\color[gray]{#1}}%
      \expandafter\def\csname LTw\endcsname{\color{white}}%
      \expandafter\def\csname LTb\endcsname{\color{black}}%
      \expandafter\def\csname LTa\endcsname{\color{black}}%
      \expandafter\def\csname LT0\endcsname{\color[rgb]{1,0,0}}%
      \expandafter\def\csname LT1\endcsname{\color[rgb]{0,1,0}}%
      \expandafter\def\csname LT2\endcsname{\color[rgb]{0,0,1}}%
      \expandafter\def\csname LT3\endcsname{\color[rgb]{1,0,1}}%
      \expandafter\def\csname LT4\endcsname{\color[rgb]{0,1,1}}%
      \expandafter\def\csname LT5\endcsname{\color[rgb]{1,1,0}}%
      \expandafter\def\csname LT6\endcsname{\color[rgb]{0,0,0}}%
      \expandafter\def\csname LT7\endcsname{\color[rgb]{1,0.3,0}}%
      \expandafter\def\csname LT8\endcsname{\color[rgb]{0.5,0.5,0.5}}%
    \else
      \def\colorrgb#1{\color{black}}%
      \def\colorgray#1{\color[gray]{#1}}%
      \expandafter\def\csname LTw\endcsname{\color{white}}%
      \expandafter\def\csname LTb\endcsname{\color{black}}%
      \expandafter\def\csname LTa\endcsname{\color{black}}%
      \expandafter\def\csname LT0\endcsname{\color{black}}%
      \expandafter\def\csname LT1\endcsname{\color{black}}%
      \expandafter\def\csname LT2\endcsname{\color{black}}%
      \expandafter\def\csname LT3\endcsname{\color{black}}%
      \expandafter\def\csname LT4\endcsname{\color{black}}%
      \expandafter\def\csname LT5\endcsname{\color{black}}%
      \expandafter\def\csname LT6\endcsname{\color{black}}%
      \expandafter\def\csname LT7\endcsname{\color{black}}%
      \expandafter\def\csname LT8\endcsname{\color{black}}%
    \fi
  \fi
    \setlength{\unitlength}{0.0500bp}%
    \ifx\gptboxheight\undefined%
      \newlength{\gptboxheight}%
      \newlength{\gptboxwidth}%
      \newsavebox{\gptboxtext}%
    \fi%
    \setlength{\fboxrule}{0.5pt}%
    \setlength{\fboxsep}{1pt}%
\begin{picture}(6802.00,3684.00)%
    \gplgaddtomacro\gplbacktext{%
      \csname LTb\endcsname%
      \put(814,704){\makebox(0,0)[r]{\strut{}$10^{0}$}}%
      \csname LTb\endcsname%
      \put(814,1247){\makebox(0,0)[r]{\strut{}$10^{1}$}}%
      \csname LTb\endcsname%
      \put(814,1790){\makebox(0,0)[r]{\strut{}$10^{2}$}}%
      \csname LTb\endcsname%
      \put(814,2333){\makebox(0,0)[r]{\strut{}$10^{3}$}}%
      \csname LTb\endcsname%
      \put(814,2876){\makebox(0,0)[r]{\strut{}$10^{4}$}}%
      \csname LTb\endcsname%
      \put(814,3419){\makebox(0,0)[r]{\strut{}$10^{5}$}}%
      \csname LTb\endcsname%
      \put(946,484){\makebox(0,0){\strut{}$10^{4}$}}%
      \csname LTb\endcsname%
      \put(2253,484){\makebox(0,0){\strut{}$10^{5}$}}%
      \csname LTb\endcsname%
      \put(3560,484){\makebox(0,0){\strut{}$10^{6}$}}%
      \csname LTb\endcsname%
      \put(4868,484){\makebox(0,0){\strut{}$10^{7}$}}%
      \csname LTb\endcsname%
      \put(6175,484){\makebox(0,0){\strut{}$10^{8}$}}%
    }%
    \gplgaddtomacro\gplfronttext{%
      \csname LTb\endcsname%
      \put(176,2061){\rotatebox{-270}{\makebox(0,0){\strut{}Runtime [s]}}}%
      \put(3675,154){\makebox(0,0){\strut{}Number $n$ of node}}%
      \csname LTb\endcsname%
      \put(5418,1097){\makebox(0,0)[r]{\strut{}Original LFR}}%
      \csname LTb\endcsname%
      \put(5418,877){\makebox(0,0)[r]{\strut{}\emlfr}}%
    }%
    \gplbacktext
    \put(0,0){\includegraphics{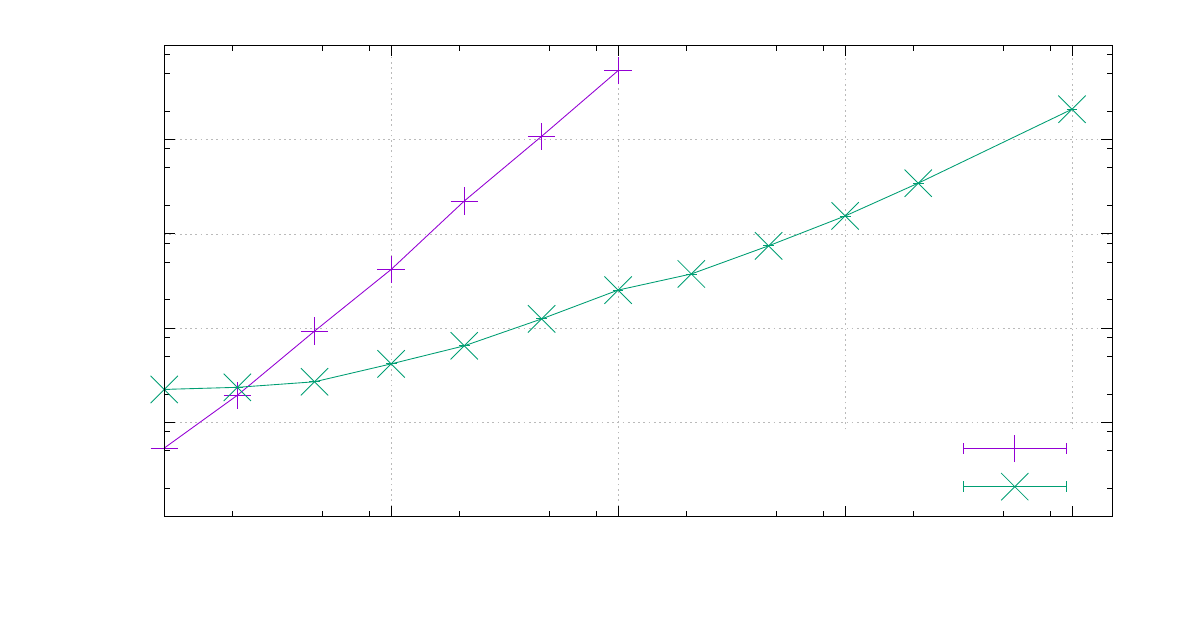}}%
    \gplfronttext
  \end{picture}%
\endgroup
}\hfill%
	\scalebox{0.6}{%
\begingroup
  \makeatletter
  \providecommand\color[2][]{%
    \GenericError{(gnuplot) \space\space\space\@spaces}{%
      Package color not loaded in conjunction with
      terminal option `colourtext'%
    }{See the gnuplot documentation for explanation.%
    }{Either use 'blacktext' in gnuplot or load the package
      color.sty in LaTeX.}%
    \renewcommand\color[2][]{}%
  }%
  \providecommand\includegraphics[2][]{%
    \GenericError{(gnuplot) \space\space\space\@spaces}{%
      Package graphicx or graphics not loaded%
    }{See the gnuplot documentation for explanation.%
    }{The gnuplot epslatex terminal needs graphicx.sty or graphics.sty.}%
    \renewcommand\includegraphics[2][]{}%
  }%
  \providecommand\rotatebox[2]{#2}%
  \@ifundefined{ifGPcolor}{%
    \newif\ifGPcolor
    \GPcolortrue
  }{}%
  \@ifundefined{ifGPblacktext}{%
    \newif\ifGPblacktext
    \GPblacktextfalse
  }{}%
  \let\gplgaddtomacro\g@addto@macro
  \gdef\gplbacktext{}%
  \gdef\gplfronttext{}%
  \makeatother
  \ifGPblacktext
    \def\colorrgb#1{}%
    \def\colorgray#1{}%
  \else
    \ifGPcolor
      \def\colorrgb#1{\color[rgb]{#1}}%
      \def\colorgray#1{\color[gray]{#1}}%
      \expandafter\def\csname LTw\endcsname{\color{white}}%
      \expandafter\def\csname LTb\endcsname{\color{black}}%
      \expandafter\def\csname LTa\endcsname{\color{black}}%
      \expandafter\def\csname LT0\endcsname{\color[rgb]{1,0,0}}%
      \expandafter\def\csname LT1\endcsname{\color[rgb]{0,1,0}}%
      \expandafter\def\csname LT2\endcsname{\color[rgb]{0,0,1}}%
      \expandafter\def\csname LT3\endcsname{\color[rgb]{1,0,1}}%
      \expandafter\def\csname LT4\endcsname{\color[rgb]{0,1,1}}%
      \expandafter\def\csname LT5\endcsname{\color[rgb]{1,1,0}}%
      \expandafter\def\csname LT6\endcsname{\color[rgb]{0,0,0}}%
      \expandafter\def\csname LT7\endcsname{\color[rgb]{1,0.3,0}}%
      \expandafter\def\csname LT8\endcsname{\color[rgb]{0.5,0.5,0.5}}%
    \else
      \def\colorrgb#1{\color{black}}%
      \def\colorgray#1{\color[gray]{#1}}%
      \expandafter\def\csname LTw\endcsname{\color{white}}%
      \expandafter\def\csname LTb\endcsname{\color{black}}%
      \expandafter\def\csname LTa\endcsname{\color{black}}%
      \expandafter\def\csname LT0\endcsname{\color{black}}%
      \expandafter\def\csname LT1\endcsname{\color{black}}%
      \expandafter\def\csname LT2\endcsname{\color{black}}%
      \expandafter\def\csname LT3\endcsname{\color{black}}%
      \expandafter\def\csname LT4\endcsname{\color{black}}%
      \expandafter\def\csname LT5\endcsname{\color{black}}%
      \expandafter\def\csname LT6\endcsname{\color{black}}%
      \expandafter\def\csname LT7\endcsname{\color{black}}%
      \expandafter\def\csname LT8\endcsname{\color{black}}%
    \fi
  \fi
    \setlength{\unitlength}{0.0500bp}%
    \ifx\gptboxheight\undefined%
      \newlength{\gptboxheight}%
      \newlength{\gptboxwidth}%
      \newsavebox{\gptboxtext}%
    \fi%
    \setlength{\fboxrule}{0.5pt}%
    \setlength{\fboxsep}{1pt}%
\begin{picture}(6802.00,3684.00)%
    \gplgaddtomacro\gplbacktext{%
      \csname LTb\endcsname%
      \put(814,704){\makebox(0,0)[r]{\strut{}$10^{1}$}}%
      \csname LTb\endcsname%
      \put(814,1247){\makebox(0,0)[r]{\strut{}$10^{2}$}}%
      \csname LTb\endcsname%
      \put(814,1790){\makebox(0,0)[r]{\strut{}$10^{3}$}}%
      \csname LTb\endcsname%
      \put(814,2333){\makebox(0,0)[r]{\strut{}$10^{4}$}}%
      \csname LTb\endcsname%
      \put(814,2876){\makebox(0,0)[r]{\strut{}$10^{5}$}}%
      \csname LTb\endcsname%
      \put(814,3419){\makebox(0,0)[r]{\strut{}$10^{6}$}}%
      \csname LTb\endcsname%
      \put(1402,484){\makebox(0,0){\strut{}$10^{7}$}}%
      \csname LTb\endcsname%
      \put(2918,484){\makebox(0,0){\strut{}$10^{8}$}}%
      \csname LTb\endcsname%
      \put(4433,484){\makebox(0,0){\strut{}$10^{9}$}}%
      \csname LTb\endcsname%
      \put(5949,484){\makebox(0,0){\strut{}$10^{10}$}}%
    }%
    \gplgaddtomacro\gplfronttext{%
      \csname LTb\endcsname%
      \put(176,2061){\rotatebox{-270}{\makebox(0,0){\strut{}Runtime [s]}}}%
      \put(3675,154){\makebox(0,0){\strut{}Number $m$ of edges}}%
      \csname LTb\endcsname%
      \put(5418,1537){\makebox(0,0)[r]{\strut{}\vles, $\bar d{=100}$}}%
      \csname LTb\endcsname%
      \put(5418,1317){\makebox(0,0)[r]{\strut{}\vles, $\bar d{=1000}$}}%
      \csname LTb\endcsname%
      \put(5418,1097){\makebox(0,0)[r]{\strut{}\emes, $\bar d{=100}$}}%
      \csname LTb\endcsname%
      \put(5418,877){\makebox(0,0)[r]{\strut{}\emes, $\bar d{=1000}$}}%
    }%
    \gplbacktext
    \put(0,0){\includegraphics{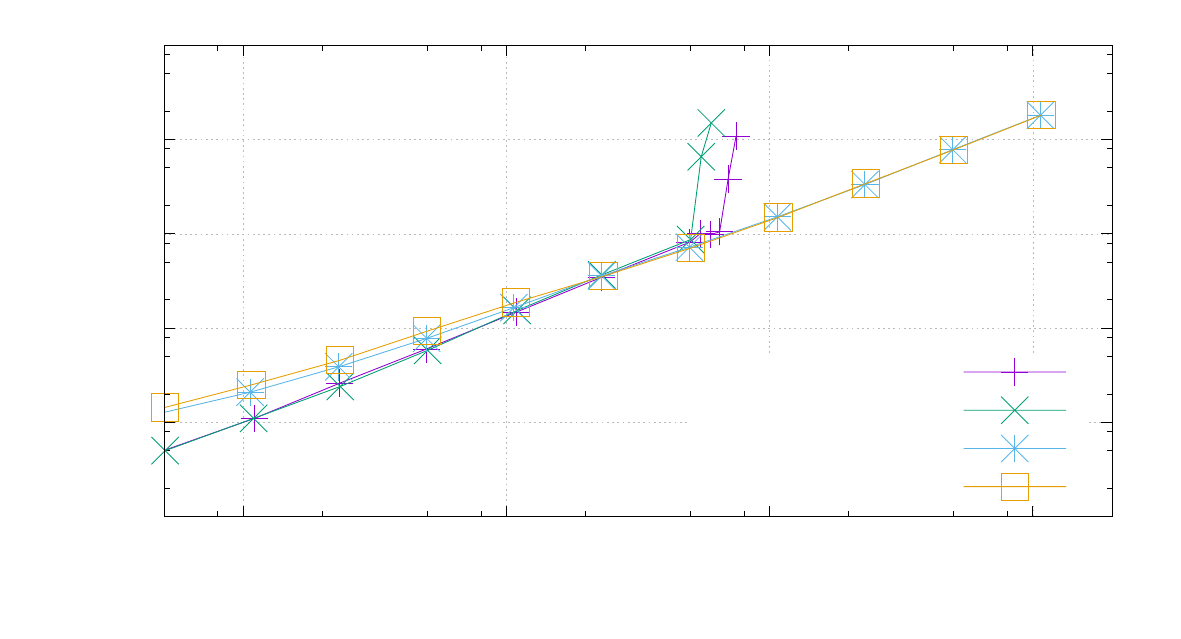}}%
    \gplfronttext
  \end{picture}%
\endgroup
}%
	
    \caption{%
        \textbf{Left:} Runtime on (SysA) of the original LFR implementation and \emlfr{} for $\mu{=}0.2$. 
        \textbf{Right: } Runtime on (SysB) of IM \vles{} and \emes{} on a graph with $m$ of edges and average degree $\bar d$ executing $k{=}10m$ swaps.
    }
    \label{fig:runtime_sweep}
    \label{fig:runtime_sweep_es}
\end{figure}

\subsection{Performance of \emhh{}} 
Our implementation of \emhh{} produces \num{180 \pm 5} million edges per second on (SysA) up to  at least $\num{2e10}$ edges.
Here, we include the computation of the input degree sequence, \emhh's compaction step, as well as the writing of the output to external memory.

\subsection{Performance of \emes{}} Figure~\ref{fig:runtime_sweep_es} presents the runtime required on (SysB) to process $k{=}10m$ swaps in an input graph with $m$ edges and an average degree $\bar d \in \{100, 1000\}$.
For reference, the performance of the existing internal memory edge swap algorithm \vles{} based on the authors' implementation~\cite{DBLP:journals/corr/abs-cs-0502085} is included.
Here we report only on the edge swapping process excluding any precomputation.
To achieve comparability, we removed connectivity tests, fixed memory management issues, and adopted the number of swaps.
Further, we extended counters for edge ids and accumulated degrees to \SI{64}{\bit} integers in order to support experiments with more than $2^{30}$ edges.
\vles{} slows down by a factor of $25$ if the data structure exceeds the available internal memory by less than \SI{10}{\percent}.
We observe an analogous behavior on machines with larger RAM.
\emes{} is faster than \vles{} for all instances with $m > \num{2.5e8}$ edges; those graphs still fit into main memory.

FDSM has applications beyond synthetic graphs, and is for instance used on real data to assess the statistical significance of observations~\cite{Schlauch2015}.
In that spirit, we execute \emes{} on an undirected version of the crawled ClueWeb12 graph's core~\cite{clueweb12}
which we obtain by deleting all nodes corresponding to uncrawled URLs.%
\footnote{We consider such vertices unrealistically (simple) as they have only degree~1 and account for $\approx $\SI{84}{\percent} of nodes in the original graph.}
Performing $k = m$ swaps on this graph with $n \approx \num{9.8e8}$ nodes and $m \approx \num{3.7e10}$ edges is feasible in less than \SI{19.1}{\hour} on (SysB).

Bhuiyan et al. propose a distributed edge switching algorithm and evaluate it on a compute cluster with 64~nodes each equipped with two Intel~Xeon~E5-2670~2.60GHz 8-core processors and 64GB RAM~\cite{DBLP:conf/icpp/BhuiyanCKM14}.
The authors report to perform $k{=}\num{1.15e11}$ swaps on a graph with $m{=}\num{e10}$ generated in a preferential attachment process in less than \SI{3}{\hour}.
We generate a preferential attachment graph using an EM~generator~\cite{doi:10.1137/1.9781611974317.4} matching the aforementioned properties and carried out edge swaps using \emes{} on (SysA).
We observe a slow down of only $8.3$ on a machine with $1/128$ the number of comparable cores and $1/64$ of internal memory.

\subsection{Performance of \emcmes{} and qualitative comparison with \emes{}}\label{subsec:emcmes-exp}
\begin{figure}
	\centering
	\includegraphics[width=0.45\textwidth]{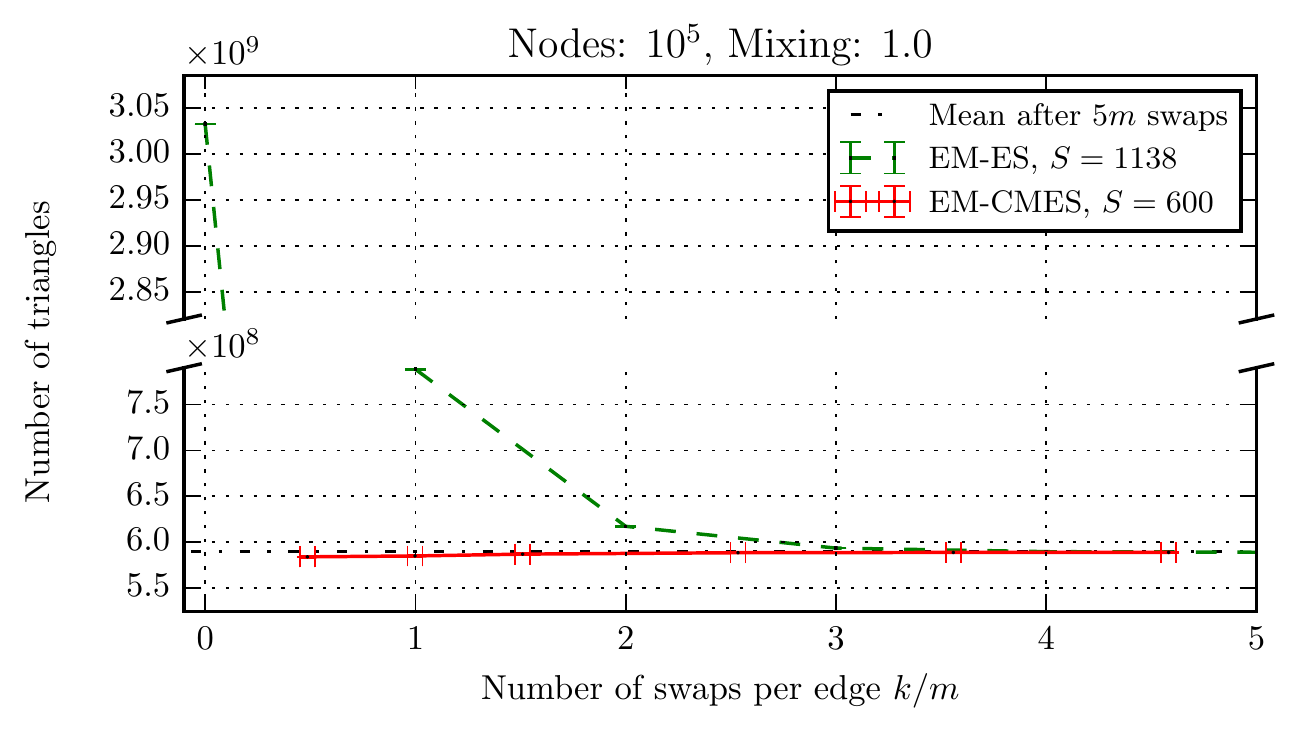}\hfill
	\includegraphics[width=0.45\textwidth]{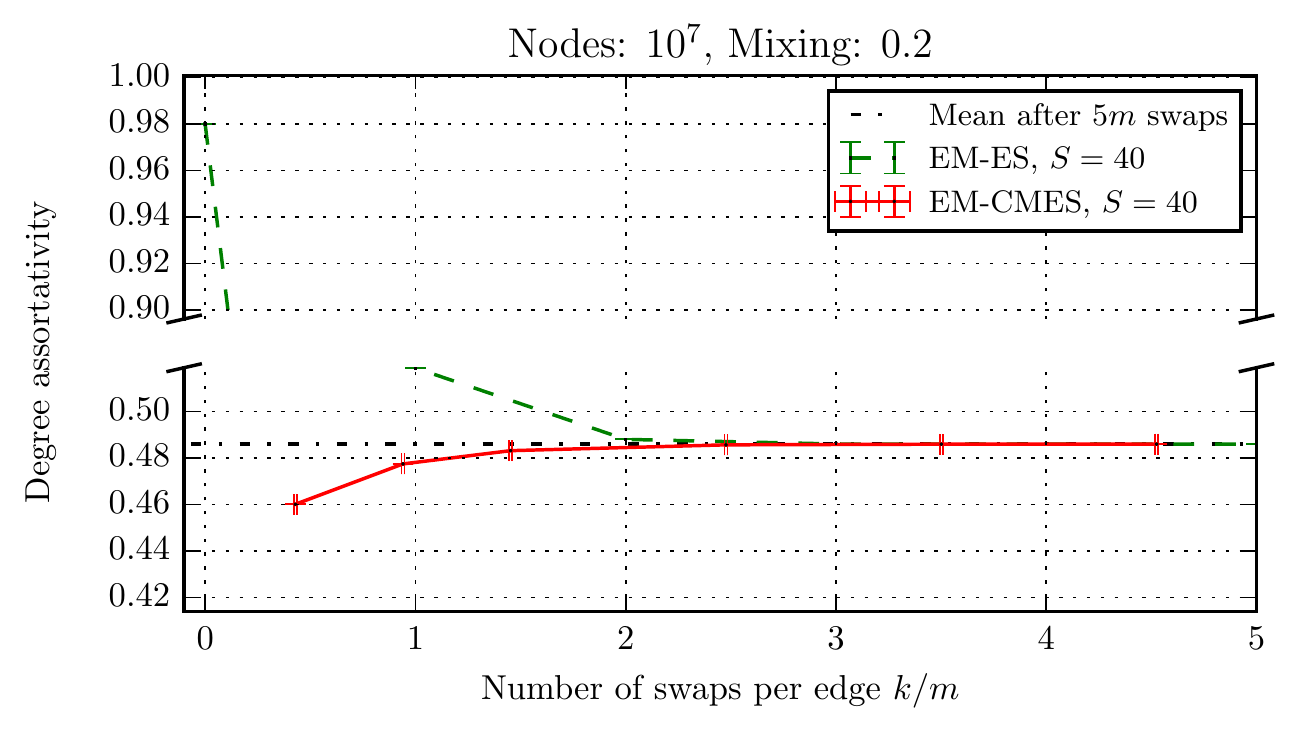}
	
	\caption{%
		\textbf{Left:} Number of triangles on \settingsConst{} with $n=\num{1e5}$ and $\mu=1.0$.
		\textbf{Right: } Degree assortativity on \settingsConst{} with $n=\num{1e7}$ and $\mu=0.2$.
		In order to factor in the increased runtime of \emcmes{} compared to \emhh{}, plots of \emcmes{} are shifted by the runtime of this phase relative to the execution of \emes.
		As \emcmes{} is a Las-Vegas algorithm, this incurs an additional error along the x-axis.
	}
	\label{fig:mixing_cmes}
\end{figure}
In section~\ref{sec:cmes}, we describe an alternative graph sampling method.
Instead of seeding \emes{} with a highly biased graph using \emhh{}, we employ the Configuration Model to quickly generate a non-simple random graph.
In order to obtain a simple graph, we then carry out several \emes{} runs in a Las-Vegas fashion.

Since \emes{} scans through the edge list in each iteration, runs with very few swaps are inefficient.
For this reason, we start the subsequent Markov Chain early: 
First identify all multi-edges and self-loops and generate swaps with random partners.
In a second step, we then introduce additional random swaps until the run contains at least $m/10$ operations\footnote{We chose this number as it yields execution times similar to the $m/8$-setting of \emes{} on simple graphs}.

For an experimental comparison between \emes{} and \emcmes{}, we consider the runtime until both yield a sufficiently uniform random sample.
Of course, the uniformity is hard to quantify; similarly to related studies (cf. section~\ref{subsec:empiricalconv}), we estimate the mixing times of both approaches as follows.
Starting from a common seed graph $G^{(0)}$, we generate an ensemble $\{G_1^{(k)}, \ldots, G_S^{(k)}\}$ of $S \gg 1$ instances by applying independent random sequences of $k \gg m$ swaps each.
During this process, we regularly export snapshots $G_i^{(jm)}$ of the intermediate instances $j \in [k/m]$ of graph $G_i$.
For \emcmes{}, we start from the same seed graph, apply the algorithm and then carry out $k$ swaps as described above.

For each snapshot, we compute a number of metrics, such as the average local clustering coefficient (ACC), the number of triangles, and degree assortativity%
\footnote{%
	In preliminary experiments, we also included spectral properties (such as extremal eigenvalues of the adjacency/laplacian matrix) and the closeness centrality of fixed nodes.
	As these measurement are more expensive to compute and yield qualitatively similar results, we decided not to include them in the larger trials.
}.
We then investigate how the distribution of these measures evolves within the ensemble as we carry out an increasing number of swaps.
We omit results for ACC since they are less sensitive compared to the other measures (cf. section~\ref{subsec:mixing-emes}).

As illustrated in Fig.~\ref{fig:mixing_cmes} and Appendix~\ref{sec:appendix-emcmes}, all proxy measures converge within $5m$ swaps with a very small variance.
No statistically significant change can be observed compared to a Markov chain with $30m$ operations (which was only computed for a subset of each ensemble).
\emhh{} generates biased instances with special properties, such as a high number of triangles and correlated node degrees, while the features of \emcmes{}'s output nearly match the converged ensemble.
This suggests that the number of swaps to obtain a sufficiently uniform sample can be reduced for \emcmes{}.

Due to computational costs, the study was carried out on multiple machines executing several tasks in parallel.
Hence, absolute running times are not meaningful, and we rather measure the computational costs in units of time required to carry out $1m$ swaps by the same process.
This accounts for the offset of \emcmes{}'s first data point.

The number of rounds required to obtain a simple graph depends on the degree distribution.
For \settingsConst{} with $n=\num{1e5}$ and $\mu=1$, a fraction of $\SI{5.1}{\percent}$ of the edges produced by the Configuration Model are illegal.
\emes{} requires $\num{18 \pm 2}$ rewiring runs in case a single swap is used per round to rewire an illegal edge.
In the default mode of operation, $\num{5.0 \pm 0.0}$ rounds suffice as the number of rewiring swaps per illegal edge is doubled in each round.
For larger graphs with $n=\num{1e7}$, only $\SI{0.07}{\percent}$ of edges are illegal and need \num{2.25 \pm 0.4} rewiring runs.

\subsection{Convergence of \emes{}}\label{subsec:mixing-emes}
\begin{figure}
	\centering
	\includegraphics[width=0.45\textwidth]{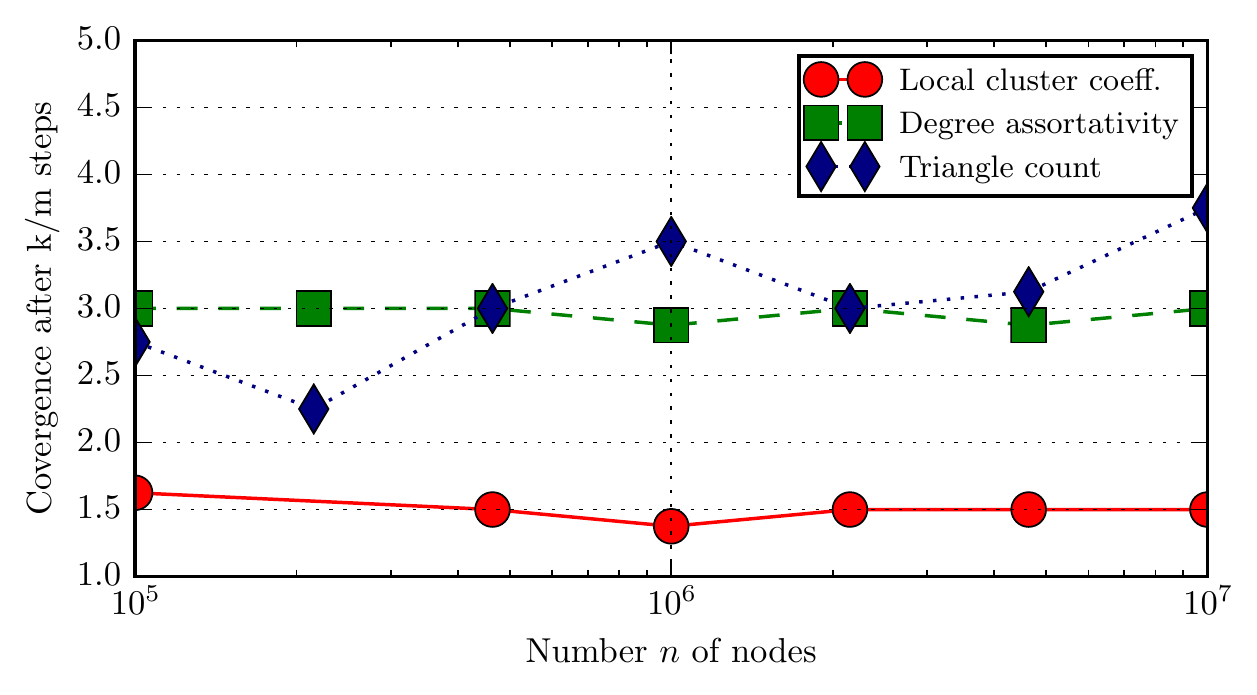}\hfill
	\includegraphics[width=0.45\textwidth]{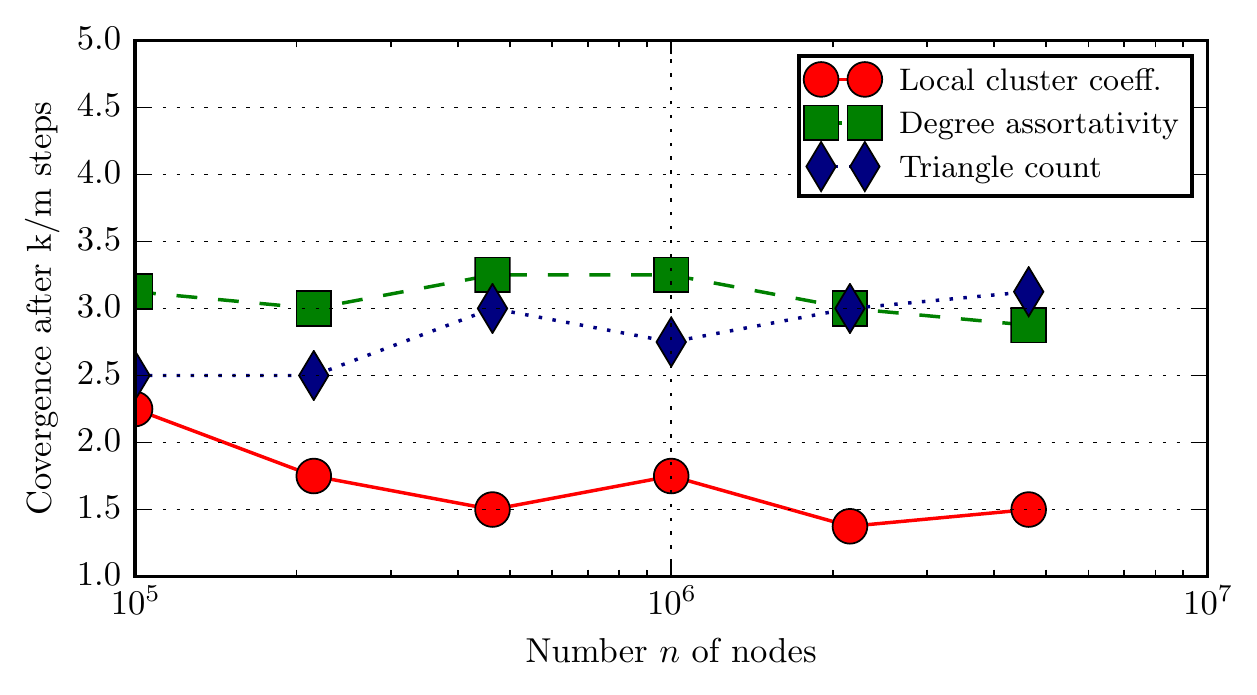}
	
	\caption{%
		Number of swaps per edge after which ensembles of graphs with \settingsConst{}, $\num{1e5} \le n \le \num{1e7}$ and $\mu = 0.4$ (left) and $\mu = 0.6$ (right) converge.
		Due to computation costs, the ensemble size is reduced from $S > 100$ to $S > 10$ for large graphs.
	}
	\label{fig:mixing_emes}
\end{figure}

In a similar spirit to the previous section, we indirectly investigate the Markov chain's mixing time as a function of the number of nodes $n$.
To do so, we generate ensembles as before with $\num{1e5} \le n \le \num{1e7}$ and compute the same graph metrics.
For each group and measure, we then search for the first snapshot $p$ in which the measure's mean is within an interval of half the standard deviation of the final values and subsequently remains there for at least three phases.
We then interpret $p$ as a proxy for the mixing time.
As depicted in Fig.~\ref{fig:mixing_emes}, no measure shows a systematic increase over the two orders of magnitude considered.
It hence seems plausible not to increase the number of swaps performed by \emlfr{} compared to the original implementation.

\subsection{Performance of \emlfr{}} Figure~\ref{fig:runtime_sweep} reports the runtime of the original LFR implementation and \emlfr{} as a function of the number of nodes $n$ and $\nu = 1$.
\emlfr{} is faster for graphs with $n \ge \num{2.5e4}$ nodes which feature approximately \num{5e5} edges and are well in the IM domain.
Further, the implementation is capable of producing graphs with more than $\num{1e10}$ edges in  \SI{17}{\hour}.%
\footnote{%
	Roughly \SI{1.5}{\hour} are spend in the end-game of the Global Rewiring (at that point less than one edge out of \num{e6} is invalid).
	In this situation, an algorithm using random I/Os may yield a speed-up.
	Alternatively, we could simply discard the few remaining invalid edges since they only constitute an insignificant fraction.
}
Using the same time budget, the original implementation generates graphs more than two orders of magnitude smaller.

\subsection{Qualitative Comparison of \emlfr{}}\label{subsec:exp-qual-compare}
When designing \emlfr{}, we made sure that it closely follows the LFR benchmark such that we can expect it to produce graphs following the same distribution as the original LFR generator.
In order to show experimentally that we achieved this goal, we generated graphs with identical parameters using the original LFR implementation and \emlfr{}.
For disjoint clusters we also compare it with the implementation that is part of NetworKit~\cite{DBLP:journals/corr/StaudtSM14}.
Using NetworKit, we evaluate the results of Infomap \cite{rab-t-09}, Louvain \cite{bgll-f-08} and OSLOM \cite{lrrf-fsscn-11}, three state-of-the-art clustering algorithms \cite{ekgb-ancac-16, Buzun2014,fh-ca-16}, and compare them using the adjusted rand measure \cite{ha-c-85} and NMI \cite{er-c-12}.

\begin{figure}[t]
  \scalebox{0.6}{%
\begingroup
  \makeatletter
  \providecommand\color[2][]{%
    \GenericError{(gnuplot) \space\space\space\@spaces}{%
      Package color not loaded in conjunction with
      terminal option `colourtext'%
    }{See the gnuplot documentation for explanation.%
    }{Either use 'blacktext' in gnuplot or load the package
      color.sty in LaTeX.}%
    \renewcommand\color[2][]{}%
  }%
  \providecommand\includegraphics[2][]{%
    \GenericError{(gnuplot) \space\space\space\@spaces}{%
      Package graphicx or graphics not loaded%
    }{See the gnuplot documentation for explanation.%
    }{The gnuplot epslatex terminal needs graphicx.sty or graphics.sty.}%
    \renewcommand\includegraphics[2][]{}%
  }%
  \providecommand\rotatebox[2]{#2}%
  \@ifundefined{ifGPcolor}{%
    \newif\ifGPcolor
    \GPcolortrue
  }{}%
  \@ifundefined{ifGPblacktext}{%
    \newif\ifGPblacktext
    \GPblacktextfalse
  }{}%
  \let\gplgaddtomacro\g@addto@macro
  \gdef\gplbacktext{}%
  \gdef\gplfronttext{}%
  \makeatother
  \ifGPblacktext
    \def\colorrgb#1{}%
    \def\colorgray#1{}%
  \else
    \ifGPcolor
      \def\colorrgb#1{\color[rgb]{#1}}%
      \def\colorgray#1{\color[gray]{#1}}%
      \expandafter\def\csname LTw\endcsname{\color{white}}%
      \expandafter\def\csname LTb\endcsname{\color{black}}%
      \expandafter\def\csname LTa\endcsname{\color{black}}%
      \expandafter\def\csname LT0\endcsname{\color[rgb]{1,0,0}}%
      \expandafter\def\csname LT1\endcsname{\color[rgb]{0,1,0}}%
      \expandafter\def\csname LT2\endcsname{\color[rgb]{0,0,1}}%
      \expandafter\def\csname LT3\endcsname{\color[rgb]{1,0,1}}%
      \expandafter\def\csname LT4\endcsname{\color[rgb]{0,1,1}}%
      \expandafter\def\csname LT5\endcsname{\color[rgb]{1,1,0}}%
      \expandafter\def\csname LT6\endcsname{\color[rgb]{0,0,0}}%
      \expandafter\def\csname LT7\endcsname{\color[rgb]{1,0.3,0}}%
      \expandafter\def\csname LT8\endcsname{\color[rgb]{0.5,0.5,0.5}}%
    \else
      \def\colorrgb#1{\color{black}}%
      \def\colorgray#1{\color[gray]{#1}}%
      \expandafter\def\csname LTw\endcsname{\color{white}}%
      \expandafter\def\csname LTb\endcsname{\color{black}}%
      \expandafter\def\csname LTa\endcsname{\color{black}}%
      \expandafter\def\csname LT0\endcsname{\color{black}}%
      \expandafter\def\csname LT1\endcsname{\color{black}}%
      \expandafter\def\csname LT2\endcsname{\color{black}}%
      \expandafter\def\csname LT3\endcsname{\color{black}}%
      \expandafter\def\csname LT4\endcsname{\color{black}}%
      \expandafter\def\csname LT5\endcsname{\color{black}}%
      \expandafter\def\csname LT6\endcsname{\color{black}}%
      \expandafter\def\csname LT7\endcsname{\color{black}}%
      \expandafter\def\csname LT8\endcsname{\color{black}}%
    \fi
  \fi
    \setlength{\unitlength}{0.0500bp}%
    \ifx\gptboxheight\undefined%
      \newlength{\gptboxheight}%
      \newlength{\gptboxwidth}%
      \newsavebox{\gptboxtext}%
    \fi%
    \setlength{\fboxrule}{0.5pt}%
    \setlength{\fboxsep}{1pt}%
\begin{picture}(6802.00,3614.00)%
    \gplgaddtomacro\gplbacktext{%
      \csname LTb\endcsname%
      \put(814,704){\makebox(0,0)[r]{\strut{}$0$}}%
      \csname LTb\endcsname%
      \put(814,1113){\makebox(0,0)[r]{\strut{}$0.2$}}%
      \csname LTb\endcsname%
      \put(814,1522){\makebox(0,0)[r]{\strut{}$0.4$}}%
      \csname LTb\endcsname%
      \put(814,1931){\makebox(0,0)[r]{\strut{}$0.6$}}%
      \csname LTb\endcsname%
      \put(814,2340){\makebox(0,0)[r]{\strut{}$0.8$}}%
      \csname LTb\endcsname%
      \put(814,2749){\makebox(0,0)[r]{\strut{}$1$}}%
      \csname LTb\endcsname%
      \put(1234,484){\makebox(0,0){\strut{}$10^{3}$}}%
      \csname LTb\endcsname%
      \put(3096,484){\makebox(0,0){\strut{}$10^{4}$}}%
      \csname LTb\endcsname%
      \put(4957,484){\makebox(0,0){\strut{}$10^{5}$}}%
    }%
    \gplgaddtomacro\gplfronttext{%
      \csname LTb\endcsname%
      \put(176,1828){\rotatebox{-270}{\makebox(0,0){\strut{}AR}}}%
      \put(3675,154){\makebox(0,0){\strut{}Number $n$ of nodes}}%
      \put(3675,3283){\makebox(0,0){\strut{}Mixing $\mu = 0.6$, Cluster: Infomap}}%
      \csname LTb\endcsname%
      \put(2266,1317){\makebox(0,0)[r]{\strut{}Orig}}%
      \csname LTb\endcsname%
      \put(2266,1097){\makebox(0,0)[r]{\strut{}NetworKit}}%
      \csname LTb\endcsname%
      \put(2266,877){\makebox(0,0)[r]{\strut{}EM}}%
    }%
    \gplbacktext
    \put(0,0){\includegraphics{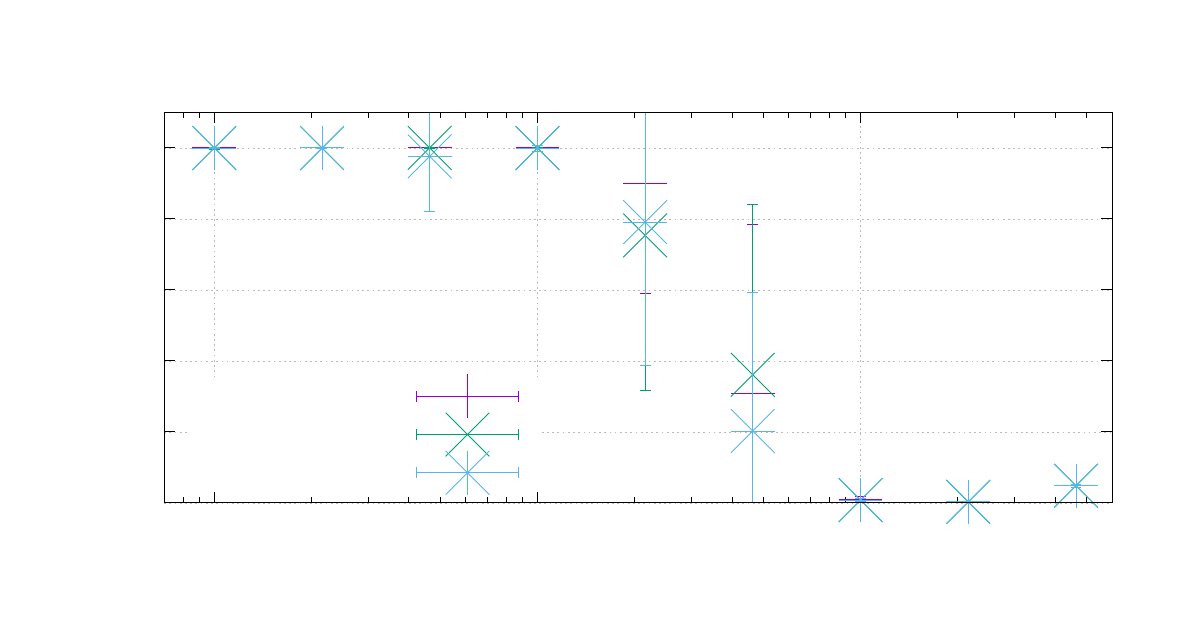}}%
    \gplfronttext
  \end{picture}%
\endgroup
}\hfill
  \scalebox{0.6}{%
\begingroup
  \makeatletter
  \providecommand\color[2][]{%
    \GenericError{(gnuplot) \space\space\space\@spaces}{%
      Package color not loaded in conjunction with
      terminal option `colourtext'%
    }{See the gnuplot documentation for explanation.%
    }{Either use 'blacktext' in gnuplot or load the package
      color.sty in LaTeX.}%
    \renewcommand\color[2][]{}%
  }%
  \providecommand\includegraphics[2][]{%
    \GenericError{(gnuplot) \space\space\space\@spaces}{%
      Package graphicx or graphics not loaded%
    }{See the gnuplot documentation for explanation.%
    }{The gnuplot epslatex terminal needs graphicx.sty or graphics.sty.}%
    \renewcommand\includegraphics[2][]{}%
  }%
  \providecommand\rotatebox[2]{#2}%
  \@ifundefined{ifGPcolor}{%
    \newif\ifGPcolor
    \GPcolortrue
  }{}%
  \@ifundefined{ifGPblacktext}{%
    \newif\ifGPblacktext
    \GPblacktextfalse
  }{}%
  \let\gplgaddtomacro\g@addto@macro
  \gdef\gplbacktext{}%
  \gdef\gplfronttext{}%
  \makeatother
  \ifGPblacktext
    \def\colorrgb#1{}%
    \def\colorgray#1{}%
  \else
    \ifGPcolor
      \def\colorrgb#1{\color[rgb]{#1}}%
      \def\colorgray#1{\color[gray]{#1}}%
      \expandafter\def\csname LTw\endcsname{\color{white}}%
      \expandafter\def\csname LTb\endcsname{\color{black}}%
      \expandafter\def\csname LTa\endcsname{\color{black}}%
      \expandafter\def\csname LT0\endcsname{\color[rgb]{1,0,0}}%
      \expandafter\def\csname LT1\endcsname{\color[rgb]{0,1,0}}%
      \expandafter\def\csname LT2\endcsname{\color[rgb]{0,0,1}}%
      \expandafter\def\csname LT3\endcsname{\color[rgb]{1,0,1}}%
      \expandafter\def\csname LT4\endcsname{\color[rgb]{0,1,1}}%
      \expandafter\def\csname LT5\endcsname{\color[rgb]{1,1,0}}%
      \expandafter\def\csname LT6\endcsname{\color[rgb]{0,0,0}}%
      \expandafter\def\csname LT7\endcsname{\color[rgb]{1,0.3,0}}%
      \expandafter\def\csname LT8\endcsname{\color[rgb]{0.5,0.5,0.5}}%
    \else
      \def\colorrgb#1{\color{black}}%
      \def\colorgray#1{\color[gray]{#1}}%
      \expandafter\def\csname LTw\endcsname{\color{white}}%
      \expandafter\def\csname LTb\endcsname{\color{black}}%
      \expandafter\def\csname LTa\endcsname{\color{black}}%
      \expandafter\def\csname LT0\endcsname{\color{black}}%
      \expandafter\def\csname LT1\endcsname{\color{black}}%
      \expandafter\def\csname LT2\endcsname{\color{black}}%
      \expandafter\def\csname LT3\endcsname{\color{black}}%
      \expandafter\def\csname LT4\endcsname{\color{black}}%
      \expandafter\def\csname LT5\endcsname{\color{black}}%
      \expandafter\def\csname LT6\endcsname{\color{black}}%
      \expandafter\def\csname LT7\endcsname{\color{black}}%
      \expandafter\def\csname LT8\endcsname{\color{black}}%
    \fi
  \fi
    \setlength{\unitlength}{0.0500bp}%
    \ifx\gptboxheight\undefined%
      \newlength{\gptboxheight}%
      \newlength{\gptboxwidth}%
      \newsavebox{\gptboxtext}%
    \fi%
    \setlength{\fboxrule}{0.5pt}%
    \setlength{\fboxsep}{1pt}%
\begin{picture}(6802.00,3614.00)%
    \gplgaddtomacro\gplbacktext{%
      \csname LTb\endcsname%
      \put(814,704){\makebox(0,0)[r]{\strut{}$0$}}%
      \csname LTb\endcsname%
      \put(814,1113){\makebox(0,0)[r]{\strut{}$0.2$}}%
      \csname LTb\endcsname%
      \put(814,1522){\makebox(0,0)[r]{\strut{}$0.4$}}%
      \csname LTb\endcsname%
      \put(814,1931){\makebox(0,0)[r]{\strut{}$0.6$}}%
      \csname LTb\endcsname%
      \put(814,2340){\makebox(0,0)[r]{\strut{}$0.8$}}%
      \csname LTb\endcsname%
      \put(814,2749){\makebox(0,0)[r]{\strut{}$1$}}%
      \csname LTb\endcsname%
      \put(1234,484){\makebox(0,0){\strut{}$10^{3}$}}%
      \csname LTb\endcsname%
      \put(3096,484){\makebox(0,0){\strut{}$10^{4}$}}%
      \csname LTb\endcsname%
      \put(4957,484){\makebox(0,0){\strut{}$10^{5}$}}%
    }%
    \gplgaddtomacro\gplfronttext{%
      \csname LTb\endcsname%
      \put(176,1828){\rotatebox{-270}{\makebox(0,0){\strut{}AR}}}%
      \put(3675,154){\makebox(0,0){\strut{}Number $n$ of nodes}}%
      \put(3675,3283){\makebox(0,0){\strut{}Mixing $\mu = 0.6$, Cluster: Louvain}}%
      \csname LTb\endcsname%
      \put(2266,1317){\makebox(0,0)[r]{\strut{}Orig}}%
      \csname LTb\endcsname%
      \put(2266,1097){\makebox(0,0)[r]{\strut{}NetworKit}}%
      \csname LTb\endcsname%
      \put(2266,877){\makebox(0,0)[r]{\strut{}EM}}%
    }%
    \gplbacktext
    \put(0,0){\includegraphics{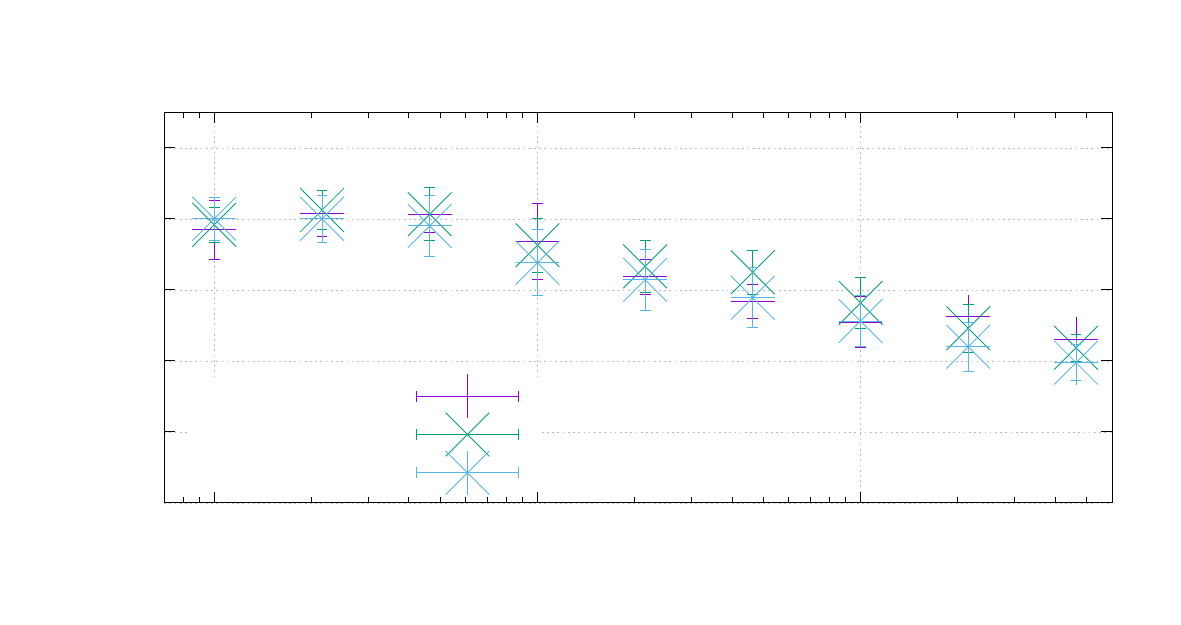}}%
    \gplfronttext
  \end{picture}%
\endgroup
}

  \caption{Adjusted rand measure of Infomap/Louvain and ground truth at $\mu = 0.6$ with disjoint clusters, $s_{\min} = 10$, $s_{\max} = n/20$.}
  \label{fig:quality_lfr_disjoint}
\end{figure}

\begin{figure}
  \scalebox{0.6}{%
\begingroup
  \makeatletter
  \providecommand\color[2][]{%
    \GenericError{(gnuplot) \space\space\space\@spaces}{%
      Package color not loaded in conjunction with
      terminal option `colourtext'%
    }{See the gnuplot documentation for explanation.%
    }{Either use 'blacktext' in gnuplot or load the package
      color.sty in LaTeX.}%
    \renewcommand\color[2][]{}%
  }%
  \providecommand\includegraphics[2][]{%
    \GenericError{(gnuplot) \space\space\space\@spaces}{%
      Package graphicx or graphics not loaded%
    }{See the gnuplot documentation for explanation.%
    }{The gnuplot epslatex terminal needs graphicx.sty or graphics.sty.}%
    \renewcommand\includegraphics[2][]{}%
  }%
  \providecommand\rotatebox[2]{#2}%
  \@ifundefined{ifGPcolor}{%
    \newif\ifGPcolor
    \GPcolortrue
  }{}%
  \@ifundefined{ifGPblacktext}{%
    \newif\ifGPblacktext
    \GPblacktextfalse
  }{}%
  \let\gplgaddtomacro\g@addto@macro
  \gdef\gplbacktext{}%
  \gdef\gplfronttext{}%
  \makeatother
  \ifGPblacktext
    \def\colorrgb#1{}%
    \def\colorgray#1{}%
  \else
    \ifGPcolor
      \def\colorrgb#1{\color[rgb]{#1}}%
      \def\colorgray#1{\color[gray]{#1}}%
      \expandafter\def\csname LTw\endcsname{\color{white}}%
      \expandafter\def\csname LTb\endcsname{\color{black}}%
      \expandafter\def\csname LTa\endcsname{\color{black}}%
      \expandafter\def\csname LT0\endcsname{\color[rgb]{1,0,0}}%
      \expandafter\def\csname LT1\endcsname{\color[rgb]{0,1,0}}%
      \expandafter\def\csname LT2\endcsname{\color[rgb]{0,0,1}}%
      \expandafter\def\csname LT3\endcsname{\color[rgb]{1,0,1}}%
      \expandafter\def\csname LT4\endcsname{\color[rgb]{0,1,1}}%
      \expandafter\def\csname LT5\endcsname{\color[rgb]{1,1,0}}%
      \expandafter\def\csname LT6\endcsname{\color[rgb]{0,0,0}}%
      \expandafter\def\csname LT7\endcsname{\color[rgb]{1,0.3,0}}%
      \expandafter\def\csname LT8\endcsname{\color[rgb]{0.5,0.5,0.5}}%
    \else
      \def\colorrgb#1{\color{black}}%
      \def\colorgray#1{\color[gray]{#1}}%
      \expandafter\def\csname LTw\endcsname{\color{white}}%
      \expandafter\def\csname LTb\endcsname{\color{black}}%
      \expandafter\def\csname LTa\endcsname{\color{black}}%
      \expandafter\def\csname LT0\endcsname{\color{black}}%
      \expandafter\def\csname LT1\endcsname{\color{black}}%
      \expandafter\def\csname LT2\endcsname{\color{black}}%
      \expandafter\def\csname LT3\endcsname{\color{black}}%
      \expandafter\def\csname LT4\endcsname{\color{black}}%
      \expandafter\def\csname LT5\endcsname{\color{black}}%
      \expandafter\def\csname LT6\endcsname{\color{black}}%
      \expandafter\def\csname LT7\endcsname{\color{black}}%
      \expandafter\def\csname LT8\endcsname{\color{black}}%
    \fi
  \fi
    \setlength{\unitlength}{0.0500bp}%
    \ifx\gptboxheight\undefined%
      \newlength{\gptboxheight}%
      \newlength{\gptboxwidth}%
      \newsavebox{\gptboxtext}%
    \fi%
    \setlength{\fboxrule}{0.5pt}%
    \setlength{\fboxsep}{1pt}%
\begin{picture}(6802.00,3614.00)%
    \gplgaddtomacro\gplbacktext{%
      \csname LTb\endcsname%
      \put(814,704){\makebox(0,0)[r]{\strut{}$0$}}%
      \csname LTb\endcsname%
      \put(814,1113){\makebox(0,0)[r]{\strut{}$0.2$}}%
      \csname LTb\endcsname%
      \put(814,1522){\makebox(0,0)[r]{\strut{}$0.4$}}%
      \csname LTb\endcsname%
      \put(814,1931){\makebox(0,0)[r]{\strut{}$0.6$}}%
      \csname LTb\endcsname%
      \put(814,2340){\makebox(0,0)[r]{\strut{}$0.8$}}%
      \csname LTb\endcsname%
      \put(814,2749){\makebox(0,0)[r]{\strut{}$1$}}%
      \csname LTb\endcsname%
      \put(1578,484){\makebox(0,0){\strut{}$10^{3}$}}%
      \csname LTb\endcsname%
      \put(3676,484){\makebox(0,0){\strut{}$10^{4}$}}%
      \csname LTb\endcsname%
      \put(5773,484){\makebox(0,0){\strut{}$10^{5}$}}%
    }%
    \gplgaddtomacro\gplfronttext{%
      \csname LTb\endcsname%
      \put(176,1828){\rotatebox{-270}{\makebox(0,0){\strut{}NMI}}}%
      \put(3675,154){\makebox(0,0){\strut{}Number $n$ of nodes}}%
      \put(3675,3283){\makebox(0,0){\strut{}Mixing: $\mu = 0.4$, Cluster: OSLOM, Overlap: $\nu = 2$}}%
      \csname LTb\endcsname%
      \put(5418,2780){\makebox(0,0)[r]{\strut{}Orig}}%
      \csname LTb\endcsname%
      \put(5418,2560){\makebox(0,0)[r]{\strut{}EM}}%
    }%
    \gplbacktext
    \put(0,0){\includegraphics{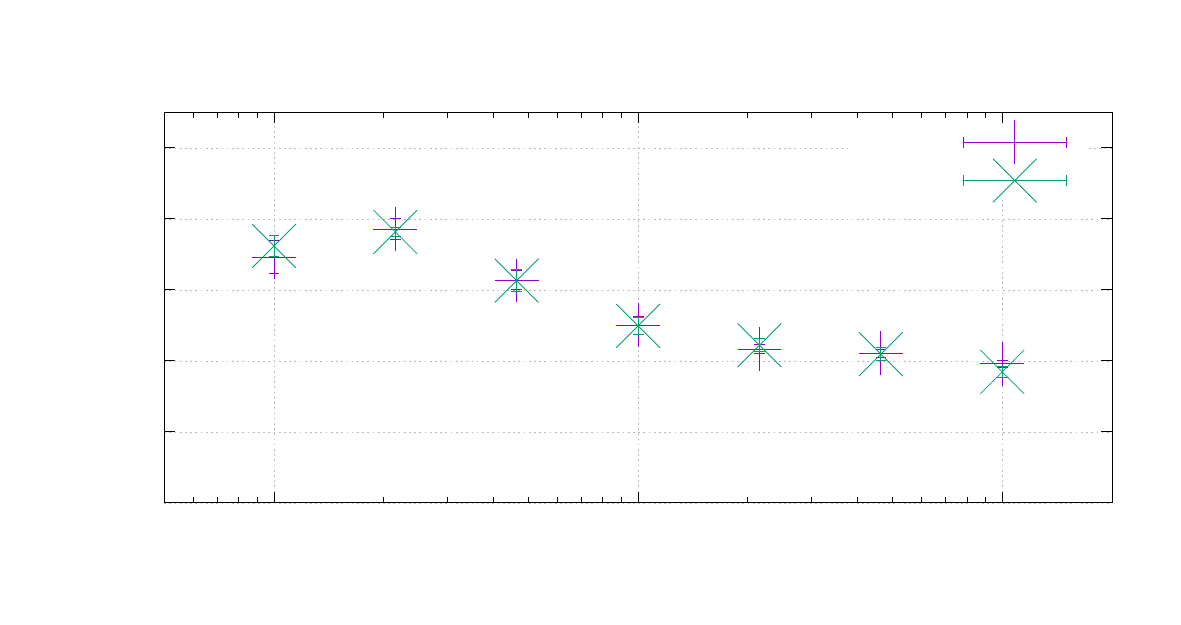}}%
    \gplfronttext
  \end{picture}%
\endgroup
}\hfill
  \scalebox{0.6}{%
\begingroup
  \makeatletter
  \providecommand\color[2][]{%
    \GenericError{(gnuplot) \space\space\space\@spaces}{%
      Package color not loaded in conjunction with
      terminal option `colourtext'%
    }{See the gnuplot documentation for explanation.%
    }{Either use 'blacktext' in gnuplot or load the package
      color.sty in LaTeX.}%
    \renewcommand\color[2][]{}%
  }%
  \providecommand\includegraphics[2][]{%
    \GenericError{(gnuplot) \space\space\space\@spaces}{%
      Package graphicx or graphics not loaded%
    }{See the gnuplot documentation for explanation.%
    }{The gnuplot epslatex terminal needs graphicx.sty or graphics.sty.}%
    \renewcommand\includegraphics[2][]{}%
  }%
  \providecommand\rotatebox[2]{#2}%
  \@ifundefined{ifGPcolor}{%
    \newif\ifGPcolor
    \GPcolortrue
  }{}%
  \@ifundefined{ifGPblacktext}{%
    \newif\ifGPblacktext
    \GPblacktextfalse
  }{}%
  \let\gplgaddtomacro\g@addto@macro
  \gdef\gplbacktext{}%
  \gdef\gplfronttext{}%
  \makeatother
  \ifGPblacktext
    \def\colorrgb#1{}%
    \def\colorgray#1{}%
  \else
    \ifGPcolor
      \def\colorrgb#1{\color[rgb]{#1}}%
      \def\colorgray#1{\color[gray]{#1}}%
      \expandafter\def\csname LTw\endcsname{\color{white}}%
      \expandafter\def\csname LTb\endcsname{\color{black}}%
      \expandafter\def\csname LTa\endcsname{\color{black}}%
      \expandafter\def\csname LT0\endcsname{\color[rgb]{1,0,0}}%
      \expandafter\def\csname LT1\endcsname{\color[rgb]{0,1,0}}%
      \expandafter\def\csname LT2\endcsname{\color[rgb]{0,0,1}}%
      \expandafter\def\csname LT3\endcsname{\color[rgb]{1,0,1}}%
      \expandafter\def\csname LT4\endcsname{\color[rgb]{0,1,1}}%
      \expandafter\def\csname LT5\endcsname{\color[rgb]{1,1,0}}%
      \expandafter\def\csname LT6\endcsname{\color[rgb]{0,0,0}}%
      \expandafter\def\csname LT7\endcsname{\color[rgb]{1,0.3,0}}%
      \expandafter\def\csname LT8\endcsname{\color[rgb]{0.5,0.5,0.5}}%
    \else
      \def\colorrgb#1{\color{black}}%
      \def\colorgray#1{\color[gray]{#1}}%
      \expandafter\def\csname LTw\endcsname{\color{white}}%
      \expandafter\def\csname LTb\endcsname{\color{black}}%
      \expandafter\def\csname LTa\endcsname{\color{black}}%
      \expandafter\def\csname LT0\endcsname{\color{black}}%
      \expandafter\def\csname LT1\endcsname{\color{black}}%
      \expandafter\def\csname LT2\endcsname{\color{black}}%
      \expandafter\def\csname LT3\endcsname{\color{black}}%
      \expandafter\def\csname LT4\endcsname{\color{black}}%
      \expandafter\def\csname LT5\endcsname{\color{black}}%
      \expandafter\def\csname LT6\endcsname{\color{black}}%
      \expandafter\def\csname LT7\endcsname{\color{black}}%
      \expandafter\def\csname LT8\endcsname{\color{black}}%
    \fi
  \fi
    \setlength{\unitlength}{0.0500bp}%
    \ifx\gptboxheight\undefined%
      \newlength{\gptboxheight}%
      \newlength{\gptboxwidth}%
      \newsavebox{\gptboxtext}%
    \fi%
    \setlength{\fboxrule}{0.5pt}%
    \setlength{\fboxsep}{1pt}%
\begin{picture}(6802.00,3614.00)%
    \gplgaddtomacro\gplbacktext{%
      \csname LTb\endcsname%
      \put(814,704){\makebox(0,0)[r]{\strut{}$0$}}%
      \csname LTb\endcsname%
      \put(814,1113){\makebox(0,0)[r]{\strut{}$0.2$}}%
      \csname LTb\endcsname%
      \put(814,1522){\makebox(0,0)[r]{\strut{}$0.4$}}%
      \csname LTb\endcsname%
      \put(814,1931){\makebox(0,0)[r]{\strut{}$0.6$}}%
      \csname LTb\endcsname%
      \put(814,2340){\makebox(0,0)[r]{\strut{}$0.8$}}%
      \csname LTb\endcsname%
      \put(814,2749){\makebox(0,0)[r]{\strut{}$1$}}%
      \csname LTb\endcsname%
      \put(1578,484){\makebox(0,0){\strut{}$10^{3}$}}%
      \csname LTb\endcsname%
      \put(3676,484){\makebox(0,0){\strut{}$10^{4}$}}%
      \csname LTb\endcsname%
      \put(5773,484){\makebox(0,0){\strut{}$10^{5}$}}%
    }%
    \gplgaddtomacro\gplfronttext{%
      \csname LTb\endcsname%
      \put(176,1828){\rotatebox{-270}{\makebox(0,0){\strut{}NMI}}}%
      \put(3675,154){\makebox(0,0){\strut{}Number $n$ of nodes}}%
      \put(3675,3283){\makebox(0,0){\strut{}Mixing: $\mu = 0.4$, Cluster: OSLOM, Overlap: $\nu = 4$}}%
      \csname LTb\endcsname%
      \put(5418,2780){\makebox(0,0)[r]{\strut{}Orig}}%
      \csname LTb\endcsname%
      \put(5418,2560){\makebox(0,0)[r]{\strut{}EM}}%
    }%
    \gplbacktext
    \put(0,0){\includegraphics{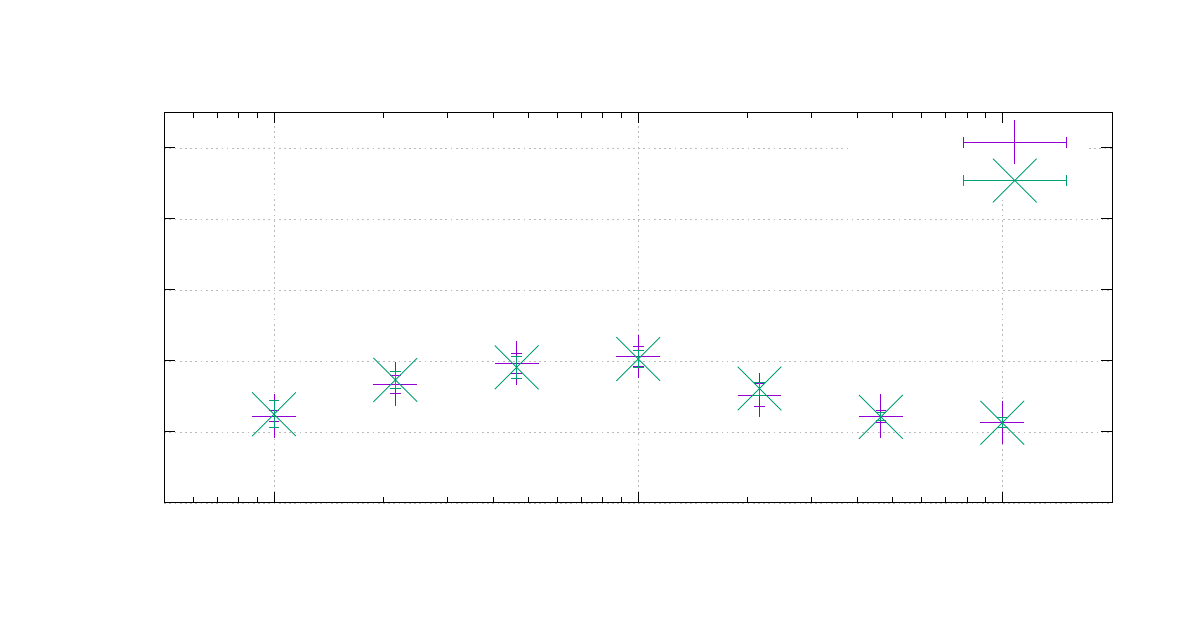}}%
    \gplfronttext
  \end{picture}%
\endgroup
}

  \caption{NMI of OSLOM and ground truth at $\mu = 0.4$ with $2$/$4$ overlapping clusters per node.}
  \label{fig:quality_lfr_overlapping}
\end{figure}

\begin{figure}[t]
	\begin{center}
		\scalebox{0.6}{%
\begingroup
  \makeatletter
  \providecommand\color[2][]{%
    \GenericError{(gnuplot) \space\space\space\@spaces}{%
      Package color not loaded in conjunction with
      terminal option `colourtext'%
    }{See the gnuplot documentation for explanation.%
    }{Either use 'blacktext' in gnuplot or load the package
      color.sty in LaTeX.}%
    \renewcommand\color[2][]{}%
  }%
  \providecommand\includegraphics[2][]{%
    \GenericError{(gnuplot) \space\space\space\@spaces}{%
      Package graphicx or graphics not loaded%
    }{See the gnuplot documentation for explanation.%
    }{The gnuplot epslatex terminal needs graphicx.sty or graphics.sty.}%
    \renewcommand\includegraphics[2][]{}%
  }%
  \providecommand\rotatebox[2]{#2}%
  \@ifundefined{ifGPcolor}{%
    \newif\ifGPcolor
    \GPcolortrue
  }{}%
  \@ifundefined{ifGPblacktext}{%
    \newif\ifGPblacktext
    \GPblacktextfalse
  }{}%
  \let\gplgaddtomacro\g@addto@macro
  \gdef\gplbacktext{}%
  \gdef\gplfronttext{}%
  \makeatother
  \ifGPblacktext
    \def\colorrgb#1{}%
    \def\colorgray#1{}%
  \else
    \ifGPcolor
      \def\colorrgb#1{\color[rgb]{#1}}%
      \def\colorgray#1{\color[gray]{#1}}%
      \expandafter\def\csname LTw\endcsname{\color{white}}%
      \expandafter\def\csname LTb\endcsname{\color{black}}%
      \expandafter\def\csname LTa\endcsname{\color{black}}%
      \expandafter\def\csname LT0\endcsname{\color[rgb]{1,0,0}}%
      \expandafter\def\csname LT1\endcsname{\color[rgb]{0,1,0}}%
      \expandafter\def\csname LT2\endcsname{\color[rgb]{0,0,1}}%
      \expandafter\def\csname LT3\endcsname{\color[rgb]{1,0,1}}%
      \expandafter\def\csname LT4\endcsname{\color[rgb]{0,1,1}}%
      \expandafter\def\csname LT5\endcsname{\color[rgb]{1,1,0}}%
      \expandafter\def\csname LT6\endcsname{\color[rgb]{0,0,0}}%
      \expandafter\def\csname LT7\endcsname{\color[rgb]{1,0.3,0}}%
      \expandafter\def\csname LT8\endcsname{\color[rgb]{0.5,0.5,0.5}}%
    \else
      \def\colorrgb#1{\color{black}}%
      \def\colorgray#1{\color[gray]{#1}}%
      \expandafter\def\csname LTw\endcsname{\color{white}}%
      \expandafter\def\csname LTb\endcsname{\color{black}}%
      \expandafter\def\csname LTa\endcsname{\color{black}}%
      \expandafter\def\csname LT0\endcsname{\color{black}}%
      \expandafter\def\csname LT1\endcsname{\color{black}}%
      \expandafter\def\csname LT2\endcsname{\color{black}}%
      \expandafter\def\csname LT3\endcsname{\color{black}}%
      \expandafter\def\csname LT4\endcsname{\color{black}}%
      \expandafter\def\csname LT5\endcsname{\color{black}}%
      \expandafter\def\csname LT6\endcsname{\color{black}}%
      \expandafter\def\csname LT7\endcsname{\color{black}}%
      \expandafter\def\csname LT8\endcsname{\color{black}}%
    \fi
  \fi
    \setlength{\unitlength}{0.0500bp}%
    \ifx\gptboxheight\undefined%
      \newlength{\gptboxheight}%
      \newlength{\gptboxwidth}%
      \newsavebox{\gptboxtext}%
    \fi%
    \setlength{\fboxrule}{0.5pt}%
    \setlength{\fboxsep}{1pt}%
\begin{picture}(6802.00,3614.00)%
    \gplgaddtomacro\gplbacktext{%
      \csname LTb\endcsname%
      \put(814,704){\makebox(0,0)[r]{\strut{}$0$}}%
      \csname LTb\endcsname%
      \put(814,1113){\makebox(0,0)[r]{\strut{}$0.2$}}%
      \csname LTb\endcsname%
      \put(814,1522){\makebox(0,0)[r]{\strut{}$0.4$}}%
      \csname LTb\endcsname%
      \put(814,1931){\makebox(0,0)[r]{\strut{}$0.6$}}%
      \csname LTb\endcsname%
      \put(814,2340){\makebox(0,0)[r]{\strut{}$0.8$}}%
      \csname LTb\endcsname%
      \put(814,2749){\makebox(0,0)[r]{\strut{}$1$}}%
      \csname LTb\endcsname%
      \put(1234,484){\makebox(0,0){\strut{}$10^{3}$}}%
      \csname LTb\endcsname%
      \put(3096,484){\makebox(0,0){\strut{}$10^{4}$}}%
      \csname LTb\endcsname%
      \put(4957,484){\makebox(0,0){\strut{}$10^{5}$}}%
    }%
    \gplgaddtomacro\gplfronttext{%
      \csname LTb\endcsname%
      \put(176,1828){\rotatebox{-270}{\makebox(0,0){\strut{}Avg. Local Clustering Coeff.}}}%
      \put(3675,154){\makebox(0,0){\strut{}Number $n$ of nodes}}%
      \put(3675,3283){\makebox(0,0){\strut{}Mixing: $\mu = 0.6$}}%
      \csname LTb\endcsname%
      \put(5418,2780){\makebox(0,0)[r]{\strut{}Orig}}%
      \csname LTb\endcsname%
      \put(5418,2560){\makebox(0,0)[r]{\strut{}NetworKit}}%
      \csname LTb\endcsname%
      \put(5418,2340){\makebox(0,0)[r]{\strut{}EM}}%
    }%
    \gplbacktext
    \put(0,0){\includegraphics{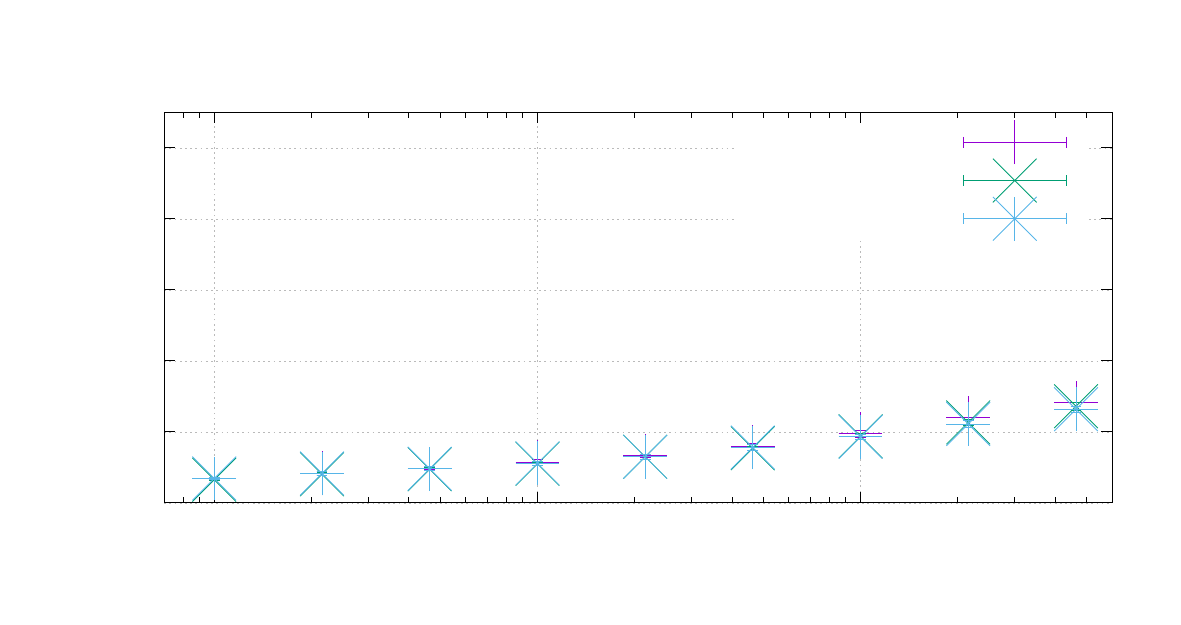}}%
    \gplfronttext
  \end{picture}%
\endgroup
}	
	\end{center}
	
	\caption{Average local clustering coefficient at $\mu = 0.6$ with disjoint clusters.}
	\label{fig:quality_lfr_avgcc}
\end{figure}

Further, we examine the average local clustering coefficient, a measure for the percentage of closed triangles which shows the presence of locally denser areas as expected in communities \cite{k-m-08}.
We report these measures for graphs ranging from $10^3$ to $10^6$ nodes.
In Fig.~\ref{fig:quality_lfr_disjoint}, \ref{fig:quality_lfr_avgcc} and \ref{fig:quality_lfr_overlapping} we present a selection of results; all of them can be found in Appendix~\ref{sec:appendix-lfr-comparison}.
There are only small differences within the range of random noise between the graphs generated by \emlfr{} and the other two implementations.
Note that due to the computational costs above $10^5$ edges, there is only one sample for the original implementation which explains the outliers in Fig.~\ref{fig:quality_lfr_disjoint}.
Similar to the results in \cite{ekgb-ancac-16}, we also observe that the performance of clustering algorithms drops significantly as the graph's size grows.
This might be due to less clearly defined community structures since the parameters are scaled, and also due to limits of current clustering algorithms.
Such behavior clearly demonstrates the necessity of \emlfr{} for being able to study this phenomenon on even larger graphs and develop algorithms that are able to handle such instances.

\section{Outlook and Conclusion}
We propose the first I/O-efficient graph generator for the LFR benchmark and the FDSM, which is the most challenging step involved:
\emhh{} materializes a graph based on a prescribed degree distribution without I/O for virtually all realistic parameters.
Including the generation of a powerlaw degree sequence and the writing of the output to disk, our implementation generates $\num{1.8e8}$ edges per second for graphs exceeding main memory.
\emes{} perturbs existing graphs with $m$ edges based on $k$ edge switches using $\Oh(k/m\cdot \sort(m))$ I/Os for $k = \Omega(m)$.

We demonstrate that \emes{} is faster than the internal memory implementation~\cite{DBLP:journals/corr/abs-cs-0502085} even for large instances still fitting in main memory and scales well beyond the limited main memory.
Compared to the distributed approach by~\cite{DBLP:conf/icpp/BhuiyanCKM14} on a cluster with 128~CPUs, \emes{} exhibits a slow-down of only $8.3$ on one CPU and hence poses a viable and cost-efficient alternative.
Our \emlfr{} implementation is orders of magnitude faster than the original LFR implementation for large instances and scales well to graphs exceeding main memory while the generated graphs are equivalent.
We further gave evidence indicating that commonly accepted parameters to derive the length of the Edge Switching Markov Chain remain valid for graph sizes approaching the external memory domain and that \emcmes{} can be used to accelerate the process. 

Currently, \emes{} does not yet fully exploit the parallelism offered by modern machines.
While dependencies between swaps make a parallelization challenging, preliminary experiments indicate that an extension is possible:
each run can be split further into smaller batches which can be parallelized in the spirit of \cite{DBLP:conf/esa/MeyerS98}.

Another possibility for a speedup could be to use the recently proposed Curveball sampling algorithm for graphs with fixed degree sequence~\cite{Carstens2016}.
Further studies are necessary to establish whether this really leads to a faster sampling in practice.
While the underlying Markov chain seems to require less steps to converge in practice, each step is more expensive.
Also, a combination of Curveball with \emcmes{} seems possible.

Further, this is just the starting point for clustering large graphs that exceed main memory using external memory.
Not only are new clustering algorithms needed, but also the evaluation of the results needs to be investigated since existing evaluation measures might not be easily computable in external memory.

\section*{Acknowledgment}
We thank Hannes Seiwert and Mark Ortmann for valuable discussions on~\emhh.

\clearpage
\bibliographystyle{plainurl}

\appendix
\newpage

\section{Summary of Definitions}
{\newcommand{\mrefto}[1]{ (section~\ref{#1})}
\begin{table}[h]
	\centering
	\begin{tabular}{|p{0.15\columnwidth}|p{0.7\columnwidth}|}\hline
		\emph{Symbol} & \emph{Description} \\\hline
		$[k]$ & $[k] \defrel \{1, \ldots, k\}$ for $k \in \mathbb N_+$\mrefto{sec:notation}\\\hline
		$[u, v]$ & Undirected simple edge with implication $u < v$\mrefto{sec:notation}\\\hline
		$B$ & Number of items in a block transferred between IM and EM\mrefto{ssec:emm}\\\hline
		$d_\text{min}$, $d_\text{max}$ & Min/max degree of nodes in LFR benchmark\mrefto{sec:lfr}\\\hline
		$d^\text{in}_v$ & $d^\text{in}_v = (1{-}\mu) \cdot d_v$, intra-community degree of node $v$\mrefto{sec:lfr}\\\hline
		$\degs$ & $\degs = (d_1, \ldots, d_n)$ with $d_{i} \le d_{i+1} \forall i$. Degree sequence of a graph\mrefto{sec:mat-degree-sequence}\\\hline
		$D_\degs$ & $D_\degs = \big| \{d_i: 1\le i \le n\} \big|$ where $\degs = (d_1, \ldots, d_n)$, degree support\mrefto{sec:mat-degree-sequence}\\\hline		
		$n$ & Number of vertices in a graph\mrefto{sec:notation}\\\hline
		$m$ & Number of edges in a graph\mrefto{sec:notation}\\\hline
		$\mu$ & Mixing parameter in LFR benchmark, i.e. ratio of neighbors that shall be in other communities\mrefto{sec:lfr}\\\hline		
		$M$ & Number of items fitting into internal memory\mrefto{ssec:emm}\\\hline
		$\pld ab\gamma$ & Powerlaw distribution with exponent $-\gamma$ on the interval $[a, b)$\mrefto{sec:notation}\\\hline
		$s_\text{min}$, $s_\text{max}$ & Min/max size of communities in LFR benchmark\mrefto{sec:lfr}\\\hline
		$\scan(n)$ & $\scan(n)=\Theta(n/B)$ I/Os, scan complexity\mrefto{ssec:emm}\\\hline
		$\sort(n)$ & $\sort(n)=\Theta((n/B) \cdot \log_{M/B}(n/B))$ I/Os, sort complexity\mrefto{ssec:emm}\\\hline
	\end{tabular}
	
	\caption{Definitions used in this paper.}
	\label{table:def-summary}
\end{table}
}
 \clearpage
\onecolumn

\section{Comparing LFR Implementations}\label{sec:appendix-lfr-comparison}
\def\threescale{0.36}
	
{
	\noindent\scalebox{\threescale}{%
\begingroup
  \makeatletter
  \providecommand\color[2][]{%
    \GenericError{(gnuplot) \space\space\space\@spaces}{%
      Package color not loaded in conjunction with
      terminal option `colourtext'%
    }{See the gnuplot documentation for explanation.%
    }{Either use 'blacktext' in gnuplot or load the package
      color.sty in LaTeX.}%
    \renewcommand\color[2][]{}%
  }%
  \providecommand\includegraphics[2][]{%
    \GenericError{(gnuplot) \space\space\space\@spaces}{%
      Package graphicx or graphics not loaded%
    }{See the gnuplot documentation for explanation.%
    }{The gnuplot epslatex terminal needs graphicx.sty or graphics.sty.}%
    \renewcommand\includegraphics[2][]{}%
  }%
  \providecommand\rotatebox[2]{#2}%
  \@ifundefined{ifGPcolor}{%
    \newif\ifGPcolor
    \GPcolortrue
  }{}%
  \@ifundefined{ifGPblacktext}{%
    \newif\ifGPblacktext
    \GPblacktextfalse
  }{}%
  \let\gplgaddtomacro\g@addto@macro
  \gdef\gplbacktext{}%
  \gdef\gplfronttext{}%
  \makeatother
  \ifGPblacktext
    \def\colorrgb#1{}%
    \def\colorgray#1{}%
  \else
    \ifGPcolor
      \def\colorrgb#1{\color[rgb]{#1}}%
      \def\colorgray#1{\color[gray]{#1}}%
      \expandafter\def\csname LTw\endcsname{\color{white}}%
      \expandafter\def\csname LTb\endcsname{\color{black}}%
      \expandafter\def\csname LTa\endcsname{\color{black}}%
      \expandafter\def\csname LT0\endcsname{\color[rgb]{1,0,0}}%
      \expandafter\def\csname LT1\endcsname{\color[rgb]{0,1,0}}%
      \expandafter\def\csname LT2\endcsname{\color[rgb]{0,0,1}}%
      \expandafter\def\csname LT3\endcsname{\color[rgb]{1,0,1}}%
      \expandafter\def\csname LT4\endcsname{\color[rgb]{0,1,1}}%
      \expandafter\def\csname LT5\endcsname{\color[rgb]{1,1,0}}%
      \expandafter\def\csname LT6\endcsname{\color[rgb]{0,0,0}}%
      \expandafter\def\csname LT7\endcsname{\color[rgb]{1,0.3,0}}%
      \expandafter\def\csname LT8\endcsname{\color[rgb]{0.5,0.5,0.5}}%
    \else
      \def\colorrgb#1{\color{black}}%
      \def\colorgray#1{\color[gray]{#1}}%
      \expandafter\def\csname LTw\endcsname{\color{white}}%
      \expandafter\def\csname LTb\endcsname{\color{black}}%
      \expandafter\def\csname LTa\endcsname{\color{black}}%
      \expandafter\def\csname LT0\endcsname{\color{black}}%
      \expandafter\def\csname LT1\endcsname{\color{black}}%
      \expandafter\def\csname LT2\endcsname{\color{black}}%
      \expandafter\def\csname LT3\endcsname{\color{black}}%
      \expandafter\def\csname LT4\endcsname{\color{black}}%
      \expandafter\def\csname LT5\endcsname{\color{black}}%
      \expandafter\def\csname LT6\endcsname{\color{black}}%
      \expandafter\def\csname LT7\endcsname{\color{black}}%
      \expandafter\def\csname LT8\endcsname{\color{black}}%
    \fi
  \fi
    \setlength{\unitlength}{0.0500bp}%
    \ifx\gptboxheight\undefined%
      \newlength{\gptboxheight}%
      \newlength{\gptboxwidth}%
      \newsavebox{\gptboxtext}%
    \fi%
    \setlength{\fboxrule}{0.5pt}%
    \setlength{\fboxsep}{1pt}%
\begin{picture}(6802.00,3614.00)%
    \gplgaddtomacro\gplbacktext{%
      \csname LTb\endcsname%
      \put(814,704){\makebox(0,0)[r]{\strut{}$0$}}%
      \csname LTb\endcsname%
      \put(814,1113){\makebox(0,0)[r]{\strut{}$0.2$}}%
      \csname LTb\endcsname%
      \put(814,1522){\makebox(0,0)[r]{\strut{}$0.4$}}%
      \csname LTb\endcsname%
      \put(814,1931){\makebox(0,0)[r]{\strut{}$0.6$}}%
      \csname LTb\endcsname%
      \put(814,2340){\makebox(0,0)[r]{\strut{}$0.8$}}%
      \csname LTb\endcsname%
      \put(814,2749){\makebox(0,0)[r]{\strut{}$1$}}%
      \csname LTb\endcsname%
      \put(1234,484){\makebox(0,0){\strut{}$10^{3}$}}%
      \csname LTb\endcsname%
      \put(3096,484){\makebox(0,0){\strut{}$10^{4}$}}%
      \csname LTb\endcsname%
      \put(4957,484){\makebox(0,0){\strut{}$10^{5}$}}%
    }%
    \gplgaddtomacro\gplfronttext{%
      \csname LTb\endcsname%
      \put(176,1828){\rotatebox{-270}{\makebox(0,0){\strut{}AR}}}%
      \put(3675,154){\makebox(0,0){\strut{}Number $n$ of nodes}}%
      \put(3675,3283){\makebox(0,0){\strut{}Mixing $\mu = 0.2$, Cluster: Infomap}}%
      \csname LTb\endcsname%
      \put(2266,1317){\makebox(0,0)[r]{\strut{}Orig}}%
      \csname LTb\endcsname%
      \put(2266,1097){\makebox(0,0)[r]{\strut{}NetworKit}}%
      \csname LTb\endcsname%
      \put(2266,877){\makebox(0,0)[r]{\strut{}EM}}%
    }%
    \gplbacktext
    \put(0,0){\includegraphics{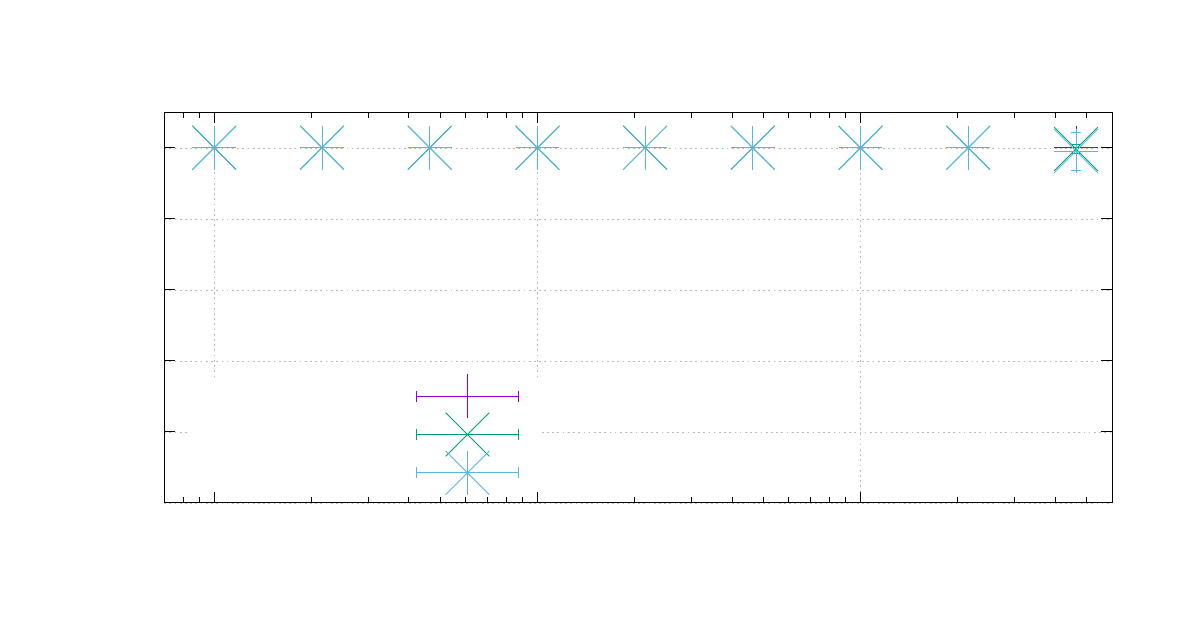}}%
    \gplfronttext
  \end{picture}%
\endgroup
}\hfill\scalebox{\threescale}{%
\begingroup
  \makeatletter
  \providecommand\color[2][]{%
    \GenericError{(gnuplot) \space\space\space\@spaces}{%
      Package color not loaded in conjunction with
      terminal option `colourtext'%
    }{See the gnuplot documentation for explanation.%
    }{Either use 'blacktext' in gnuplot or load the package
      color.sty in LaTeX.}%
    \renewcommand\color[2][]{}%
  }%
  \providecommand\includegraphics[2][]{%
    \GenericError{(gnuplot) \space\space\space\@spaces}{%
      Package graphicx or graphics not loaded%
    }{See the gnuplot documentation for explanation.%
    }{The gnuplot epslatex terminal needs graphicx.sty or graphics.sty.}%
    \renewcommand\includegraphics[2][]{}%
  }%
  \providecommand\rotatebox[2]{#2}%
  \@ifundefined{ifGPcolor}{%
    \newif\ifGPcolor
    \GPcolortrue
  }{}%
  \@ifundefined{ifGPblacktext}{%
    \newif\ifGPblacktext
    \GPblacktextfalse
  }{}%
  \let\gplgaddtomacro\g@addto@macro
  \gdef\gplbacktext{}%
  \gdef\gplfronttext{}%
  \makeatother
  \ifGPblacktext
    \def\colorrgb#1{}%
    \def\colorgray#1{}%
  \else
    \ifGPcolor
      \def\colorrgb#1{\color[rgb]{#1}}%
      \def\colorgray#1{\color[gray]{#1}}%
      \expandafter\def\csname LTw\endcsname{\color{white}}%
      \expandafter\def\csname LTb\endcsname{\color{black}}%
      \expandafter\def\csname LTa\endcsname{\color{black}}%
      \expandafter\def\csname LT0\endcsname{\color[rgb]{1,0,0}}%
      \expandafter\def\csname LT1\endcsname{\color[rgb]{0,1,0}}%
      \expandafter\def\csname LT2\endcsname{\color[rgb]{0,0,1}}%
      \expandafter\def\csname LT3\endcsname{\color[rgb]{1,0,1}}%
      \expandafter\def\csname LT4\endcsname{\color[rgb]{0,1,1}}%
      \expandafter\def\csname LT5\endcsname{\color[rgb]{1,1,0}}%
      \expandafter\def\csname LT6\endcsname{\color[rgb]{0,0,0}}%
      \expandafter\def\csname LT7\endcsname{\color[rgb]{1,0.3,0}}%
      \expandafter\def\csname LT8\endcsname{\color[rgb]{0.5,0.5,0.5}}%
    \else
      \def\colorrgb#1{\color{black}}%
      \def\colorgray#1{\color[gray]{#1}}%
      \expandafter\def\csname LTw\endcsname{\color{white}}%
      \expandafter\def\csname LTb\endcsname{\color{black}}%
      \expandafter\def\csname LTa\endcsname{\color{black}}%
      \expandafter\def\csname LT0\endcsname{\color{black}}%
      \expandafter\def\csname LT1\endcsname{\color{black}}%
      \expandafter\def\csname LT2\endcsname{\color{black}}%
      \expandafter\def\csname LT3\endcsname{\color{black}}%
      \expandafter\def\csname LT4\endcsname{\color{black}}%
      \expandafter\def\csname LT5\endcsname{\color{black}}%
      \expandafter\def\csname LT6\endcsname{\color{black}}%
      \expandafter\def\csname LT7\endcsname{\color{black}}%
      \expandafter\def\csname LT8\endcsname{\color{black}}%
    \fi
  \fi
    \setlength{\unitlength}{0.0500bp}%
    \ifx\gptboxheight\undefined%
      \newlength{\gptboxheight}%
      \newlength{\gptboxwidth}%
      \newsavebox{\gptboxtext}%
    \fi%
    \setlength{\fboxrule}{0.5pt}%
    \setlength{\fboxsep}{1pt}%
\begin{picture}(6802.00,3614.00)%
    \gplgaddtomacro\gplbacktext{%
      \csname LTb\endcsname%
      \put(814,704){\makebox(0,0)[r]{\strut{}$0$}}%
      \csname LTb\endcsname%
      \put(814,1113){\makebox(0,0)[r]{\strut{}$0.2$}}%
      \csname LTb\endcsname%
      \put(814,1522){\makebox(0,0)[r]{\strut{}$0.4$}}%
      \csname LTb\endcsname%
      \put(814,1931){\makebox(0,0)[r]{\strut{}$0.6$}}%
      \csname LTb\endcsname%
      \put(814,2340){\makebox(0,0)[r]{\strut{}$0.8$}}%
      \csname LTb\endcsname%
      \put(814,2749){\makebox(0,0)[r]{\strut{}$1$}}%
      \csname LTb\endcsname%
      \put(1234,484){\makebox(0,0){\strut{}$10^{3}$}}%
      \csname LTb\endcsname%
      \put(3096,484){\makebox(0,0){\strut{}$10^{4}$}}%
      \csname LTb\endcsname%
      \put(4957,484){\makebox(0,0){\strut{}$10^{5}$}}%
    }%
    \gplgaddtomacro\gplfronttext{%
      \csname LTb\endcsname%
      \put(176,1828){\rotatebox{-270}{\makebox(0,0){\strut{}AR}}}%
      \put(3675,154){\makebox(0,0){\strut{}Number $n$ of nodes}}%
      \put(3675,3283){\makebox(0,0){\strut{}Mixing $\mu = 0.4$, Cluster: Infomap}}%
      \csname LTb\endcsname%
      \put(2266,1317){\makebox(0,0)[r]{\strut{}Orig}}%
      \csname LTb\endcsname%
      \put(2266,1097){\makebox(0,0)[r]{\strut{}NetworKit}}%
      \csname LTb\endcsname%
      \put(2266,877){\makebox(0,0)[r]{\strut{}EM}}%
    }%
    \gplbacktext
    \put(0,0){\includegraphics{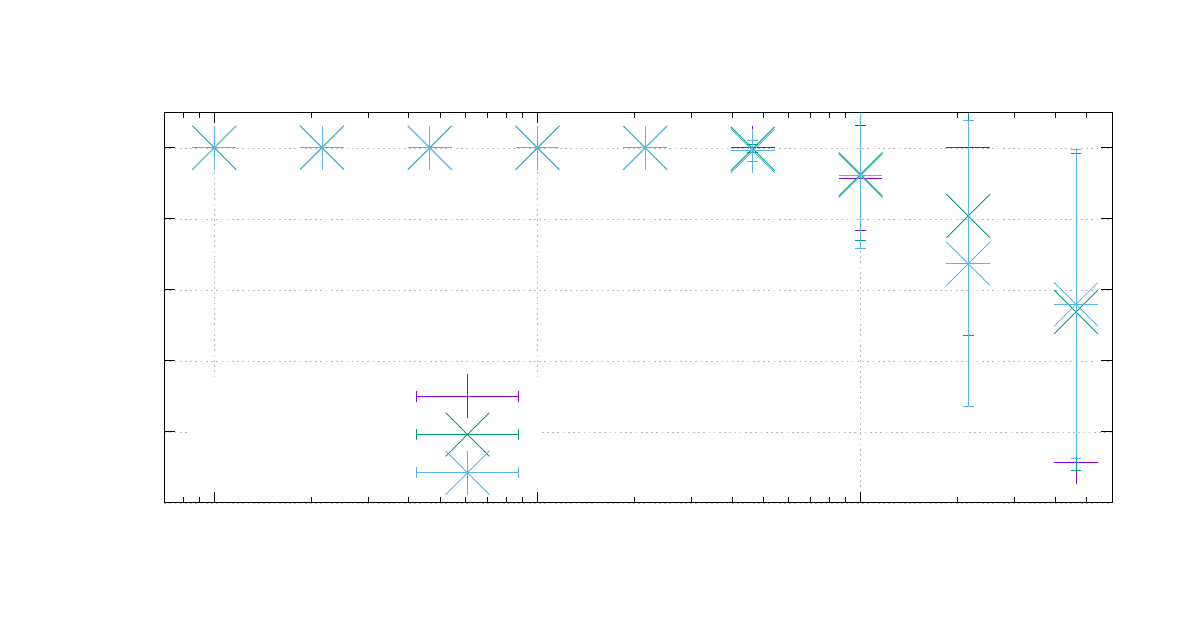}}%
    \gplfronttext
  \end{picture}%
\endgroup
}\hfill\scalebox{\threescale}{%
\begingroup
  \makeatletter
  \providecommand\color[2][]{%
    \GenericError{(gnuplot) \space\space\space\@spaces}{%
      Package color not loaded in conjunction with
      terminal option `colourtext'%
    }{See the gnuplot documentation for explanation.%
    }{Either use 'blacktext' in gnuplot or load the package
      color.sty in LaTeX.}%
    \renewcommand\color[2][]{}%
  }%
  \providecommand\includegraphics[2][]{%
    \GenericError{(gnuplot) \space\space\space\@spaces}{%
      Package graphicx or graphics not loaded%
    }{See the gnuplot documentation for explanation.%
    }{The gnuplot epslatex terminal needs graphicx.sty or graphics.sty.}%
    \renewcommand\includegraphics[2][]{}%
  }%
  \providecommand\rotatebox[2]{#2}%
  \@ifundefined{ifGPcolor}{%
    \newif\ifGPcolor
    \GPcolortrue
  }{}%
  \@ifundefined{ifGPblacktext}{%
    \newif\ifGPblacktext
    \GPblacktextfalse
  }{}%
  \let\gplgaddtomacro\g@addto@macro
  \gdef\gplbacktext{}%
  \gdef\gplfronttext{}%
  \makeatother
  \ifGPblacktext
    \def\colorrgb#1{}%
    \def\colorgray#1{}%
  \else
    \ifGPcolor
      \def\colorrgb#1{\color[rgb]{#1}}%
      \def\colorgray#1{\color[gray]{#1}}%
      \expandafter\def\csname LTw\endcsname{\color{white}}%
      \expandafter\def\csname LTb\endcsname{\color{black}}%
      \expandafter\def\csname LTa\endcsname{\color{black}}%
      \expandafter\def\csname LT0\endcsname{\color[rgb]{1,0,0}}%
      \expandafter\def\csname LT1\endcsname{\color[rgb]{0,1,0}}%
      \expandafter\def\csname LT2\endcsname{\color[rgb]{0,0,1}}%
      \expandafter\def\csname LT3\endcsname{\color[rgb]{1,0,1}}%
      \expandafter\def\csname LT4\endcsname{\color[rgb]{0,1,1}}%
      \expandafter\def\csname LT5\endcsname{\color[rgb]{1,1,0}}%
      \expandafter\def\csname LT6\endcsname{\color[rgb]{0,0,0}}%
      \expandafter\def\csname LT7\endcsname{\color[rgb]{1,0.3,0}}%
      \expandafter\def\csname LT8\endcsname{\color[rgb]{0.5,0.5,0.5}}%
    \else
      \def\colorrgb#1{\color{black}}%
      \def\colorgray#1{\color[gray]{#1}}%
      \expandafter\def\csname LTw\endcsname{\color{white}}%
      \expandafter\def\csname LTb\endcsname{\color{black}}%
      \expandafter\def\csname LTa\endcsname{\color{black}}%
      \expandafter\def\csname LT0\endcsname{\color{black}}%
      \expandafter\def\csname LT1\endcsname{\color{black}}%
      \expandafter\def\csname LT2\endcsname{\color{black}}%
      \expandafter\def\csname LT3\endcsname{\color{black}}%
      \expandafter\def\csname LT4\endcsname{\color{black}}%
      \expandafter\def\csname LT5\endcsname{\color{black}}%
      \expandafter\def\csname LT6\endcsname{\color{black}}%
      \expandafter\def\csname LT7\endcsname{\color{black}}%
      \expandafter\def\csname LT8\endcsname{\color{black}}%
    \fi
  \fi
    \setlength{\unitlength}{0.0500bp}%
    \ifx\gptboxheight\undefined%
      \newlength{\gptboxheight}%
      \newlength{\gptboxwidth}%
      \newsavebox{\gptboxtext}%
    \fi%
    \setlength{\fboxrule}{0.5pt}%
    \setlength{\fboxsep}{1pt}%
\begin{picture}(6802.00,3614.00)%
    \gplgaddtomacro\gplbacktext{%
      \csname LTb\endcsname%
      \put(814,704){\makebox(0,0)[r]{\strut{}$0$}}%
      \csname LTb\endcsname%
      \put(814,1113){\makebox(0,0)[r]{\strut{}$0.2$}}%
      \csname LTb\endcsname%
      \put(814,1522){\makebox(0,0)[r]{\strut{}$0.4$}}%
      \csname LTb\endcsname%
      \put(814,1931){\makebox(0,0)[r]{\strut{}$0.6$}}%
      \csname LTb\endcsname%
      \put(814,2340){\makebox(0,0)[r]{\strut{}$0.8$}}%
      \csname LTb\endcsname%
      \put(814,2749){\makebox(0,0)[r]{\strut{}$1$}}%
      \csname LTb\endcsname%
      \put(1234,484){\makebox(0,0){\strut{}$10^{3}$}}%
      \csname LTb\endcsname%
      \put(3096,484){\makebox(0,0){\strut{}$10^{4}$}}%
      \csname LTb\endcsname%
      \put(4957,484){\makebox(0,0){\strut{}$10^{5}$}}%
    }%
    \gplgaddtomacro\gplfronttext{%
      \csname LTb\endcsname%
      \put(176,1828){\rotatebox{-270}{\makebox(0,0){\strut{}AR}}}%
      \put(3675,154){\makebox(0,0){\strut{}Number $n$ of nodes}}%
      \put(3675,3283){\makebox(0,0){\strut{}Mixing $\mu = 0.6$, Cluster: Infomap}}%
      \csname LTb\endcsname%
      \put(2266,1317){\makebox(0,0)[r]{\strut{}Orig}}%
      \csname LTb\endcsname%
      \put(2266,1097){\makebox(0,0)[r]{\strut{}NetworKit}}%
      \csname LTb\endcsname%
      \put(2266,877){\makebox(0,0)[r]{\strut{}EM}}%
    }%
    \gplbacktext
    \put(0,0){\includegraphics{lfr_no_Infomap_AR_6}}%
    \gplfronttext
  \end{picture}%
\endgroup
}\hfill\\ %
	\noindent\scalebox{\threescale}{%
\begingroup
  \makeatletter
  \providecommand\color[2][]{%
    \GenericError{(gnuplot) \space\space\space\@spaces}{%
      Package color not loaded in conjunction with
      terminal option `colourtext'%
    }{See the gnuplot documentation for explanation.%
    }{Either use 'blacktext' in gnuplot or load the package
      color.sty in LaTeX.}%
    \renewcommand\color[2][]{}%
  }%
  \providecommand\includegraphics[2][]{%
    \GenericError{(gnuplot) \space\space\space\@spaces}{%
      Package graphicx or graphics not loaded%
    }{See the gnuplot documentation for explanation.%
    }{The gnuplot epslatex terminal needs graphicx.sty or graphics.sty.}%
    \renewcommand\includegraphics[2][]{}%
  }%
  \providecommand\rotatebox[2]{#2}%
  \@ifundefined{ifGPcolor}{%
    \newif\ifGPcolor
    \GPcolortrue
  }{}%
  \@ifundefined{ifGPblacktext}{%
    \newif\ifGPblacktext
    \GPblacktextfalse
  }{}%
  \let\gplgaddtomacro\g@addto@macro
  \gdef\gplbacktext{}%
  \gdef\gplfronttext{}%
  \makeatother
  \ifGPblacktext
    \def\colorrgb#1{}%
    \def\colorgray#1{}%
  \else
    \ifGPcolor
      \def\colorrgb#1{\color[rgb]{#1}}%
      \def\colorgray#1{\color[gray]{#1}}%
      \expandafter\def\csname LTw\endcsname{\color{white}}%
      \expandafter\def\csname LTb\endcsname{\color{black}}%
      \expandafter\def\csname LTa\endcsname{\color{black}}%
      \expandafter\def\csname LT0\endcsname{\color[rgb]{1,0,0}}%
      \expandafter\def\csname LT1\endcsname{\color[rgb]{0,1,0}}%
      \expandafter\def\csname LT2\endcsname{\color[rgb]{0,0,1}}%
      \expandafter\def\csname LT3\endcsname{\color[rgb]{1,0,1}}%
      \expandafter\def\csname LT4\endcsname{\color[rgb]{0,1,1}}%
      \expandafter\def\csname LT5\endcsname{\color[rgb]{1,1,0}}%
      \expandafter\def\csname LT6\endcsname{\color[rgb]{0,0,0}}%
      \expandafter\def\csname LT7\endcsname{\color[rgb]{1,0.3,0}}%
      \expandafter\def\csname LT8\endcsname{\color[rgb]{0.5,0.5,0.5}}%
    \else
      \def\colorrgb#1{\color{black}}%
      \def\colorgray#1{\color[gray]{#1}}%
      \expandafter\def\csname LTw\endcsname{\color{white}}%
      \expandafter\def\csname LTb\endcsname{\color{black}}%
      \expandafter\def\csname LTa\endcsname{\color{black}}%
      \expandafter\def\csname LT0\endcsname{\color{black}}%
      \expandafter\def\csname LT1\endcsname{\color{black}}%
      \expandafter\def\csname LT2\endcsname{\color{black}}%
      \expandafter\def\csname LT3\endcsname{\color{black}}%
      \expandafter\def\csname LT4\endcsname{\color{black}}%
      \expandafter\def\csname LT5\endcsname{\color{black}}%
      \expandafter\def\csname LT6\endcsname{\color{black}}%
      \expandafter\def\csname LT7\endcsname{\color{black}}%
      \expandafter\def\csname LT8\endcsname{\color{black}}%
    \fi
  \fi
    \setlength{\unitlength}{0.0500bp}%
    \ifx\gptboxheight\undefined%
      \newlength{\gptboxheight}%
      \newlength{\gptboxwidth}%
      \newsavebox{\gptboxtext}%
    \fi%
    \setlength{\fboxrule}{0.5pt}%
    \setlength{\fboxsep}{1pt}%
\begin{picture}(6802.00,3614.00)%
    \gplgaddtomacro\gplbacktext{%
      \csname LTb\endcsname%
      \put(814,704){\makebox(0,0)[r]{\strut{}$0$}}%
      \csname LTb\endcsname%
      \put(814,1113){\makebox(0,0)[r]{\strut{}$0.2$}}%
      \csname LTb\endcsname%
      \put(814,1522){\makebox(0,0)[r]{\strut{}$0.4$}}%
      \csname LTb\endcsname%
      \put(814,1931){\makebox(0,0)[r]{\strut{}$0.6$}}%
      \csname LTb\endcsname%
      \put(814,2340){\makebox(0,0)[r]{\strut{}$0.8$}}%
      \csname LTb\endcsname%
      \put(814,2749){\makebox(0,0)[r]{\strut{}$1$}}%
      \csname LTb\endcsname%
      \put(1234,484){\makebox(0,0){\strut{}$10^{3}$}}%
      \csname LTb\endcsname%
      \put(3096,484){\makebox(0,0){\strut{}$10^{4}$}}%
      \csname LTb\endcsname%
      \put(4957,484){\makebox(0,0){\strut{}$10^{5}$}}%
    }%
    \gplgaddtomacro\gplfronttext{%
      \csname LTb\endcsname%
      \put(176,1828){\rotatebox{-270}{\makebox(0,0){\strut{}NMI}}}%
      \put(3675,154){\makebox(0,0){\strut{}Number $n$ of nodes}}%
      \put(3675,3283){\makebox(0,0){\strut{}Mixing $\mu = 0.2$, Cluster: Infomap}}%
      \csname LTb\endcsname%
      \put(2266,1317){\makebox(0,0)[r]{\strut{}Orig}}%
      \csname LTb\endcsname%
      \put(2266,1097){\makebox(0,0)[r]{\strut{}NetworKit}}%
      \csname LTb\endcsname%
      \put(2266,877){\makebox(0,0)[r]{\strut{}EM}}%
    }%
    \gplbacktext
    \put(0,0){\includegraphics{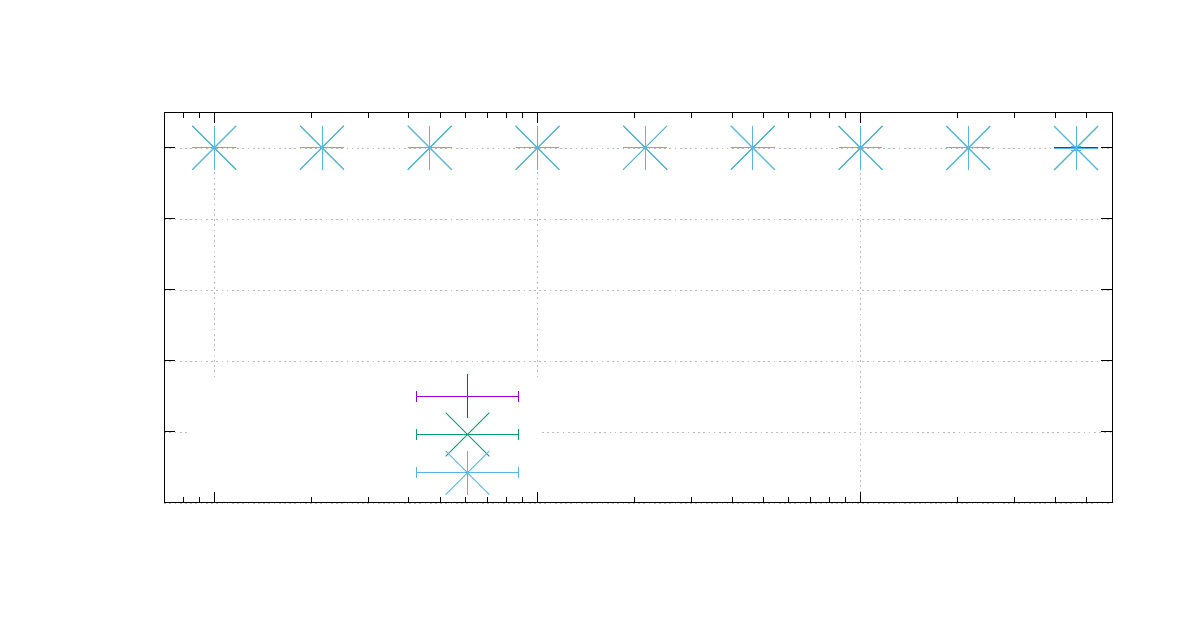}}%
    \gplfronttext
  \end{picture}%
\endgroup
}\hfill\scalebox{\threescale}{%
\begingroup
  \makeatletter
  \providecommand\color[2][]{%
    \GenericError{(gnuplot) \space\space\space\@spaces}{%
      Package color not loaded in conjunction with
      terminal option `colourtext'%
    }{See the gnuplot documentation for explanation.%
    }{Either use 'blacktext' in gnuplot or load the package
      color.sty in LaTeX.}%
    \renewcommand\color[2][]{}%
  }%
  \providecommand\includegraphics[2][]{%
    \GenericError{(gnuplot) \space\space\space\@spaces}{%
      Package graphicx or graphics not loaded%
    }{See the gnuplot documentation for explanation.%
    }{The gnuplot epslatex terminal needs graphicx.sty or graphics.sty.}%
    \renewcommand\includegraphics[2][]{}%
  }%
  \providecommand\rotatebox[2]{#2}%
  \@ifundefined{ifGPcolor}{%
    \newif\ifGPcolor
    \GPcolortrue
  }{}%
  \@ifundefined{ifGPblacktext}{%
    \newif\ifGPblacktext
    \GPblacktextfalse
  }{}%
  \let\gplgaddtomacro\g@addto@macro
  \gdef\gplbacktext{}%
  \gdef\gplfronttext{}%
  \makeatother
  \ifGPblacktext
    \def\colorrgb#1{}%
    \def\colorgray#1{}%
  \else
    \ifGPcolor
      \def\colorrgb#1{\color[rgb]{#1}}%
      \def\colorgray#1{\color[gray]{#1}}%
      \expandafter\def\csname LTw\endcsname{\color{white}}%
      \expandafter\def\csname LTb\endcsname{\color{black}}%
      \expandafter\def\csname LTa\endcsname{\color{black}}%
      \expandafter\def\csname LT0\endcsname{\color[rgb]{1,0,0}}%
      \expandafter\def\csname LT1\endcsname{\color[rgb]{0,1,0}}%
      \expandafter\def\csname LT2\endcsname{\color[rgb]{0,0,1}}%
      \expandafter\def\csname LT3\endcsname{\color[rgb]{1,0,1}}%
      \expandafter\def\csname LT4\endcsname{\color[rgb]{0,1,1}}%
      \expandafter\def\csname LT5\endcsname{\color[rgb]{1,1,0}}%
      \expandafter\def\csname LT6\endcsname{\color[rgb]{0,0,0}}%
      \expandafter\def\csname LT7\endcsname{\color[rgb]{1,0.3,0}}%
      \expandafter\def\csname LT8\endcsname{\color[rgb]{0.5,0.5,0.5}}%
    \else
      \def\colorrgb#1{\color{black}}%
      \def\colorgray#1{\color[gray]{#1}}%
      \expandafter\def\csname LTw\endcsname{\color{white}}%
      \expandafter\def\csname LTb\endcsname{\color{black}}%
      \expandafter\def\csname LTa\endcsname{\color{black}}%
      \expandafter\def\csname LT0\endcsname{\color{black}}%
      \expandafter\def\csname LT1\endcsname{\color{black}}%
      \expandafter\def\csname LT2\endcsname{\color{black}}%
      \expandafter\def\csname LT3\endcsname{\color{black}}%
      \expandafter\def\csname LT4\endcsname{\color{black}}%
      \expandafter\def\csname LT5\endcsname{\color{black}}%
      \expandafter\def\csname LT6\endcsname{\color{black}}%
      \expandafter\def\csname LT7\endcsname{\color{black}}%
      \expandafter\def\csname LT8\endcsname{\color{black}}%
    \fi
  \fi
    \setlength{\unitlength}{0.0500bp}%
    \ifx\gptboxheight\undefined%
      \newlength{\gptboxheight}%
      \newlength{\gptboxwidth}%
      \newsavebox{\gptboxtext}%
    \fi%
    \setlength{\fboxrule}{0.5pt}%
    \setlength{\fboxsep}{1pt}%
\begin{picture}(6802.00,3614.00)%
    \gplgaddtomacro\gplbacktext{%
      \csname LTb\endcsname%
      \put(814,704){\makebox(0,0)[r]{\strut{}$0$}}%
      \csname LTb\endcsname%
      \put(814,1113){\makebox(0,0)[r]{\strut{}$0.2$}}%
      \csname LTb\endcsname%
      \put(814,1522){\makebox(0,0)[r]{\strut{}$0.4$}}%
      \csname LTb\endcsname%
      \put(814,1931){\makebox(0,0)[r]{\strut{}$0.6$}}%
      \csname LTb\endcsname%
      \put(814,2340){\makebox(0,0)[r]{\strut{}$0.8$}}%
      \csname LTb\endcsname%
      \put(814,2749){\makebox(0,0)[r]{\strut{}$1$}}%
      \csname LTb\endcsname%
      \put(1234,484){\makebox(0,0){\strut{}$10^{3}$}}%
      \csname LTb\endcsname%
      \put(3096,484){\makebox(0,0){\strut{}$10^{4}$}}%
      \csname LTb\endcsname%
      \put(4957,484){\makebox(0,0){\strut{}$10^{5}$}}%
    }%
    \gplgaddtomacro\gplfronttext{%
      \csname LTb\endcsname%
      \put(176,1828){\rotatebox{-270}{\makebox(0,0){\strut{}NMI}}}%
      \put(3675,154){\makebox(0,0){\strut{}Number $n$ of nodes}}%
      \put(3675,3283){\makebox(0,0){\strut{}Mixing $\mu = 0.4$, Cluster: Infomap}}%
      \csname LTb\endcsname%
      \put(2266,1317){\makebox(0,0)[r]{\strut{}Orig}}%
      \csname LTb\endcsname%
      \put(2266,1097){\makebox(0,0)[r]{\strut{}NetworKit}}%
      \csname LTb\endcsname%
      \put(2266,877){\makebox(0,0)[r]{\strut{}EM}}%
    }%
    \gplbacktext
    \put(0,0){\includegraphics{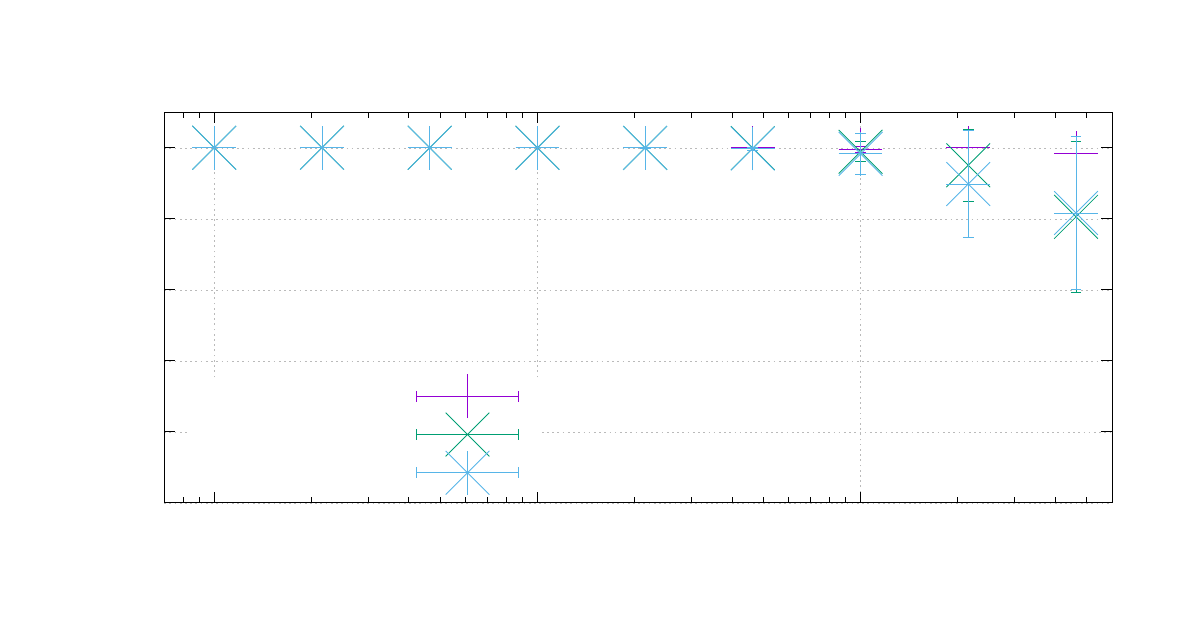}}%
    \gplfronttext
  \end{picture}%
\endgroup
}\hfill\scalebox{\threescale}{%
\begingroup
  \makeatletter
  \providecommand\color[2][]{%
    \GenericError{(gnuplot) \space\space\space\@spaces}{%
      Package color not loaded in conjunction with
      terminal option `colourtext'%
    }{See the gnuplot documentation for explanation.%
    }{Either use 'blacktext' in gnuplot or load the package
      color.sty in LaTeX.}%
    \renewcommand\color[2][]{}%
  }%
  \providecommand\includegraphics[2][]{%
    \GenericError{(gnuplot) \space\space\space\@spaces}{%
      Package graphicx or graphics not loaded%
    }{See the gnuplot documentation for explanation.%
    }{The gnuplot epslatex terminal needs graphicx.sty or graphics.sty.}%
    \renewcommand\includegraphics[2][]{}%
  }%
  \providecommand\rotatebox[2]{#2}%
  \@ifundefined{ifGPcolor}{%
    \newif\ifGPcolor
    \GPcolortrue
  }{}%
  \@ifundefined{ifGPblacktext}{%
    \newif\ifGPblacktext
    \GPblacktextfalse
  }{}%
  \let\gplgaddtomacro\g@addto@macro
  \gdef\gplbacktext{}%
  \gdef\gplfronttext{}%
  \makeatother
  \ifGPblacktext
    \def\colorrgb#1{}%
    \def\colorgray#1{}%
  \else
    \ifGPcolor
      \def\colorrgb#1{\color[rgb]{#1}}%
      \def\colorgray#1{\color[gray]{#1}}%
      \expandafter\def\csname LTw\endcsname{\color{white}}%
      \expandafter\def\csname LTb\endcsname{\color{black}}%
      \expandafter\def\csname LTa\endcsname{\color{black}}%
      \expandafter\def\csname LT0\endcsname{\color[rgb]{1,0,0}}%
      \expandafter\def\csname LT1\endcsname{\color[rgb]{0,1,0}}%
      \expandafter\def\csname LT2\endcsname{\color[rgb]{0,0,1}}%
      \expandafter\def\csname LT3\endcsname{\color[rgb]{1,0,1}}%
      \expandafter\def\csname LT4\endcsname{\color[rgb]{0,1,1}}%
      \expandafter\def\csname LT5\endcsname{\color[rgb]{1,1,0}}%
      \expandafter\def\csname LT6\endcsname{\color[rgb]{0,0,0}}%
      \expandafter\def\csname LT7\endcsname{\color[rgb]{1,0.3,0}}%
      \expandafter\def\csname LT8\endcsname{\color[rgb]{0.5,0.5,0.5}}%
    \else
      \def\colorrgb#1{\color{black}}%
      \def\colorgray#1{\color[gray]{#1}}%
      \expandafter\def\csname LTw\endcsname{\color{white}}%
      \expandafter\def\csname LTb\endcsname{\color{black}}%
      \expandafter\def\csname LTa\endcsname{\color{black}}%
      \expandafter\def\csname LT0\endcsname{\color{black}}%
      \expandafter\def\csname LT1\endcsname{\color{black}}%
      \expandafter\def\csname LT2\endcsname{\color{black}}%
      \expandafter\def\csname LT3\endcsname{\color{black}}%
      \expandafter\def\csname LT4\endcsname{\color{black}}%
      \expandafter\def\csname LT5\endcsname{\color{black}}%
      \expandafter\def\csname LT6\endcsname{\color{black}}%
      \expandafter\def\csname LT7\endcsname{\color{black}}%
      \expandafter\def\csname LT8\endcsname{\color{black}}%
    \fi
  \fi
    \setlength{\unitlength}{0.0500bp}%
    \ifx\gptboxheight\undefined%
      \newlength{\gptboxheight}%
      \newlength{\gptboxwidth}%
      \newsavebox{\gptboxtext}%
    \fi%
    \setlength{\fboxrule}{0.5pt}%
    \setlength{\fboxsep}{1pt}%
\begin{picture}(6802.00,3614.00)%
    \gplgaddtomacro\gplbacktext{%
      \csname LTb\endcsname%
      \put(814,704){\makebox(0,0)[r]{\strut{}$0$}}%
      \csname LTb\endcsname%
      \put(814,1113){\makebox(0,0)[r]{\strut{}$0.2$}}%
      \csname LTb\endcsname%
      \put(814,1522){\makebox(0,0)[r]{\strut{}$0.4$}}%
      \csname LTb\endcsname%
      \put(814,1931){\makebox(0,0)[r]{\strut{}$0.6$}}%
      \csname LTb\endcsname%
      \put(814,2340){\makebox(0,0)[r]{\strut{}$0.8$}}%
      \csname LTb\endcsname%
      \put(814,2749){\makebox(0,0)[r]{\strut{}$1$}}%
      \csname LTb\endcsname%
      \put(1234,484){\makebox(0,0){\strut{}$10^{3}$}}%
      \csname LTb\endcsname%
      \put(3096,484){\makebox(0,0){\strut{}$10^{4}$}}%
      \csname LTb\endcsname%
      \put(4957,484){\makebox(0,0){\strut{}$10^{5}$}}%
    }%
    \gplgaddtomacro\gplfronttext{%
      \csname LTb\endcsname%
      \put(176,1828){\rotatebox{-270}{\makebox(0,0){\strut{}NMI}}}%
      \put(3675,154){\makebox(0,0){\strut{}Number $n$ of nodes}}%
      \put(3675,3283){\makebox(0,0){\strut{}Mixing $\mu = 0.6$, Cluster: Infomap}}%
      \csname LTb\endcsname%
      \put(2266,1317){\makebox(0,0)[r]{\strut{}Orig}}%
      \csname LTb\endcsname%
      \put(2266,1097){\makebox(0,0)[r]{\strut{}NetworKit}}%
      \csname LTb\endcsname%
      \put(2266,877){\makebox(0,0)[r]{\strut{}EM}}%
    }%
    \gplbacktext
    \put(0,0){\includegraphics{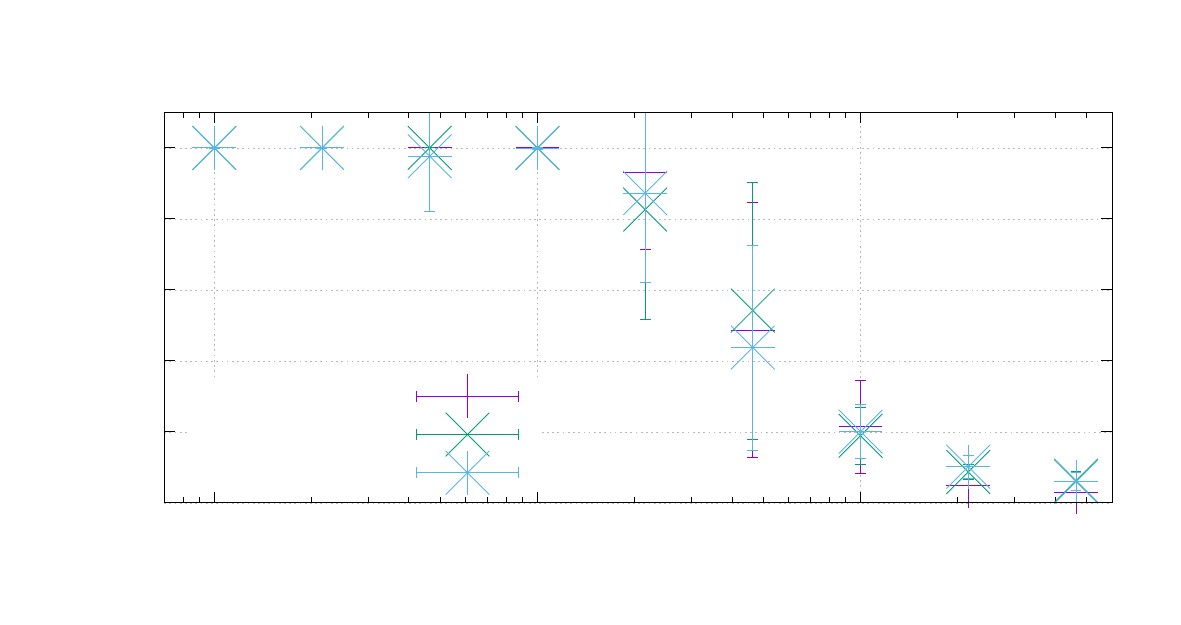}}%
    \gplfronttext
  \end{picture}%
\endgroup
}\hfill\\ %
	
	\vspace{0.1em}
	
	\noindent\scalebox{\threescale}{%
\begingroup
  \makeatletter
  \providecommand\color[2][]{%
    \GenericError{(gnuplot) \space\space\space\@spaces}{%
      Package color not loaded in conjunction with
      terminal option `colourtext'%
    }{See the gnuplot documentation for explanation.%
    }{Either use 'blacktext' in gnuplot or load the package
      color.sty in LaTeX.}%
    \renewcommand\color[2][]{}%
  }%
  \providecommand\includegraphics[2][]{%
    \GenericError{(gnuplot) \space\space\space\@spaces}{%
      Package graphicx or graphics not loaded%
    }{See the gnuplot documentation for explanation.%
    }{The gnuplot epslatex terminal needs graphicx.sty or graphics.sty.}%
    \renewcommand\includegraphics[2][]{}%
  }%
  \providecommand\rotatebox[2]{#2}%
  \@ifundefined{ifGPcolor}{%
    \newif\ifGPcolor
    \GPcolortrue
  }{}%
  \@ifundefined{ifGPblacktext}{%
    \newif\ifGPblacktext
    \GPblacktextfalse
  }{}%
  \let\gplgaddtomacro\g@addto@macro
  \gdef\gplbacktext{}%
  \gdef\gplfronttext{}%
  \makeatother
  \ifGPblacktext
    \def\colorrgb#1{}%
    \def\colorgray#1{}%
  \else
    \ifGPcolor
      \def\colorrgb#1{\color[rgb]{#1}}%
      \def\colorgray#1{\color[gray]{#1}}%
      \expandafter\def\csname LTw\endcsname{\color{white}}%
      \expandafter\def\csname LTb\endcsname{\color{black}}%
      \expandafter\def\csname LTa\endcsname{\color{black}}%
      \expandafter\def\csname LT0\endcsname{\color[rgb]{1,0,0}}%
      \expandafter\def\csname LT1\endcsname{\color[rgb]{0,1,0}}%
      \expandafter\def\csname LT2\endcsname{\color[rgb]{0,0,1}}%
      \expandafter\def\csname LT3\endcsname{\color[rgb]{1,0,1}}%
      \expandafter\def\csname LT4\endcsname{\color[rgb]{0,1,1}}%
      \expandafter\def\csname LT5\endcsname{\color[rgb]{1,1,0}}%
      \expandafter\def\csname LT6\endcsname{\color[rgb]{0,0,0}}%
      \expandafter\def\csname LT7\endcsname{\color[rgb]{1,0.3,0}}%
      \expandafter\def\csname LT8\endcsname{\color[rgb]{0.5,0.5,0.5}}%
    \else
      \def\colorrgb#1{\color{black}}%
      \def\colorgray#1{\color[gray]{#1}}%
      \expandafter\def\csname LTw\endcsname{\color{white}}%
      \expandafter\def\csname LTb\endcsname{\color{black}}%
      \expandafter\def\csname LTa\endcsname{\color{black}}%
      \expandafter\def\csname LT0\endcsname{\color{black}}%
      \expandafter\def\csname LT1\endcsname{\color{black}}%
      \expandafter\def\csname LT2\endcsname{\color{black}}%
      \expandafter\def\csname LT3\endcsname{\color{black}}%
      \expandafter\def\csname LT4\endcsname{\color{black}}%
      \expandafter\def\csname LT5\endcsname{\color{black}}%
      \expandafter\def\csname LT6\endcsname{\color{black}}%
      \expandafter\def\csname LT7\endcsname{\color{black}}%
      \expandafter\def\csname LT8\endcsname{\color{black}}%
    \fi
  \fi
    \setlength{\unitlength}{0.0500bp}%
    \ifx\gptboxheight\undefined%
      \newlength{\gptboxheight}%
      \newlength{\gptboxwidth}%
      \newsavebox{\gptboxtext}%
    \fi%
    \setlength{\fboxrule}{0.5pt}%
    \setlength{\fboxsep}{1pt}%
\begin{picture}(6802.00,3614.00)%
    \gplgaddtomacro\gplbacktext{%
      \csname LTb\endcsname%
      \put(814,704){\makebox(0,0)[r]{\strut{}$0$}}%
      \csname LTb\endcsname%
      \put(814,1113){\makebox(0,0)[r]{\strut{}$0.2$}}%
      \csname LTb\endcsname%
      \put(814,1522){\makebox(0,0)[r]{\strut{}$0.4$}}%
      \csname LTb\endcsname%
      \put(814,1931){\makebox(0,0)[r]{\strut{}$0.6$}}%
      \csname LTb\endcsname%
      \put(814,2340){\makebox(0,0)[r]{\strut{}$0.8$}}%
      \csname LTb\endcsname%
      \put(814,2749){\makebox(0,0)[r]{\strut{}$1$}}%
      \csname LTb\endcsname%
      \put(1234,484){\makebox(0,0){\strut{}$10^{3}$}}%
      \csname LTb\endcsname%
      \put(3096,484){\makebox(0,0){\strut{}$10^{4}$}}%
      \csname LTb\endcsname%
      \put(4957,484){\makebox(0,0){\strut{}$10^{5}$}}%
    }%
    \gplgaddtomacro\gplfronttext{%
      \csname LTb\endcsname%
      \put(176,1828){\rotatebox{-270}{\makebox(0,0){\strut{}AR}}}%
      \put(3675,154){\makebox(0,0){\strut{}Number $n$ of nodes}}%
      \put(3675,3283){\makebox(0,0){\strut{}Mixing $\mu = 0.2$, Cluster: Louvain}}%
      \csname LTb\endcsname%
      \put(2266,1317){\makebox(0,0)[r]{\strut{}Orig}}%
      \csname LTb\endcsname%
      \put(2266,1097){\makebox(0,0)[r]{\strut{}NetworKit}}%
      \csname LTb\endcsname%
      \put(2266,877){\makebox(0,0)[r]{\strut{}EM}}%
    }%
    \gplbacktext
    \put(0,0){\includegraphics{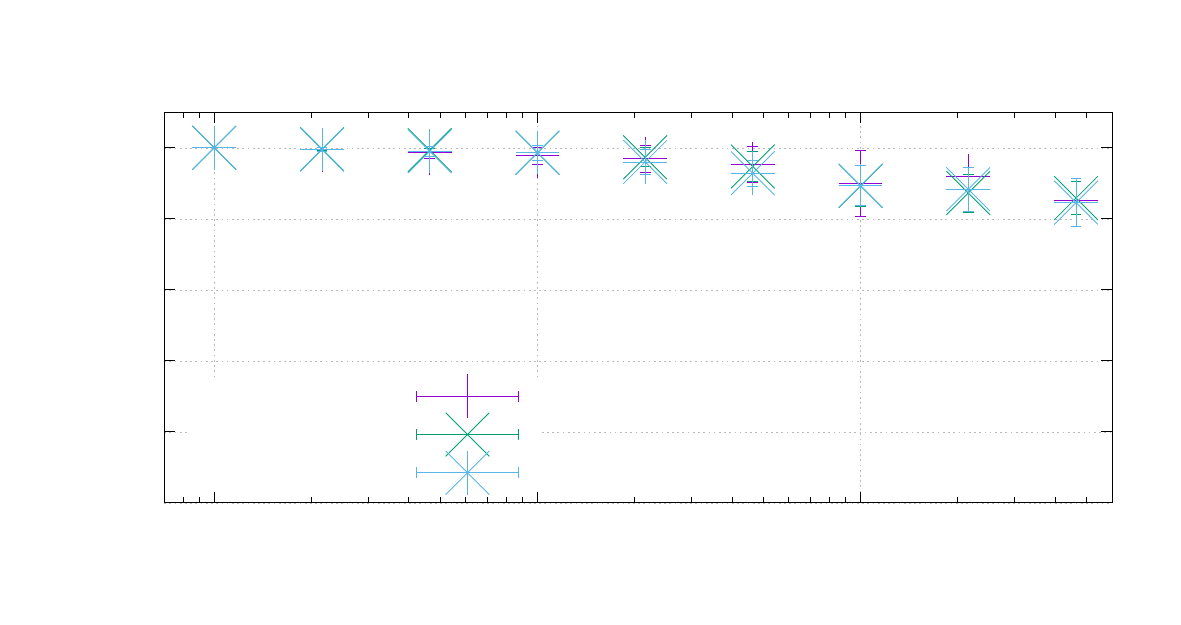}}%
    \gplfronttext
  \end{picture}%
\endgroup
}\hfill\scalebox{\threescale}{%
\begingroup
  \makeatletter
  \providecommand\color[2][]{%
    \GenericError{(gnuplot) \space\space\space\@spaces}{%
      Package color not loaded in conjunction with
      terminal option `colourtext'%
    }{See the gnuplot documentation for explanation.%
    }{Either use 'blacktext' in gnuplot or load the package
      color.sty in LaTeX.}%
    \renewcommand\color[2][]{}%
  }%
  \providecommand\includegraphics[2][]{%
    \GenericError{(gnuplot) \space\space\space\@spaces}{%
      Package graphicx or graphics not loaded%
    }{See the gnuplot documentation for explanation.%
    }{The gnuplot epslatex terminal needs graphicx.sty or graphics.sty.}%
    \renewcommand\includegraphics[2][]{}%
  }%
  \providecommand\rotatebox[2]{#2}%
  \@ifundefined{ifGPcolor}{%
    \newif\ifGPcolor
    \GPcolortrue
  }{}%
  \@ifundefined{ifGPblacktext}{%
    \newif\ifGPblacktext
    \GPblacktextfalse
  }{}%
  \let\gplgaddtomacro\g@addto@macro
  \gdef\gplbacktext{}%
  \gdef\gplfronttext{}%
  \makeatother
  \ifGPblacktext
    \def\colorrgb#1{}%
    \def\colorgray#1{}%
  \else
    \ifGPcolor
      \def\colorrgb#1{\color[rgb]{#1}}%
      \def\colorgray#1{\color[gray]{#1}}%
      \expandafter\def\csname LTw\endcsname{\color{white}}%
      \expandafter\def\csname LTb\endcsname{\color{black}}%
      \expandafter\def\csname LTa\endcsname{\color{black}}%
      \expandafter\def\csname LT0\endcsname{\color[rgb]{1,0,0}}%
      \expandafter\def\csname LT1\endcsname{\color[rgb]{0,1,0}}%
      \expandafter\def\csname LT2\endcsname{\color[rgb]{0,0,1}}%
      \expandafter\def\csname LT3\endcsname{\color[rgb]{1,0,1}}%
      \expandafter\def\csname LT4\endcsname{\color[rgb]{0,1,1}}%
      \expandafter\def\csname LT5\endcsname{\color[rgb]{1,1,0}}%
      \expandafter\def\csname LT6\endcsname{\color[rgb]{0,0,0}}%
      \expandafter\def\csname LT7\endcsname{\color[rgb]{1,0.3,0}}%
      \expandafter\def\csname LT8\endcsname{\color[rgb]{0.5,0.5,0.5}}%
    \else
      \def\colorrgb#1{\color{black}}%
      \def\colorgray#1{\color[gray]{#1}}%
      \expandafter\def\csname LTw\endcsname{\color{white}}%
      \expandafter\def\csname LTb\endcsname{\color{black}}%
      \expandafter\def\csname LTa\endcsname{\color{black}}%
      \expandafter\def\csname LT0\endcsname{\color{black}}%
      \expandafter\def\csname LT1\endcsname{\color{black}}%
      \expandafter\def\csname LT2\endcsname{\color{black}}%
      \expandafter\def\csname LT3\endcsname{\color{black}}%
      \expandafter\def\csname LT4\endcsname{\color{black}}%
      \expandafter\def\csname LT5\endcsname{\color{black}}%
      \expandafter\def\csname LT6\endcsname{\color{black}}%
      \expandafter\def\csname LT7\endcsname{\color{black}}%
      \expandafter\def\csname LT8\endcsname{\color{black}}%
    \fi
  \fi
    \setlength{\unitlength}{0.0500bp}%
    \ifx\gptboxheight\undefined%
      \newlength{\gptboxheight}%
      \newlength{\gptboxwidth}%
      \newsavebox{\gptboxtext}%
    \fi%
    \setlength{\fboxrule}{0.5pt}%
    \setlength{\fboxsep}{1pt}%
\begin{picture}(6802.00,3614.00)%
    \gplgaddtomacro\gplbacktext{%
      \csname LTb\endcsname%
      \put(814,704){\makebox(0,0)[r]{\strut{}$0$}}%
      \csname LTb\endcsname%
      \put(814,1113){\makebox(0,0)[r]{\strut{}$0.2$}}%
      \csname LTb\endcsname%
      \put(814,1522){\makebox(0,0)[r]{\strut{}$0.4$}}%
      \csname LTb\endcsname%
      \put(814,1931){\makebox(0,0)[r]{\strut{}$0.6$}}%
      \csname LTb\endcsname%
      \put(814,2340){\makebox(0,0)[r]{\strut{}$0.8$}}%
      \csname LTb\endcsname%
      \put(814,2749){\makebox(0,0)[r]{\strut{}$1$}}%
      \csname LTb\endcsname%
      \put(1234,484){\makebox(0,0){\strut{}$10^{3}$}}%
      \csname LTb\endcsname%
      \put(3096,484){\makebox(0,0){\strut{}$10^{4}$}}%
      \csname LTb\endcsname%
      \put(4957,484){\makebox(0,0){\strut{}$10^{5}$}}%
    }%
    \gplgaddtomacro\gplfronttext{%
      \csname LTb\endcsname%
      \put(176,1828){\rotatebox{-270}{\makebox(0,0){\strut{}AR}}}%
      \put(3675,154){\makebox(0,0){\strut{}Number $n$ of nodes}}%
      \put(3675,3283){\makebox(0,0){\strut{}Mixing $\mu = 0.4$, Cluster: Louvain}}%
      \csname LTb\endcsname%
      \put(2266,1317){\makebox(0,0)[r]{\strut{}Orig}}%
      \csname LTb\endcsname%
      \put(2266,1097){\makebox(0,0)[r]{\strut{}NetworKit}}%
      \csname LTb\endcsname%
      \put(2266,877){\makebox(0,0)[r]{\strut{}EM}}%
    }%
    \gplbacktext
    \put(0,0){\includegraphics{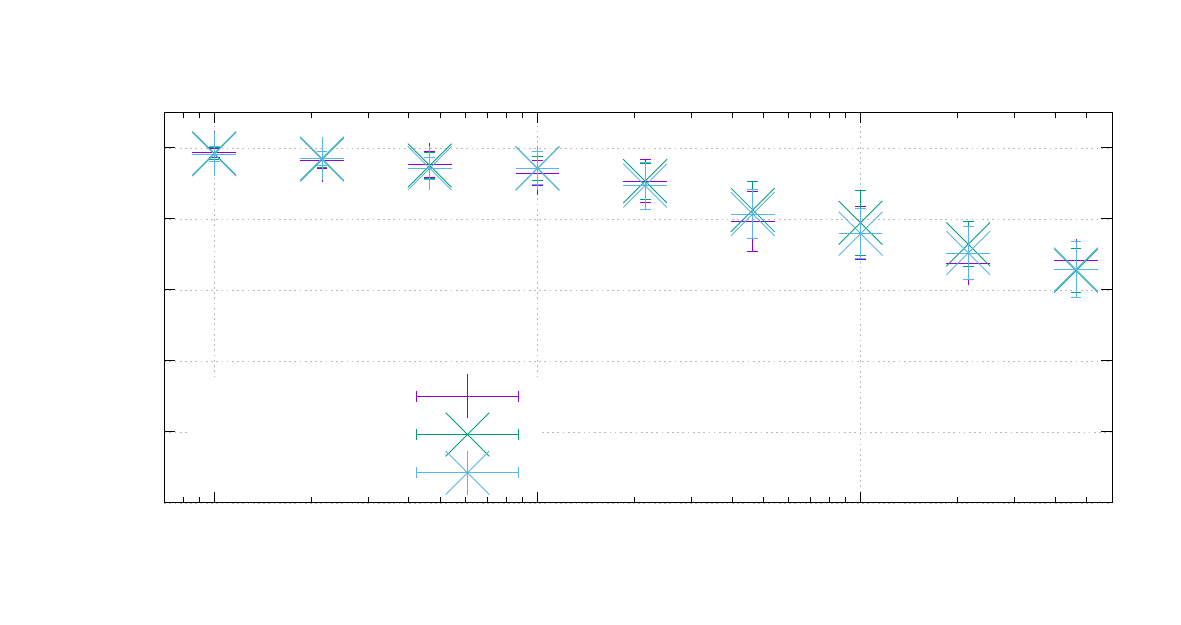}}%
    \gplfronttext
  \end{picture}%
\endgroup
}\hfill\scalebox{\threescale}{%
\begingroup
  \makeatletter
  \providecommand\color[2][]{%
    \GenericError{(gnuplot) \space\space\space\@spaces}{%
      Package color not loaded in conjunction with
      terminal option `colourtext'%
    }{See the gnuplot documentation for explanation.%
    }{Either use 'blacktext' in gnuplot or load the package
      color.sty in LaTeX.}%
    \renewcommand\color[2][]{}%
  }%
  \providecommand\includegraphics[2][]{%
    \GenericError{(gnuplot) \space\space\space\@spaces}{%
      Package graphicx or graphics not loaded%
    }{See the gnuplot documentation for explanation.%
    }{The gnuplot epslatex terminal needs graphicx.sty or graphics.sty.}%
    \renewcommand\includegraphics[2][]{}%
  }%
  \providecommand\rotatebox[2]{#2}%
  \@ifundefined{ifGPcolor}{%
    \newif\ifGPcolor
    \GPcolortrue
  }{}%
  \@ifundefined{ifGPblacktext}{%
    \newif\ifGPblacktext
    \GPblacktextfalse
  }{}%
  \let\gplgaddtomacro\g@addto@macro
  \gdef\gplbacktext{}%
  \gdef\gplfronttext{}%
  \makeatother
  \ifGPblacktext
    \def\colorrgb#1{}%
    \def\colorgray#1{}%
  \else
    \ifGPcolor
      \def\colorrgb#1{\color[rgb]{#1}}%
      \def\colorgray#1{\color[gray]{#1}}%
      \expandafter\def\csname LTw\endcsname{\color{white}}%
      \expandafter\def\csname LTb\endcsname{\color{black}}%
      \expandafter\def\csname LTa\endcsname{\color{black}}%
      \expandafter\def\csname LT0\endcsname{\color[rgb]{1,0,0}}%
      \expandafter\def\csname LT1\endcsname{\color[rgb]{0,1,0}}%
      \expandafter\def\csname LT2\endcsname{\color[rgb]{0,0,1}}%
      \expandafter\def\csname LT3\endcsname{\color[rgb]{1,0,1}}%
      \expandafter\def\csname LT4\endcsname{\color[rgb]{0,1,1}}%
      \expandafter\def\csname LT5\endcsname{\color[rgb]{1,1,0}}%
      \expandafter\def\csname LT6\endcsname{\color[rgb]{0,0,0}}%
      \expandafter\def\csname LT7\endcsname{\color[rgb]{1,0.3,0}}%
      \expandafter\def\csname LT8\endcsname{\color[rgb]{0.5,0.5,0.5}}%
    \else
      \def\colorrgb#1{\color{black}}%
      \def\colorgray#1{\color[gray]{#1}}%
      \expandafter\def\csname LTw\endcsname{\color{white}}%
      \expandafter\def\csname LTb\endcsname{\color{black}}%
      \expandafter\def\csname LTa\endcsname{\color{black}}%
      \expandafter\def\csname LT0\endcsname{\color{black}}%
      \expandafter\def\csname LT1\endcsname{\color{black}}%
      \expandafter\def\csname LT2\endcsname{\color{black}}%
      \expandafter\def\csname LT3\endcsname{\color{black}}%
      \expandafter\def\csname LT4\endcsname{\color{black}}%
      \expandafter\def\csname LT5\endcsname{\color{black}}%
      \expandafter\def\csname LT6\endcsname{\color{black}}%
      \expandafter\def\csname LT7\endcsname{\color{black}}%
      \expandafter\def\csname LT8\endcsname{\color{black}}%
    \fi
  \fi
    \setlength{\unitlength}{0.0500bp}%
    \ifx\gptboxheight\undefined%
      \newlength{\gptboxheight}%
      \newlength{\gptboxwidth}%
      \newsavebox{\gptboxtext}%
    \fi%
    \setlength{\fboxrule}{0.5pt}%
    \setlength{\fboxsep}{1pt}%
\begin{picture}(6802.00,3614.00)%
    \gplgaddtomacro\gplbacktext{%
      \csname LTb\endcsname%
      \put(814,704){\makebox(0,0)[r]{\strut{}$0$}}%
      \csname LTb\endcsname%
      \put(814,1113){\makebox(0,0)[r]{\strut{}$0.2$}}%
      \csname LTb\endcsname%
      \put(814,1522){\makebox(0,0)[r]{\strut{}$0.4$}}%
      \csname LTb\endcsname%
      \put(814,1931){\makebox(0,0)[r]{\strut{}$0.6$}}%
      \csname LTb\endcsname%
      \put(814,2340){\makebox(0,0)[r]{\strut{}$0.8$}}%
      \csname LTb\endcsname%
      \put(814,2749){\makebox(0,0)[r]{\strut{}$1$}}%
      \csname LTb\endcsname%
      \put(1234,484){\makebox(0,0){\strut{}$10^{3}$}}%
      \csname LTb\endcsname%
      \put(3096,484){\makebox(0,0){\strut{}$10^{4}$}}%
      \csname LTb\endcsname%
      \put(4957,484){\makebox(0,0){\strut{}$10^{5}$}}%
    }%
    \gplgaddtomacro\gplfronttext{%
      \csname LTb\endcsname%
      \put(176,1828){\rotatebox{-270}{\makebox(0,0){\strut{}AR}}}%
      \put(3675,154){\makebox(0,0){\strut{}Number $n$ of nodes}}%
      \put(3675,3283){\makebox(0,0){\strut{}Mixing $\mu = 0.6$, Cluster: Louvain}}%
      \csname LTb\endcsname%
      \put(2266,1317){\makebox(0,0)[r]{\strut{}Orig}}%
      \csname LTb\endcsname%
      \put(2266,1097){\makebox(0,0)[r]{\strut{}NetworKit}}%
      \csname LTb\endcsname%
      \put(2266,877){\makebox(0,0)[r]{\strut{}EM}}%
    }%
    \gplbacktext
    \put(0,0){\includegraphics{lfr_no_Louvain_AR_6}}%
    \gplfronttext
  \end{picture}%
\endgroup
}\hfill\\ %
	\noindent\scalebox{\threescale}{%
\begingroup
  \makeatletter
  \providecommand\color[2][]{%
    \GenericError{(gnuplot) \space\space\space\@spaces}{%
      Package color not loaded in conjunction with
      terminal option `colourtext'%
    }{See the gnuplot documentation for explanation.%
    }{Either use 'blacktext' in gnuplot or load the package
      color.sty in LaTeX.}%
    \renewcommand\color[2][]{}%
  }%
  \providecommand\includegraphics[2][]{%
    \GenericError{(gnuplot) \space\space\space\@spaces}{%
      Package graphicx or graphics not loaded%
    }{See the gnuplot documentation for explanation.%
    }{The gnuplot epslatex terminal needs graphicx.sty or graphics.sty.}%
    \renewcommand\includegraphics[2][]{}%
  }%
  \providecommand\rotatebox[2]{#2}%
  \@ifundefined{ifGPcolor}{%
    \newif\ifGPcolor
    \GPcolortrue
  }{}%
  \@ifundefined{ifGPblacktext}{%
    \newif\ifGPblacktext
    \GPblacktextfalse
  }{}%
  \let\gplgaddtomacro\g@addto@macro
  \gdef\gplbacktext{}%
  \gdef\gplfronttext{}%
  \makeatother
  \ifGPblacktext
    \def\colorrgb#1{}%
    \def\colorgray#1{}%
  \else
    \ifGPcolor
      \def\colorrgb#1{\color[rgb]{#1}}%
      \def\colorgray#1{\color[gray]{#1}}%
      \expandafter\def\csname LTw\endcsname{\color{white}}%
      \expandafter\def\csname LTb\endcsname{\color{black}}%
      \expandafter\def\csname LTa\endcsname{\color{black}}%
      \expandafter\def\csname LT0\endcsname{\color[rgb]{1,0,0}}%
      \expandafter\def\csname LT1\endcsname{\color[rgb]{0,1,0}}%
      \expandafter\def\csname LT2\endcsname{\color[rgb]{0,0,1}}%
      \expandafter\def\csname LT3\endcsname{\color[rgb]{1,0,1}}%
      \expandafter\def\csname LT4\endcsname{\color[rgb]{0,1,1}}%
      \expandafter\def\csname LT5\endcsname{\color[rgb]{1,1,0}}%
      \expandafter\def\csname LT6\endcsname{\color[rgb]{0,0,0}}%
      \expandafter\def\csname LT7\endcsname{\color[rgb]{1,0.3,0}}%
      \expandafter\def\csname LT8\endcsname{\color[rgb]{0.5,0.5,0.5}}%
    \else
      \def\colorrgb#1{\color{black}}%
      \def\colorgray#1{\color[gray]{#1}}%
      \expandafter\def\csname LTw\endcsname{\color{white}}%
      \expandafter\def\csname LTb\endcsname{\color{black}}%
      \expandafter\def\csname LTa\endcsname{\color{black}}%
      \expandafter\def\csname LT0\endcsname{\color{black}}%
      \expandafter\def\csname LT1\endcsname{\color{black}}%
      \expandafter\def\csname LT2\endcsname{\color{black}}%
      \expandafter\def\csname LT3\endcsname{\color{black}}%
      \expandafter\def\csname LT4\endcsname{\color{black}}%
      \expandafter\def\csname LT5\endcsname{\color{black}}%
      \expandafter\def\csname LT6\endcsname{\color{black}}%
      \expandafter\def\csname LT7\endcsname{\color{black}}%
      \expandafter\def\csname LT8\endcsname{\color{black}}%
    \fi
  \fi
    \setlength{\unitlength}{0.0500bp}%
    \ifx\gptboxheight\undefined%
      \newlength{\gptboxheight}%
      \newlength{\gptboxwidth}%
      \newsavebox{\gptboxtext}%
    \fi%
    \setlength{\fboxrule}{0.5pt}%
    \setlength{\fboxsep}{1pt}%
\begin{picture}(6802.00,3614.00)%
    \gplgaddtomacro\gplbacktext{%
      \csname LTb\endcsname%
      \put(814,704){\makebox(0,0)[r]{\strut{}$0$}}%
      \csname LTb\endcsname%
      \put(814,1113){\makebox(0,0)[r]{\strut{}$0.2$}}%
      \csname LTb\endcsname%
      \put(814,1522){\makebox(0,0)[r]{\strut{}$0.4$}}%
      \csname LTb\endcsname%
      \put(814,1931){\makebox(0,0)[r]{\strut{}$0.6$}}%
      \csname LTb\endcsname%
      \put(814,2340){\makebox(0,0)[r]{\strut{}$0.8$}}%
      \csname LTb\endcsname%
      \put(814,2749){\makebox(0,0)[r]{\strut{}$1$}}%
      \csname LTb\endcsname%
      \put(1234,484){\makebox(0,0){\strut{}$10^{3}$}}%
      \csname LTb\endcsname%
      \put(3096,484){\makebox(0,0){\strut{}$10^{4}$}}%
      \csname LTb\endcsname%
      \put(4957,484){\makebox(0,0){\strut{}$10^{5}$}}%
    }%
    \gplgaddtomacro\gplfronttext{%
      \csname LTb\endcsname%
      \put(176,1828){\rotatebox{-270}{\makebox(0,0){\strut{}NMI}}}%
      \put(3675,154){\makebox(0,0){\strut{}Number $n$ of nodes}}%
      \put(3675,3283){\makebox(0,0){\strut{}Mixing $\mu = 0.2$, Cluster: Louvain}}%
      \csname LTb\endcsname%
      \put(2266,1317){\makebox(0,0)[r]{\strut{}Orig}}%
      \csname LTb\endcsname%
      \put(2266,1097){\makebox(0,0)[r]{\strut{}NetworKit}}%
      \csname LTb\endcsname%
      \put(2266,877){\makebox(0,0)[r]{\strut{}EM}}%
    }%
    \gplbacktext
    \put(0,0){\includegraphics{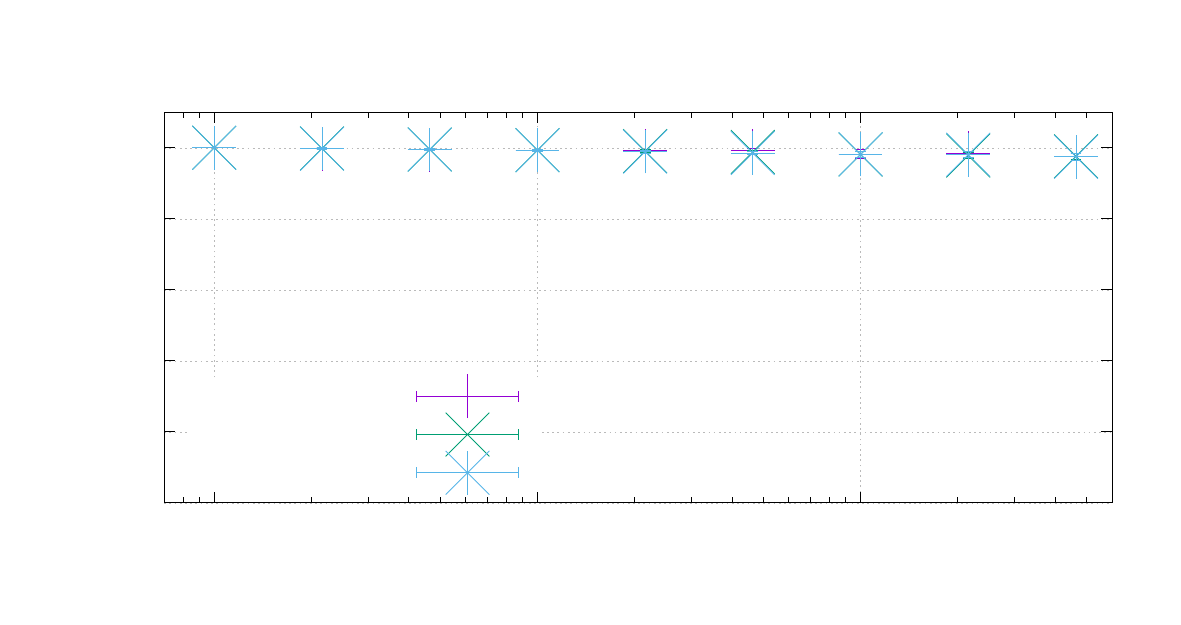}}%
    \gplfronttext
  \end{picture}%
\endgroup
}\hfill\scalebox{\threescale}{%
\begingroup
  \makeatletter
  \providecommand\color[2][]{%
    \GenericError{(gnuplot) \space\space\space\@spaces}{%
      Package color not loaded in conjunction with
      terminal option `colourtext'%
    }{See the gnuplot documentation for explanation.%
    }{Either use 'blacktext' in gnuplot or load the package
      color.sty in LaTeX.}%
    \renewcommand\color[2][]{}%
  }%
  \providecommand\includegraphics[2][]{%
    \GenericError{(gnuplot) \space\space\space\@spaces}{%
      Package graphicx or graphics not loaded%
    }{See the gnuplot documentation for explanation.%
    }{The gnuplot epslatex terminal needs graphicx.sty or graphics.sty.}%
    \renewcommand\includegraphics[2][]{}%
  }%
  \providecommand\rotatebox[2]{#2}%
  \@ifundefined{ifGPcolor}{%
    \newif\ifGPcolor
    \GPcolortrue
  }{}%
  \@ifundefined{ifGPblacktext}{%
    \newif\ifGPblacktext
    \GPblacktextfalse
  }{}%
  \let\gplgaddtomacro\g@addto@macro
  \gdef\gplbacktext{}%
  \gdef\gplfronttext{}%
  \makeatother
  \ifGPblacktext
    \def\colorrgb#1{}%
    \def\colorgray#1{}%
  \else
    \ifGPcolor
      \def\colorrgb#1{\color[rgb]{#1}}%
      \def\colorgray#1{\color[gray]{#1}}%
      \expandafter\def\csname LTw\endcsname{\color{white}}%
      \expandafter\def\csname LTb\endcsname{\color{black}}%
      \expandafter\def\csname LTa\endcsname{\color{black}}%
      \expandafter\def\csname LT0\endcsname{\color[rgb]{1,0,0}}%
      \expandafter\def\csname LT1\endcsname{\color[rgb]{0,1,0}}%
      \expandafter\def\csname LT2\endcsname{\color[rgb]{0,0,1}}%
      \expandafter\def\csname LT3\endcsname{\color[rgb]{1,0,1}}%
      \expandafter\def\csname LT4\endcsname{\color[rgb]{0,1,1}}%
      \expandafter\def\csname LT5\endcsname{\color[rgb]{1,1,0}}%
      \expandafter\def\csname LT6\endcsname{\color[rgb]{0,0,0}}%
      \expandafter\def\csname LT7\endcsname{\color[rgb]{1,0.3,0}}%
      \expandafter\def\csname LT8\endcsname{\color[rgb]{0.5,0.5,0.5}}%
    \else
      \def\colorrgb#1{\color{black}}%
      \def\colorgray#1{\color[gray]{#1}}%
      \expandafter\def\csname LTw\endcsname{\color{white}}%
      \expandafter\def\csname LTb\endcsname{\color{black}}%
      \expandafter\def\csname LTa\endcsname{\color{black}}%
      \expandafter\def\csname LT0\endcsname{\color{black}}%
      \expandafter\def\csname LT1\endcsname{\color{black}}%
      \expandafter\def\csname LT2\endcsname{\color{black}}%
      \expandafter\def\csname LT3\endcsname{\color{black}}%
      \expandafter\def\csname LT4\endcsname{\color{black}}%
      \expandafter\def\csname LT5\endcsname{\color{black}}%
      \expandafter\def\csname LT6\endcsname{\color{black}}%
      \expandafter\def\csname LT7\endcsname{\color{black}}%
      \expandafter\def\csname LT8\endcsname{\color{black}}%
    \fi
  \fi
    \setlength{\unitlength}{0.0500bp}%
    \ifx\gptboxheight\undefined%
      \newlength{\gptboxheight}%
      \newlength{\gptboxwidth}%
      \newsavebox{\gptboxtext}%
    \fi%
    \setlength{\fboxrule}{0.5pt}%
    \setlength{\fboxsep}{1pt}%
\begin{picture}(6802.00,3614.00)%
    \gplgaddtomacro\gplbacktext{%
      \csname LTb\endcsname%
      \put(814,704){\makebox(0,0)[r]{\strut{}$0$}}%
      \csname LTb\endcsname%
      \put(814,1113){\makebox(0,0)[r]{\strut{}$0.2$}}%
      \csname LTb\endcsname%
      \put(814,1522){\makebox(0,0)[r]{\strut{}$0.4$}}%
      \csname LTb\endcsname%
      \put(814,1931){\makebox(0,0)[r]{\strut{}$0.6$}}%
      \csname LTb\endcsname%
      \put(814,2340){\makebox(0,0)[r]{\strut{}$0.8$}}%
      \csname LTb\endcsname%
      \put(814,2749){\makebox(0,0)[r]{\strut{}$1$}}%
      \csname LTb\endcsname%
      \put(1234,484){\makebox(0,0){\strut{}$10^{3}$}}%
      \csname LTb\endcsname%
      \put(3096,484){\makebox(0,0){\strut{}$10^{4}$}}%
      \csname LTb\endcsname%
      \put(4957,484){\makebox(0,0){\strut{}$10^{5}$}}%
    }%
    \gplgaddtomacro\gplfronttext{%
      \csname LTb\endcsname%
      \put(176,1828){\rotatebox{-270}{\makebox(0,0){\strut{}NMI}}}%
      \put(3675,154){\makebox(0,0){\strut{}Number $n$ of nodes}}%
      \put(3675,3283){\makebox(0,0){\strut{}Mixing $\mu = 0.4$, Cluster: Louvain}}%
      \csname LTb\endcsname%
      \put(2266,1317){\makebox(0,0)[r]{\strut{}Orig}}%
      \csname LTb\endcsname%
      \put(2266,1097){\makebox(0,0)[r]{\strut{}NetworKit}}%
      \csname LTb\endcsname%
      \put(2266,877){\makebox(0,0)[r]{\strut{}EM}}%
    }%
    \gplbacktext
    \put(0,0){\includegraphics{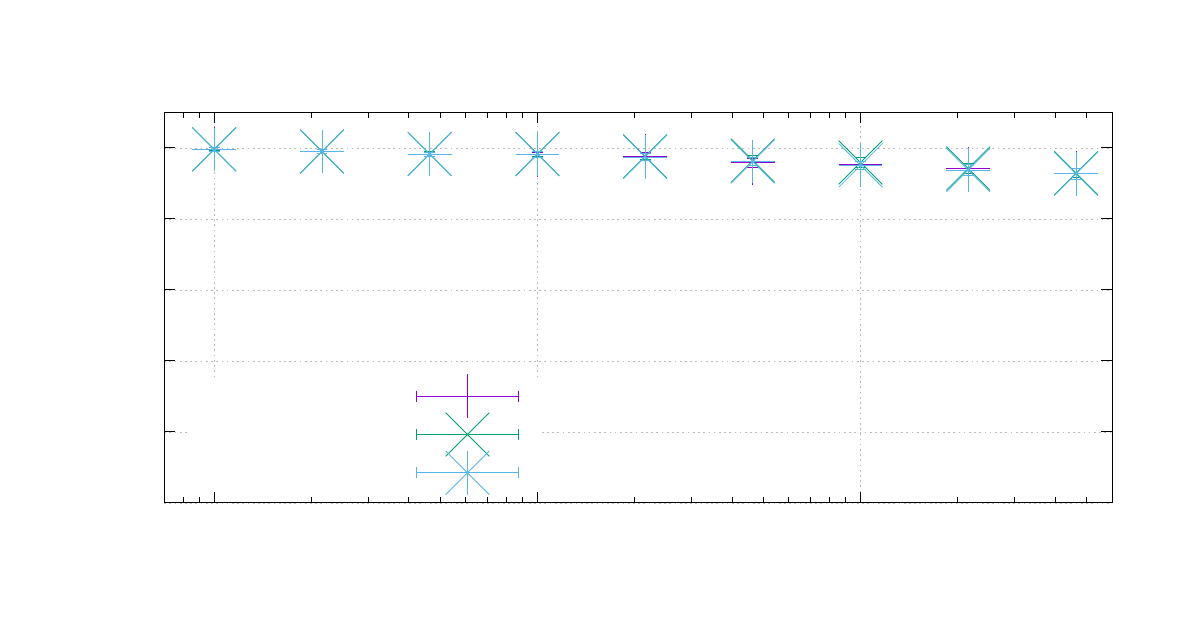}}%
    \gplfronttext
  \end{picture}%
\endgroup
}\hfill\scalebox{\threescale}{%
\begingroup
  \makeatletter
  \providecommand\color[2][]{%
    \GenericError{(gnuplot) \space\space\space\@spaces}{%
      Package color not loaded in conjunction with
      terminal option `colourtext'%
    }{See the gnuplot documentation for explanation.%
    }{Either use 'blacktext' in gnuplot or load the package
      color.sty in LaTeX.}%
    \renewcommand\color[2][]{}%
  }%
  \providecommand\includegraphics[2][]{%
    \GenericError{(gnuplot) \space\space\space\@spaces}{%
      Package graphicx or graphics not loaded%
    }{See the gnuplot documentation for explanation.%
    }{The gnuplot epslatex terminal needs graphicx.sty or graphics.sty.}%
    \renewcommand\includegraphics[2][]{}%
  }%
  \providecommand\rotatebox[2]{#2}%
  \@ifundefined{ifGPcolor}{%
    \newif\ifGPcolor
    \GPcolortrue
  }{}%
  \@ifundefined{ifGPblacktext}{%
    \newif\ifGPblacktext
    \GPblacktextfalse
  }{}%
  \let\gplgaddtomacro\g@addto@macro
  \gdef\gplbacktext{}%
  \gdef\gplfronttext{}%
  \makeatother
  \ifGPblacktext
    \def\colorrgb#1{}%
    \def\colorgray#1{}%
  \else
    \ifGPcolor
      \def\colorrgb#1{\color[rgb]{#1}}%
      \def\colorgray#1{\color[gray]{#1}}%
      \expandafter\def\csname LTw\endcsname{\color{white}}%
      \expandafter\def\csname LTb\endcsname{\color{black}}%
      \expandafter\def\csname LTa\endcsname{\color{black}}%
      \expandafter\def\csname LT0\endcsname{\color[rgb]{1,0,0}}%
      \expandafter\def\csname LT1\endcsname{\color[rgb]{0,1,0}}%
      \expandafter\def\csname LT2\endcsname{\color[rgb]{0,0,1}}%
      \expandafter\def\csname LT3\endcsname{\color[rgb]{1,0,1}}%
      \expandafter\def\csname LT4\endcsname{\color[rgb]{0,1,1}}%
      \expandafter\def\csname LT5\endcsname{\color[rgb]{1,1,0}}%
      \expandafter\def\csname LT6\endcsname{\color[rgb]{0,0,0}}%
      \expandafter\def\csname LT7\endcsname{\color[rgb]{1,0.3,0}}%
      \expandafter\def\csname LT8\endcsname{\color[rgb]{0.5,0.5,0.5}}%
    \else
      \def\colorrgb#1{\color{black}}%
      \def\colorgray#1{\color[gray]{#1}}%
      \expandafter\def\csname LTw\endcsname{\color{white}}%
      \expandafter\def\csname LTb\endcsname{\color{black}}%
      \expandafter\def\csname LTa\endcsname{\color{black}}%
      \expandafter\def\csname LT0\endcsname{\color{black}}%
      \expandafter\def\csname LT1\endcsname{\color{black}}%
      \expandafter\def\csname LT2\endcsname{\color{black}}%
      \expandafter\def\csname LT3\endcsname{\color{black}}%
      \expandafter\def\csname LT4\endcsname{\color{black}}%
      \expandafter\def\csname LT5\endcsname{\color{black}}%
      \expandafter\def\csname LT6\endcsname{\color{black}}%
      \expandafter\def\csname LT7\endcsname{\color{black}}%
      \expandafter\def\csname LT8\endcsname{\color{black}}%
    \fi
  \fi
    \setlength{\unitlength}{0.0500bp}%
    \ifx\gptboxheight\undefined%
      \newlength{\gptboxheight}%
      \newlength{\gptboxwidth}%
      \newsavebox{\gptboxtext}%
    \fi%
    \setlength{\fboxrule}{0.5pt}%
    \setlength{\fboxsep}{1pt}%
\begin{picture}(6802.00,3614.00)%
    \gplgaddtomacro\gplbacktext{%
      \csname LTb\endcsname%
      \put(814,704){\makebox(0,0)[r]{\strut{}$0$}}%
      \csname LTb\endcsname%
      \put(814,1113){\makebox(0,0)[r]{\strut{}$0.2$}}%
      \csname LTb\endcsname%
      \put(814,1522){\makebox(0,0)[r]{\strut{}$0.4$}}%
      \csname LTb\endcsname%
      \put(814,1931){\makebox(0,0)[r]{\strut{}$0.6$}}%
      \csname LTb\endcsname%
      \put(814,2340){\makebox(0,0)[r]{\strut{}$0.8$}}%
      \csname LTb\endcsname%
      \put(814,2749){\makebox(0,0)[r]{\strut{}$1$}}%
      \csname LTb\endcsname%
      \put(1234,484){\makebox(0,0){\strut{}$10^{3}$}}%
      \csname LTb\endcsname%
      \put(3096,484){\makebox(0,0){\strut{}$10^{4}$}}%
      \csname LTb\endcsname%
      \put(4957,484){\makebox(0,0){\strut{}$10^{5}$}}%
    }%
    \gplgaddtomacro\gplfronttext{%
      \csname LTb\endcsname%
      \put(176,1828){\rotatebox{-270}{\makebox(0,0){\strut{}NMI}}}%
      \put(3675,154){\makebox(0,0){\strut{}Number $n$ of nodes}}%
      \put(3675,3283){\makebox(0,0){\strut{}Mixing $\mu = 0.6$, Cluster: Louvain}}%
      \csname LTb\endcsname%
      \put(2266,1317){\makebox(0,0)[r]{\strut{}Orig}}%
      \csname LTb\endcsname%
      \put(2266,1097){\makebox(0,0)[r]{\strut{}NetworKit}}%
      \csname LTb\endcsname%
      \put(2266,877){\makebox(0,0)[r]{\strut{}EM}}%
    }%
    \gplbacktext
    \put(0,0){\includegraphics{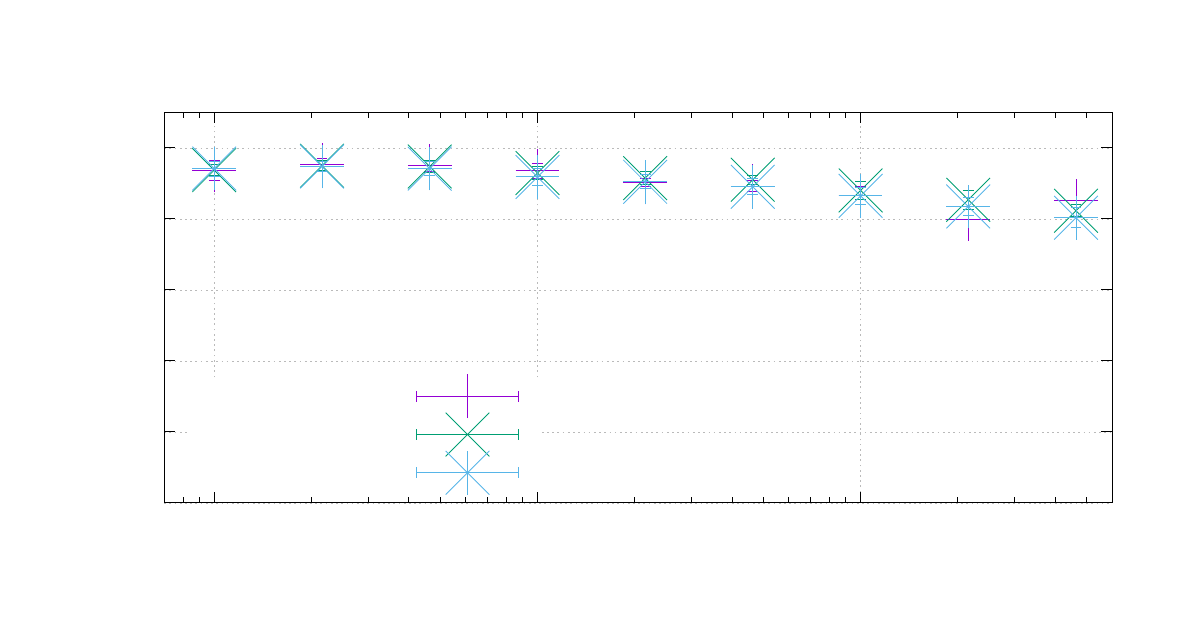}}%
    \gplfronttext
  \end{picture}%
\endgroup
}\hfill\\ %

	\vspace{0.1em}

	\noindent\scalebox{\threescale}{%
\begingroup
  \makeatletter
  \providecommand\color[2][]{%
    \GenericError{(gnuplot) \space\space\space\@spaces}{%
      Package color not loaded in conjunction with
      terminal option `colourtext'%
    }{See the gnuplot documentation for explanation.%
    }{Either use 'blacktext' in gnuplot or load the package
      color.sty in LaTeX.}%
    \renewcommand\color[2][]{}%
  }%
  \providecommand\includegraphics[2][]{%
    \GenericError{(gnuplot) \space\space\space\@spaces}{%
      Package graphicx or graphics not loaded%
    }{See the gnuplot documentation for explanation.%
    }{The gnuplot epslatex terminal needs graphicx.sty or graphics.sty.}%
    \renewcommand\includegraphics[2][]{}%
  }%
  \providecommand\rotatebox[2]{#2}%
  \@ifundefined{ifGPcolor}{%
    \newif\ifGPcolor
    \GPcolortrue
  }{}%
  \@ifundefined{ifGPblacktext}{%
    \newif\ifGPblacktext
    \GPblacktextfalse
  }{}%
  \let\gplgaddtomacro\g@addto@macro
  \gdef\gplbacktext{}%
  \gdef\gplfronttext{}%
  \makeatother
  \ifGPblacktext
    \def\colorrgb#1{}%
    \def\colorgray#1{}%
  \else
    \ifGPcolor
      \def\colorrgb#1{\color[rgb]{#1}}%
      \def\colorgray#1{\color[gray]{#1}}%
      \expandafter\def\csname LTw\endcsname{\color{white}}%
      \expandafter\def\csname LTb\endcsname{\color{black}}%
      \expandafter\def\csname LTa\endcsname{\color{black}}%
      \expandafter\def\csname LT0\endcsname{\color[rgb]{1,0,0}}%
      \expandafter\def\csname LT1\endcsname{\color[rgb]{0,1,0}}%
      \expandafter\def\csname LT2\endcsname{\color[rgb]{0,0,1}}%
      \expandafter\def\csname LT3\endcsname{\color[rgb]{1,0,1}}%
      \expandafter\def\csname LT4\endcsname{\color[rgb]{0,1,1}}%
      \expandafter\def\csname LT5\endcsname{\color[rgb]{1,1,0}}%
      \expandafter\def\csname LT6\endcsname{\color[rgb]{0,0,0}}%
      \expandafter\def\csname LT7\endcsname{\color[rgb]{1,0.3,0}}%
      \expandafter\def\csname LT8\endcsname{\color[rgb]{0.5,0.5,0.5}}%
    \else
      \def\colorrgb#1{\color{black}}%
      \def\colorgray#1{\color[gray]{#1}}%
      \expandafter\def\csname LTw\endcsname{\color{white}}%
      \expandafter\def\csname LTb\endcsname{\color{black}}%
      \expandafter\def\csname LTa\endcsname{\color{black}}%
      \expandafter\def\csname LT0\endcsname{\color{black}}%
      \expandafter\def\csname LT1\endcsname{\color{black}}%
      \expandafter\def\csname LT2\endcsname{\color{black}}%
      \expandafter\def\csname LT3\endcsname{\color{black}}%
      \expandafter\def\csname LT4\endcsname{\color{black}}%
      \expandafter\def\csname LT5\endcsname{\color{black}}%
      \expandafter\def\csname LT6\endcsname{\color{black}}%
      \expandafter\def\csname LT7\endcsname{\color{black}}%
      \expandafter\def\csname LT8\endcsname{\color{black}}%
    \fi
  \fi
    \setlength{\unitlength}{0.0500bp}%
    \ifx\gptboxheight\undefined%
      \newlength{\gptboxheight}%
      \newlength{\gptboxwidth}%
      \newsavebox{\gptboxtext}%
    \fi%
    \setlength{\fboxrule}{0.5pt}%
    \setlength{\fboxsep}{1pt}%
\begin{picture}(6802.00,3614.00)%
    \gplgaddtomacro\gplbacktext{%
      \csname LTb\endcsname%
      \put(814,704){\makebox(0,0)[r]{\strut{}$10^{3}$}}%
      \csname LTb\endcsname%
      \put(814,1154){\makebox(0,0)[r]{\strut{}$10^{4}$}}%
      \csname LTb\endcsname%
      \put(814,1604){\makebox(0,0)[r]{\strut{}$10^{5}$}}%
      \csname LTb\endcsname%
      \put(814,2053){\makebox(0,0)[r]{\strut{}$10^{6}$}}%
      \csname LTb\endcsname%
      \put(814,2503){\makebox(0,0)[r]{\strut{}$10^{7}$}}%
      \csname LTb\endcsname%
      \put(814,2953){\makebox(0,0)[r]{\strut{}$10^{8}$}}%
      \csname LTb\endcsname%
      \put(1234,484){\makebox(0,0){\strut{}$10^{3}$}}%
      \csname LTb\endcsname%
      \put(3096,484){\makebox(0,0){\strut{}$10^{4}$}}%
      \csname LTb\endcsname%
      \put(4957,484){\makebox(0,0){\strut{}$10^{5}$}}%
    }%
    \gplgaddtomacro\gplfronttext{%
      \csname LTb\endcsname%
      \put(176,1828){\rotatebox{-270}{\makebox(0,0){\strut{}Edges}}}%
      \put(3675,154){\makebox(0,0){\strut{}Number $n$ of nodes}}%
      \put(3675,3283){\makebox(0,0){\strut{}Mixing: $\mu = 0.2$}}%
      \csname LTb\endcsname%
      \put(5418,1317){\makebox(0,0)[r]{\strut{}Orig}}%
      \csname LTb\endcsname%
      \put(5418,1097){\makebox(0,0)[r]{\strut{}NetworKit}}%
      \csname LTb\endcsname%
      \put(5418,877){\makebox(0,0)[r]{\strut{}EM}}%
    }%
    \gplbacktext
    \put(0,0){\includegraphics{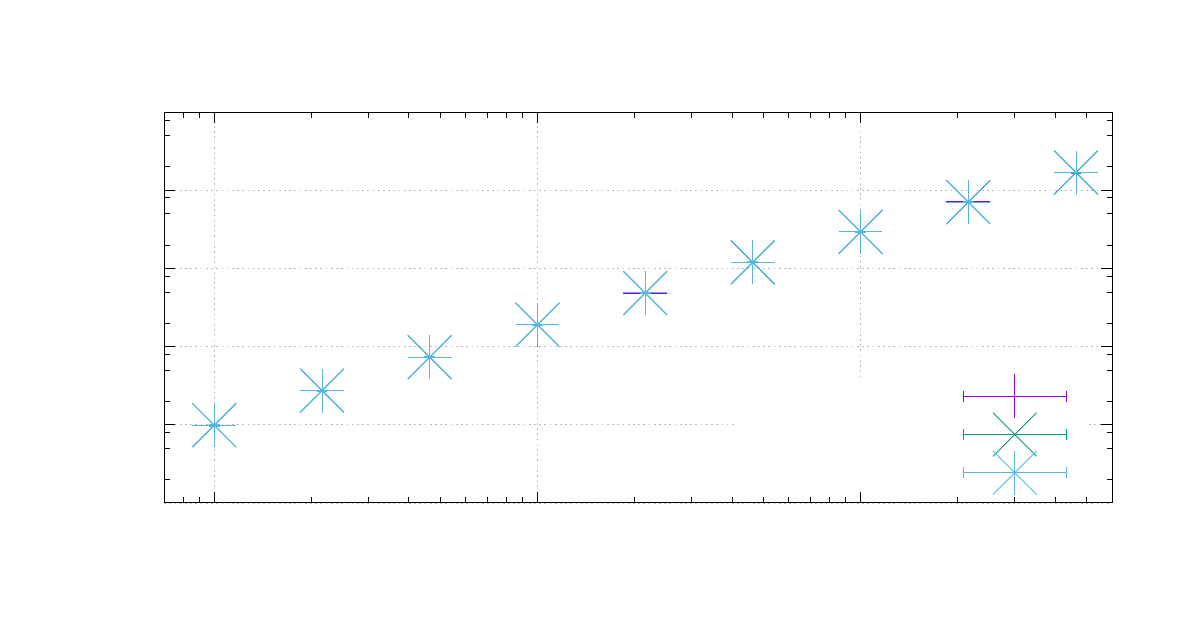}}%
    \gplfronttext
  \end{picture}%
\endgroup
}\hfill\scalebox{\threescale}{%
\begingroup
  \makeatletter
  \providecommand\color[2][]{%
    \GenericError{(gnuplot) \space\space\space\@spaces}{%
      Package color not loaded in conjunction with
      terminal option `colourtext'%
    }{See the gnuplot documentation for explanation.%
    }{Either use 'blacktext' in gnuplot or load the package
      color.sty in LaTeX.}%
    \renewcommand\color[2][]{}%
  }%
  \providecommand\includegraphics[2][]{%
    \GenericError{(gnuplot) \space\space\space\@spaces}{%
      Package graphicx or graphics not loaded%
    }{See the gnuplot documentation for explanation.%
    }{The gnuplot epslatex terminal needs graphicx.sty or graphics.sty.}%
    \renewcommand\includegraphics[2][]{}%
  }%
  \providecommand\rotatebox[2]{#2}%
  \@ifundefined{ifGPcolor}{%
    \newif\ifGPcolor
    \GPcolortrue
  }{}%
  \@ifundefined{ifGPblacktext}{%
    \newif\ifGPblacktext
    \GPblacktextfalse
  }{}%
  \let\gplgaddtomacro\g@addto@macro
  \gdef\gplbacktext{}%
  \gdef\gplfronttext{}%
  \makeatother
  \ifGPblacktext
    \def\colorrgb#1{}%
    \def\colorgray#1{}%
  \else
    \ifGPcolor
      \def\colorrgb#1{\color[rgb]{#1}}%
      \def\colorgray#1{\color[gray]{#1}}%
      \expandafter\def\csname LTw\endcsname{\color{white}}%
      \expandafter\def\csname LTb\endcsname{\color{black}}%
      \expandafter\def\csname LTa\endcsname{\color{black}}%
      \expandafter\def\csname LT0\endcsname{\color[rgb]{1,0,0}}%
      \expandafter\def\csname LT1\endcsname{\color[rgb]{0,1,0}}%
      \expandafter\def\csname LT2\endcsname{\color[rgb]{0,0,1}}%
      \expandafter\def\csname LT3\endcsname{\color[rgb]{1,0,1}}%
      \expandafter\def\csname LT4\endcsname{\color[rgb]{0,1,1}}%
      \expandafter\def\csname LT5\endcsname{\color[rgb]{1,1,0}}%
      \expandafter\def\csname LT6\endcsname{\color[rgb]{0,0,0}}%
      \expandafter\def\csname LT7\endcsname{\color[rgb]{1,0.3,0}}%
      \expandafter\def\csname LT8\endcsname{\color[rgb]{0.5,0.5,0.5}}%
    \else
      \def\colorrgb#1{\color{black}}%
      \def\colorgray#1{\color[gray]{#1}}%
      \expandafter\def\csname LTw\endcsname{\color{white}}%
      \expandafter\def\csname LTb\endcsname{\color{black}}%
      \expandafter\def\csname LTa\endcsname{\color{black}}%
      \expandafter\def\csname LT0\endcsname{\color{black}}%
      \expandafter\def\csname LT1\endcsname{\color{black}}%
      \expandafter\def\csname LT2\endcsname{\color{black}}%
      \expandafter\def\csname LT3\endcsname{\color{black}}%
      \expandafter\def\csname LT4\endcsname{\color{black}}%
      \expandafter\def\csname LT5\endcsname{\color{black}}%
      \expandafter\def\csname LT6\endcsname{\color{black}}%
      \expandafter\def\csname LT7\endcsname{\color{black}}%
      \expandafter\def\csname LT8\endcsname{\color{black}}%
    \fi
  \fi
    \setlength{\unitlength}{0.0500bp}%
    \ifx\gptboxheight\undefined%
      \newlength{\gptboxheight}%
      \newlength{\gptboxwidth}%
      \newsavebox{\gptboxtext}%
    \fi%
    \setlength{\fboxrule}{0.5pt}%
    \setlength{\fboxsep}{1pt}%
\begin{picture}(6802.00,3614.00)%
    \gplgaddtomacro\gplbacktext{%
      \csname LTb\endcsname%
      \put(814,704){\makebox(0,0)[r]{\strut{}$10^{3}$}}%
      \csname LTb\endcsname%
      \put(814,1154){\makebox(0,0)[r]{\strut{}$10^{4}$}}%
      \csname LTb\endcsname%
      \put(814,1604){\makebox(0,0)[r]{\strut{}$10^{5}$}}%
      \csname LTb\endcsname%
      \put(814,2053){\makebox(0,0)[r]{\strut{}$10^{6}$}}%
      \csname LTb\endcsname%
      \put(814,2503){\makebox(0,0)[r]{\strut{}$10^{7}$}}%
      \csname LTb\endcsname%
      \put(814,2953){\makebox(0,0)[r]{\strut{}$10^{8}$}}%
      \csname LTb\endcsname%
      \put(1234,484){\makebox(0,0){\strut{}$10^{3}$}}%
      \csname LTb\endcsname%
      \put(3096,484){\makebox(0,0){\strut{}$10^{4}$}}%
      \csname LTb\endcsname%
      \put(4957,484){\makebox(0,0){\strut{}$10^{5}$}}%
    }%
    \gplgaddtomacro\gplfronttext{%
      \csname LTb\endcsname%
      \put(176,1828){\rotatebox{-270}{\makebox(0,0){\strut{}Edges}}}%
      \put(3675,154){\makebox(0,0){\strut{}Number $n$ of nodes}}%
      \put(3675,3283){\makebox(0,0){\strut{}Mixing: $\mu = 0.4$}}%
      \csname LTb\endcsname%
      \put(5418,1317){\makebox(0,0)[r]{\strut{}Orig}}%
      \csname LTb\endcsname%
      \put(5418,1097){\makebox(0,0)[r]{\strut{}NetworKit}}%
      \csname LTb\endcsname%
      \put(5418,877){\makebox(0,0)[r]{\strut{}EM}}%
    }%
    \gplbacktext
    \put(0,0){\includegraphics{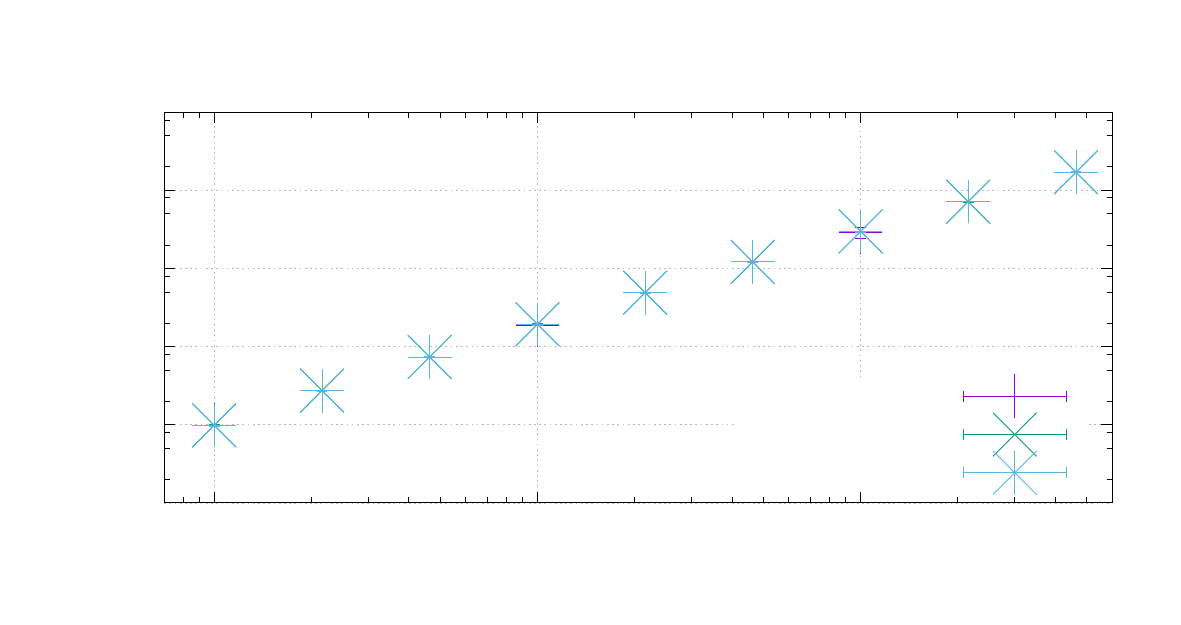}}%
    \gplfronttext
  \end{picture}%
\endgroup
}\hfill\scalebox{\threescale}{%
\begingroup
  \makeatletter
  \providecommand\color[2][]{%
    \GenericError{(gnuplot) \space\space\space\@spaces}{%
      Package color not loaded in conjunction with
      terminal option `colourtext'%
    }{See the gnuplot documentation for explanation.%
    }{Either use 'blacktext' in gnuplot or load the package
      color.sty in LaTeX.}%
    \renewcommand\color[2][]{}%
  }%
  \providecommand\includegraphics[2][]{%
    \GenericError{(gnuplot) \space\space\space\@spaces}{%
      Package graphicx or graphics not loaded%
    }{See the gnuplot documentation for explanation.%
    }{The gnuplot epslatex terminal needs graphicx.sty or graphics.sty.}%
    \renewcommand\includegraphics[2][]{}%
  }%
  \providecommand\rotatebox[2]{#2}%
  \@ifundefined{ifGPcolor}{%
    \newif\ifGPcolor
    \GPcolortrue
  }{}%
  \@ifundefined{ifGPblacktext}{%
    \newif\ifGPblacktext
    \GPblacktextfalse
  }{}%
  \let\gplgaddtomacro\g@addto@macro
  \gdef\gplbacktext{}%
  \gdef\gplfronttext{}%
  \makeatother
  \ifGPblacktext
    \def\colorrgb#1{}%
    \def\colorgray#1{}%
  \else
    \ifGPcolor
      \def\colorrgb#1{\color[rgb]{#1}}%
      \def\colorgray#1{\color[gray]{#1}}%
      \expandafter\def\csname LTw\endcsname{\color{white}}%
      \expandafter\def\csname LTb\endcsname{\color{black}}%
      \expandafter\def\csname LTa\endcsname{\color{black}}%
      \expandafter\def\csname LT0\endcsname{\color[rgb]{1,0,0}}%
      \expandafter\def\csname LT1\endcsname{\color[rgb]{0,1,0}}%
      \expandafter\def\csname LT2\endcsname{\color[rgb]{0,0,1}}%
      \expandafter\def\csname LT3\endcsname{\color[rgb]{1,0,1}}%
      \expandafter\def\csname LT4\endcsname{\color[rgb]{0,1,1}}%
      \expandafter\def\csname LT5\endcsname{\color[rgb]{1,1,0}}%
      \expandafter\def\csname LT6\endcsname{\color[rgb]{0,0,0}}%
      \expandafter\def\csname LT7\endcsname{\color[rgb]{1,0.3,0}}%
      \expandafter\def\csname LT8\endcsname{\color[rgb]{0.5,0.5,0.5}}%
    \else
      \def\colorrgb#1{\color{black}}%
      \def\colorgray#1{\color[gray]{#1}}%
      \expandafter\def\csname LTw\endcsname{\color{white}}%
      \expandafter\def\csname LTb\endcsname{\color{black}}%
      \expandafter\def\csname LTa\endcsname{\color{black}}%
      \expandafter\def\csname LT0\endcsname{\color{black}}%
      \expandafter\def\csname LT1\endcsname{\color{black}}%
      \expandafter\def\csname LT2\endcsname{\color{black}}%
      \expandafter\def\csname LT3\endcsname{\color{black}}%
      \expandafter\def\csname LT4\endcsname{\color{black}}%
      \expandafter\def\csname LT5\endcsname{\color{black}}%
      \expandafter\def\csname LT6\endcsname{\color{black}}%
      \expandafter\def\csname LT7\endcsname{\color{black}}%
      \expandafter\def\csname LT8\endcsname{\color{black}}%
    \fi
  \fi
    \setlength{\unitlength}{0.0500bp}%
    \ifx\gptboxheight\undefined%
      \newlength{\gptboxheight}%
      \newlength{\gptboxwidth}%
      \newsavebox{\gptboxtext}%
    \fi%
    \setlength{\fboxrule}{0.5pt}%
    \setlength{\fboxsep}{1pt}%
\begin{picture}(6802.00,3614.00)%
    \gplgaddtomacro\gplbacktext{%
      \csname LTb\endcsname%
      \put(814,704){\makebox(0,0)[r]{\strut{}$10^{3}$}}%
      \csname LTb\endcsname%
      \put(814,1154){\makebox(0,0)[r]{\strut{}$10^{4}$}}%
      \csname LTb\endcsname%
      \put(814,1604){\makebox(0,0)[r]{\strut{}$10^{5}$}}%
      \csname LTb\endcsname%
      \put(814,2053){\makebox(0,0)[r]{\strut{}$10^{6}$}}%
      \csname LTb\endcsname%
      \put(814,2503){\makebox(0,0)[r]{\strut{}$10^{7}$}}%
      \csname LTb\endcsname%
      \put(814,2953){\makebox(0,0)[r]{\strut{}$10^{8}$}}%
      \csname LTb\endcsname%
      \put(1234,484){\makebox(0,0){\strut{}$10^{3}$}}%
      \csname LTb\endcsname%
      \put(3096,484){\makebox(0,0){\strut{}$10^{4}$}}%
      \csname LTb\endcsname%
      \put(4957,484){\makebox(0,0){\strut{}$10^{5}$}}%
    }%
    \gplgaddtomacro\gplfronttext{%
      \csname LTb\endcsname%
      \put(176,1828){\rotatebox{-270}{\makebox(0,0){\strut{}Edges}}}%
      \put(3675,154){\makebox(0,0){\strut{}Number $n$ of nodes}}%
      \put(3675,3283){\makebox(0,0){\strut{}Mixing: $\mu = 0.6$}}%
      \csname LTb\endcsname%
      \put(5418,1317){\makebox(0,0)[r]{\strut{}Orig}}%
      \csname LTb\endcsname%
      \put(5418,1097){\makebox(0,0)[r]{\strut{}NetworKit}}%
      \csname LTb\endcsname%
      \put(5418,877){\makebox(0,0)[r]{\strut{}EM}}%
    }%
    \gplbacktext
    \put(0,0){\includegraphics{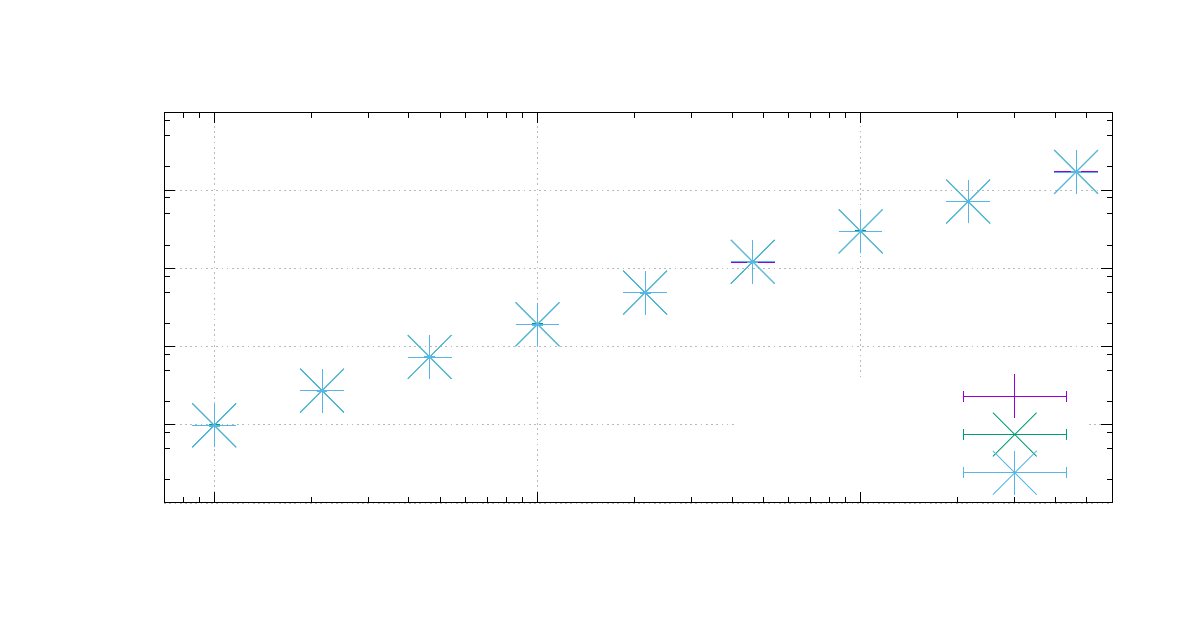}}%
    \gplfronttext
  \end{picture}%
\endgroup
}\hfill\\ %
	\noindent\scalebox{\threescale}{%
\begingroup
  \makeatletter
  \providecommand\color[2][]{%
    \GenericError{(gnuplot) \space\space\space\@spaces}{%
      Package color not loaded in conjunction with
      terminal option `colourtext'%
    }{See the gnuplot documentation for explanation.%
    }{Either use 'blacktext' in gnuplot or load the package
      color.sty in LaTeX.}%
    \renewcommand\color[2][]{}%
  }%
  \providecommand\includegraphics[2][]{%
    \GenericError{(gnuplot) \space\space\space\@spaces}{%
      Package graphicx or graphics not loaded%
    }{See the gnuplot documentation for explanation.%
    }{The gnuplot epslatex terminal needs graphicx.sty or graphics.sty.}%
    \renewcommand\includegraphics[2][]{}%
  }%
  \providecommand\rotatebox[2]{#2}%
  \@ifundefined{ifGPcolor}{%
    \newif\ifGPcolor
    \GPcolortrue
  }{}%
  \@ifundefined{ifGPblacktext}{%
    \newif\ifGPblacktext
    \GPblacktextfalse
  }{}%
  \let\gplgaddtomacro\g@addto@macro
  \gdef\gplbacktext{}%
  \gdef\gplfronttext{}%
  \makeatother
  \ifGPblacktext
    \def\colorrgb#1{}%
    \def\colorgray#1{}%
  \else
    \ifGPcolor
      \def\colorrgb#1{\color[rgb]{#1}}%
      \def\colorgray#1{\color[gray]{#1}}%
      \expandafter\def\csname LTw\endcsname{\color{white}}%
      \expandafter\def\csname LTb\endcsname{\color{black}}%
      \expandafter\def\csname LTa\endcsname{\color{black}}%
      \expandafter\def\csname LT0\endcsname{\color[rgb]{1,0,0}}%
      \expandafter\def\csname LT1\endcsname{\color[rgb]{0,1,0}}%
      \expandafter\def\csname LT2\endcsname{\color[rgb]{0,0,1}}%
      \expandafter\def\csname LT3\endcsname{\color[rgb]{1,0,1}}%
      \expandafter\def\csname LT4\endcsname{\color[rgb]{0,1,1}}%
      \expandafter\def\csname LT5\endcsname{\color[rgb]{1,1,0}}%
      \expandafter\def\csname LT6\endcsname{\color[rgb]{0,0,0}}%
      \expandafter\def\csname LT7\endcsname{\color[rgb]{1,0.3,0}}%
      \expandafter\def\csname LT8\endcsname{\color[rgb]{0.5,0.5,0.5}}%
    \else
      \def\colorrgb#1{\color{black}}%
      \def\colorgray#1{\color[gray]{#1}}%
      \expandafter\def\csname LTw\endcsname{\color{white}}%
      \expandafter\def\csname LTb\endcsname{\color{black}}%
      \expandafter\def\csname LTa\endcsname{\color{black}}%
      \expandafter\def\csname LT0\endcsname{\color{black}}%
      \expandafter\def\csname LT1\endcsname{\color{black}}%
      \expandafter\def\csname LT2\endcsname{\color{black}}%
      \expandafter\def\csname LT3\endcsname{\color{black}}%
      \expandafter\def\csname LT4\endcsname{\color{black}}%
      \expandafter\def\csname LT5\endcsname{\color{black}}%
      \expandafter\def\csname LT6\endcsname{\color{black}}%
      \expandafter\def\csname LT7\endcsname{\color{black}}%
      \expandafter\def\csname LT8\endcsname{\color{black}}%
    \fi
  \fi
    \setlength{\unitlength}{0.0500bp}%
    \ifx\gptboxheight\undefined%
      \newlength{\gptboxheight}%
      \newlength{\gptboxwidth}%
      \newsavebox{\gptboxtext}%
    \fi%
    \setlength{\fboxrule}{0.5pt}%
    \setlength{\fboxsep}{1pt}%
\begin{picture}(6802.00,3614.00)%
    \gplgaddtomacro\gplbacktext{%
      \csname LTb\endcsname%
      \put(814,704){\makebox(0,0)[r]{\strut{}$0$}}%
      \csname LTb\endcsname%
      \put(814,1113){\makebox(0,0)[r]{\strut{}$0.2$}}%
      \csname LTb\endcsname%
      \put(814,1522){\makebox(0,0)[r]{\strut{}$0.4$}}%
      \csname LTb\endcsname%
      \put(814,1931){\makebox(0,0)[r]{\strut{}$0.6$}}%
      \csname LTb\endcsname%
      \put(814,2340){\makebox(0,0)[r]{\strut{}$0.8$}}%
      \csname LTb\endcsname%
      \put(814,2749){\makebox(0,0)[r]{\strut{}$1$}}%
      \csname LTb\endcsname%
      \put(1234,484){\makebox(0,0){\strut{}$10^{3}$}}%
      \csname LTb\endcsname%
      \put(3096,484){\makebox(0,0){\strut{}$10^{4}$}}%
      \csname LTb\endcsname%
      \put(4957,484){\makebox(0,0){\strut{}$10^{5}$}}%
    }%
    \gplgaddtomacro\gplfronttext{%
      \csname LTb\endcsname%
      \put(176,1828){\rotatebox{-270}{\makebox(0,0){\strut{}Avg. Local Clustering Coeff.}}}%
      \put(3675,154){\makebox(0,0){\strut{}Number $n$ of nodes}}%
      \put(3675,3283){\makebox(0,0){\strut{}Mixing: $\mu = 0.2$}}%
      \csname LTb\endcsname%
      \put(5418,2780){\makebox(0,0)[r]{\strut{}Orig}}%
      \csname LTb\endcsname%
      \put(5418,2560){\makebox(0,0)[r]{\strut{}NetworKit}}%
      \csname LTb\endcsname%
      \put(5418,2340){\makebox(0,0)[r]{\strut{}EM}}%
    }%
    \gplbacktext
    \put(0,0){\includegraphics{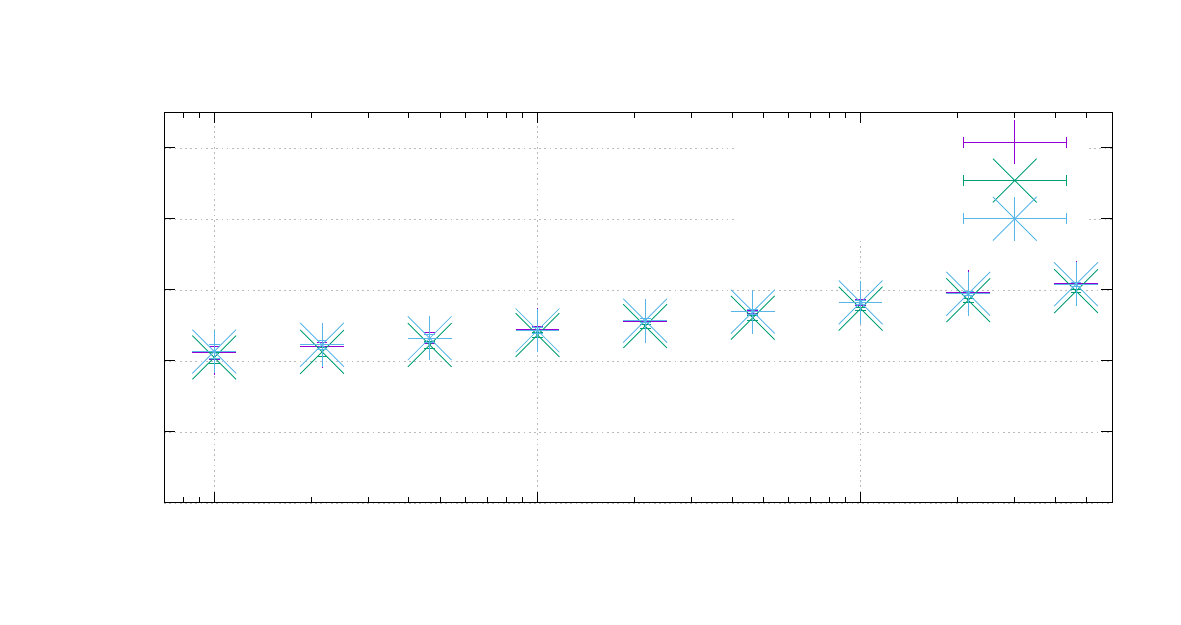}}%
    \gplfronttext
  \end{picture}%
\endgroup
}\hfill\scalebox{\threescale}{%
\begingroup
  \makeatletter
  \providecommand\color[2][]{%
    \GenericError{(gnuplot) \space\space\space\@spaces}{%
      Package color not loaded in conjunction with
      terminal option `colourtext'%
    }{See the gnuplot documentation for explanation.%
    }{Either use 'blacktext' in gnuplot or load the package
      color.sty in LaTeX.}%
    \renewcommand\color[2][]{}%
  }%
  \providecommand\includegraphics[2][]{%
    \GenericError{(gnuplot) \space\space\space\@spaces}{%
      Package graphicx or graphics not loaded%
    }{See the gnuplot documentation for explanation.%
    }{The gnuplot epslatex terminal needs graphicx.sty or graphics.sty.}%
    \renewcommand\includegraphics[2][]{}%
  }%
  \providecommand\rotatebox[2]{#2}%
  \@ifundefined{ifGPcolor}{%
    \newif\ifGPcolor
    \GPcolortrue
  }{}%
  \@ifundefined{ifGPblacktext}{%
    \newif\ifGPblacktext
    \GPblacktextfalse
  }{}%
  \let\gplgaddtomacro\g@addto@macro
  \gdef\gplbacktext{}%
  \gdef\gplfronttext{}%
  \makeatother
  \ifGPblacktext
    \def\colorrgb#1{}%
    \def\colorgray#1{}%
  \else
    \ifGPcolor
      \def\colorrgb#1{\color[rgb]{#1}}%
      \def\colorgray#1{\color[gray]{#1}}%
      \expandafter\def\csname LTw\endcsname{\color{white}}%
      \expandafter\def\csname LTb\endcsname{\color{black}}%
      \expandafter\def\csname LTa\endcsname{\color{black}}%
      \expandafter\def\csname LT0\endcsname{\color[rgb]{1,0,0}}%
      \expandafter\def\csname LT1\endcsname{\color[rgb]{0,1,0}}%
      \expandafter\def\csname LT2\endcsname{\color[rgb]{0,0,1}}%
      \expandafter\def\csname LT3\endcsname{\color[rgb]{1,0,1}}%
      \expandafter\def\csname LT4\endcsname{\color[rgb]{0,1,1}}%
      \expandafter\def\csname LT5\endcsname{\color[rgb]{1,1,0}}%
      \expandafter\def\csname LT6\endcsname{\color[rgb]{0,0,0}}%
      \expandafter\def\csname LT7\endcsname{\color[rgb]{1,0.3,0}}%
      \expandafter\def\csname LT8\endcsname{\color[rgb]{0.5,0.5,0.5}}%
    \else
      \def\colorrgb#1{\color{black}}%
      \def\colorgray#1{\color[gray]{#1}}%
      \expandafter\def\csname LTw\endcsname{\color{white}}%
      \expandafter\def\csname LTb\endcsname{\color{black}}%
      \expandafter\def\csname LTa\endcsname{\color{black}}%
      \expandafter\def\csname LT0\endcsname{\color{black}}%
      \expandafter\def\csname LT1\endcsname{\color{black}}%
      \expandafter\def\csname LT2\endcsname{\color{black}}%
      \expandafter\def\csname LT3\endcsname{\color{black}}%
      \expandafter\def\csname LT4\endcsname{\color{black}}%
      \expandafter\def\csname LT5\endcsname{\color{black}}%
      \expandafter\def\csname LT6\endcsname{\color{black}}%
      \expandafter\def\csname LT7\endcsname{\color{black}}%
      \expandafter\def\csname LT8\endcsname{\color{black}}%
    \fi
  \fi
    \setlength{\unitlength}{0.0500bp}%
    \ifx\gptboxheight\undefined%
      \newlength{\gptboxheight}%
      \newlength{\gptboxwidth}%
      \newsavebox{\gptboxtext}%
    \fi%
    \setlength{\fboxrule}{0.5pt}%
    \setlength{\fboxsep}{1pt}%
\begin{picture}(6802.00,3614.00)%
    \gplgaddtomacro\gplbacktext{%
      \csname LTb\endcsname%
      \put(814,704){\makebox(0,0)[r]{\strut{}$0$}}%
      \csname LTb\endcsname%
      \put(814,1113){\makebox(0,0)[r]{\strut{}$0.2$}}%
      \csname LTb\endcsname%
      \put(814,1522){\makebox(0,0)[r]{\strut{}$0.4$}}%
      \csname LTb\endcsname%
      \put(814,1931){\makebox(0,0)[r]{\strut{}$0.6$}}%
      \csname LTb\endcsname%
      \put(814,2340){\makebox(0,0)[r]{\strut{}$0.8$}}%
      \csname LTb\endcsname%
      \put(814,2749){\makebox(0,0)[r]{\strut{}$1$}}%
      \csname LTb\endcsname%
      \put(1234,484){\makebox(0,0){\strut{}$10^{3}$}}%
      \csname LTb\endcsname%
      \put(3096,484){\makebox(0,0){\strut{}$10^{4}$}}%
      \csname LTb\endcsname%
      \put(4957,484){\makebox(0,0){\strut{}$10^{5}$}}%
    }%
    \gplgaddtomacro\gplfronttext{%
      \csname LTb\endcsname%
      \put(176,1828){\rotatebox{-270}{\makebox(0,0){\strut{}Avg. Local Clustering Coeff.}}}%
      \put(3675,154){\makebox(0,0){\strut{}Number $n$ of nodes}}%
      \put(3675,3283){\makebox(0,0){\strut{}Mixing: $\mu = 0.4$}}%
      \csname LTb\endcsname%
      \put(5418,2780){\makebox(0,0)[r]{\strut{}Orig}}%
      \csname LTb\endcsname%
      \put(5418,2560){\makebox(0,0)[r]{\strut{}NetworKit}}%
      \csname LTb\endcsname%
      \put(5418,2340){\makebox(0,0)[r]{\strut{}EM}}%
    }%
    \gplbacktext
    \put(0,0){\includegraphics{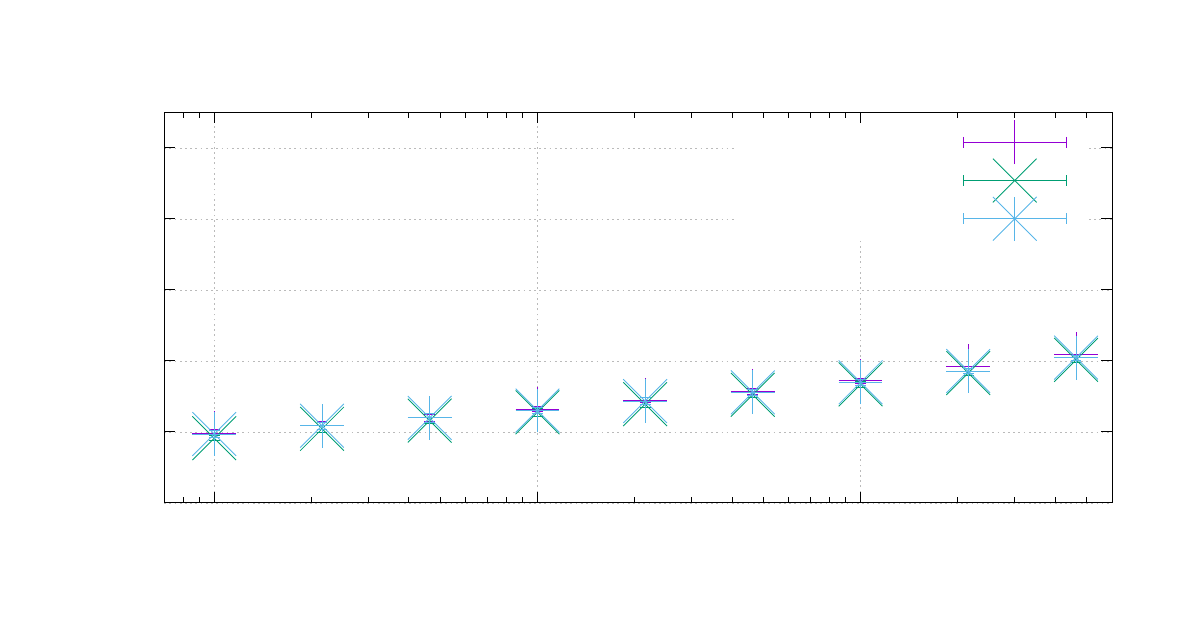}}%
    \gplfronttext
  \end{picture}%
\endgroup
}\hfill\scalebox{\threescale}{%
\begingroup
  \makeatletter
  \providecommand\color[2][]{%
    \GenericError{(gnuplot) \space\space\space\@spaces}{%
      Package color not loaded in conjunction with
      terminal option `colourtext'%
    }{See the gnuplot documentation for explanation.%
    }{Either use 'blacktext' in gnuplot or load the package
      color.sty in LaTeX.}%
    \renewcommand\color[2][]{}%
  }%
  \providecommand\includegraphics[2][]{%
    \GenericError{(gnuplot) \space\space\space\@spaces}{%
      Package graphicx or graphics not loaded%
    }{See the gnuplot documentation for explanation.%
    }{The gnuplot epslatex terminal needs graphicx.sty or graphics.sty.}%
    \renewcommand\includegraphics[2][]{}%
  }%
  \providecommand\rotatebox[2]{#2}%
  \@ifundefined{ifGPcolor}{%
    \newif\ifGPcolor
    \GPcolortrue
  }{}%
  \@ifundefined{ifGPblacktext}{%
    \newif\ifGPblacktext
    \GPblacktextfalse
  }{}%
  \let\gplgaddtomacro\g@addto@macro
  \gdef\gplbacktext{}%
  \gdef\gplfronttext{}%
  \makeatother
  \ifGPblacktext
    \def\colorrgb#1{}%
    \def\colorgray#1{}%
  \else
    \ifGPcolor
      \def\colorrgb#1{\color[rgb]{#1}}%
      \def\colorgray#1{\color[gray]{#1}}%
      \expandafter\def\csname LTw\endcsname{\color{white}}%
      \expandafter\def\csname LTb\endcsname{\color{black}}%
      \expandafter\def\csname LTa\endcsname{\color{black}}%
      \expandafter\def\csname LT0\endcsname{\color[rgb]{1,0,0}}%
      \expandafter\def\csname LT1\endcsname{\color[rgb]{0,1,0}}%
      \expandafter\def\csname LT2\endcsname{\color[rgb]{0,0,1}}%
      \expandafter\def\csname LT3\endcsname{\color[rgb]{1,0,1}}%
      \expandafter\def\csname LT4\endcsname{\color[rgb]{0,1,1}}%
      \expandafter\def\csname LT5\endcsname{\color[rgb]{1,1,0}}%
      \expandafter\def\csname LT6\endcsname{\color[rgb]{0,0,0}}%
      \expandafter\def\csname LT7\endcsname{\color[rgb]{1,0.3,0}}%
      \expandafter\def\csname LT8\endcsname{\color[rgb]{0.5,0.5,0.5}}%
    \else
      \def\colorrgb#1{\color{black}}%
      \def\colorgray#1{\color[gray]{#1}}%
      \expandafter\def\csname LTw\endcsname{\color{white}}%
      \expandafter\def\csname LTb\endcsname{\color{black}}%
      \expandafter\def\csname LTa\endcsname{\color{black}}%
      \expandafter\def\csname LT0\endcsname{\color{black}}%
      \expandafter\def\csname LT1\endcsname{\color{black}}%
      \expandafter\def\csname LT2\endcsname{\color{black}}%
      \expandafter\def\csname LT3\endcsname{\color{black}}%
      \expandafter\def\csname LT4\endcsname{\color{black}}%
      \expandafter\def\csname LT5\endcsname{\color{black}}%
      \expandafter\def\csname LT6\endcsname{\color{black}}%
      \expandafter\def\csname LT7\endcsname{\color{black}}%
      \expandafter\def\csname LT8\endcsname{\color{black}}%
    \fi
  \fi
    \setlength{\unitlength}{0.0500bp}%
    \ifx\gptboxheight\undefined%
      \newlength{\gptboxheight}%
      \newlength{\gptboxwidth}%
      \newsavebox{\gptboxtext}%
    \fi%
    \setlength{\fboxrule}{0.5pt}%
    \setlength{\fboxsep}{1pt}%
\begin{picture}(6802.00,3614.00)%
    \gplgaddtomacro\gplbacktext{%
      \csname LTb\endcsname%
      \put(814,704){\makebox(0,0)[r]{\strut{}$0$}}%
      \csname LTb\endcsname%
      \put(814,1113){\makebox(0,0)[r]{\strut{}$0.2$}}%
      \csname LTb\endcsname%
      \put(814,1522){\makebox(0,0)[r]{\strut{}$0.4$}}%
      \csname LTb\endcsname%
      \put(814,1931){\makebox(0,0)[r]{\strut{}$0.6$}}%
      \csname LTb\endcsname%
      \put(814,2340){\makebox(0,0)[r]{\strut{}$0.8$}}%
      \csname LTb\endcsname%
      \put(814,2749){\makebox(0,0)[r]{\strut{}$1$}}%
      \csname LTb\endcsname%
      \put(1234,484){\makebox(0,0){\strut{}$10^{3}$}}%
      \csname LTb\endcsname%
      \put(3096,484){\makebox(0,0){\strut{}$10^{4}$}}%
      \csname LTb\endcsname%
      \put(4957,484){\makebox(0,0){\strut{}$10^{5}$}}%
    }%
    \gplgaddtomacro\gplfronttext{%
      \csname LTb\endcsname%
      \put(176,1828){\rotatebox{-270}{\makebox(0,0){\strut{}Avg. Local Clustering Coeff.}}}%
      \put(3675,154){\makebox(0,0){\strut{}Number $n$ of nodes}}%
      \put(3675,3283){\makebox(0,0){\strut{}Mixing: $\mu = 0.6$}}%
      \csname LTb\endcsname%
      \put(5418,2780){\makebox(0,0)[r]{\strut{}Orig}}%
      \csname LTb\endcsname%
      \put(5418,2560){\makebox(0,0)[r]{\strut{}NetworKit}}%
      \csname LTb\endcsname%
      \put(5418,2340){\makebox(0,0)[r]{\strut{}EM}}%
    }%
    \gplbacktext
    \put(0,0){\includegraphics{lfr_no_avgcc_1_6}}%
    \gplfronttext
  \end{picture}%
\endgroup
}\hfill\\ %

	\label{fig:lfr_no}
	Comparison of the original LFR implementation, the NetworKit implementation and our EM solution for values of 
		$10^3 \le n \le 10^6$, 
		$\mu{\in}\{0.2, 0.4, 0.6\}$, 
		$\gamma{=}2$, $\beta{=}-1$
		$d_\text{min}{=}10$, $d_\text{max}{=}n/20$,
		$s_\text{min}{=}10$, $s_\text{max}{=}n/20$.
	Clustering is performed using Infomap and Louvain and compared to the ground-truth emitted by the generator using AdjustedRandMeasure (AR) and Normalized Mutual Information (NMI); $S \ge 8$.
	Due to the computational costs, graphs with $n \ge 10^5$ have a reduced multiplicity.
	In case of the original implementation it may be based on a single run which accounts for the few outliers.
}

\clearpage

{
\noindent\scalebox{\threescale}{%
\begingroup
  \makeatletter
  \providecommand\color[2][]{%
    \GenericError{(gnuplot) \space\space\space\@spaces}{%
      Package color not loaded in conjunction with
      terminal option `colourtext'%
    }{See the gnuplot documentation for explanation.%
    }{Either use 'blacktext' in gnuplot or load the package
      color.sty in LaTeX.}%
    \renewcommand\color[2][]{}%
  }%
  \providecommand\includegraphics[2][]{%
    \GenericError{(gnuplot) \space\space\space\@spaces}{%
      Package graphicx or graphics not loaded%
    }{See the gnuplot documentation for explanation.%
    }{The gnuplot epslatex terminal needs graphicx.sty or graphics.sty.}%
    \renewcommand\includegraphics[2][]{}%
  }%
  \providecommand\rotatebox[2]{#2}%
  \@ifundefined{ifGPcolor}{%
    \newif\ifGPcolor
    \GPcolortrue
  }{}%
  \@ifundefined{ifGPblacktext}{%
    \newif\ifGPblacktext
    \GPblacktextfalse
  }{}%
  \let\gplgaddtomacro\g@addto@macro
  \gdef\gplbacktext{}%
  \gdef\gplfronttext{}%
  \makeatother
  \ifGPblacktext
    \def\colorrgb#1{}%
    \def\colorgray#1{}%
  \else
    \ifGPcolor
      \def\colorrgb#1{\color[rgb]{#1}}%
      \def\colorgray#1{\color[gray]{#1}}%
      \expandafter\def\csname LTw\endcsname{\color{white}}%
      \expandafter\def\csname LTb\endcsname{\color{black}}%
      \expandafter\def\csname LTa\endcsname{\color{black}}%
      \expandafter\def\csname LT0\endcsname{\color[rgb]{1,0,0}}%
      \expandafter\def\csname LT1\endcsname{\color[rgb]{0,1,0}}%
      \expandafter\def\csname LT2\endcsname{\color[rgb]{0,0,1}}%
      \expandafter\def\csname LT3\endcsname{\color[rgb]{1,0,1}}%
      \expandafter\def\csname LT4\endcsname{\color[rgb]{0,1,1}}%
      \expandafter\def\csname LT5\endcsname{\color[rgb]{1,1,0}}%
      \expandafter\def\csname LT6\endcsname{\color[rgb]{0,0,0}}%
      \expandafter\def\csname LT7\endcsname{\color[rgb]{1,0.3,0}}%
      \expandafter\def\csname LT8\endcsname{\color[rgb]{0.5,0.5,0.5}}%
    \else
      \def\colorrgb#1{\color{black}}%
      \def\colorgray#1{\color[gray]{#1}}%
      \expandafter\def\csname LTw\endcsname{\color{white}}%
      \expandafter\def\csname LTb\endcsname{\color{black}}%
      \expandafter\def\csname LTa\endcsname{\color{black}}%
      \expandafter\def\csname LT0\endcsname{\color{black}}%
      \expandafter\def\csname LT1\endcsname{\color{black}}%
      \expandafter\def\csname LT2\endcsname{\color{black}}%
      \expandafter\def\csname LT3\endcsname{\color{black}}%
      \expandafter\def\csname LT4\endcsname{\color{black}}%
      \expandafter\def\csname LT5\endcsname{\color{black}}%
      \expandafter\def\csname LT6\endcsname{\color{black}}%
      \expandafter\def\csname LT7\endcsname{\color{black}}%
      \expandafter\def\csname LT8\endcsname{\color{black}}%
    \fi
  \fi
    \setlength{\unitlength}{0.0500bp}%
    \ifx\gptboxheight\undefined%
      \newlength{\gptboxheight}%
      \newlength{\gptboxwidth}%
      \newsavebox{\gptboxtext}%
    \fi%
    \setlength{\fboxrule}{0.5pt}%
    \setlength{\fboxsep}{1pt}%
\begin{picture}(6802.00,3614.00)%
    \gplgaddtomacro\gplbacktext{%
      \csname LTb\endcsname%
      \put(814,704){\makebox(0,0)[r]{\strut{}$0$}}%
      \csname LTb\endcsname%
      \put(814,1113){\makebox(0,0)[r]{\strut{}$0.2$}}%
      \csname LTb\endcsname%
      \put(814,1522){\makebox(0,0)[r]{\strut{}$0.4$}}%
      \csname LTb\endcsname%
      \put(814,1931){\makebox(0,0)[r]{\strut{}$0.6$}}%
      \csname LTb\endcsname%
      \put(814,2340){\makebox(0,0)[r]{\strut{}$0.8$}}%
      \csname LTb\endcsname%
      \put(814,2749){\makebox(0,0)[r]{\strut{}$1$}}%
      \csname LTb\endcsname%
      \put(1578,484){\makebox(0,0){\strut{}$10^{3}$}}%
      \csname LTb\endcsname%
      \put(3676,484){\makebox(0,0){\strut{}$10^{4}$}}%
      \csname LTb\endcsname%
      \put(5773,484){\makebox(0,0){\strut{}$10^{5}$}}%
    }%
    \gplgaddtomacro\gplfronttext{%
      \csname LTb\endcsname%
      \put(176,1828){\rotatebox{-270}{\makebox(0,0){\strut{}NMI}}}%
      \put(3675,154){\makebox(0,0){\strut{}Number $n$ of nodes}}%
      \put(3675,3283){\makebox(0,0){\strut{}Mixing: $\mu = 0.2$, Cluster: OSLOM, Overlap: $\nu = 2$}}%
      \csname LTb\endcsname%
      \put(5418,2780){\makebox(0,0)[r]{\strut{}Orig}}%
      \csname LTb\endcsname%
      \put(5418,2560){\makebox(0,0)[r]{\strut{}EM}}%
    }%
    \gplbacktext
    \put(0,0){\includegraphics{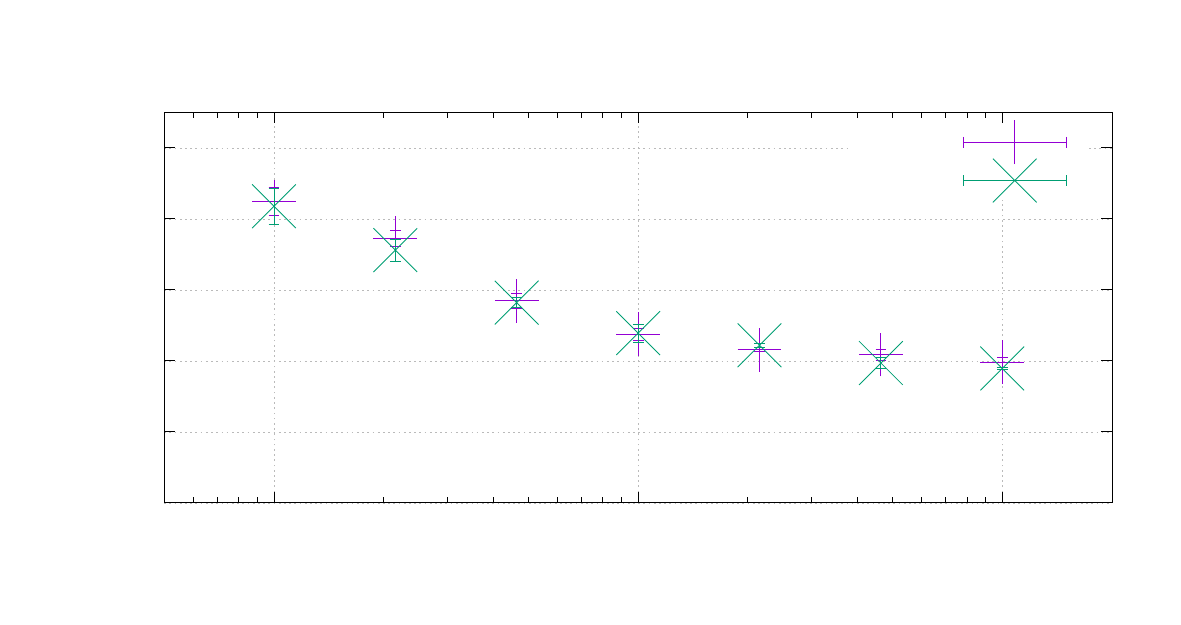}}%
    \gplfronttext
  \end{picture}%
\endgroup
}\hfill\scalebox{\threescale}{%
\begingroup
  \makeatletter
  \providecommand\color[2][]{%
    \GenericError{(gnuplot) \space\space\space\@spaces}{%
      Package color not loaded in conjunction with
      terminal option `colourtext'%
    }{See the gnuplot documentation for explanation.%
    }{Either use 'blacktext' in gnuplot or load the package
      color.sty in LaTeX.}%
    \renewcommand\color[2][]{}%
  }%
  \providecommand\includegraphics[2][]{%
    \GenericError{(gnuplot) \space\space\space\@spaces}{%
      Package graphicx or graphics not loaded%
    }{See the gnuplot documentation for explanation.%
    }{The gnuplot epslatex terminal needs graphicx.sty or graphics.sty.}%
    \renewcommand\includegraphics[2][]{}%
  }%
  \providecommand\rotatebox[2]{#2}%
  \@ifundefined{ifGPcolor}{%
    \newif\ifGPcolor
    \GPcolortrue
  }{}%
  \@ifundefined{ifGPblacktext}{%
    \newif\ifGPblacktext
    \GPblacktextfalse
  }{}%
  \let\gplgaddtomacro\g@addto@macro
  \gdef\gplbacktext{}%
  \gdef\gplfronttext{}%
  \makeatother
  \ifGPblacktext
    \def\colorrgb#1{}%
    \def\colorgray#1{}%
  \else
    \ifGPcolor
      \def\colorrgb#1{\color[rgb]{#1}}%
      \def\colorgray#1{\color[gray]{#1}}%
      \expandafter\def\csname LTw\endcsname{\color{white}}%
      \expandafter\def\csname LTb\endcsname{\color{black}}%
      \expandafter\def\csname LTa\endcsname{\color{black}}%
      \expandafter\def\csname LT0\endcsname{\color[rgb]{1,0,0}}%
      \expandafter\def\csname LT1\endcsname{\color[rgb]{0,1,0}}%
      \expandafter\def\csname LT2\endcsname{\color[rgb]{0,0,1}}%
      \expandafter\def\csname LT3\endcsname{\color[rgb]{1,0,1}}%
      \expandafter\def\csname LT4\endcsname{\color[rgb]{0,1,1}}%
      \expandafter\def\csname LT5\endcsname{\color[rgb]{1,1,0}}%
      \expandafter\def\csname LT6\endcsname{\color[rgb]{0,0,0}}%
      \expandafter\def\csname LT7\endcsname{\color[rgb]{1,0.3,0}}%
      \expandafter\def\csname LT8\endcsname{\color[rgb]{0.5,0.5,0.5}}%
    \else
      \def\colorrgb#1{\color{black}}%
      \def\colorgray#1{\color[gray]{#1}}%
      \expandafter\def\csname LTw\endcsname{\color{white}}%
      \expandafter\def\csname LTb\endcsname{\color{black}}%
      \expandafter\def\csname LTa\endcsname{\color{black}}%
      \expandafter\def\csname LT0\endcsname{\color{black}}%
      \expandafter\def\csname LT1\endcsname{\color{black}}%
      \expandafter\def\csname LT2\endcsname{\color{black}}%
      \expandafter\def\csname LT3\endcsname{\color{black}}%
      \expandafter\def\csname LT4\endcsname{\color{black}}%
      \expandafter\def\csname LT5\endcsname{\color{black}}%
      \expandafter\def\csname LT6\endcsname{\color{black}}%
      \expandafter\def\csname LT7\endcsname{\color{black}}%
      \expandafter\def\csname LT8\endcsname{\color{black}}%
    \fi
  \fi
    \setlength{\unitlength}{0.0500bp}%
    \ifx\gptboxheight\undefined%
      \newlength{\gptboxheight}%
      \newlength{\gptboxwidth}%
      \newsavebox{\gptboxtext}%
    \fi%
    \setlength{\fboxrule}{0.5pt}%
    \setlength{\fboxsep}{1pt}%
\begin{picture}(6802.00,3614.00)%
    \gplgaddtomacro\gplbacktext{%
      \csname LTb\endcsname%
      \put(814,704){\makebox(0,0)[r]{\strut{}$0$}}%
      \csname LTb\endcsname%
      \put(814,1113){\makebox(0,0)[r]{\strut{}$0.2$}}%
      \csname LTb\endcsname%
      \put(814,1522){\makebox(0,0)[r]{\strut{}$0.4$}}%
      \csname LTb\endcsname%
      \put(814,1931){\makebox(0,0)[r]{\strut{}$0.6$}}%
      \csname LTb\endcsname%
      \put(814,2340){\makebox(0,0)[r]{\strut{}$0.8$}}%
      \csname LTb\endcsname%
      \put(814,2749){\makebox(0,0)[r]{\strut{}$1$}}%
      \csname LTb\endcsname%
      \put(1578,484){\makebox(0,0){\strut{}$10^{3}$}}%
      \csname LTb\endcsname%
      \put(3676,484){\makebox(0,0){\strut{}$10^{4}$}}%
      \csname LTb\endcsname%
      \put(5773,484){\makebox(0,0){\strut{}$10^{5}$}}%
    }%
    \gplgaddtomacro\gplfronttext{%
      \csname LTb\endcsname%
      \put(176,1828){\rotatebox{-270}{\makebox(0,0){\strut{}NMI}}}%
      \put(3675,154){\makebox(0,0){\strut{}Number $n$ of nodes}}%
      \put(3675,3283){\makebox(0,0){\strut{}Mixing: $\mu = 0.4$, Cluster: OSLOM, Overlap: $\nu = 2$}}%
      \csname LTb\endcsname%
      \put(5418,2780){\makebox(0,0)[r]{\strut{}Orig}}%
      \csname LTb\endcsname%
      \put(5418,2560){\makebox(0,0)[r]{\strut{}EM}}%
    }%
    \gplbacktext
    \put(0,0){\includegraphics{lfr_nmi_2_4}}%
    \gplfronttext
  \end{picture}%
\endgroup
}\hfill\scalebox{\threescale}{%
\begingroup
  \makeatletter
  \providecommand\color[2][]{%
    \GenericError{(gnuplot) \space\space\space\@spaces}{%
      Package color not loaded in conjunction with
      terminal option `colourtext'%
    }{See the gnuplot documentation for explanation.%
    }{Either use 'blacktext' in gnuplot or load the package
      color.sty in LaTeX.}%
    \renewcommand\color[2][]{}%
  }%
  \providecommand\includegraphics[2][]{%
    \GenericError{(gnuplot) \space\space\space\@spaces}{%
      Package graphicx or graphics not loaded%
    }{See the gnuplot documentation for explanation.%
    }{The gnuplot epslatex terminal needs graphicx.sty or graphics.sty.}%
    \renewcommand\includegraphics[2][]{}%
  }%
  \providecommand\rotatebox[2]{#2}%
  \@ifundefined{ifGPcolor}{%
    \newif\ifGPcolor
    \GPcolortrue
  }{}%
  \@ifundefined{ifGPblacktext}{%
    \newif\ifGPblacktext
    \GPblacktextfalse
  }{}%
  \let\gplgaddtomacro\g@addto@macro
  \gdef\gplbacktext{}%
  \gdef\gplfronttext{}%
  \makeatother
  \ifGPblacktext
    \def\colorrgb#1{}%
    \def\colorgray#1{}%
  \else
    \ifGPcolor
      \def\colorrgb#1{\color[rgb]{#1}}%
      \def\colorgray#1{\color[gray]{#1}}%
      \expandafter\def\csname LTw\endcsname{\color{white}}%
      \expandafter\def\csname LTb\endcsname{\color{black}}%
      \expandafter\def\csname LTa\endcsname{\color{black}}%
      \expandafter\def\csname LT0\endcsname{\color[rgb]{1,0,0}}%
      \expandafter\def\csname LT1\endcsname{\color[rgb]{0,1,0}}%
      \expandafter\def\csname LT2\endcsname{\color[rgb]{0,0,1}}%
      \expandafter\def\csname LT3\endcsname{\color[rgb]{1,0,1}}%
      \expandafter\def\csname LT4\endcsname{\color[rgb]{0,1,1}}%
      \expandafter\def\csname LT5\endcsname{\color[rgb]{1,1,0}}%
      \expandafter\def\csname LT6\endcsname{\color[rgb]{0,0,0}}%
      \expandafter\def\csname LT7\endcsname{\color[rgb]{1,0.3,0}}%
      \expandafter\def\csname LT8\endcsname{\color[rgb]{0.5,0.5,0.5}}%
    \else
      \def\colorrgb#1{\color{black}}%
      \def\colorgray#1{\color[gray]{#1}}%
      \expandafter\def\csname LTw\endcsname{\color{white}}%
      \expandafter\def\csname LTb\endcsname{\color{black}}%
      \expandafter\def\csname LTa\endcsname{\color{black}}%
      \expandafter\def\csname LT0\endcsname{\color{black}}%
      \expandafter\def\csname LT1\endcsname{\color{black}}%
      \expandafter\def\csname LT2\endcsname{\color{black}}%
      \expandafter\def\csname LT3\endcsname{\color{black}}%
      \expandafter\def\csname LT4\endcsname{\color{black}}%
      \expandafter\def\csname LT5\endcsname{\color{black}}%
      \expandafter\def\csname LT6\endcsname{\color{black}}%
      \expandafter\def\csname LT7\endcsname{\color{black}}%
      \expandafter\def\csname LT8\endcsname{\color{black}}%
    \fi
  \fi
    \setlength{\unitlength}{0.0500bp}%
    \ifx\gptboxheight\undefined%
      \newlength{\gptboxheight}%
      \newlength{\gptboxwidth}%
      \newsavebox{\gptboxtext}%
    \fi%
    \setlength{\fboxrule}{0.5pt}%
    \setlength{\fboxsep}{1pt}%
\begin{picture}(6802.00,3614.00)%
    \gplgaddtomacro\gplbacktext{%
      \csname LTb\endcsname%
      \put(814,704){\makebox(0,0)[r]{\strut{}$0$}}%
      \csname LTb\endcsname%
      \put(814,1113){\makebox(0,0)[r]{\strut{}$0.2$}}%
      \csname LTb\endcsname%
      \put(814,1522){\makebox(0,0)[r]{\strut{}$0.4$}}%
      \csname LTb\endcsname%
      \put(814,1931){\makebox(0,0)[r]{\strut{}$0.6$}}%
      \csname LTb\endcsname%
      \put(814,2340){\makebox(0,0)[r]{\strut{}$0.8$}}%
      \csname LTb\endcsname%
      \put(814,2749){\makebox(0,0)[r]{\strut{}$1$}}%
      \csname LTb\endcsname%
      \put(1578,484){\makebox(0,0){\strut{}$10^{3}$}}%
      \csname LTb\endcsname%
      \put(3676,484){\makebox(0,0){\strut{}$10^{4}$}}%
      \csname LTb\endcsname%
      \put(5773,484){\makebox(0,0){\strut{}$10^{5}$}}%
    }%
    \gplgaddtomacro\gplfronttext{%
      \csname LTb\endcsname%
      \put(176,1828){\rotatebox{-270}{\makebox(0,0){\strut{}NMI}}}%
      \put(3675,154){\makebox(0,0){\strut{}Number $n$ of nodes}}%
      \put(3675,3283){\makebox(0,0){\strut{}Mixing: $\mu = 0.6$, Cluster: OSLOM, Overlap: $\nu = 2$}}%
      \csname LTb\endcsname%
      \put(5418,2780){\makebox(0,0)[r]{\strut{}Orig}}%
      \csname LTb\endcsname%
      \put(5418,2560){\makebox(0,0)[r]{\strut{}EM}}%
    }%
    \gplbacktext
    \put(0,0){\includegraphics{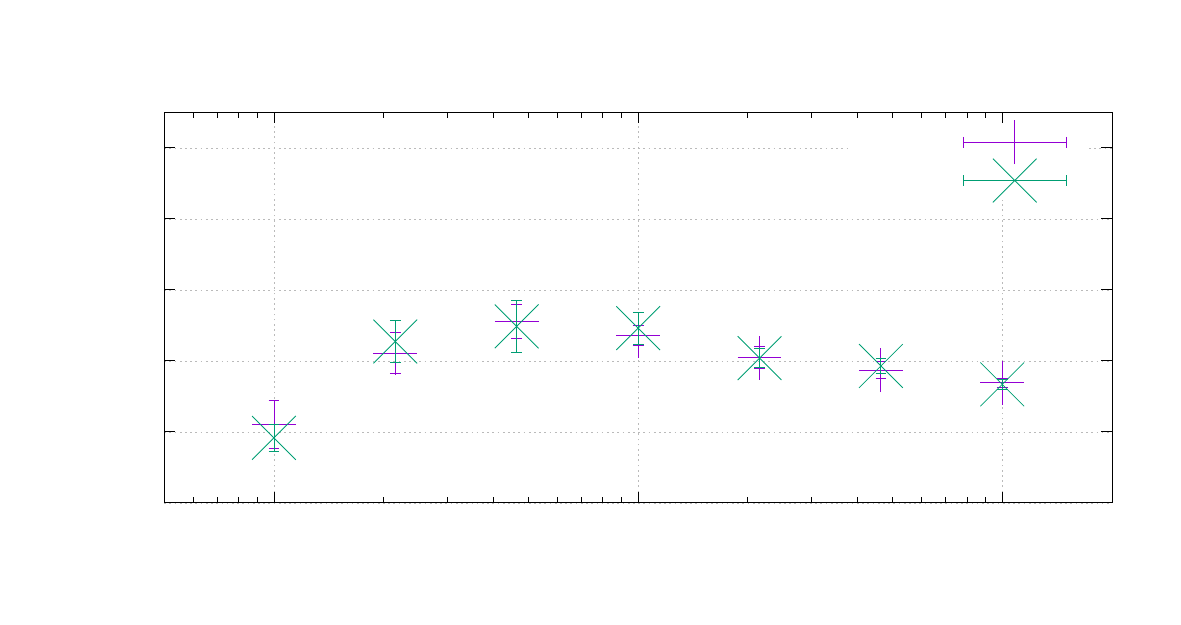}}%
    \gplfronttext
  \end{picture}%
\endgroup
}\hfill\\ %
\noindent\scalebox{\threescale}{%
\begingroup
  \makeatletter
  \providecommand\color[2][]{%
    \GenericError{(gnuplot) \space\space\space\@spaces}{%
      Package color not loaded in conjunction with
      terminal option `colourtext'%
    }{See the gnuplot documentation for explanation.%
    }{Either use 'blacktext' in gnuplot or load the package
      color.sty in LaTeX.}%
    \renewcommand\color[2][]{}%
  }%
  \providecommand\includegraphics[2][]{%
    \GenericError{(gnuplot) \space\space\space\@spaces}{%
      Package graphicx or graphics not loaded%
    }{See the gnuplot documentation for explanation.%
    }{The gnuplot epslatex terminal needs graphicx.sty or graphics.sty.}%
    \renewcommand\includegraphics[2][]{}%
  }%
  \providecommand\rotatebox[2]{#2}%
  \@ifundefined{ifGPcolor}{%
    \newif\ifGPcolor
    \GPcolortrue
  }{}%
  \@ifundefined{ifGPblacktext}{%
    \newif\ifGPblacktext
    \GPblacktextfalse
  }{}%
  \let\gplgaddtomacro\g@addto@macro
  \gdef\gplbacktext{}%
  \gdef\gplfronttext{}%
  \makeatother
  \ifGPblacktext
    \def\colorrgb#1{}%
    \def\colorgray#1{}%
  \else
    \ifGPcolor
      \def\colorrgb#1{\color[rgb]{#1}}%
      \def\colorgray#1{\color[gray]{#1}}%
      \expandafter\def\csname LTw\endcsname{\color{white}}%
      \expandafter\def\csname LTb\endcsname{\color{black}}%
      \expandafter\def\csname LTa\endcsname{\color{black}}%
      \expandafter\def\csname LT0\endcsname{\color[rgb]{1,0,0}}%
      \expandafter\def\csname LT1\endcsname{\color[rgb]{0,1,0}}%
      \expandafter\def\csname LT2\endcsname{\color[rgb]{0,0,1}}%
      \expandafter\def\csname LT3\endcsname{\color[rgb]{1,0,1}}%
      \expandafter\def\csname LT4\endcsname{\color[rgb]{0,1,1}}%
      \expandafter\def\csname LT5\endcsname{\color[rgb]{1,1,0}}%
      \expandafter\def\csname LT6\endcsname{\color[rgb]{0,0,0}}%
      \expandafter\def\csname LT7\endcsname{\color[rgb]{1,0.3,0}}%
      \expandafter\def\csname LT8\endcsname{\color[rgb]{0.5,0.5,0.5}}%
    \else
      \def\colorrgb#1{\color{black}}%
      \def\colorgray#1{\color[gray]{#1}}%
      \expandafter\def\csname LTw\endcsname{\color{white}}%
      \expandafter\def\csname LTb\endcsname{\color{black}}%
      \expandafter\def\csname LTa\endcsname{\color{black}}%
      \expandafter\def\csname LT0\endcsname{\color{black}}%
      \expandafter\def\csname LT1\endcsname{\color{black}}%
      \expandafter\def\csname LT2\endcsname{\color{black}}%
      \expandafter\def\csname LT3\endcsname{\color{black}}%
      \expandafter\def\csname LT4\endcsname{\color{black}}%
      \expandafter\def\csname LT5\endcsname{\color{black}}%
      \expandafter\def\csname LT6\endcsname{\color{black}}%
      \expandafter\def\csname LT7\endcsname{\color{black}}%
      \expandafter\def\csname LT8\endcsname{\color{black}}%
    \fi
  \fi
    \setlength{\unitlength}{0.0500bp}%
    \ifx\gptboxheight\undefined%
      \newlength{\gptboxheight}%
      \newlength{\gptboxwidth}%
      \newsavebox{\gptboxtext}%
    \fi%
    \setlength{\fboxrule}{0.5pt}%
    \setlength{\fboxsep}{1pt}%
\begin{picture}(6802.00,3614.00)%
    \gplgaddtomacro\gplbacktext{%
      \csname LTb\endcsname%
      \put(814,704){\makebox(0,0)[r]{\strut{}$0$}}%
      \csname LTb\endcsname%
      \put(814,1113){\makebox(0,0)[r]{\strut{}$0.2$}}%
      \csname LTb\endcsname%
      \put(814,1522){\makebox(0,0)[r]{\strut{}$0.4$}}%
      \csname LTb\endcsname%
      \put(814,1931){\makebox(0,0)[r]{\strut{}$0.6$}}%
      \csname LTb\endcsname%
      \put(814,2340){\makebox(0,0)[r]{\strut{}$0.8$}}%
      \csname LTb\endcsname%
      \put(814,2749){\makebox(0,0)[r]{\strut{}$1$}}%
      \csname LTb\endcsname%
      \put(1578,484){\makebox(0,0){\strut{}$10^{3}$}}%
      \csname LTb\endcsname%
      \put(3676,484){\makebox(0,0){\strut{}$10^{4}$}}%
      \csname LTb\endcsname%
      \put(5773,484){\makebox(0,0){\strut{}$10^{5}$}}%
    }%
    \gplgaddtomacro\gplfronttext{%
      \csname LTb\endcsname%
      \put(176,1828){\rotatebox{-270}{\makebox(0,0){\strut{}Avg. Local Clustering Coeff.}}}%
      \put(3675,154){\makebox(0,0){\strut{}Number $n$ of nodes}}%
      \put(3675,3283){\makebox(0,0){\strut{}Mixing: $\mu = 0.2$}}%
      \csname LTb\endcsname%
      \put(5418,2780){\makebox(0,0)[r]{\strut{}Orig}}%
      \csname LTb\endcsname%
      \put(5418,2560){\makebox(0,0)[r]{\strut{}EM}}%
    }%
    \gplbacktext
    \put(0,0){\includegraphics{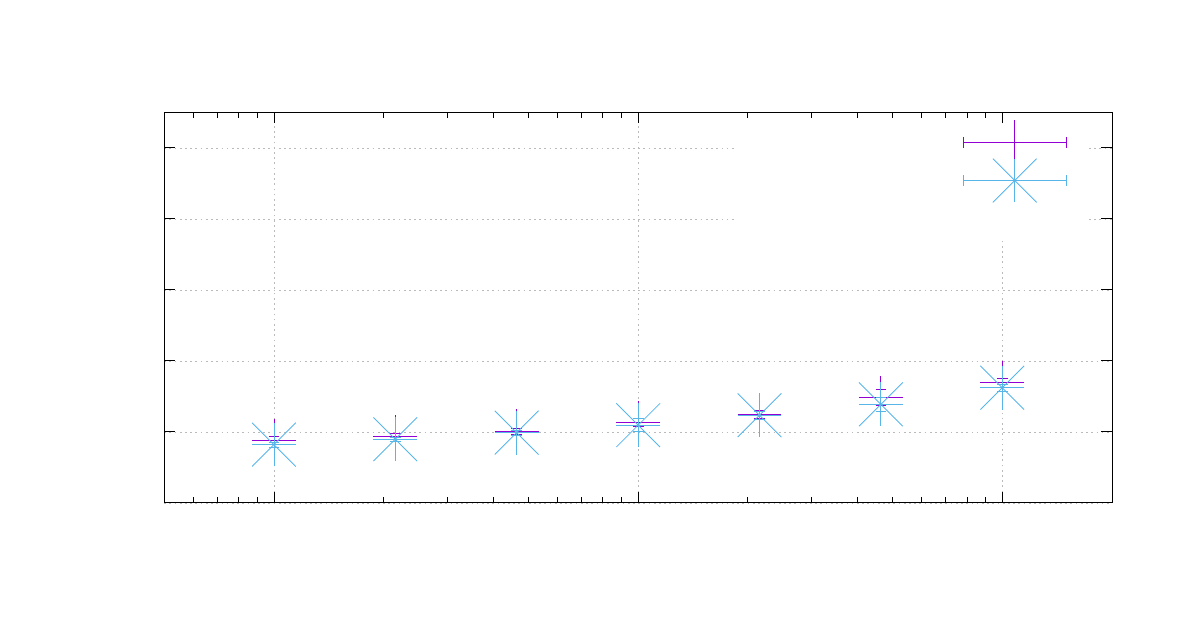}}%
    \gplfronttext
  \end{picture}%
\endgroup
}\hfill\scalebox{\threescale}{%
\begingroup
  \makeatletter
  \providecommand\color[2][]{%
    \GenericError{(gnuplot) \space\space\space\@spaces}{%
      Package color not loaded in conjunction with
      terminal option `colourtext'%
    }{See the gnuplot documentation for explanation.%
    }{Either use 'blacktext' in gnuplot or load the package
      color.sty in LaTeX.}%
    \renewcommand\color[2][]{}%
  }%
  \providecommand\includegraphics[2][]{%
    \GenericError{(gnuplot) \space\space\space\@spaces}{%
      Package graphicx or graphics not loaded%
    }{See the gnuplot documentation for explanation.%
    }{The gnuplot epslatex terminal needs graphicx.sty or graphics.sty.}%
    \renewcommand\includegraphics[2][]{}%
  }%
  \providecommand\rotatebox[2]{#2}%
  \@ifundefined{ifGPcolor}{%
    \newif\ifGPcolor
    \GPcolortrue
  }{}%
  \@ifundefined{ifGPblacktext}{%
    \newif\ifGPblacktext
    \GPblacktextfalse
  }{}%
  \let\gplgaddtomacro\g@addto@macro
  \gdef\gplbacktext{}%
  \gdef\gplfronttext{}%
  \makeatother
  \ifGPblacktext
    \def\colorrgb#1{}%
    \def\colorgray#1{}%
  \else
    \ifGPcolor
      \def\colorrgb#1{\color[rgb]{#1}}%
      \def\colorgray#1{\color[gray]{#1}}%
      \expandafter\def\csname LTw\endcsname{\color{white}}%
      \expandafter\def\csname LTb\endcsname{\color{black}}%
      \expandafter\def\csname LTa\endcsname{\color{black}}%
      \expandafter\def\csname LT0\endcsname{\color[rgb]{1,0,0}}%
      \expandafter\def\csname LT1\endcsname{\color[rgb]{0,1,0}}%
      \expandafter\def\csname LT2\endcsname{\color[rgb]{0,0,1}}%
      \expandafter\def\csname LT3\endcsname{\color[rgb]{1,0,1}}%
      \expandafter\def\csname LT4\endcsname{\color[rgb]{0,1,1}}%
      \expandafter\def\csname LT5\endcsname{\color[rgb]{1,1,0}}%
      \expandafter\def\csname LT6\endcsname{\color[rgb]{0,0,0}}%
      \expandafter\def\csname LT7\endcsname{\color[rgb]{1,0.3,0}}%
      \expandafter\def\csname LT8\endcsname{\color[rgb]{0.5,0.5,0.5}}%
    \else
      \def\colorrgb#1{\color{black}}%
      \def\colorgray#1{\color[gray]{#1}}%
      \expandafter\def\csname LTw\endcsname{\color{white}}%
      \expandafter\def\csname LTb\endcsname{\color{black}}%
      \expandafter\def\csname LTa\endcsname{\color{black}}%
      \expandafter\def\csname LT0\endcsname{\color{black}}%
      \expandafter\def\csname LT1\endcsname{\color{black}}%
      \expandafter\def\csname LT2\endcsname{\color{black}}%
      \expandafter\def\csname LT3\endcsname{\color{black}}%
      \expandafter\def\csname LT4\endcsname{\color{black}}%
      \expandafter\def\csname LT5\endcsname{\color{black}}%
      \expandafter\def\csname LT6\endcsname{\color{black}}%
      \expandafter\def\csname LT7\endcsname{\color{black}}%
      \expandafter\def\csname LT8\endcsname{\color{black}}%
    \fi
  \fi
    \setlength{\unitlength}{0.0500bp}%
    \ifx\gptboxheight\undefined%
      \newlength{\gptboxheight}%
      \newlength{\gptboxwidth}%
      \newsavebox{\gptboxtext}%
    \fi%
    \setlength{\fboxrule}{0.5pt}%
    \setlength{\fboxsep}{1pt}%
\begin{picture}(6802.00,3614.00)%
    \gplgaddtomacro\gplbacktext{%
      \csname LTb\endcsname%
      \put(814,704){\makebox(0,0)[r]{\strut{}$0$}}%
      \csname LTb\endcsname%
      \put(814,1113){\makebox(0,0)[r]{\strut{}$0.2$}}%
      \csname LTb\endcsname%
      \put(814,1522){\makebox(0,0)[r]{\strut{}$0.4$}}%
      \csname LTb\endcsname%
      \put(814,1931){\makebox(0,0)[r]{\strut{}$0.6$}}%
      \csname LTb\endcsname%
      \put(814,2340){\makebox(0,0)[r]{\strut{}$0.8$}}%
      \csname LTb\endcsname%
      \put(814,2749){\makebox(0,0)[r]{\strut{}$1$}}%
      \csname LTb\endcsname%
      \put(1578,484){\makebox(0,0){\strut{}$10^{3}$}}%
      \csname LTb\endcsname%
      \put(3676,484){\makebox(0,0){\strut{}$10^{4}$}}%
      \csname LTb\endcsname%
      \put(5773,484){\makebox(0,0){\strut{}$10^{5}$}}%
    }%
    \gplgaddtomacro\gplfronttext{%
      \csname LTb\endcsname%
      \put(176,1828){\rotatebox{-270}{\makebox(0,0){\strut{}Avg. Local Clustering Coeff.}}}%
      \put(3675,154){\makebox(0,0){\strut{}Number $n$ of nodes}}%
      \put(3675,3283){\makebox(0,0){\strut{}Mixing: $\mu = 0.4$}}%
      \csname LTb\endcsname%
      \put(5418,2780){\makebox(0,0)[r]{\strut{}Orig}}%
      \csname LTb\endcsname%
      \put(5418,2560){\makebox(0,0)[r]{\strut{}EM}}%
    }%
    \gplbacktext
    \put(0,0){\includegraphics{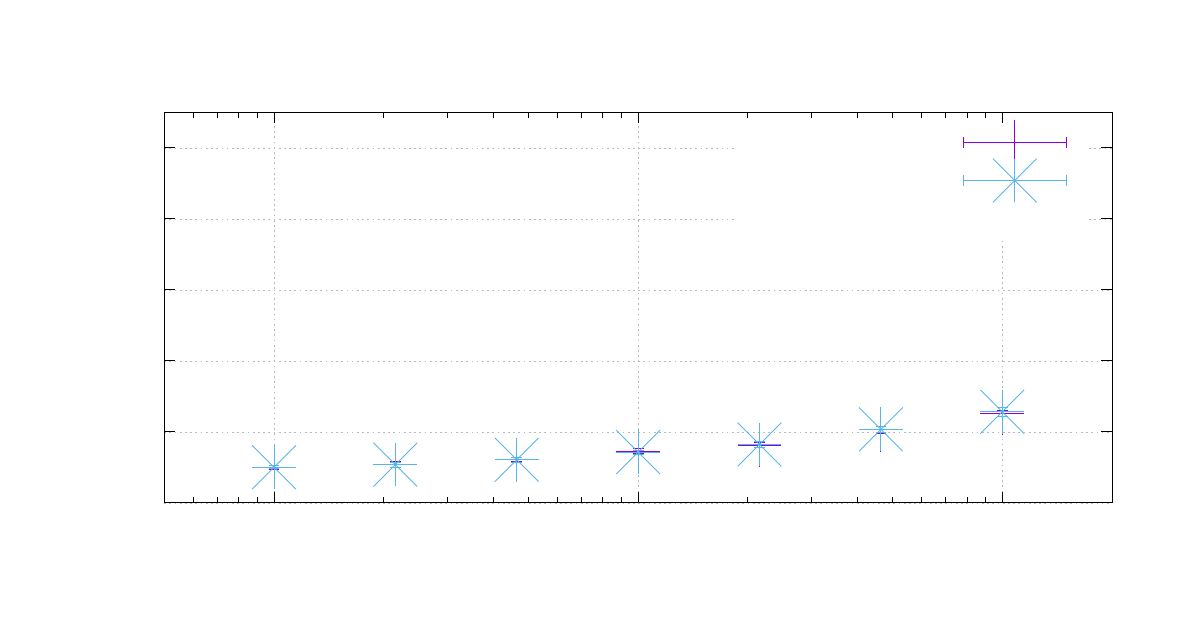}}%
    \gplfronttext
  \end{picture}%
\endgroup
}\hfill\scalebox{\threescale}{%
\begingroup
  \makeatletter
  \providecommand\color[2][]{%
    \GenericError{(gnuplot) \space\space\space\@spaces}{%
      Package color not loaded in conjunction with
      terminal option `colourtext'%
    }{See the gnuplot documentation for explanation.%
    }{Either use 'blacktext' in gnuplot or load the package
      color.sty in LaTeX.}%
    \renewcommand\color[2][]{}%
  }%
  \providecommand\includegraphics[2][]{%
    \GenericError{(gnuplot) \space\space\space\@spaces}{%
      Package graphicx or graphics not loaded%
    }{See the gnuplot documentation for explanation.%
    }{The gnuplot epslatex terminal needs graphicx.sty or graphics.sty.}%
    \renewcommand\includegraphics[2][]{}%
  }%
  \providecommand\rotatebox[2]{#2}%
  \@ifundefined{ifGPcolor}{%
    \newif\ifGPcolor
    \GPcolortrue
  }{}%
  \@ifundefined{ifGPblacktext}{%
    \newif\ifGPblacktext
    \GPblacktextfalse
  }{}%
  \let\gplgaddtomacro\g@addto@macro
  \gdef\gplbacktext{}%
  \gdef\gplfronttext{}%
  \makeatother
  \ifGPblacktext
    \def\colorrgb#1{}%
    \def\colorgray#1{}%
  \else
    \ifGPcolor
      \def\colorrgb#1{\color[rgb]{#1}}%
      \def\colorgray#1{\color[gray]{#1}}%
      \expandafter\def\csname LTw\endcsname{\color{white}}%
      \expandafter\def\csname LTb\endcsname{\color{black}}%
      \expandafter\def\csname LTa\endcsname{\color{black}}%
      \expandafter\def\csname LT0\endcsname{\color[rgb]{1,0,0}}%
      \expandafter\def\csname LT1\endcsname{\color[rgb]{0,1,0}}%
      \expandafter\def\csname LT2\endcsname{\color[rgb]{0,0,1}}%
      \expandafter\def\csname LT3\endcsname{\color[rgb]{1,0,1}}%
      \expandafter\def\csname LT4\endcsname{\color[rgb]{0,1,1}}%
      \expandafter\def\csname LT5\endcsname{\color[rgb]{1,1,0}}%
      \expandafter\def\csname LT6\endcsname{\color[rgb]{0,0,0}}%
      \expandafter\def\csname LT7\endcsname{\color[rgb]{1,0.3,0}}%
      \expandafter\def\csname LT8\endcsname{\color[rgb]{0.5,0.5,0.5}}%
    \else
      \def\colorrgb#1{\color{black}}%
      \def\colorgray#1{\color[gray]{#1}}%
      \expandafter\def\csname LTw\endcsname{\color{white}}%
      \expandafter\def\csname LTb\endcsname{\color{black}}%
      \expandafter\def\csname LTa\endcsname{\color{black}}%
      \expandafter\def\csname LT0\endcsname{\color{black}}%
      \expandafter\def\csname LT1\endcsname{\color{black}}%
      \expandafter\def\csname LT2\endcsname{\color{black}}%
      \expandafter\def\csname LT3\endcsname{\color{black}}%
      \expandafter\def\csname LT4\endcsname{\color{black}}%
      \expandafter\def\csname LT5\endcsname{\color{black}}%
      \expandafter\def\csname LT6\endcsname{\color{black}}%
      \expandafter\def\csname LT7\endcsname{\color{black}}%
      \expandafter\def\csname LT8\endcsname{\color{black}}%
    \fi
  \fi
    \setlength{\unitlength}{0.0500bp}%
    \ifx\gptboxheight\undefined%
      \newlength{\gptboxheight}%
      \newlength{\gptboxwidth}%
      \newsavebox{\gptboxtext}%
    \fi%
    \setlength{\fboxrule}{0.5pt}%
    \setlength{\fboxsep}{1pt}%
\begin{picture}(6802.00,3614.00)%
    \gplgaddtomacro\gplbacktext{%
      \csname LTb\endcsname%
      \put(814,704){\makebox(0,0)[r]{\strut{}$0$}}%
      \csname LTb\endcsname%
      \put(814,1113){\makebox(0,0)[r]{\strut{}$0.2$}}%
      \csname LTb\endcsname%
      \put(814,1522){\makebox(0,0)[r]{\strut{}$0.4$}}%
      \csname LTb\endcsname%
      \put(814,1931){\makebox(0,0)[r]{\strut{}$0.6$}}%
      \csname LTb\endcsname%
      \put(814,2340){\makebox(0,0)[r]{\strut{}$0.8$}}%
      \csname LTb\endcsname%
      \put(814,2749){\makebox(0,0)[r]{\strut{}$1$}}%
      \csname LTb\endcsname%
      \put(1578,484){\makebox(0,0){\strut{}$10^{3}$}}%
      \csname LTb\endcsname%
      \put(3676,484){\makebox(0,0){\strut{}$10^{4}$}}%
      \csname LTb\endcsname%
      \put(5773,484){\makebox(0,0){\strut{}$10^{5}$}}%
    }%
    \gplgaddtomacro\gplfronttext{%
      \csname LTb\endcsname%
      \put(176,1828){\rotatebox{-270}{\makebox(0,0){\strut{}Avg. Local Clustering Coeff.}}}%
      \put(3675,154){\makebox(0,0){\strut{}Number $n$ of nodes}}%
      \put(3675,3283){\makebox(0,0){\strut{}Mixing: $\mu = 0.6$}}%
      \csname LTb\endcsname%
      \put(5418,2780){\makebox(0,0)[r]{\strut{}Orig}}%
      \csname LTb\endcsname%
      \put(5418,2560){\makebox(0,0)[r]{\strut{}EM}}%
    }%
    \gplbacktext
    \put(0,0){\includegraphics{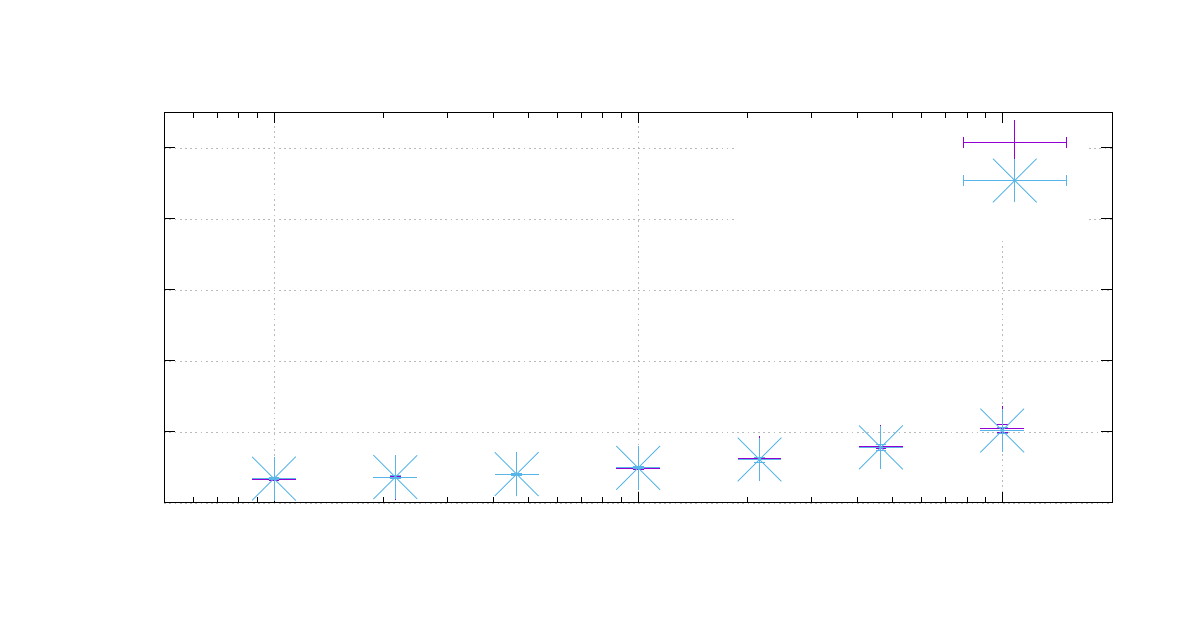}}%
    \gplfronttext
  \end{picture}%
\endgroup
}\hfill\\ %
\noindent\scalebox{\threescale}{%
\begingroup
  \makeatletter
  \providecommand\color[2][]{%
    \GenericError{(gnuplot) \space\space\space\@spaces}{%
      Package color not loaded in conjunction with
      terminal option `colourtext'%
    }{See the gnuplot documentation for explanation.%
    }{Either use 'blacktext' in gnuplot or load the package
      color.sty in LaTeX.}%
    \renewcommand\color[2][]{}%
  }%
  \providecommand\includegraphics[2][]{%
    \GenericError{(gnuplot) \space\space\space\@spaces}{%
      Package graphicx or graphics not loaded%
    }{See the gnuplot documentation for explanation.%
    }{The gnuplot epslatex terminal needs graphicx.sty or graphics.sty.}%
    \renewcommand\includegraphics[2][]{}%
  }%
  \providecommand\rotatebox[2]{#2}%
  \@ifundefined{ifGPcolor}{%
    \newif\ifGPcolor
    \GPcolortrue
  }{}%
  \@ifundefined{ifGPblacktext}{%
    \newif\ifGPblacktext
    \GPblacktextfalse
  }{}%
  \let\gplgaddtomacro\g@addto@macro
  \gdef\gplbacktext{}%
  \gdef\gplfronttext{}%
  \makeatother
  \ifGPblacktext
    \def\colorrgb#1{}%
    \def\colorgray#1{}%
  \else
    \ifGPcolor
      \def\colorrgb#1{\color[rgb]{#1}}%
      \def\colorgray#1{\color[gray]{#1}}%
      \expandafter\def\csname LTw\endcsname{\color{white}}%
      \expandafter\def\csname LTb\endcsname{\color{black}}%
      \expandafter\def\csname LTa\endcsname{\color{black}}%
      \expandafter\def\csname LT0\endcsname{\color[rgb]{1,0,0}}%
      \expandafter\def\csname LT1\endcsname{\color[rgb]{0,1,0}}%
      \expandafter\def\csname LT2\endcsname{\color[rgb]{0,0,1}}%
      \expandafter\def\csname LT3\endcsname{\color[rgb]{1,0,1}}%
      \expandafter\def\csname LT4\endcsname{\color[rgb]{0,1,1}}%
      \expandafter\def\csname LT5\endcsname{\color[rgb]{1,1,0}}%
      \expandafter\def\csname LT6\endcsname{\color[rgb]{0,0,0}}%
      \expandafter\def\csname LT7\endcsname{\color[rgb]{1,0.3,0}}%
      \expandafter\def\csname LT8\endcsname{\color[rgb]{0.5,0.5,0.5}}%
    \else
      \def\colorrgb#1{\color{black}}%
      \def\colorgray#1{\color[gray]{#1}}%
      \expandafter\def\csname LTw\endcsname{\color{white}}%
      \expandafter\def\csname LTb\endcsname{\color{black}}%
      \expandafter\def\csname LTa\endcsname{\color{black}}%
      \expandafter\def\csname LT0\endcsname{\color{black}}%
      \expandafter\def\csname LT1\endcsname{\color{black}}%
      \expandafter\def\csname LT2\endcsname{\color{black}}%
      \expandafter\def\csname LT3\endcsname{\color{black}}%
      \expandafter\def\csname LT4\endcsname{\color{black}}%
      \expandafter\def\csname LT5\endcsname{\color{black}}%
      \expandafter\def\csname LT6\endcsname{\color{black}}%
      \expandafter\def\csname LT7\endcsname{\color{black}}%
      \expandafter\def\csname LT8\endcsname{\color{black}}%
    \fi
  \fi
    \setlength{\unitlength}{0.0500bp}%
    \ifx\gptboxheight\undefined%
      \newlength{\gptboxheight}%
      \newlength{\gptboxwidth}%
      \newsavebox{\gptboxtext}%
    \fi%
    \setlength{\fboxrule}{0.5pt}%
    \setlength{\fboxsep}{1pt}%
\begin{picture}(6802.00,3614.00)%
    \gplgaddtomacro\gplbacktext{%
      \csname LTb\endcsname%
      \put(814,704){\makebox(0,0)[r]{\strut{}$0$}}%
      \csname LTb\endcsname%
      \put(814,1113){\makebox(0,0)[r]{\strut{}$0.2$}}%
      \csname LTb\endcsname%
      \put(814,1522){\makebox(0,0)[r]{\strut{}$0.4$}}%
      \csname LTb\endcsname%
      \put(814,1931){\makebox(0,0)[r]{\strut{}$0.6$}}%
      \csname LTb\endcsname%
      \put(814,2340){\makebox(0,0)[r]{\strut{}$0.8$}}%
      \csname LTb\endcsname%
      \put(814,2749){\makebox(0,0)[r]{\strut{}$1$}}%
      \csname LTb\endcsname%
      \put(1578,484){\makebox(0,0){\strut{}$10^{3}$}}%
      \csname LTb\endcsname%
      \put(3676,484){\makebox(0,0){\strut{}$10^{4}$}}%
      \csname LTb\endcsname%
      \put(5773,484){\makebox(0,0){\strut{}$10^{5}$}}%
    }%
    \gplgaddtomacro\gplfronttext{%
      \csname LTb\endcsname%
      \put(176,1828){\rotatebox{-270}{\makebox(0,0){\strut{}Degree Assortativity}}}%
      \put(3675,154){\makebox(0,0){\strut{}Number $n$ of nodes}}%
      \put(3675,3283){\makebox(0,0){\strut{}Mixing: $\mu = 0.2$, Degree Assortativity, Overlap: $\nu = 2$}}%
      \csname LTb\endcsname%
      \put(5418,2780){\makebox(0,0)[r]{\strut{}Orig}}%
      \csname LTb\endcsname%
      \put(5418,2560){\makebox(0,0)[r]{\strut{}EM}}%
    }%
    \gplbacktext
    \put(0,0){\includegraphics{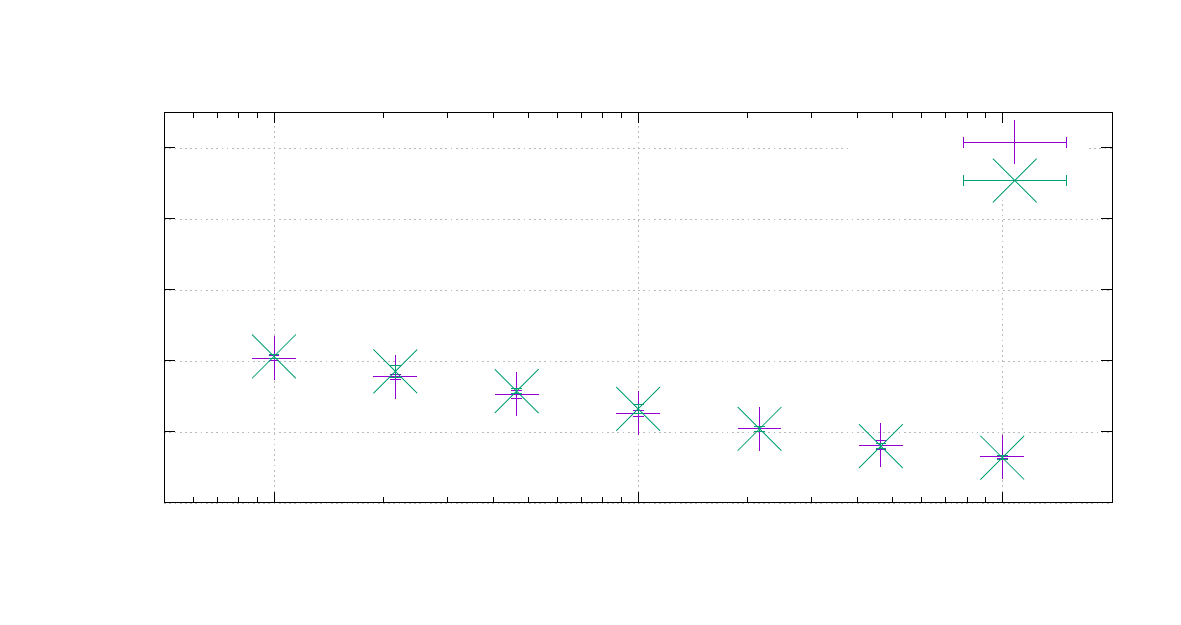}}%
    \gplfronttext
  \end{picture}%
\endgroup
}\hfill\scalebox{\threescale}{%
\begingroup
  \makeatletter
  \providecommand\color[2][]{%
    \GenericError{(gnuplot) \space\space\space\@spaces}{%
      Package color not loaded in conjunction with
      terminal option `colourtext'%
    }{See the gnuplot documentation for explanation.%
    }{Either use 'blacktext' in gnuplot or load the package
      color.sty in LaTeX.}%
    \renewcommand\color[2][]{}%
  }%
  \providecommand\includegraphics[2][]{%
    \GenericError{(gnuplot) \space\space\space\@spaces}{%
      Package graphicx or graphics not loaded%
    }{See the gnuplot documentation for explanation.%
    }{The gnuplot epslatex terminal needs graphicx.sty or graphics.sty.}%
    \renewcommand\includegraphics[2][]{}%
  }%
  \providecommand\rotatebox[2]{#2}%
  \@ifundefined{ifGPcolor}{%
    \newif\ifGPcolor
    \GPcolortrue
  }{}%
  \@ifundefined{ifGPblacktext}{%
    \newif\ifGPblacktext
    \GPblacktextfalse
  }{}%
  \let\gplgaddtomacro\g@addto@macro
  \gdef\gplbacktext{}%
  \gdef\gplfronttext{}%
  \makeatother
  \ifGPblacktext
    \def\colorrgb#1{}%
    \def\colorgray#1{}%
  \else
    \ifGPcolor
      \def\colorrgb#1{\color[rgb]{#1}}%
      \def\colorgray#1{\color[gray]{#1}}%
      \expandafter\def\csname LTw\endcsname{\color{white}}%
      \expandafter\def\csname LTb\endcsname{\color{black}}%
      \expandafter\def\csname LTa\endcsname{\color{black}}%
      \expandafter\def\csname LT0\endcsname{\color[rgb]{1,0,0}}%
      \expandafter\def\csname LT1\endcsname{\color[rgb]{0,1,0}}%
      \expandafter\def\csname LT2\endcsname{\color[rgb]{0,0,1}}%
      \expandafter\def\csname LT3\endcsname{\color[rgb]{1,0,1}}%
      \expandafter\def\csname LT4\endcsname{\color[rgb]{0,1,1}}%
      \expandafter\def\csname LT5\endcsname{\color[rgb]{1,1,0}}%
      \expandafter\def\csname LT6\endcsname{\color[rgb]{0,0,0}}%
      \expandafter\def\csname LT7\endcsname{\color[rgb]{1,0.3,0}}%
      \expandafter\def\csname LT8\endcsname{\color[rgb]{0.5,0.5,0.5}}%
    \else
      \def\colorrgb#1{\color{black}}%
      \def\colorgray#1{\color[gray]{#1}}%
      \expandafter\def\csname LTw\endcsname{\color{white}}%
      \expandafter\def\csname LTb\endcsname{\color{black}}%
      \expandafter\def\csname LTa\endcsname{\color{black}}%
      \expandafter\def\csname LT0\endcsname{\color{black}}%
      \expandafter\def\csname LT1\endcsname{\color{black}}%
      \expandafter\def\csname LT2\endcsname{\color{black}}%
      \expandafter\def\csname LT3\endcsname{\color{black}}%
      \expandafter\def\csname LT4\endcsname{\color{black}}%
      \expandafter\def\csname LT5\endcsname{\color{black}}%
      \expandafter\def\csname LT6\endcsname{\color{black}}%
      \expandafter\def\csname LT7\endcsname{\color{black}}%
      \expandafter\def\csname LT8\endcsname{\color{black}}%
    \fi
  \fi
    \setlength{\unitlength}{0.0500bp}%
    \ifx\gptboxheight\undefined%
      \newlength{\gptboxheight}%
      \newlength{\gptboxwidth}%
      \newsavebox{\gptboxtext}%
    \fi%
    \setlength{\fboxrule}{0.5pt}%
    \setlength{\fboxsep}{1pt}%
\begin{picture}(6802.00,3614.00)%
    \gplgaddtomacro\gplbacktext{%
      \csname LTb\endcsname%
      \put(814,704){\makebox(0,0)[r]{\strut{}$0$}}%
      \csname LTb\endcsname%
      \put(814,1113){\makebox(0,0)[r]{\strut{}$0.2$}}%
      \csname LTb\endcsname%
      \put(814,1522){\makebox(0,0)[r]{\strut{}$0.4$}}%
      \csname LTb\endcsname%
      \put(814,1931){\makebox(0,0)[r]{\strut{}$0.6$}}%
      \csname LTb\endcsname%
      \put(814,2340){\makebox(0,0)[r]{\strut{}$0.8$}}%
      \csname LTb\endcsname%
      \put(814,2749){\makebox(0,0)[r]{\strut{}$1$}}%
      \csname LTb\endcsname%
      \put(1578,484){\makebox(0,0){\strut{}$10^{3}$}}%
      \csname LTb\endcsname%
      \put(3676,484){\makebox(0,0){\strut{}$10^{4}$}}%
      \csname LTb\endcsname%
      \put(5773,484){\makebox(0,0){\strut{}$10^{5}$}}%
    }%
    \gplgaddtomacro\gplfronttext{%
      \csname LTb\endcsname%
      \put(176,1828){\rotatebox{-270}{\makebox(0,0){\strut{}Degree Assortativity}}}%
      \put(3675,154){\makebox(0,0){\strut{}Number $n$ of nodes}}%
      \put(3675,3283){\makebox(0,0){\strut{}Mixing: $\mu = 0.4$, Degree Assortativity, Overlap: $\nu = 2$}}%
      \csname LTb\endcsname%
      \put(5418,2780){\makebox(0,0)[r]{\strut{}Orig}}%
      \csname LTb\endcsname%
      \put(5418,2560){\makebox(0,0)[r]{\strut{}EM}}%
    }%
    \gplbacktext
    \put(0,0){\includegraphics{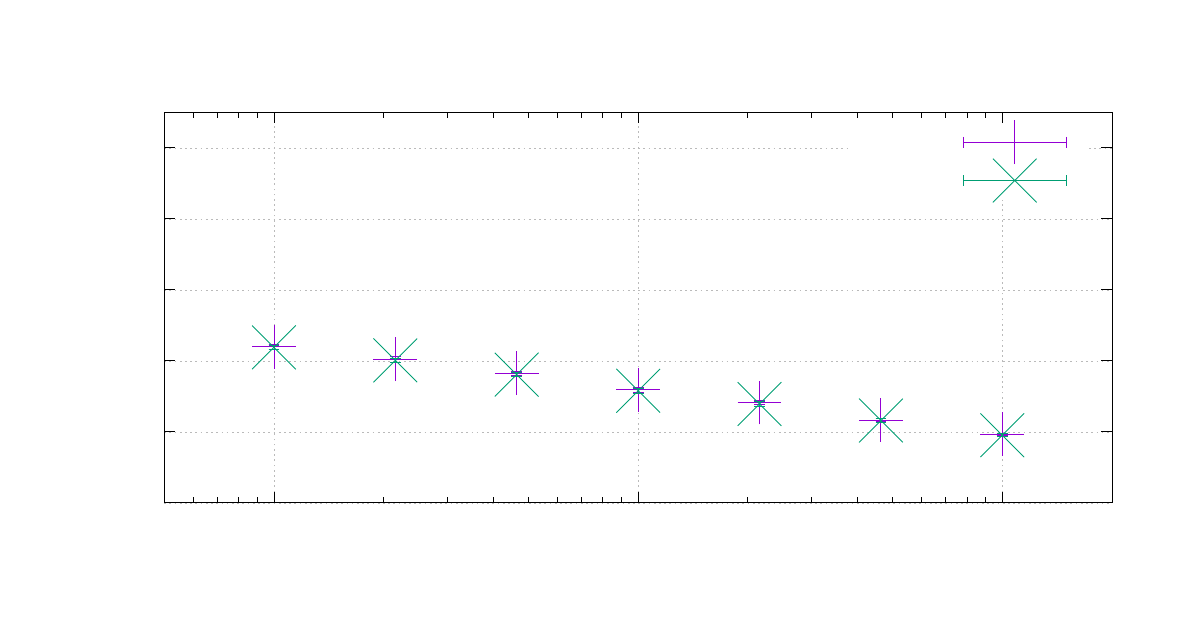}}%
    \gplfronttext
  \end{picture}%
\endgroup
}\hfill\scalebox{\threescale}{%
\begingroup
  \makeatletter
  \providecommand\color[2][]{%
    \GenericError{(gnuplot) \space\space\space\@spaces}{%
      Package color not loaded in conjunction with
      terminal option `colourtext'%
    }{See the gnuplot documentation for explanation.%
    }{Either use 'blacktext' in gnuplot or load the package
      color.sty in LaTeX.}%
    \renewcommand\color[2][]{}%
  }%
  \providecommand\includegraphics[2][]{%
    \GenericError{(gnuplot) \space\space\space\@spaces}{%
      Package graphicx or graphics not loaded%
    }{See the gnuplot documentation for explanation.%
    }{The gnuplot epslatex terminal needs graphicx.sty or graphics.sty.}%
    \renewcommand\includegraphics[2][]{}%
  }%
  \providecommand\rotatebox[2]{#2}%
  \@ifundefined{ifGPcolor}{%
    \newif\ifGPcolor
    \GPcolortrue
  }{}%
  \@ifundefined{ifGPblacktext}{%
    \newif\ifGPblacktext
    \GPblacktextfalse
  }{}%
  \let\gplgaddtomacro\g@addto@macro
  \gdef\gplbacktext{}%
  \gdef\gplfronttext{}%
  \makeatother
  \ifGPblacktext
    \def\colorrgb#1{}%
    \def\colorgray#1{}%
  \else
    \ifGPcolor
      \def\colorrgb#1{\color[rgb]{#1}}%
      \def\colorgray#1{\color[gray]{#1}}%
      \expandafter\def\csname LTw\endcsname{\color{white}}%
      \expandafter\def\csname LTb\endcsname{\color{black}}%
      \expandafter\def\csname LTa\endcsname{\color{black}}%
      \expandafter\def\csname LT0\endcsname{\color[rgb]{1,0,0}}%
      \expandafter\def\csname LT1\endcsname{\color[rgb]{0,1,0}}%
      \expandafter\def\csname LT2\endcsname{\color[rgb]{0,0,1}}%
      \expandafter\def\csname LT3\endcsname{\color[rgb]{1,0,1}}%
      \expandafter\def\csname LT4\endcsname{\color[rgb]{0,1,1}}%
      \expandafter\def\csname LT5\endcsname{\color[rgb]{1,1,0}}%
      \expandafter\def\csname LT6\endcsname{\color[rgb]{0,0,0}}%
      \expandafter\def\csname LT7\endcsname{\color[rgb]{1,0.3,0}}%
      \expandafter\def\csname LT8\endcsname{\color[rgb]{0.5,0.5,0.5}}%
    \else
      \def\colorrgb#1{\color{black}}%
      \def\colorgray#1{\color[gray]{#1}}%
      \expandafter\def\csname LTw\endcsname{\color{white}}%
      \expandafter\def\csname LTb\endcsname{\color{black}}%
      \expandafter\def\csname LTa\endcsname{\color{black}}%
      \expandafter\def\csname LT0\endcsname{\color{black}}%
      \expandafter\def\csname LT1\endcsname{\color{black}}%
      \expandafter\def\csname LT2\endcsname{\color{black}}%
      \expandafter\def\csname LT3\endcsname{\color{black}}%
      \expandafter\def\csname LT4\endcsname{\color{black}}%
      \expandafter\def\csname LT5\endcsname{\color{black}}%
      \expandafter\def\csname LT6\endcsname{\color{black}}%
      \expandafter\def\csname LT7\endcsname{\color{black}}%
      \expandafter\def\csname LT8\endcsname{\color{black}}%
    \fi
  \fi
    \setlength{\unitlength}{0.0500bp}%
    \ifx\gptboxheight\undefined%
      \newlength{\gptboxheight}%
      \newlength{\gptboxwidth}%
      \newsavebox{\gptboxtext}%
    \fi%
    \setlength{\fboxrule}{0.5pt}%
    \setlength{\fboxsep}{1pt}%
\begin{picture}(6802.00,3614.00)%
    \gplgaddtomacro\gplbacktext{%
      \csname LTb\endcsname%
      \put(814,704){\makebox(0,0)[r]{\strut{}$0$}}%
      \csname LTb\endcsname%
      \put(814,1113){\makebox(0,0)[r]{\strut{}$0.2$}}%
      \csname LTb\endcsname%
      \put(814,1522){\makebox(0,0)[r]{\strut{}$0.4$}}%
      \csname LTb\endcsname%
      \put(814,1931){\makebox(0,0)[r]{\strut{}$0.6$}}%
      \csname LTb\endcsname%
      \put(814,2340){\makebox(0,0)[r]{\strut{}$0.8$}}%
      \csname LTb\endcsname%
      \put(814,2749){\makebox(0,0)[r]{\strut{}$1$}}%
      \csname LTb\endcsname%
      \put(1578,484){\makebox(0,0){\strut{}$10^{3}$}}%
      \csname LTb\endcsname%
      \put(3676,484){\makebox(0,0){\strut{}$10^{4}$}}%
      \csname LTb\endcsname%
      \put(5773,484){\makebox(0,0){\strut{}$10^{5}$}}%
    }%
    \gplgaddtomacro\gplfronttext{%
      \csname LTb\endcsname%
      \put(176,1828){\rotatebox{-270}{\makebox(0,0){\strut{}Degree Assortativity}}}%
      \put(3675,154){\makebox(0,0){\strut{}Number $n$ of nodes}}%
      \put(3675,3283){\makebox(0,0){\strut{}Mixing: $\mu = 0.6$, Degree Assortativity, Overlap: $\nu = 2$}}%
      \csname LTb\endcsname%
      \put(5418,2780){\makebox(0,0)[r]{\strut{}Orig}}%
      \csname LTb\endcsname%
      \put(5418,2560){\makebox(0,0)[r]{\strut{}EM}}%
    }%
    \gplbacktext
    \put(0,0){\includegraphics{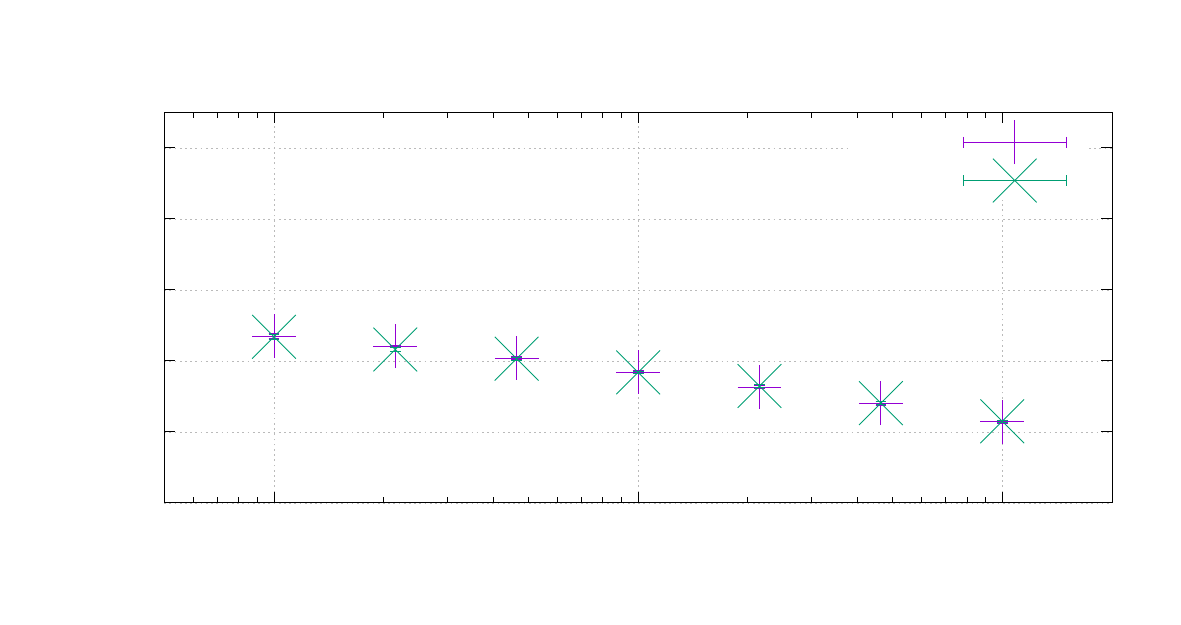}}%
    \gplfronttext
  \end{picture}%
\endgroup
}\hfill\\ %

\vspace{2em}	
\hrule	
\vspace{2em}	

\noindent\scalebox{\threescale}{%
\begingroup
  \makeatletter
  \providecommand\color[2][]{%
    \GenericError{(gnuplot) \space\space\space\@spaces}{%
      Package color not loaded in conjunction with
      terminal option `colourtext'%
    }{See the gnuplot documentation for explanation.%
    }{Either use 'blacktext' in gnuplot or load the package
      color.sty in LaTeX.}%
    \renewcommand\color[2][]{}%
  }%
  \providecommand\includegraphics[2][]{%
    \GenericError{(gnuplot) \space\space\space\@spaces}{%
      Package graphicx or graphics not loaded%
    }{See the gnuplot documentation for explanation.%
    }{The gnuplot epslatex terminal needs graphicx.sty or graphics.sty.}%
    \renewcommand\includegraphics[2][]{}%
  }%
  \providecommand\rotatebox[2]{#2}%
  \@ifundefined{ifGPcolor}{%
    \newif\ifGPcolor
    \GPcolortrue
  }{}%
  \@ifundefined{ifGPblacktext}{%
    \newif\ifGPblacktext
    \GPblacktextfalse
  }{}%
  \let\gplgaddtomacro\g@addto@macro
  \gdef\gplbacktext{}%
  \gdef\gplfronttext{}%
  \makeatother
  \ifGPblacktext
    \def\colorrgb#1{}%
    \def\colorgray#1{}%
  \else
    \ifGPcolor
      \def\colorrgb#1{\color[rgb]{#1}}%
      \def\colorgray#1{\color[gray]{#1}}%
      \expandafter\def\csname LTw\endcsname{\color{white}}%
      \expandafter\def\csname LTb\endcsname{\color{black}}%
      \expandafter\def\csname LTa\endcsname{\color{black}}%
      \expandafter\def\csname LT0\endcsname{\color[rgb]{1,0,0}}%
      \expandafter\def\csname LT1\endcsname{\color[rgb]{0,1,0}}%
      \expandafter\def\csname LT2\endcsname{\color[rgb]{0,0,1}}%
      \expandafter\def\csname LT3\endcsname{\color[rgb]{1,0,1}}%
      \expandafter\def\csname LT4\endcsname{\color[rgb]{0,1,1}}%
      \expandafter\def\csname LT5\endcsname{\color[rgb]{1,1,0}}%
      \expandafter\def\csname LT6\endcsname{\color[rgb]{0,0,0}}%
      \expandafter\def\csname LT7\endcsname{\color[rgb]{1,0.3,0}}%
      \expandafter\def\csname LT8\endcsname{\color[rgb]{0.5,0.5,0.5}}%
    \else
      \def\colorrgb#1{\color{black}}%
      \def\colorgray#1{\color[gray]{#1}}%
      \expandafter\def\csname LTw\endcsname{\color{white}}%
      \expandafter\def\csname LTb\endcsname{\color{black}}%
      \expandafter\def\csname LTa\endcsname{\color{black}}%
      \expandafter\def\csname LT0\endcsname{\color{black}}%
      \expandafter\def\csname LT1\endcsname{\color{black}}%
      \expandafter\def\csname LT2\endcsname{\color{black}}%
      \expandafter\def\csname LT3\endcsname{\color{black}}%
      \expandafter\def\csname LT4\endcsname{\color{black}}%
      \expandafter\def\csname LT5\endcsname{\color{black}}%
      \expandafter\def\csname LT6\endcsname{\color{black}}%
      \expandafter\def\csname LT7\endcsname{\color{black}}%
      \expandafter\def\csname LT8\endcsname{\color{black}}%
    \fi
  \fi
    \setlength{\unitlength}{0.0500bp}%
    \ifx\gptboxheight\undefined%
      \newlength{\gptboxheight}%
      \newlength{\gptboxwidth}%
      \newsavebox{\gptboxtext}%
    \fi%
    \setlength{\fboxrule}{0.5pt}%
    \setlength{\fboxsep}{1pt}%
\begin{picture}(6802.00,3614.00)%
    \gplgaddtomacro\gplbacktext{%
      \csname LTb\endcsname%
      \put(814,704){\makebox(0,0)[r]{\strut{}$0$}}%
      \csname LTb\endcsname%
      \put(814,1113){\makebox(0,0)[r]{\strut{}$0.2$}}%
      \csname LTb\endcsname%
      \put(814,1522){\makebox(0,0)[r]{\strut{}$0.4$}}%
      \csname LTb\endcsname%
      \put(814,1931){\makebox(0,0)[r]{\strut{}$0.6$}}%
      \csname LTb\endcsname%
      \put(814,2340){\makebox(0,0)[r]{\strut{}$0.8$}}%
      \csname LTb\endcsname%
      \put(814,2749){\makebox(0,0)[r]{\strut{}$1$}}%
      \csname LTb\endcsname%
      \put(1578,484){\makebox(0,0){\strut{}$10^{3}$}}%
      \csname LTb\endcsname%
      \put(3676,484){\makebox(0,0){\strut{}$10^{4}$}}%
      \csname LTb\endcsname%
      \put(5773,484){\makebox(0,0){\strut{}$10^{5}$}}%
    }%
    \gplgaddtomacro\gplfronttext{%
      \csname LTb\endcsname%
      \put(176,1828){\rotatebox{-270}{\makebox(0,0){\strut{}NMI}}}%
      \put(3675,154){\makebox(0,0){\strut{}Number $n$ of nodes}}%
      \put(3675,3283){\makebox(0,0){\strut{}Mixing: $\mu = 0.2$, Cluster: OSLOM, Overlap: $\nu = 3$}}%
      \csname LTb\endcsname%
      \put(5418,2780){\makebox(0,0)[r]{\strut{}Orig}}%
      \csname LTb\endcsname%
      \put(5418,2560){\makebox(0,0)[r]{\strut{}EM}}%
    }%
    \gplbacktext
    \put(0,0){\includegraphics{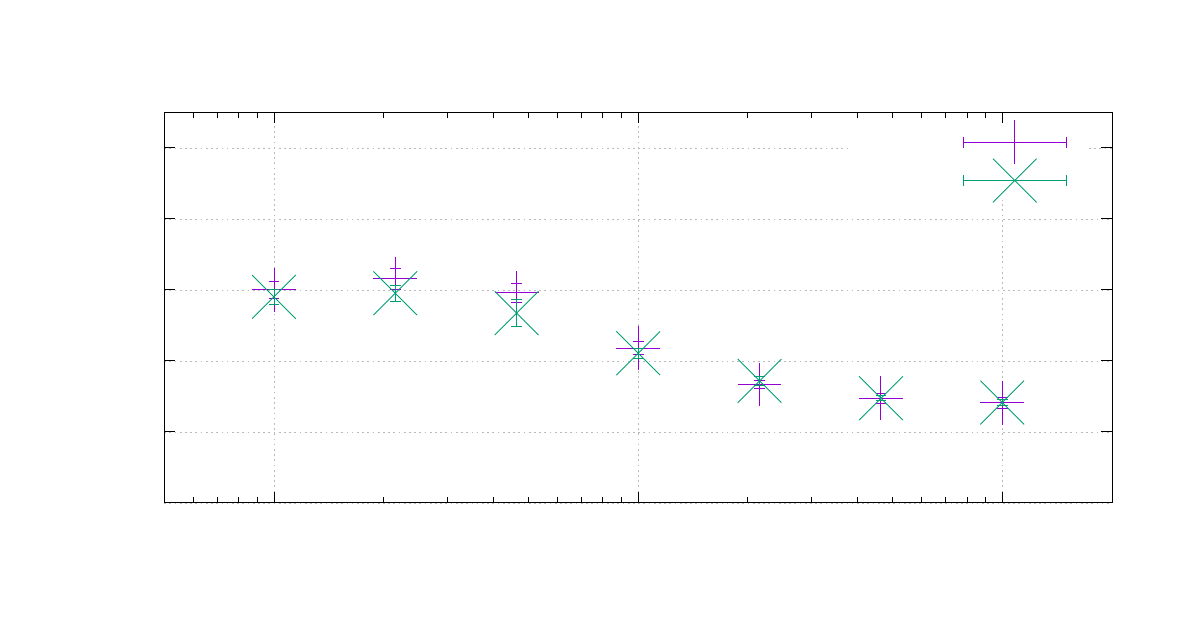}}%
    \gplfronttext
  \end{picture}%
\endgroup
}\hfill\scalebox{\threescale}{%
\begingroup
  \makeatletter
  \providecommand\color[2][]{%
    \GenericError{(gnuplot) \space\space\space\@spaces}{%
      Package color not loaded in conjunction with
      terminal option `colourtext'%
    }{See the gnuplot documentation for explanation.%
    }{Either use 'blacktext' in gnuplot or load the package
      color.sty in LaTeX.}%
    \renewcommand\color[2][]{}%
  }%
  \providecommand\includegraphics[2][]{%
    \GenericError{(gnuplot) \space\space\space\@spaces}{%
      Package graphicx or graphics not loaded%
    }{See the gnuplot documentation for explanation.%
    }{The gnuplot epslatex terminal needs graphicx.sty or graphics.sty.}%
    \renewcommand\includegraphics[2][]{}%
  }%
  \providecommand\rotatebox[2]{#2}%
  \@ifundefined{ifGPcolor}{%
    \newif\ifGPcolor
    \GPcolortrue
  }{}%
  \@ifundefined{ifGPblacktext}{%
    \newif\ifGPblacktext
    \GPblacktextfalse
  }{}%
  \let\gplgaddtomacro\g@addto@macro
  \gdef\gplbacktext{}%
  \gdef\gplfronttext{}%
  \makeatother
  \ifGPblacktext
    \def\colorrgb#1{}%
    \def\colorgray#1{}%
  \else
    \ifGPcolor
      \def\colorrgb#1{\color[rgb]{#1}}%
      \def\colorgray#1{\color[gray]{#1}}%
      \expandafter\def\csname LTw\endcsname{\color{white}}%
      \expandafter\def\csname LTb\endcsname{\color{black}}%
      \expandafter\def\csname LTa\endcsname{\color{black}}%
      \expandafter\def\csname LT0\endcsname{\color[rgb]{1,0,0}}%
      \expandafter\def\csname LT1\endcsname{\color[rgb]{0,1,0}}%
      \expandafter\def\csname LT2\endcsname{\color[rgb]{0,0,1}}%
      \expandafter\def\csname LT3\endcsname{\color[rgb]{1,0,1}}%
      \expandafter\def\csname LT4\endcsname{\color[rgb]{0,1,1}}%
      \expandafter\def\csname LT5\endcsname{\color[rgb]{1,1,0}}%
      \expandafter\def\csname LT6\endcsname{\color[rgb]{0,0,0}}%
      \expandafter\def\csname LT7\endcsname{\color[rgb]{1,0.3,0}}%
      \expandafter\def\csname LT8\endcsname{\color[rgb]{0.5,0.5,0.5}}%
    \else
      \def\colorrgb#1{\color{black}}%
      \def\colorgray#1{\color[gray]{#1}}%
      \expandafter\def\csname LTw\endcsname{\color{white}}%
      \expandafter\def\csname LTb\endcsname{\color{black}}%
      \expandafter\def\csname LTa\endcsname{\color{black}}%
      \expandafter\def\csname LT0\endcsname{\color{black}}%
      \expandafter\def\csname LT1\endcsname{\color{black}}%
      \expandafter\def\csname LT2\endcsname{\color{black}}%
      \expandafter\def\csname LT3\endcsname{\color{black}}%
      \expandafter\def\csname LT4\endcsname{\color{black}}%
      \expandafter\def\csname LT5\endcsname{\color{black}}%
      \expandafter\def\csname LT6\endcsname{\color{black}}%
      \expandafter\def\csname LT7\endcsname{\color{black}}%
      \expandafter\def\csname LT8\endcsname{\color{black}}%
    \fi
  \fi
    \setlength{\unitlength}{0.0500bp}%
    \ifx\gptboxheight\undefined%
      \newlength{\gptboxheight}%
      \newlength{\gptboxwidth}%
      \newsavebox{\gptboxtext}%
    \fi%
    \setlength{\fboxrule}{0.5pt}%
    \setlength{\fboxsep}{1pt}%
\begin{picture}(6802.00,3614.00)%
    \gplgaddtomacro\gplbacktext{%
      \csname LTb\endcsname%
      \put(814,704){\makebox(0,0)[r]{\strut{}$0$}}%
      \csname LTb\endcsname%
      \put(814,1113){\makebox(0,0)[r]{\strut{}$0.2$}}%
      \csname LTb\endcsname%
      \put(814,1522){\makebox(0,0)[r]{\strut{}$0.4$}}%
      \csname LTb\endcsname%
      \put(814,1931){\makebox(0,0)[r]{\strut{}$0.6$}}%
      \csname LTb\endcsname%
      \put(814,2340){\makebox(0,0)[r]{\strut{}$0.8$}}%
      \csname LTb\endcsname%
      \put(814,2749){\makebox(0,0)[r]{\strut{}$1$}}%
      \csname LTb\endcsname%
      \put(1578,484){\makebox(0,0){\strut{}$10^{3}$}}%
      \csname LTb\endcsname%
      \put(3676,484){\makebox(0,0){\strut{}$10^{4}$}}%
      \csname LTb\endcsname%
      \put(5773,484){\makebox(0,0){\strut{}$10^{5}$}}%
    }%
    \gplgaddtomacro\gplfronttext{%
      \csname LTb\endcsname%
      \put(176,1828){\rotatebox{-270}{\makebox(0,0){\strut{}NMI}}}%
      \put(3675,154){\makebox(0,0){\strut{}Number $n$ of nodes}}%
      \put(3675,3283){\makebox(0,0){\strut{}Mixing: $\mu = 0.4$, Cluster: OSLOM, Overlap: $\nu = 3$}}%
      \csname LTb\endcsname%
      \put(5418,2780){\makebox(0,0)[r]{\strut{}Orig}}%
      \csname LTb\endcsname%
      \put(5418,2560){\makebox(0,0)[r]{\strut{}EM}}%
    }%
    \gplbacktext
    \put(0,0){\includegraphics{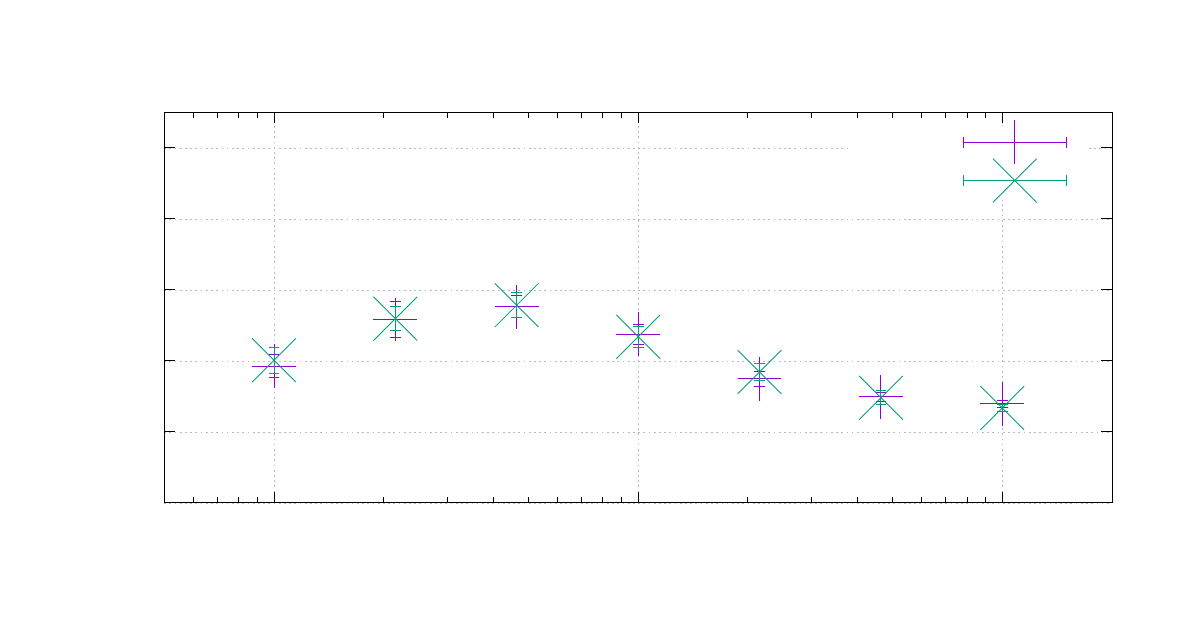}}%
    \gplfronttext
  \end{picture}%
\endgroup
}\hfill\scalebox{\threescale}{%
\begingroup
  \makeatletter
  \providecommand\color[2][]{%
    \GenericError{(gnuplot) \space\space\space\@spaces}{%
      Package color not loaded in conjunction with
      terminal option `colourtext'%
    }{See the gnuplot documentation for explanation.%
    }{Either use 'blacktext' in gnuplot or load the package
      color.sty in LaTeX.}%
    \renewcommand\color[2][]{}%
  }%
  \providecommand\includegraphics[2][]{%
    \GenericError{(gnuplot) \space\space\space\@spaces}{%
      Package graphicx or graphics not loaded%
    }{See the gnuplot documentation for explanation.%
    }{The gnuplot epslatex terminal needs graphicx.sty or graphics.sty.}%
    \renewcommand\includegraphics[2][]{}%
  }%
  \providecommand\rotatebox[2]{#2}%
  \@ifundefined{ifGPcolor}{%
    \newif\ifGPcolor
    \GPcolortrue
  }{}%
  \@ifundefined{ifGPblacktext}{%
    \newif\ifGPblacktext
    \GPblacktextfalse
  }{}%
  \let\gplgaddtomacro\g@addto@macro
  \gdef\gplbacktext{}%
  \gdef\gplfronttext{}%
  \makeatother
  \ifGPblacktext
    \def\colorrgb#1{}%
    \def\colorgray#1{}%
  \else
    \ifGPcolor
      \def\colorrgb#1{\color[rgb]{#1}}%
      \def\colorgray#1{\color[gray]{#1}}%
      \expandafter\def\csname LTw\endcsname{\color{white}}%
      \expandafter\def\csname LTb\endcsname{\color{black}}%
      \expandafter\def\csname LTa\endcsname{\color{black}}%
      \expandafter\def\csname LT0\endcsname{\color[rgb]{1,0,0}}%
      \expandafter\def\csname LT1\endcsname{\color[rgb]{0,1,0}}%
      \expandafter\def\csname LT2\endcsname{\color[rgb]{0,0,1}}%
      \expandafter\def\csname LT3\endcsname{\color[rgb]{1,0,1}}%
      \expandafter\def\csname LT4\endcsname{\color[rgb]{0,1,1}}%
      \expandafter\def\csname LT5\endcsname{\color[rgb]{1,1,0}}%
      \expandafter\def\csname LT6\endcsname{\color[rgb]{0,0,0}}%
      \expandafter\def\csname LT7\endcsname{\color[rgb]{1,0.3,0}}%
      \expandafter\def\csname LT8\endcsname{\color[rgb]{0.5,0.5,0.5}}%
    \else
      \def\colorrgb#1{\color{black}}%
      \def\colorgray#1{\color[gray]{#1}}%
      \expandafter\def\csname LTw\endcsname{\color{white}}%
      \expandafter\def\csname LTb\endcsname{\color{black}}%
      \expandafter\def\csname LTa\endcsname{\color{black}}%
      \expandafter\def\csname LT0\endcsname{\color{black}}%
      \expandafter\def\csname LT1\endcsname{\color{black}}%
      \expandafter\def\csname LT2\endcsname{\color{black}}%
      \expandafter\def\csname LT3\endcsname{\color{black}}%
      \expandafter\def\csname LT4\endcsname{\color{black}}%
      \expandafter\def\csname LT5\endcsname{\color{black}}%
      \expandafter\def\csname LT6\endcsname{\color{black}}%
      \expandafter\def\csname LT7\endcsname{\color{black}}%
      \expandafter\def\csname LT8\endcsname{\color{black}}%
    \fi
  \fi
    \setlength{\unitlength}{0.0500bp}%
    \ifx\gptboxheight\undefined%
      \newlength{\gptboxheight}%
      \newlength{\gptboxwidth}%
      \newsavebox{\gptboxtext}%
    \fi%
    \setlength{\fboxrule}{0.5pt}%
    \setlength{\fboxsep}{1pt}%
\begin{picture}(6802.00,3614.00)%
    \gplgaddtomacro\gplbacktext{%
      \csname LTb\endcsname%
      \put(814,704){\makebox(0,0)[r]{\strut{}$0$}}%
      \csname LTb\endcsname%
      \put(814,1113){\makebox(0,0)[r]{\strut{}$0.2$}}%
      \csname LTb\endcsname%
      \put(814,1522){\makebox(0,0)[r]{\strut{}$0.4$}}%
      \csname LTb\endcsname%
      \put(814,1931){\makebox(0,0)[r]{\strut{}$0.6$}}%
      \csname LTb\endcsname%
      \put(814,2340){\makebox(0,0)[r]{\strut{}$0.8$}}%
      \csname LTb\endcsname%
      \put(814,2749){\makebox(0,0)[r]{\strut{}$1$}}%
      \csname LTb\endcsname%
      \put(1578,484){\makebox(0,0){\strut{}$10^{3}$}}%
      \csname LTb\endcsname%
      \put(3676,484){\makebox(0,0){\strut{}$10^{4}$}}%
      \csname LTb\endcsname%
      \put(5773,484){\makebox(0,0){\strut{}$10^{5}$}}%
    }%
    \gplgaddtomacro\gplfronttext{%
      \csname LTb\endcsname%
      \put(176,1828){\rotatebox{-270}{\makebox(0,0){\strut{}NMI}}}%
      \put(3675,154){\makebox(0,0){\strut{}Number $n$ of nodes}}%
      \put(3675,3283){\makebox(0,0){\strut{}Mixing: $\mu = 0.6$, Cluster: OSLOM, Overlap: $\nu = 3$}}%
      \csname LTb\endcsname%
      \put(5418,2780){\makebox(0,0)[r]{\strut{}Orig}}%
      \csname LTb\endcsname%
      \put(5418,2560){\makebox(0,0)[r]{\strut{}EM}}%
    }%
    \gplbacktext
    \put(0,0){\includegraphics{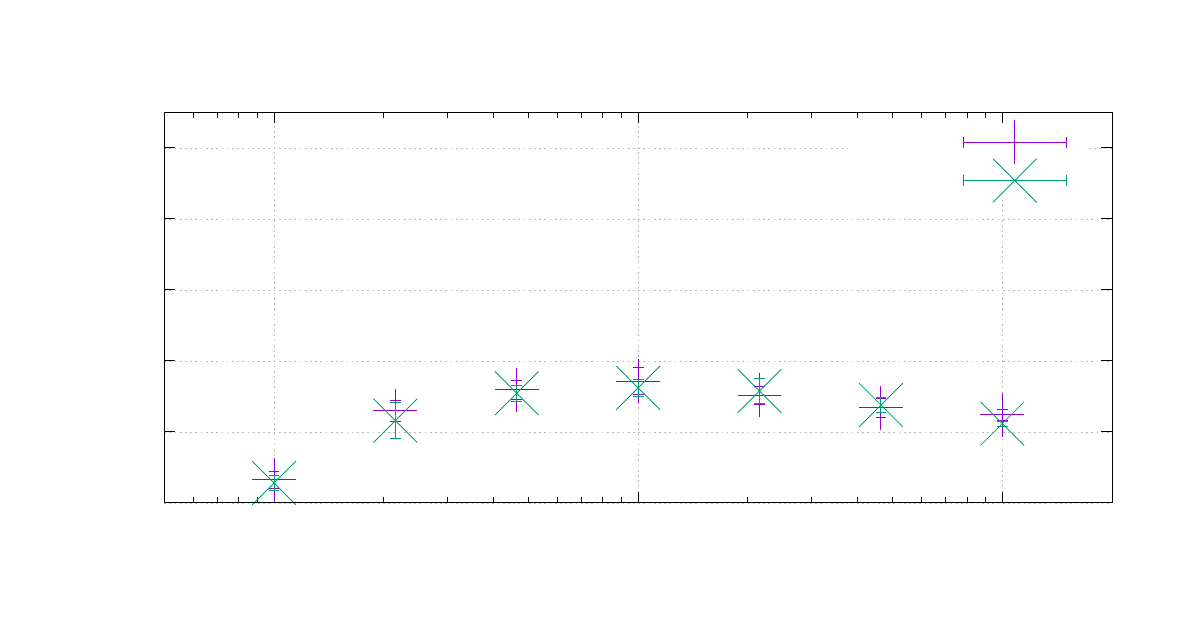}}%
    \gplfronttext
  \end{picture}%
\endgroup
}\hfill\\ %
\noindent\scalebox{\threescale}{%
\begingroup
  \makeatletter
  \providecommand\color[2][]{%
    \GenericError{(gnuplot) \space\space\space\@spaces}{%
      Package color not loaded in conjunction with
      terminal option `colourtext'%
    }{See the gnuplot documentation for explanation.%
    }{Either use 'blacktext' in gnuplot or load the package
      color.sty in LaTeX.}%
    \renewcommand\color[2][]{}%
  }%
  \providecommand\includegraphics[2][]{%
    \GenericError{(gnuplot) \space\space\space\@spaces}{%
      Package graphicx or graphics not loaded%
    }{See the gnuplot documentation for explanation.%
    }{The gnuplot epslatex terminal needs graphicx.sty or graphics.sty.}%
    \renewcommand\includegraphics[2][]{}%
  }%
  \providecommand\rotatebox[2]{#2}%
  \@ifundefined{ifGPcolor}{%
    \newif\ifGPcolor
    \GPcolortrue
  }{}%
  \@ifundefined{ifGPblacktext}{%
    \newif\ifGPblacktext
    \GPblacktextfalse
  }{}%
  \let\gplgaddtomacro\g@addto@macro
  \gdef\gplbacktext{}%
  \gdef\gplfronttext{}%
  \makeatother
  \ifGPblacktext
    \def\colorrgb#1{}%
    \def\colorgray#1{}%
  \else
    \ifGPcolor
      \def\colorrgb#1{\color[rgb]{#1}}%
      \def\colorgray#1{\color[gray]{#1}}%
      \expandafter\def\csname LTw\endcsname{\color{white}}%
      \expandafter\def\csname LTb\endcsname{\color{black}}%
      \expandafter\def\csname LTa\endcsname{\color{black}}%
      \expandafter\def\csname LT0\endcsname{\color[rgb]{1,0,0}}%
      \expandafter\def\csname LT1\endcsname{\color[rgb]{0,1,0}}%
      \expandafter\def\csname LT2\endcsname{\color[rgb]{0,0,1}}%
      \expandafter\def\csname LT3\endcsname{\color[rgb]{1,0,1}}%
      \expandafter\def\csname LT4\endcsname{\color[rgb]{0,1,1}}%
      \expandafter\def\csname LT5\endcsname{\color[rgb]{1,1,0}}%
      \expandafter\def\csname LT6\endcsname{\color[rgb]{0,0,0}}%
      \expandafter\def\csname LT7\endcsname{\color[rgb]{1,0.3,0}}%
      \expandafter\def\csname LT8\endcsname{\color[rgb]{0.5,0.5,0.5}}%
    \else
      \def\colorrgb#1{\color{black}}%
      \def\colorgray#1{\color[gray]{#1}}%
      \expandafter\def\csname LTw\endcsname{\color{white}}%
      \expandafter\def\csname LTb\endcsname{\color{black}}%
      \expandafter\def\csname LTa\endcsname{\color{black}}%
      \expandafter\def\csname LT0\endcsname{\color{black}}%
      \expandafter\def\csname LT1\endcsname{\color{black}}%
      \expandafter\def\csname LT2\endcsname{\color{black}}%
      \expandafter\def\csname LT3\endcsname{\color{black}}%
      \expandafter\def\csname LT4\endcsname{\color{black}}%
      \expandafter\def\csname LT5\endcsname{\color{black}}%
      \expandafter\def\csname LT6\endcsname{\color{black}}%
      \expandafter\def\csname LT7\endcsname{\color{black}}%
      \expandafter\def\csname LT8\endcsname{\color{black}}%
    \fi
  \fi
    \setlength{\unitlength}{0.0500bp}%
    \ifx\gptboxheight\undefined%
      \newlength{\gptboxheight}%
      \newlength{\gptboxwidth}%
      \newsavebox{\gptboxtext}%
    \fi%
    \setlength{\fboxrule}{0.5pt}%
    \setlength{\fboxsep}{1pt}%
\begin{picture}(6802.00,3614.00)%
    \gplgaddtomacro\gplbacktext{%
      \csname LTb\endcsname%
      \put(814,704){\makebox(0,0)[r]{\strut{}$0$}}%
      \csname LTb\endcsname%
      \put(814,1113){\makebox(0,0)[r]{\strut{}$0.2$}}%
      \csname LTb\endcsname%
      \put(814,1522){\makebox(0,0)[r]{\strut{}$0.4$}}%
      \csname LTb\endcsname%
      \put(814,1931){\makebox(0,0)[r]{\strut{}$0.6$}}%
      \csname LTb\endcsname%
      \put(814,2340){\makebox(0,0)[r]{\strut{}$0.8$}}%
      \csname LTb\endcsname%
      \put(814,2749){\makebox(0,0)[r]{\strut{}$1$}}%
      \csname LTb\endcsname%
      \put(1578,484){\makebox(0,0){\strut{}$10^{3}$}}%
      \csname LTb\endcsname%
      \put(3676,484){\makebox(0,0){\strut{}$10^{4}$}}%
      \csname LTb\endcsname%
      \put(5773,484){\makebox(0,0){\strut{}$10^{5}$}}%
    }%
    \gplgaddtomacro\gplfronttext{%
      \csname LTb\endcsname%
      \put(176,1828){\rotatebox{-270}{\makebox(0,0){\strut{}Avg. Local Clustering Coeff.}}}%
      \put(3675,154){\makebox(0,0){\strut{}Number $n$ of nodes}}%
      \put(3675,3283){\makebox(0,0){\strut{}Mixing: $\mu = 0.2$}}%
      \csname LTb\endcsname%
      \put(5418,2780){\makebox(0,0)[r]{\strut{}Orig}}%
      \csname LTb\endcsname%
      \put(5418,2560){\makebox(0,0)[r]{\strut{}EM}}%
    }%
    \gplbacktext
    \put(0,0){\includegraphics{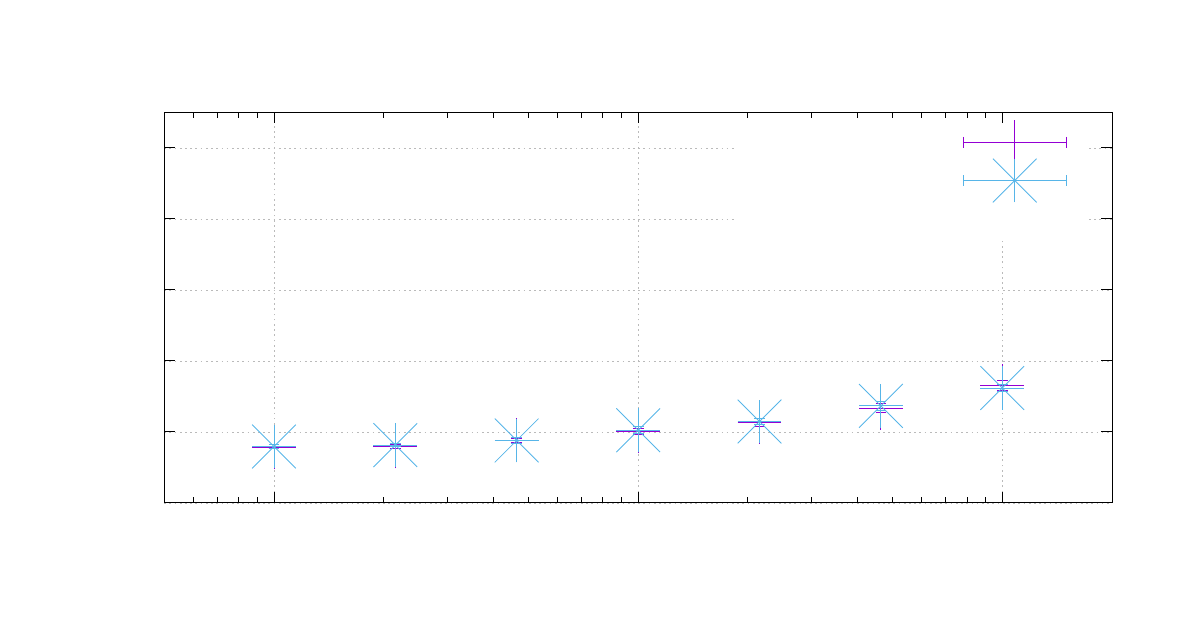}}%
    \gplfronttext
  \end{picture}%
\endgroup
}\hfill\scalebox{\threescale}{%
\begingroup
  \makeatletter
  \providecommand\color[2][]{%
    \GenericError{(gnuplot) \space\space\space\@spaces}{%
      Package color not loaded in conjunction with
      terminal option `colourtext'%
    }{See the gnuplot documentation for explanation.%
    }{Either use 'blacktext' in gnuplot or load the package
      color.sty in LaTeX.}%
    \renewcommand\color[2][]{}%
  }%
  \providecommand\includegraphics[2][]{%
    \GenericError{(gnuplot) \space\space\space\@spaces}{%
      Package graphicx or graphics not loaded%
    }{See the gnuplot documentation for explanation.%
    }{The gnuplot epslatex terminal needs graphicx.sty or graphics.sty.}%
    \renewcommand\includegraphics[2][]{}%
  }%
  \providecommand\rotatebox[2]{#2}%
  \@ifundefined{ifGPcolor}{%
    \newif\ifGPcolor
    \GPcolortrue
  }{}%
  \@ifundefined{ifGPblacktext}{%
    \newif\ifGPblacktext
    \GPblacktextfalse
  }{}%
  \let\gplgaddtomacro\g@addto@macro
  \gdef\gplbacktext{}%
  \gdef\gplfronttext{}%
  \makeatother
  \ifGPblacktext
    \def\colorrgb#1{}%
    \def\colorgray#1{}%
  \else
    \ifGPcolor
      \def\colorrgb#1{\color[rgb]{#1}}%
      \def\colorgray#1{\color[gray]{#1}}%
      \expandafter\def\csname LTw\endcsname{\color{white}}%
      \expandafter\def\csname LTb\endcsname{\color{black}}%
      \expandafter\def\csname LTa\endcsname{\color{black}}%
      \expandafter\def\csname LT0\endcsname{\color[rgb]{1,0,0}}%
      \expandafter\def\csname LT1\endcsname{\color[rgb]{0,1,0}}%
      \expandafter\def\csname LT2\endcsname{\color[rgb]{0,0,1}}%
      \expandafter\def\csname LT3\endcsname{\color[rgb]{1,0,1}}%
      \expandafter\def\csname LT4\endcsname{\color[rgb]{0,1,1}}%
      \expandafter\def\csname LT5\endcsname{\color[rgb]{1,1,0}}%
      \expandafter\def\csname LT6\endcsname{\color[rgb]{0,0,0}}%
      \expandafter\def\csname LT7\endcsname{\color[rgb]{1,0.3,0}}%
      \expandafter\def\csname LT8\endcsname{\color[rgb]{0.5,0.5,0.5}}%
    \else
      \def\colorrgb#1{\color{black}}%
      \def\colorgray#1{\color[gray]{#1}}%
      \expandafter\def\csname LTw\endcsname{\color{white}}%
      \expandafter\def\csname LTb\endcsname{\color{black}}%
      \expandafter\def\csname LTa\endcsname{\color{black}}%
      \expandafter\def\csname LT0\endcsname{\color{black}}%
      \expandafter\def\csname LT1\endcsname{\color{black}}%
      \expandafter\def\csname LT2\endcsname{\color{black}}%
      \expandafter\def\csname LT3\endcsname{\color{black}}%
      \expandafter\def\csname LT4\endcsname{\color{black}}%
      \expandafter\def\csname LT5\endcsname{\color{black}}%
      \expandafter\def\csname LT6\endcsname{\color{black}}%
      \expandafter\def\csname LT7\endcsname{\color{black}}%
      \expandafter\def\csname LT8\endcsname{\color{black}}%
    \fi
  \fi
    \setlength{\unitlength}{0.0500bp}%
    \ifx\gptboxheight\undefined%
      \newlength{\gptboxheight}%
      \newlength{\gptboxwidth}%
      \newsavebox{\gptboxtext}%
    \fi%
    \setlength{\fboxrule}{0.5pt}%
    \setlength{\fboxsep}{1pt}%
\begin{picture}(6802.00,3614.00)%
    \gplgaddtomacro\gplbacktext{%
      \csname LTb\endcsname%
      \put(814,704){\makebox(0,0)[r]{\strut{}$0$}}%
      \csname LTb\endcsname%
      \put(814,1113){\makebox(0,0)[r]{\strut{}$0.2$}}%
      \csname LTb\endcsname%
      \put(814,1522){\makebox(0,0)[r]{\strut{}$0.4$}}%
      \csname LTb\endcsname%
      \put(814,1931){\makebox(0,0)[r]{\strut{}$0.6$}}%
      \csname LTb\endcsname%
      \put(814,2340){\makebox(0,0)[r]{\strut{}$0.8$}}%
      \csname LTb\endcsname%
      \put(814,2749){\makebox(0,0)[r]{\strut{}$1$}}%
      \csname LTb\endcsname%
      \put(1578,484){\makebox(0,0){\strut{}$10^{3}$}}%
      \csname LTb\endcsname%
      \put(3676,484){\makebox(0,0){\strut{}$10^{4}$}}%
      \csname LTb\endcsname%
      \put(5773,484){\makebox(0,0){\strut{}$10^{5}$}}%
    }%
    \gplgaddtomacro\gplfronttext{%
      \csname LTb\endcsname%
      \put(176,1828){\rotatebox{-270}{\makebox(0,0){\strut{}Avg. Local Clustering Coeff.}}}%
      \put(3675,154){\makebox(0,0){\strut{}Number $n$ of nodes}}%
      \put(3675,3283){\makebox(0,0){\strut{}Mixing: $\mu = 0.4$}}%
      \csname LTb\endcsname%
      \put(5418,2780){\makebox(0,0)[r]{\strut{}Orig}}%
      \csname LTb\endcsname%
      \put(5418,2560){\makebox(0,0)[r]{\strut{}EM}}%
    }%
    \gplbacktext
    \put(0,0){\includegraphics{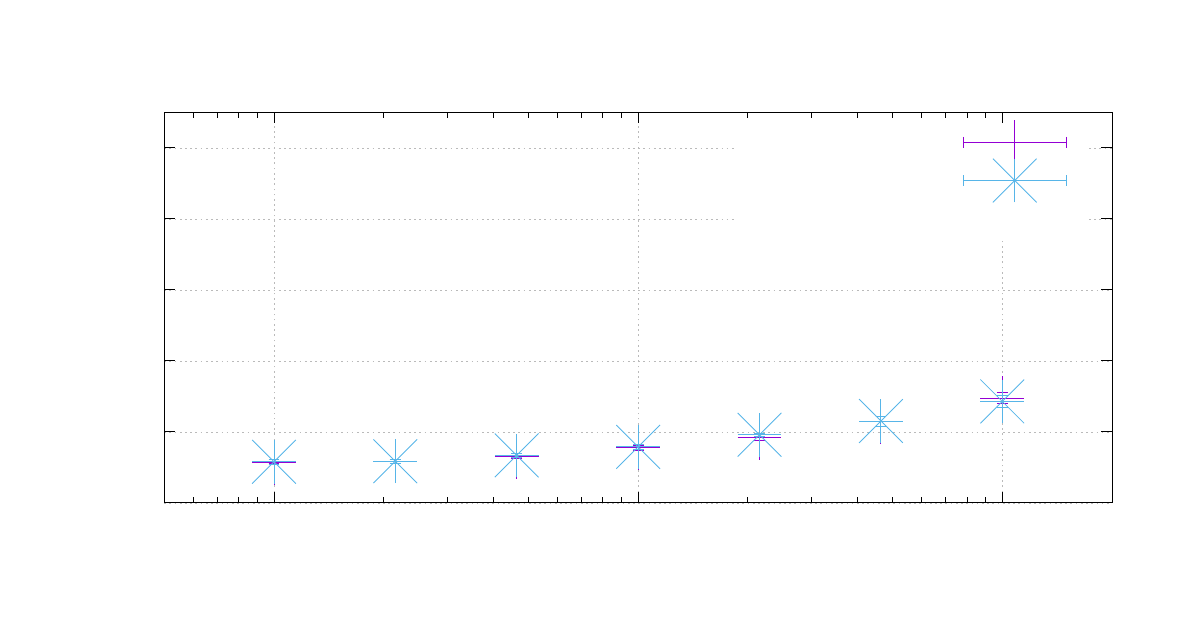}}%
    \gplfronttext
  \end{picture}%
\endgroup
}\hfill\scalebox{\threescale}{%
\begingroup
  \makeatletter
  \providecommand\color[2][]{%
    \GenericError{(gnuplot) \space\space\space\@spaces}{%
      Package color not loaded in conjunction with
      terminal option `colourtext'%
    }{See the gnuplot documentation for explanation.%
    }{Either use 'blacktext' in gnuplot or load the package
      color.sty in LaTeX.}%
    \renewcommand\color[2][]{}%
  }%
  \providecommand\includegraphics[2][]{%
    \GenericError{(gnuplot) \space\space\space\@spaces}{%
      Package graphicx or graphics not loaded%
    }{See the gnuplot documentation for explanation.%
    }{The gnuplot epslatex terminal needs graphicx.sty or graphics.sty.}%
    \renewcommand\includegraphics[2][]{}%
  }%
  \providecommand\rotatebox[2]{#2}%
  \@ifundefined{ifGPcolor}{%
    \newif\ifGPcolor
    \GPcolortrue
  }{}%
  \@ifundefined{ifGPblacktext}{%
    \newif\ifGPblacktext
    \GPblacktextfalse
  }{}%
  \let\gplgaddtomacro\g@addto@macro
  \gdef\gplbacktext{}%
  \gdef\gplfronttext{}%
  \makeatother
  \ifGPblacktext
    \def\colorrgb#1{}%
    \def\colorgray#1{}%
  \else
    \ifGPcolor
      \def\colorrgb#1{\color[rgb]{#1}}%
      \def\colorgray#1{\color[gray]{#1}}%
      \expandafter\def\csname LTw\endcsname{\color{white}}%
      \expandafter\def\csname LTb\endcsname{\color{black}}%
      \expandafter\def\csname LTa\endcsname{\color{black}}%
      \expandafter\def\csname LT0\endcsname{\color[rgb]{1,0,0}}%
      \expandafter\def\csname LT1\endcsname{\color[rgb]{0,1,0}}%
      \expandafter\def\csname LT2\endcsname{\color[rgb]{0,0,1}}%
      \expandafter\def\csname LT3\endcsname{\color[rgb]{1,0,1}}%
      \expandafter\def\csname LT4\endcsname{\color[rgb]{0,1,1}}%
      \expandafter\def\csname LT5\endcsname{\color[rgb]{1,1,0}}%
      \expandafter\def\csname LT6\endcsname{\color[rgb]{0,0,0}}%
      \expandafter\def\csname LT7\endcsname{\color[rgb]{1,0.3,0}}%
      \expandafter\def\csname LT8\endcsname{\color[rgb]{0.5,0.5,0.5}}%
    \else
      \def\colorrgb#1{\color{black}}%
      \def\colorgray#1{\color[gray]{#1}}%
      \expandafter\def\csname LTw\endcsname{\color{white}}%
      \expandafter\def\csname LTb\endcsname{\color{black}}%
      \expandafter\def\csname LTa\endcsname{\color{black}}%
      \expandafter\def\csname LT0\endcsname{\color{black}}%
      \expandafter\def\csname LT1\endcsname{\color{black}}%
      \expandafter\def\csname LT2\endcsname{\color{black}}%
      \expandafter\def\csname LT3\endcsname{\color{black}}%
      \expandafter\def\csname LT4\endcsname{\color{black}}%
      \expandafter\def\csname LT5\endcsname{\color{black}}%
      \expandafter\def\csname LT6\endcsname{\color{black}}%
      \expandafter\def\csname LT7\endcsname{\color{black}}%
      \expandafter\def\csname LT8\endcsname{\color{black}}%
    \fi
  \fi
    \setlength{\unitlength}{0.0500bp}%
    \ifx\gptboxheight\undefined%
      \newlength{\gptboxheight}%
      \newlength{\gptboxwidth}%
      \newsavebox{\gptboxtext}%
    \fi%
    \setlength{\fboxrule}{0.5pt}%
    \setlength{\fboxsep}{1pt}%
\begin{picture}(6802.00,3614.00)%
    \gplgaddtomacro\gplbacktext{%
      \csname LTb\endcsname%
      \put(814,704){\makebox(0,0)[r]{\strut{}$0$}}%
      \csname LTb\endcsname%
      \put(814,1113){\makebox(0,0)[r]{\strut{}$0.2$}}%
      \csname LTb\endcsname%
      \put(814,1522){\makebox(0,0)[r]{\strut{}$0.4$}}%
      \csname LTb\endcsname%
      \put(814,1931){\makebox(0,0)[r]{\strut{}$0.6$}}%
      \csname LTb\endcsname%
      \put(814,2340){\makebox(0,0)[r]{\strut{}$0.8$}}%
      \csname LTb\endcsname%
      \put(814,2749){\makebox(0,0)[r]{\strut{}$1$}}%
      \csname LTb\endcsname%
      \put(1578,484){\makebox(0,0){\strut{}$10^{3}$}}%
      \csname LTb\endcsname%
      \put(3676,484){\makebox(0,0){\strut{}$10^{4}$}}%
      \csname LTb\endcsname%
      \put(5773,484){\makebox(0,0){\strut{}$10^{5}$}}%
    }%
    \gplgaddtomacro\gplfronttext{%
      \csname LTb\endcsname%
      \put(176,1828){\rotatebox{-270}{\makebox(0,0){\strut{}Avg. Local Clustering Coeff.}}}%
      \put(3675,154){\makebox(0,0){\strut{}Number $n$ of nodes}}%
      \put(3675,3283){\makebox(0,0){\strut{}Mixing: $\mu = 0.6$}}%
      \csname LTb\endcsname%
      \put(5418,2780){\makebox(0,0)[r]{\strut{}Orig}}%
      \csname LTb\endcsname%
      \put(5418,2560){\makebox(0,0)[r]{\strut{}EM}}%
    }%
    \gplbacktext
    \put(0,0){\includegraphics{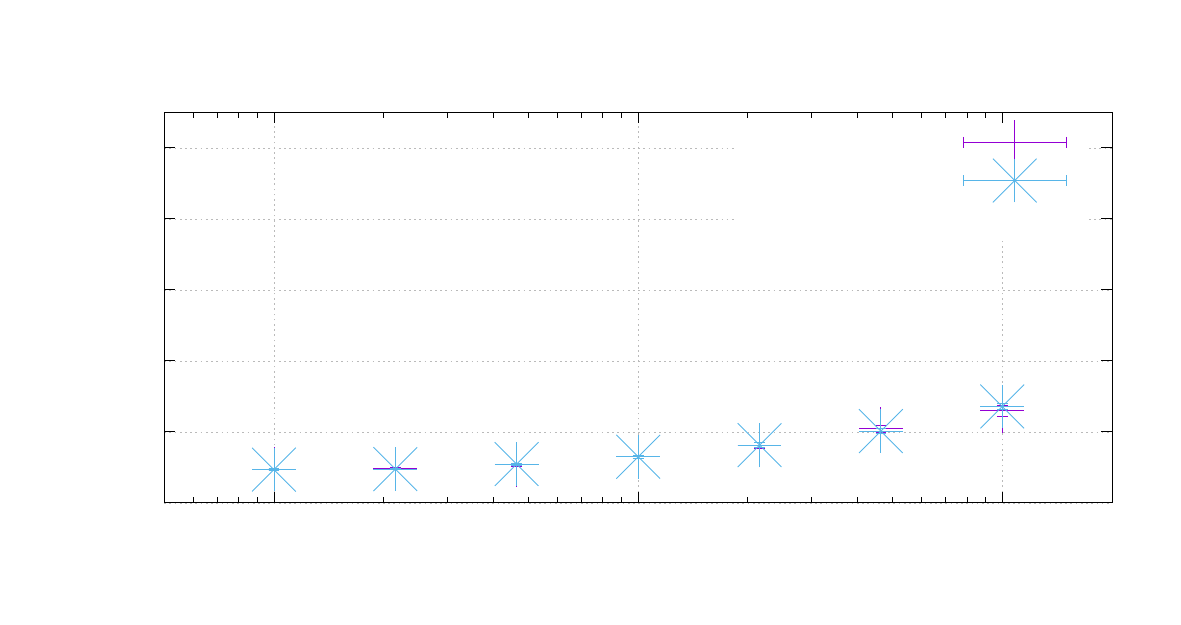}}%
    \gplfronttext
  \end{picture}%
\endgroup
}\hfill\\ %
\noindent\scalebox{\threescale}{%
\begingroup
  \makeatletter
  \providecommand\color[2][]{%
    \GenericError{(gnuplot) \space\space\space\@spaces}{%
      Package color not loaded in conjunction with
      terminal option `colourtext'%
    }{See the gnuplot documentation for explanation.%
    }{Either use 'blacktext' in gnuplot or load the package
      color.sty in LaTeX.}%
    \renewcommand\color[2][]{}%
  }%
  \providecommand\includegraphics[2][]{%
    \GenericError{(gnuplot) \space\space\space\@spaces}{%
      Package graphicx or graphics not loaded%
    }{See the gnuplot documentation for explanation.%
    }{The gnuplot epslatex terminal needs graphicx.sty or graphics.sty.}%
    \renewcommand\includegraphics[2][]{}%
  }%
  \providecommand\rotatebox[2]{#2}%
  \@ifundefined{ifGPcolor}{%
    \newif\ifGPcolor
    \GPcolortrue
  }{}%
  \@ifundefined{ifGPblacktext}{%
    \newif\ifGPblacktext
    \GPblacktextfalse
  }{}%
  \let\gplgaddtomacro\g@addto@macro
  \gdef\gplbacktext{}%
  \gdef\gplfronttext{}%
  \makeatother
  \ifGPblacktext
    \def\colorrgb#1{}%
    \def\colorgray#1{}%
  \else
    \ifGPcolor
      \def\colorrgb#1{\color[rgb]{#1}}%
      \def\colorgray#1{\color[gray]{#1}}%
      \expandafter\def\csname LTw\endcsname{\color{white}}%
      \expandafter\def\csname LTb\endcsname{\color{black}}%
      \expandafter\def\csname LTa\endcsname{\color{black}}%
      \expandafter\def\csname LT0\endcsname{\color[rgb]{1,0,0}}%
      \expandafter\def\csname LT1\endcsname{\color[rgb]{0,1,0}}%
      \expandafter\def\csname LT2\endcsname{\color[rgb]{0,0,1}}%
      \expandafter\def\csname LT3\endcsname{\color[rgb]{1,0,1}}%
      \expandafter\def\csname LT4\endcsname{\color[rgb]{0,1,1}}%
      \expandafter\def\csname LT5\endcsname{\color[rgb]{1,1,0}}%
      \expandafter\def\csname LT6\endcsname{\color[rgb]{0,0,0}}%
      \expandafter\def\csname LT7\endcsname{\color[rgb]{1,0.3,0}}%
      \expandafter\def\csname LT8\endcsname{\color[rgb]{0.5,0.5,0.5}}%
    \else
      \def\colorrgb#1{\color{black}}%
      \def\colorgray#1{\color[gray]{#1}}%
      \expandafter\def\csname LTw\endcsname{\color{white}}%
      \expandafter\def\csname LTb\endcsname{\color{black}}%
      \expandafter\def\csname LTa\endcsname{\color{black}}%
      \expandafter\def\csname LT0\endcsname{\color{black}}%
      \expandafter\def\csname LT1\endcsname{\color{black}}%
      \expandafter\def\csname LT2\endcsname{\color{black}}%
      \expandafter\def\csname LT3\endcsname{\color{black}}%
      \expandafter\def\csname LT4\endcsname{\color{black}}%
      \expandafter\def\csname LT5\endcsname{\color{black}}%
      \expandafter\def\csname LT6\endcsname{\color{black}}%
      \expandafter\def\csname LT7\endcsname{\color{black}}%
      \expandafter\def\csname LT8\endcsname{\color{black}}%
    \fi
  \fi
    \setlength{\unitlength}{0.0500bp}%
    \ifx\gptboxheight\undefined%
      \newlength{\gptboxheight}%
      \newlength{\gptboxwidth}%
      \newsavebox{\gptboxtext}%
    \fi%
    \setlength{\fboxrule}{0.5pt}%
    \setlength{\fboxsep}{1pt}%
\begin{picture}(6802.00,3614.00)%
    \gplgaddtomacro\gplbacktext{%
      \csname LTb\endcsname%
      \put(814,704){\makebox(0,0)[r]{\strut{}$0$}}%
      \csname LTb\endcsname%
      \put(814,1113){\makebox(0,0)[r]{\strut{}$0.2$}}%
      \csname LTb\endcsname%
      \put(814,1522){\makebox(0,0)[r]{\strut{}$0.4$}}%
      \csname LTb\endcsname%
      \put(814,1931){\makebox(0,0)[r]{\strut{}$0.6$}}%
      \csname LTb\endcsname%
      \put(814,2340){\makebox(0,0)[r]{\strut{}$0.8$}}%
      \csname LTb\endcsname%
      \put(814,2749){\makebox(0,0)[r]{\strut{}$1$}}%
      \csname LTb\endcsname%
      \put(1578,484){\makebox(0,0){\strut{}$10^{3}$}}%
      \csname LTb\endcsname%
      \put(3676,484){\makebox(0,0){\strut{}$10^{4}$}}%
      \csname LTb\endcsname%
      \put(5773,484){\makebox(0,0){\strut{}$10^{5}$}}%
    }%
    \gplgaddtomacro\gplfronttext{%
      \csname LTb\endcsname%
      \put(176,1828){\rotatebox{-270}{\makebox(0,0){\strut{}Degree Assortativity}}}%
      \put(3675,154){\makebox(0,0){\strut{}Number $n$ of nodes}}%
      \put(3675,3283){\makebox(0,0){\strut{}Mixing: $\mu = 0.2$, Degree Assortativity, Overlap: $\nu = 3$}}%
      \csname LTb\endcsname%
      \put(5418,2780){\makebox(0,0)[r]{\strut{}Orig}}%
      \csname LTb\endcsname%
      \put(5418,2560){\makebox(0,0)[r]{\strut{}EM}}%
    }%
    \gplbacktext
    \put(0,0){\includegraphics{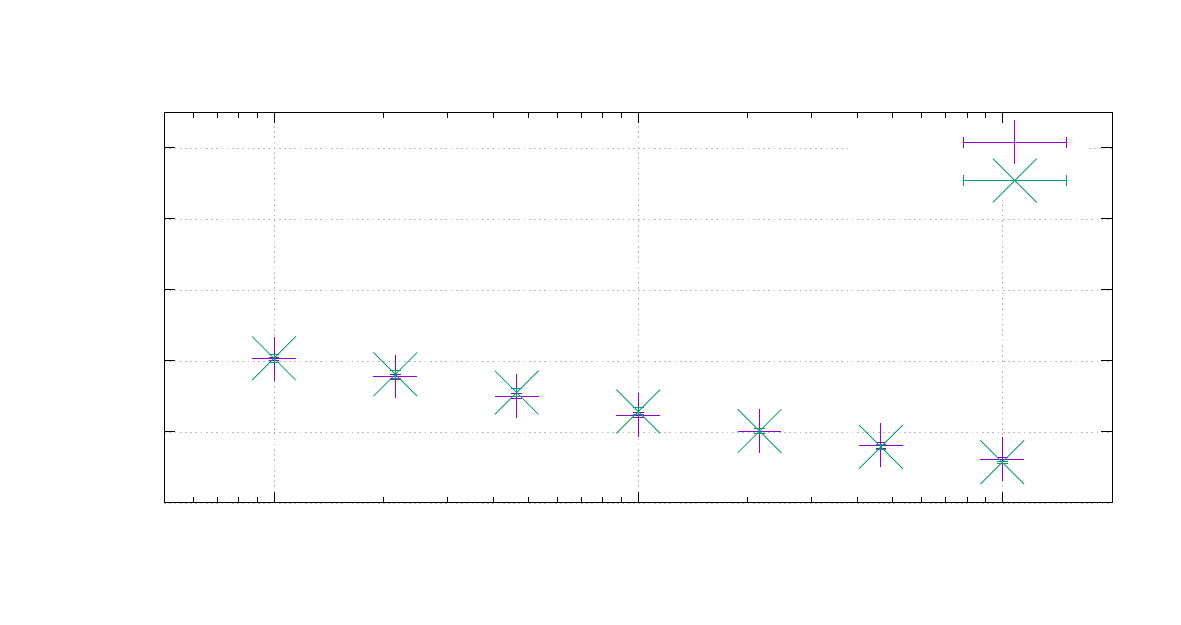}}%
    \gplfronttext
  \end{picture}%
\endgroup
}\hfill\scalebox{\threescale}{%
\begingroup
  \makeatletter
  \providecommand\color[2][]{%
    \GenericError{(gnuplot) \space\space\space\@spaces}{%
      Package color not loaded in conjunction with
      terminal option `colourtext'%
    }{See the gnuplot documentation for explanation.%
    }{Either use 'blacktext' in gnuplot or load the package
      color.sty in LaTeX.}%
    \renewcommand\color[2][]{}%
  }%
  \providecommand\includegraphics[2][]{%
    \GenericError{(gnuplot) \space\space\space\@spaces}{%
      Package graphicx or graphics not loaded%
    }{See the gnuplot documentation for explanation.%
    }{The gnuplot epslatex terminal needs graphicx.sty or graphics.sty.}%
    \renewcommand\includegraphics[2][]{}%
  }%
  \providecommand\rotatebox[2]{#2}%
  \@ifundefined{ifGPcolor}{%
    \newif\ifGPcolor
    \GPcolortrue
  }{}%
  \@ifundefined{ifGPblacktext}{%
    \newif\ifGPblacktext
    \GPblacktextfalse
  }{}%
  \let\gplgaddtomacro\g@addto@macro
  \gdef\gplbacktext{}%
  \gdef\gplfronttext{}%
  \makeatother
  \ifGPblacktext
    \def\colorrgb#1{}%
    \def\colorgray#1{}%
  \else
    \ifGPcolor
      \def\colorrgb#1{\color[rgb]{#1}}%
      \def\colorgray#1{\color[gray]{#1}}%
      \expandafter\def\csname LTw\endcsname{\color{white}}%
      \expandafter\def\csname LTb\endcsname{\color{black}}%
      \expandafter\def\csname LTa\endcsname{\color{black}}%
      \expandafter\def\csname LT0\endcsname{\color[rgb]{1,0,0}}%
      \expandafter\def\csname LT1\endcsname{\color[rgb]{0,1,0}}%
      \expandafter\def\csname LT2\endcsname{\color[rgb]{0,0,1}}%
      \expandafter\def\csname LT3\endcsname{\color[rgb]{1,0,1}}%
      \expandafter\def\csname LT4\endcsname{\color[rgb]{0,1,1}}%
      \expandafter\def\csname LT5\endcsname{\color[rgb]{1,1,0}}%
      \expandafter\def\csname LT6\endcsname{\color[rgb]{0,0,0}}%
      \expandafter\def\csname LT7\endcsname{\color[rgb]{1,0.3,0}}%
      \expandafter\def\csname LT8\endcsname{\color[rgb]{0.5,0.5,0.5}}%
    \else
      \def\colorrgb#1{\color{black}}%
      \def\colorgray#1{\color[gray]{#1}}%
      \expandafter\def\csname LTw\endcsname{\color{white}}%
      \expandafter\def\csname LTb\endcsname{\color{black}}%
      \expandafter\def\csname LTa\endcsname{\color{black}}%
      \expandafter\def\csname LT0\endcsname{\color{black}}%
      \expandafter\def\csname LT1\endcsname{\color{black}}%
      \expandafter\def\csname LT2\endcsname{\color{black}}%
      \expandafter\def\csname LT3\endcsname{\color{black}}%
      \expandafter\def\csname LT4\endcsname{\color{black}}%
      \expandafter\def\csname LT5\endcsname{\color{black}}%
      \expandafter\def\csname LT6\endcsname{\color{black}}%
      \expandafter\def\csname LT7\endcsname{\color{black}}%
      \expandafter\def\csname LT8\endcsname{\color{black}}%
    \fi
  \fi
    \setlength{\unitlength}{0.0500bp}%
    \ifx\gptboxheight\undefined%
      \newlength{\gptboxheight}%
      \newlength{\gptboxwidth}%
      \newsavebox{\gptboxtext}%
    \fi%
    \setlength{\fboxrule}{0.5pt}%
    \setlength{\fboxsep}{1pt}%
\begin{picture}(6802.00,3614.00)%
    \gplgaddtomacro\gplbacktext{%
      \csname LTb\endcsname%
      \put(814,704){\makebox(0,0)[r]{\strut{}$0$}}%
      \csname LTb\endcsname%
      \put(814,1113){\makebox(0,0)[r]{\strut{}$0.2$}}%
      \csname LTb\endcsname%
      \put(814,1522){\makebox(0,0)[r]{\strut{}$0.4$}}%
      \csname LTb\endcsname%
      \put(814,1931){\makebox(0,0)[r]{\strut{}$0.6$}}%
      \csname LTb\endcsname%
      \put(814,2340){\makebox(0,0)[r]{\strut{}$0.8$}}%
      \csname LTb\endcsname%
      \put(814,2749){\makebox(0,0)[r]{\strut{}$1$}}%
      \csname LTb\endcsname%
      \put(1578,484){\makebox(0,0){\strut{}$10^{3}$}}%
      \csname LTb\endcsname%
      \put(3676,484){\makebox(0,0){\strut{}$10^{4}$}}%
      \csname LTb\endcsname%
      \put(5773,484){\makebox(0,0){\strut{}$10^{5}$}}%
    }%
    \gplgaddtomacro\gplfronttext{%
      \csname LTb\endcsname%
      \put(176,1828){\rotatebox{-270}{\makebox(0,0){\strut{}Degree Assortativity}}}%
      \put(3675,154){\makebox(0,0){\strut{}Number $n$ of nodes}}%
      \put(3675,3283){\makebox(0,0){\strut{}Mixing: $\mu = 0.4$, Degree Assortativity, Overlap: $\nu = 3$}}%
      \csname LTb\endcsname%
      \put(5418,2780){\makebox(0,0)[r]{\strut{}Orig}}%
      \csname LTb\endcsname%
      \put(5418,2560){\makebox(0,0)[r]{\strut{}EM}}%
    }%
    \gplbacktext
    \put(0,0){\includegraphics{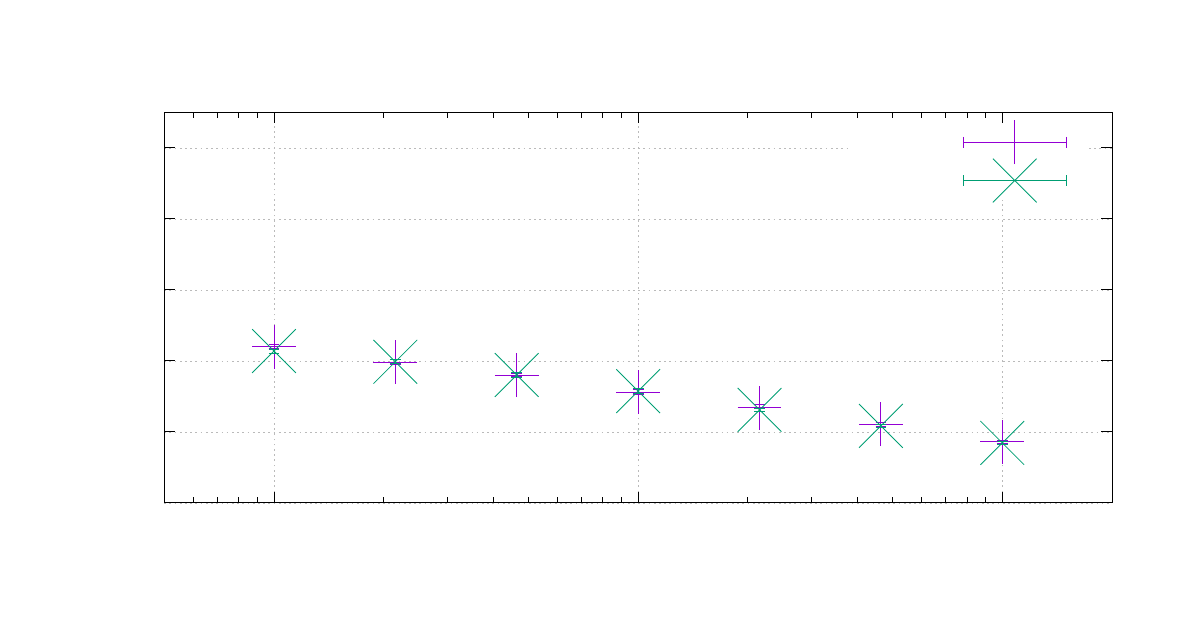}}%
    \gplfronttext
  \end{picture}%
\endgroup
}\hfill\scalebox{\threescale}{%
\begingroup
  \makeatletter
  \providecommand\color[2][]{%
    \GenericError{(gnuplot) \space\space\space\@spaces}{%
      Package color not loaded in conjunction with
      terminal option `colourtext'%
    }{See the gnuplot documentation for explanation.%
    }{Either use 'blacktext' in gnuplot or load the package
      color.sty in LaTeX.}%
    \renewcommand\color[2][]{}%
  }%
  \providecommand\includegraphics[2][]{%
    \GenericError{(gnuplot) \space\space\space\@spaces}{%
      Package graphicx or graphics not loaded%
    }{See the gnuplot documentation for explanation.%
    }{The gnuplot epslatex terminal needs graphicx.sty or graphics.sty.}%
    \renewcommand\includegraphics[2][]{}%
  }%
  \providecommand\rotatebox[2]{#2}%
  \@ifundefined{ifGPcolor}{%
    \newif\ifGPcolor
    \GPcolortrue
  }{}%
  \@ifundefined{ifGPblacktext}{%
    \newif\ifGPblacktext
    \GPblacktextfalse
  }{}%
  \let\gplgaddtomacro\g@addto@macro
  \gdef\gplbacktext{}%
  \gdef\gplfronttext{}%
  \makeatother
  \ifGPblacktext
    \def\colorrgb#1{}%
    \def\colorgray#1{}%
  \else
    \ifGPcolor
      \def\colorrgb#1{\color[rgb]{#1}}%
      \def\colorgray#1{\color[gray]{#1}}%
      \expandafter\def\csname LTw\endcsname{\color{white}}%
      \expandafter\def\csname LTb\endcsname{\color{black}}%
      \expandafter\def\csname LTa\endcsname{\color{black}}%
      \expandafter\def\csname LT0\endcsname{\color[rgb]{1,0,0}}%
      \expandafter\def\csname LT1\endcsname{\color[rgb]{0,1,0}}%
      \expandafter\def\csname LT2\endcsname{\color[rgb]{0,0,1}}%
      \expandafter\def\csname LT3\endcsname{\color[rgb]{1,0,1}}%
      \expandafter\def\csname LT4\endcsname{\color[rgb]{0,1,1}}%
      \expandafter\def\csname LT5\endcsname{\color[rgb]{1,1,0}}%
      \expandafter\def\csname LT6\endcsname{\color[rgb]{0,0,0}}%
      \expandafter\def\csname LT7\endcsname{\color[rgb]{1,0.3,0}}%
      \expandafter\def\csname LT8\endcsname{\color[rgb]{0.5,0.5,0.5}}%
    \else
      \def\colorrgb#1{\color{black}}%
      \def\colorgray#1{\color[gray]{#1}}%
      \expandafter\def\csname LTw\endcsname{\color{white}}%
      \expandafter\def\csname LTb\endcsname{\color{black}}%
      \expandafter\def\csname LTa\endcsname{\color{black}}%
      \expandafter\def\csname LT0\endcsname{\color{black}}%
      \expandafter\def\csname LT1\endcsname{\color{black}}%
      \expandafter\def\csname LT2\endcsname{\color{black}}%
      \expandafter\def\csname LT3\endcsname{\color{black}}%
      \expandafter\def\csname LT4\endcsname{\color{black}}%
      \expandafter\def\csname LT5\endcsname{\color{black}}%
      \expandafter\def\csname LT6\endcsname{\color{black}}%
      \expandafter\def\csname LT7\endcsname{\color{black}}%
      \expandafter\def\csname LT8\endcsname{\color{black}}%
    \fi
  \fi
    \setlength{\unitlength}{0.0500bp}%
    \ifx\gptboxheight\undefined%
      \newlength{\gptboxheight}%
      \newlength{\gptboxwidth}%
      \newsavebox{\gptboxtext}%
    \fi%
    \setlength{\fboxrule}{0.5pt}%
    \setlength{\fboxsep}{1pt}%
\begin{picture}(6802.00,3614.00)%
    \gplgaddtomacro\gplbacktext{%
      \csname LTb\endcsname%
      \put(814,704){\makebox(0,0)[r]{\strut{}$0$}}%
      \csname LTb\endcsname%
      \put(814,1113){\makebox(0,0)[r]{\strut{}$0.2$}}%
      \csname LTb\endcsname%
      \put(814,1522){\makebox(0,0)[r]{\strut{}$0.4$}}%
      \csname LTb\endcsname%
      \put(814,1931){\makebox(0,0)[r]{\strut{}$0.6$}}%
      \csname LTb\endcsname%
      \put(814,2340){\makebox(0,0)[r]{\strut{}$0.8$}}%
      \csname LTb\endcsname%
      \put(814,2749){\makebox(0,0)[r]{\strut{}$1$}}%
      \csname LTb\endcsname%
      \put(1578,484){\makebox(0,0){\strut{}$10^{3}$}}%
      \csname LTb\endcsname%
      \put(3676,484){\makebox(0,0){\strut{}$10^{4}$}}%
      \csname LTb\endcsname%
      \put(5773,484){\makebox(0,0){\strut{}$10^{5}$}}%
    }%
    \gplgaddtomacro\gplfronttext{%
      \csname LTb\endcsname%
      \put(176,1828){\rotatebox{-270}{\makebox(0,0){\strut{}Degree Assortativity}}}%
      \put(3675,154){\makebox(0,0){\strut{}Number $n$ of nodes}}%
      \put(3675,3283){\makebox(0,0){\strut{}Mixing: $\mu = 0.6$, Degree Assortativity, Overlap: $\nu = 3$}}%
      \csname LTb\endcsname%
      \put(5418,2780){\makebox(0,0)[r]{\strut{}Orig}}%
      \csname LTb\endcsname%
      \put(5418,2560){\makebox(0,0)[r]{\strut{}EM}}%
    }%
    \gplbacktext
    \put(0,0){\includegraphics{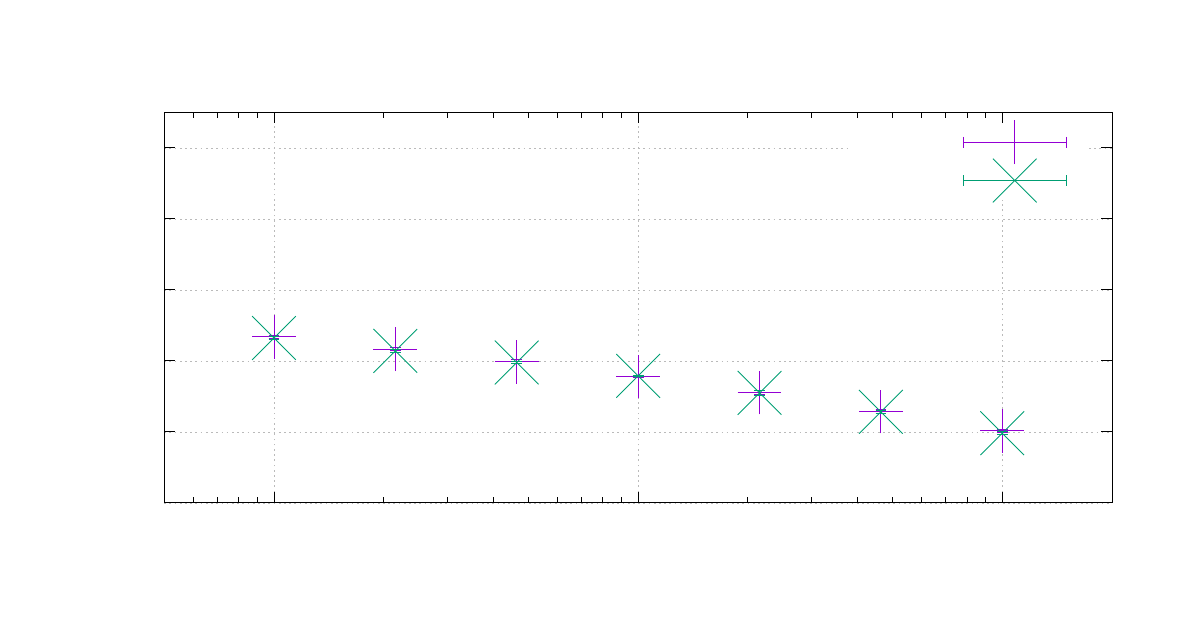}}%
    \gplfronttext
  \end{picture}%
\endgroup
}\hfill\\ %
	
\vspace{2em}	
\hrule	
\vspace{2em}	

\noindent\scalebox{\threescale}{%
\begingroup
  \makeatletter
  \providecommand\color[2][]{%
    \GenericError{(gnuplot) \space\space\space\@spaces}{%
      Package color not loaded in conjunction with
      terminal option `colourtext'%
    }{See the gnuplot documentation for explanation.%
    }{Either use 'blacktext' in gnuplot or load the package
      color.sty in LaTeX.}%
    \renewcommand\color[2][]{}%
  }%
  \providecommand\includegraphics[2][]{%
    \GenericError{(gnuplot) \space\space\space\@spaces}{%
      Package graphicx or graphics not loaded%
    }{See the gnuplot documentation for explanation.%
    }{The gnuplot epslatex terminal needs graphicx.sty or graphics.sty.}%
    \renewcommand\includegraphics[2][]{}%
  }%
  \providecommand\rotatebox[2]{#2}%
  \@ifundefined{ifGPcolor}{%
    \newif\ifGPcolor
    \GPcolortrue
  }{}%
  \@ifundefined{ifGPblacktext}{%
    \newif\ifGPblacktext
    \GPblacktextfalse
  }{}%
  \let\gplgaddtomacro\g@addto@macro
  \gdef\gplbacktext{}%
  \gdef\gplfronttext{}%
  \makeatother
  \ifGPblacktext
    \def\colorrgb#1{}%
    \def\colorgray#1{}%
  \else
    \ifGPcolor
      \def\colorrgb#1{\color[rgb]{#1}}%
      \def\colorgray#1{\color[gray]{#1}}%
      \expandafter\def\csname LTw\endcsname{\color{white}}%
      \expandafter\def\csname LTb\endcsname{\color{black}}%
      \expandafter\def\csname LTa\endcsname{\color{black}}%
      \expandafter\def\csname LT0\endcsname{\color[rgb]{1,0,0}}%
      \expandafter\def\csname LT1\endcsname{\color[rgb]{0,1,0}}%
      \expandafter\def\csname LT2\endcsname{\color[rgb]{0,0,1}}%
      \expandafter\def\csname LT3\endcsname{\color[rgb]{1,0,1}}%
      \expandafter\def\csname LT4\endcsname{\color[rgb]{0,1,1}}%
      \expandafter\def\csname LT5\endcsname{\color[rgb]{1,1,0}}%
      \expandafter\def\csname LT6\endcsname{\color[rgb]{0,0,0}}%
      \expandafter\def\csname LT7\endcsname{\color[rgb]{1,0.3,0}}%
      \expandafter\def\csname LT8\endcsname{\color[rgb]{0.5,0.5,0.5}}%
    \else
      \def\colorrgb#1{\color{black}}%
      \def\colorgray#1{\color[gray]{#1}}%
      \expandafter\def\csname LTw\endcsname{\color{white}}%
      \expandafter\def\csname LTb\endcsname{\color{black}}%
      \expandafter\def\csname LTa\endcsname{\color{black}}%
      \expandafter\def\csname LT0\endcsname{\color{black}}%
      \expandafter\def\csname LT1\endcsname{\color{black}}%
      \expandafter\def\csname LT2\endcsname{\color{black}}%
      \expandafter\def\csname LT3\endcsname{\color{black}}%
      \expandafter\def\csname LT4\endcsname{\color{black}}%
      \expandafter\def\csname LT5\endcsname{\color{black}}%
      \expandafter\def\csname LT6\endcsname{\color{black}}%
      \expandafter\def\csname LT7\endcsname{\color{black}}%
      \expandafter\def\csname LT8\endcsname{\color{black}}%
    \fi
  \fi
    \setlength{\unitlength}{0.0500bp}%
    \ifx\gptboxheight\undefined%
      \newlength{\gptboxheight}%
      \newlength{\gptboxwidth}%
      \newsavebox{\gptboxtext}%
    \fi%
    \setlength{\fboxrule}{0.5pt}%
    \setlength{\fboxsep}{1pt}%
\begin{picture}(6802.00,3614.00)%
    \gplgaddtomacro\gplbacktext{%
      \csname LTb\endcsname%
      \put(814,704){\makebox(0,0)[r]{\strut{}$0$}}%
      \csname LTb\endcsname%
      \put(814,1113){\makebox(0,0)[r]{\strut{}$0.2$}}%
      \csname LTb\endcsname%
      \put(814,1522){\makebox(0,0)[r]{\strut{}$0.4$}}%
      \csname LTb\endcsname%
      \put(814,1931){\makebox(0,0)[r]{\strut{}$0.6$}}%
      \csname LTb\endcsname%
      \put(814,2340){\makebox(0,0)[r]{\strut{}$0.8$}}%
      \csname LTb\endcsname%
      \put(814,2749){\makebox(0,0)[r]{\strut{}$1$}}%
      \csname LTb\endcsname%
      \put(1578,484){\makebox(0,0){\strut{}$10^{3}$}}%
      \csname LTb\endcsname%
      \put(3676,484){\makebox(0,0){\strut{}$10^{4}$}}%
      \csname LTb\endcsname%
      \put(5773,484){\makebox(0,0){\strut{}$10^{5}$}}%
    }%
    \gplgaddtomacro\gplfronttext{%
      \csname LTb\endcsname%
      \put(176,1828){\rotatebox{-270}{\makebox(0,0){\strut{}NMI}}}%
      \put(3675,154){\makebox(0,0){\strut{}Number $n$ of nodes}}%
      \put(3675,3283){\makebox(0,0){\strut{}Mixing: $\mu = 0.2$, Cluster: OSLOM, Overlap: $\nu = 4$}}%
      \csname LTb\endcsname%
      \put(5418,2780){\makebox(0,0)[r]{\strut{}Orig}}%
      \csname LTb\endcsname%
      \put(5418,2560){\makebox(0,0)[r]{\strut{}EM}}%
    }%
    \gplbacktext
    \put(0,0){\includegraphics{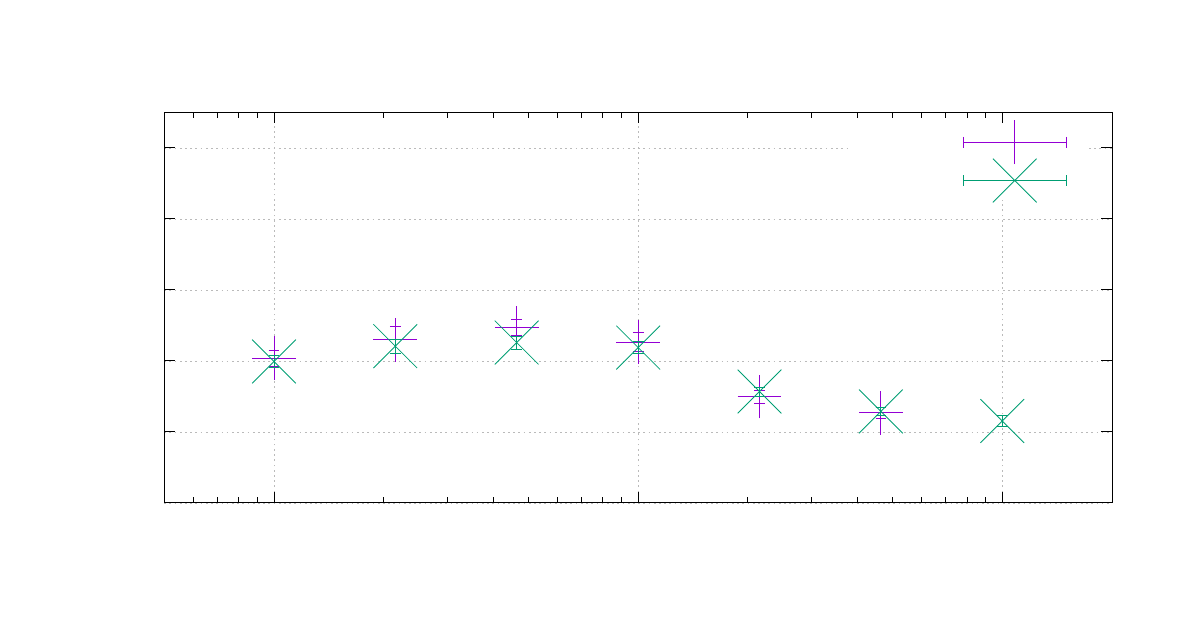}}%
    \gplfronttext
  \end{picture}%
\endgroup
}\hfill\scalebox{\threescale}{%
\begingroup
  \makeatletter
  \providecommand\color[2][]{%
    \GenericError{(gnuplot) \space\space\space\@spaces}{%
      Package color not loaded in conjunction with
      terminal option `colourtext'%
    }{See the gnuplot documentation for explanation.%
    }{Either use 'blacktext' in gnuplot or load the package
      color.sty in LaTeX.}%
    \renewcommand\color[2][]{}%
  }%
  \providecommand\includegraphics[2][]{%
    \GenericError{(gnuplot) \space\space\space\@spaces}{%
      Package graphicx or graphics not loaded%
    }{See the gnuplot documentation for explanation.%
    }{The gnuplot epslatex terminal needs graphicx.sty or graphics.sty.}%
    \renewcommand\includegraphics[2][]{}%
  }%
  \providecommand\rotatebox[2]{#2}%
  \@ifundefined{ifGPcolor}{%
    \newif\ifGPcolor
    \GPcolortrue
  }{}%
  \@ifundefined{ifGPblacktext}{%
    \newif\ifGPblacktext
    \GPblacktextfalse
  }{}%
  \let\gplgaddtomacro\g@addto@macro
  \gdef\gplbacktext{}%
  \gdef\gplfronttext{}%
  \makeatother
  \ifGPblacktext
    \def\colorrgb#1{}%
    \def\colorgray#1{}%
  \else
    \ifGPcolor
      \def\colorrgb#1{\color[rgb]{#1}}%
      \def\colorgray#1{\color[gray]{#1}}%
      \expandafter\def\csname LTw\endcsname{\color{white}}%
      \expandafter\def\csname LTb\endcsname{\color{black}}%
      \expandafter\def\csname LTa\endcsname{\color{black}}%
      \expandafter\def\csname LT0\endcsname{\color[rgb]{1,0,0}}%
      \expandafter\def\csname LT1\endcsname{\color[rgb]{0,1,0}}%
      \expandafter\def\csname LT2\endcsname{\color[rgb]{0,0,1}}%
      \expandafter\def\csname LT3\endcsname{\color[rgb]{1,0,1}}%
      \expandafter\def\csname LT4\endcsname{\color[rgb]{0,1,1}}%
      \expandafter\def\csname LT5\endcsname{\color[rgb]{1,1,0}}%
      \expandafter\def\csname LT6\endcsname{\color[rgb]{0,0,0}}%
      \expandafter\def\csname LT7\endcsname{\color[rgb]{1,0.3,0}}%
      \expandafter\def\csname LT8\endcsname{\color[rgb]{0.5,0.5,0.5}}%
    \else
      \def\colorrgb#1{\color{black}}%
      \def\colorgray#1{\color[gray]{#1}}%
      \expandafter\def\csname LTw\endcsname{\color{white}}%
      \expandafter\def\csname LTb\endcsname{\color{black}}%
      \expandafter\def\csname LTa\endcsname{\color{black}}%
      \expandafter\def\csname LT0\endcsname{\color{black}}%
      \expandafter\def\csname LT1\endcsname{\color{black}}%
      \expandafter\def\csname LT2\endcsname{\color{black}}%
      \expandafter\def\csname LT3\endcsname{\color{black}}%
      \expandafter\def\csname LT4\endcsname{\color{black}}%
      \expandafter\def\csname LT5\endcsname{\color{black}}%
      \expandafter\def\csname LT6\endcsname{\color{black}}%
      \expandafter\def\csname LT7\endcsname{\color{black}}%
      \expandafter\def\csname LT8\endcsname{\color{black}}%
    \fi
  \fi
    \setlength{\unitlength}{0.0500bp}%
    \ifx\gptboxheight\undefined%
      \newlength{\gptboxheight}%
      \newlength{\gptboxwidth}%
      \newsavebox{\gptboxtext}%
    \fi%
    \setlength{\fboxrule}{0.5pt}%
    \setlength{\fboxsep}{1pt}%
\begin{picture}(6802.00,3614.00)%
    \gplgaddtomacro\gplbacktext{%
      \csname LTb\endcsname%
      \put(814,704){\makebox(0,0)[r]{\strut{}$0$}}%
      \csname LTb\endcsname%
      \put(814,1113){\makebox(0,0)[r]{\strut{}$0.2$}}%
      \csname LTb\endcsname%
      \put(814,1522){\makebox(0,0)[r]{\strut{}$0.4$}}%
      \csname LTb\endcsname%
      \put(814,1931){\makebox(0,0)[r]{\strut{}$0.6$}}%
      \csname LTb\endcsname%
      \put(814,2340){\makebox(0,0)[r]{\strut{}$0.8$}}%
      \csname LTb\endcsname%
      \put(814,2749){\makebox(0,0)[r]{\strut{}$1$}}%
      \csname LTb\endcsname%
      \put(1578,484){\makebox(0,0){\strut{}$10^{3}$}}%
      \csname LTb\endcsname%
      \put(3676,484){\makebox(0,0){\strut{}$10^{4}$}}%
      \csname LTb\endcsname%
      \put(5773,484){\makebox(0,0){\strut{}$10^{5}$}}%
    }%
    \gplgaddtomacro\gplfronttext{%
      \csname LTb\endcsname%
      \put(176,1828){\rotatebox{-270}{\makebox(0,0){\strut{}NMI}}}%
      \put(3675,154){\makebox(0,0){\strut{}Number $n$ of nodes}}%
      \put(3675,3283){\makebox(0,0){\strut{}Mixing: $\mu = 0.4$, Cluster: OSLOM, Overlap: $\nu = 4$}}%
      \csname LTb\endcsname%
      \put(5418,2780){\makebox(0,0)[r]{\strut{}Orig}}%
      \csname LTb\endcsname%
      \put(5418,2560){\makebox(0,0)[r]{\strut{}EM}}%
    }%
    \gplbacktext
    \put(0,0){\includegraphics{lfr_nmi_4_4}}%
    \gplfronttext
  \end{picture}%
\endgroup
}\hfill\scalebox{\threescale}{%
\begingroup
  \makeatletter
  \providecommand\color[2][]{%
    \GenericError{(gnuplot) \space\space\space\@spaces}{%
      Package color not loaded in conjunction with
      terminal option `colourtext'%
    }{See the gnuplot documentation for explanation.%
    }{Either use 'blacktext' in gnuplot or load the package
      color.sty in LaTeX.}%
    \renewcommand\color[2][]{}%
  }%
  \providecommand\includegraphics[2][]{%
    \GenericError{(gnuplot) \space\space\space\@spaces}{%
      Package graphicx or graphics not loaded%
    }{See the gnuplot documentation for explanation.%
    }{The gnuplot epslatex terminal needs graphicx.sty or graphics.sty.}%
    \renewcommand\includegraphics[2][]{}%
  }%
  \providecommand\rotatebox[2]{#2}%
  \@ifundefined{ifGPcolor}{%
    \newif\ifGPcolor
    \GPcolortrue
  }{}%
  \@ifundefined{ifGPblacktext}{%
    \newif\ifGPblacktext
    \GPblacktextfalse
  }{}%
  \let\gplgaddtomacro\g@addto@macro
  \gdef\gplbacktext{}%
  \gdef\gplfronttext{}%
  \makeatother
  \ifGPblacktext
    \def\colorrgb#1{}%
    \def\colorgray#1{}%
  \else
    \ifGPcolor
      \def\colorrgb#1{\color[rgb]{#1}}%
      \def\colorgray#1{\color[gray]{#1}}%
      \expandafter\def\csname LTw\endcsname{\color{white}}%
      \expandafter\def\csname LTb\endcsname{\color{black}}%
      \expandafter\def\csname LTa\endcsname{\color{black}}%
      \expandafter\def\csname LT0\endcsname{\color[rgb]{1,0,0}}%
      \expandafter\def\csname LT1\endcsname{\color[rgb]{0,1,0}}%
      \expandafter\def\csname LT2\endcsname{\color[rgb]{0,0,1}}%
      \expandafter\def\csname LT3\endcsname{\color[rgb]{1,0,1}}%
      \expandafter\def\csname LT4\endcsname{\color[rgb]{0,1,1}}%
      \expandafter\def\csname LT5\endcsname{\color[rgb]{1,1,0}}%
      \expandafter\def\csname LT6\endcsname{\color[rgb]{0,0,0}}%
      \expandafter\def\csname LT7\endcsname{\color[rgb]{1,0.3,0}}%
      \expandafter\def\csname LT8\endcsname{\color[rgb]{0.5,0.5,0.5}}%
    \else
      \def\colorrgb#1{\color{black}}%
      \def\colorgray#1{\color[gray]{#1}}%
      \expandafter\def\csname LTw\endcsname{\color{white}}%
      \expandafter\def\csname LTb\endcsname{\color{black}}%
      \expandafter\def\csname LTa\endcsname{\color{black}}%
      \expandafter\def\csname LT0\endcsname{\color{black}}%
      \expandafter\def\csname LT1\endcsname{\color{black}}%
      \expandafter\def\csname LT2\endcsname{\color{black}}%
      \expandafter\def\csname LT3\endcsname{\color{black}}%
      \expandafter\def\csname LT4\endcsname{\color{black}}%
      \expandafter\def\csname LT5\endcsname{\color{black}}%
      \expandafter\def\csname LT6\endcsname{\color{black}}%
      \expandafter\def\csname LT7\endcsname{\color{black}}%
      \expandafter\def\csname LT8\endcsname{\color{black}}%
    \fi
  \fi
    \setlength{\unitlength}{0.0500bp}%
    \ifx\gptboxheight\undefined%
      \newlength{\gptboxheight}%
      \newlength{\gptboxwidth}%
      \newsavebox{\gptboxtext}%
    \fi%
    \setlength{\fboxrule}{0.5pt}%
    \setlength{\fboxsep}{1pt}%
\begin{picture}(6802.00,3614.00)%
    \gplgaddtomacro\gplbacktext{%
      \csname LTb\endcsname%
      \put(814,704){\makebox(0,0)[r]{\strut{}$0$}}%
      \csname LTb\endcsname%
      \put(814,1113){\makebox(0,0)[r]{\strut{}$0.2$}}%
      \csname LTb\endcsname%
      \put(814,1522){\makebox(0,0)[r]{\strut{}$0.4$}}%
      \csname LTb\endcsname%
      \put(814,1931){\makebox(0,0)[r]{\strut{}$0.6$}}%
      \csname LTb\endcsname%
      \put(814,2340){\makebox(0,0)[r]{\strut{}$0.8$}}%
      \csname LTb\endcsname%
      \put(814,2749){\makebox(0,0)[r]{\strut{}$1$}}%
      \csname LTb\endcsname%
      \put(1578,484){\makebox(0,0){\strut{}$10^{3}$}}%
      \csname LTb\endcsname%
      \put(3676,484){\makebox(0,0){\strut{}$10^{4}$}}%
      \csname LTb\endcsname%
      \put(5773,484){\makebox(0,0){\strut{}$10^{5}$}}%
    }%
    \gplgaddtomacro\gplfronttext{%
      \csname LTb\endcsname%
      \put(176,1828){\rotatebox{-270}{\makebox(0,0){\strut{}NMI}}}%
      \put(3675,154){\makebox(0,0){\strut{}Number $n$ of nodes}}%
      \put(3675,3283){\makebox(0,0){\strut{}Mixing: $\mu = 0.6$, Cluster: OSLOM, Overlap: $\nu = 4$}}%
      \csname LTb\endcsname%
      \put(5418,2780){\makebox(0,0)[r]{\strut{}Orig}}%
      \csname LTb\endcsname%
      \put(5418,2560){\makebox(0,0)[r]{\strut{}EM}}%
    }%
    \gplbacktext
    \put(0,0){\includegraphics{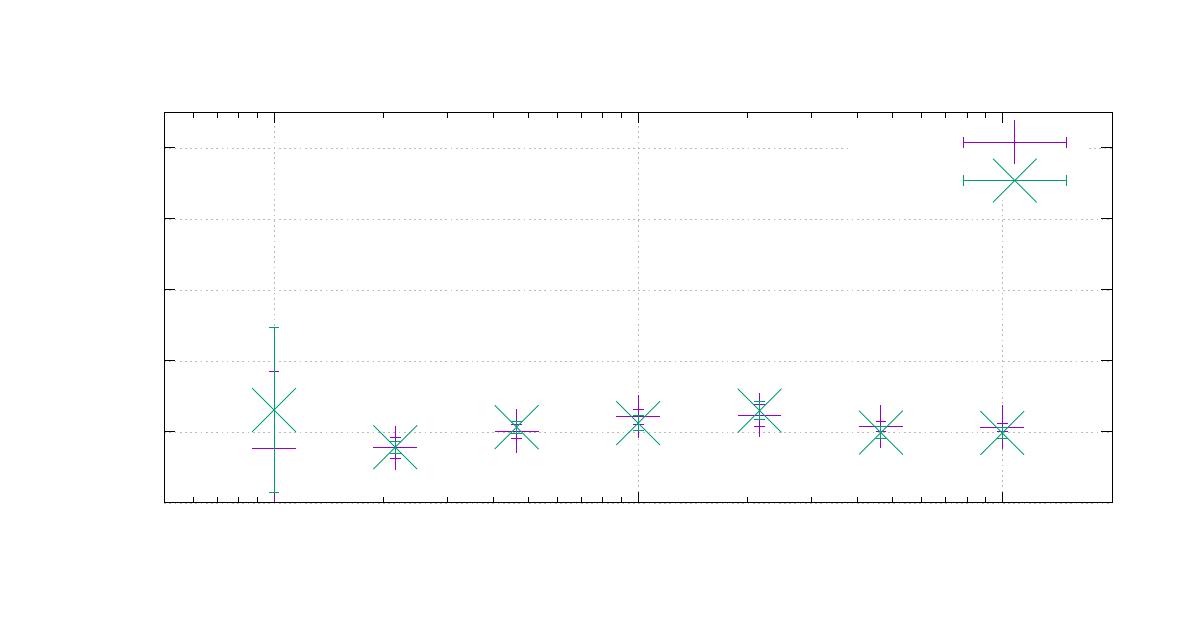}}%
    \gplfronttext
  \end{picture}%
\endgroup
}\hfill\\ %
\noindent\scalebox{\threescale}{%
\begingroup
  \makeatletter
  \providecommand\color[2][]{%
    \GenericError{(gnuplot) \space\space\space\@spaces}{%
      Package color not loaded in conjunction with
      terminal option `colourtext'%
    }{See the gnuplot documentation for explanation.%
    }{Either use 'blacktext' in gnuplot or load the package
      color.sty in LaTeX.}%
    \renewcommand\color[2][]{}%
  }%
  \providecommand\includegraphics[2][]{%
    \GenericError{(gnuplot) \space\space\space\@spaces}{%
      Package graphicx or graphics not loaded%
    }{See the gnuplot documentation for explanation.%
    }{The gnuplot epslatex terminal needs graphicx.sty or graphics.sty.}%
    \renewcommand\includegraphics[2][]{}%
  }%
  \providecommand\rotatebox[2]{#2}%
  \@ifundefined{ifGPcolor}{%
    \newif\ifGPcolor
    \GPcolortrue
  }{}%
  \@ifundefined{ifGPblacktext}{%
    \newif\ifGPblacktext
    \GPblacktextfalse
  }{}%
  \let\gplgaddtomacro\g@addto@macro
  \gdef\gplbacktext{}%
  \gdef\gplfronttext{}%
  \makeatother
  \ifGPblacktext
    \def\colorrgb#1{}%
    \def\colorgray#1{}%
  \else
    \ifGPcolor
      \def\colorrgb#1{\color[rgb]{#1}}%
      \def\colorgray#1{\color[gray]{#1}}%
      \expandafter\def\csname LTw\endcsname{\color{white}}%
      \expandafter\def\csname LTb\endcsname{\color{black}}%
      \expandafter\def\csname LTa\endcsname{\color{black}}%
      \expandafter\def\csname LT0\endcsname{\color[rgb]{1,0,0}}%
      \expandafter\def\csname LT1\endcsname{\color[rgb]{0,1,0}}%
      \expandafter\def\csname LT2\endcsname{\color[rgb]{0,0,1}}%
      \expandafter\def\csname LT3\endcsname{\color[rgb]{1,0,1}}%
      \expandafter\def\csname LT4\endcsname{\color[rgb]{0,1,1}}%
      \expandafter\def\csname LT5\endcsname{\color[rgb]{1,1,0}}%
      \expandafter\def\csname LT6\endcsname{\color[rgb]{0,0,0}}%
      \expandafter\def\csname LT7\endcsname{\color[rgb]{1,0.3,0}}%
      \expandafter\def\csname LT8\endcsname{\color[rgb]{0.5,0.5,0.5}}%
    \else
      \def\colorrgb#1{\color{black}}%
      \def\colorgray#1{\color[gray]{#1}}%
      \expandafter\def\csname LTw\endcsname{\color{white}}%
      \expandafter\def\csname LTb\endcsname{\color{black}}%
      \expandafter\def\csname LTa\endcsname{\color{black}}%
      \expandafter\def\csname LT0\endcsname{\color{black}}%
      \expandafter\def\csname LT1\endcsname{\color{black}}%
      \expandafter\def\csname LT2\endcsname{\color{black}}%
      \expandafter\def\csname LT3\endcsname{\color{black}}%
      \expandafter\def\csname LT4\endcsname{\color{black}}%
      \expandafter\def\csname LT5\endcsname{\color{black}}%
      \expandafter\def\csname LT6\endcsname{\color{black}}%
      \expandafter\def\csname LT7\endcsname{\color{black}}%
      \expandafter\def\csname LT8\endcsname{\color{black}}%
    \fi
  \fi
    \setlength{\unitlength}{0.0500bp}%
    \ifx\gptboxheight\undefined%
      \newlength{\gptboxheight}%
      \newlength{\gptboxwidth}%
      \newsavebox{\gptboxtext}%
    \fi%
    \setlength{\fboxrule}{0.5pt}%
    \setlength{\fboxsep}{1pt}%
\begin{picture}(6802.00,3614.00)%
    \gplgaddtomacro\gplbacktext{%
      \csname LTb\endcsname%
      \put(814,704){\makebox(0,0)[r]{\strut{}$0$}}%
      \csname LTb\endcsname%
      \put(814,1113){\makebox(0,0)[r]{\strut{}$0.2$}}%
      \csname LTb\endcsname%
      \put(814,1522){\makebox(0,0)[r]{\strut{}$0.4$}}%
      \csname LTb\endcsname%
      \put(814,1931){\makebox(0,0)[r]{\strut{}$0.6$}}%
      \csname LTb\endcsname%
      \put(814,2340){\makebox(0,0)[r]{\strut{}$0.8$}}%
      \csname LTb\endcsname%
      \put(814,2749){\makebox(0,0)[r]{\strut{}$1$}}%
      \csname LTb\endcsname%
      \put(1578,484){\makebox(0,0){\strut{}$10^{3}$}}%
      \csname LTb\endcsname%
      \put(3676,484){\makebox(0,0){\strut{}$10^{4}$}}%
      \csname LTb\endcsname%
      \put(5773,484){\makebox(0,0){\strut{}$10^{5}$}}%
    }%
    \gplgaddtomacro\gplfronttext{%
      \csname LTb\endcsname%
      \put(176,1828){\rotatebox{-270}{\makebox(0,0){\strut{}Avg. Local Clustering Coeff.}}}%
      \put(3675,154){\makebox(0,0){\strut{}Number $n$ of nodes}}%
      \put(3675,3283){\makebox(0,0){\strut{}Mixing: $\mu = 0.2$}}%
      \csname LTb\endcsname%
      \put(5418,2780){\makebox(0,0)[r]{\strut{}Orig}}%
      \csname LTb\endcsname%
      \put(5418,2560){\makebox(0,0)[r]{\strut{}EM}}%
    }%
    \gplbacktext
    \put(0,0){\includegraphics{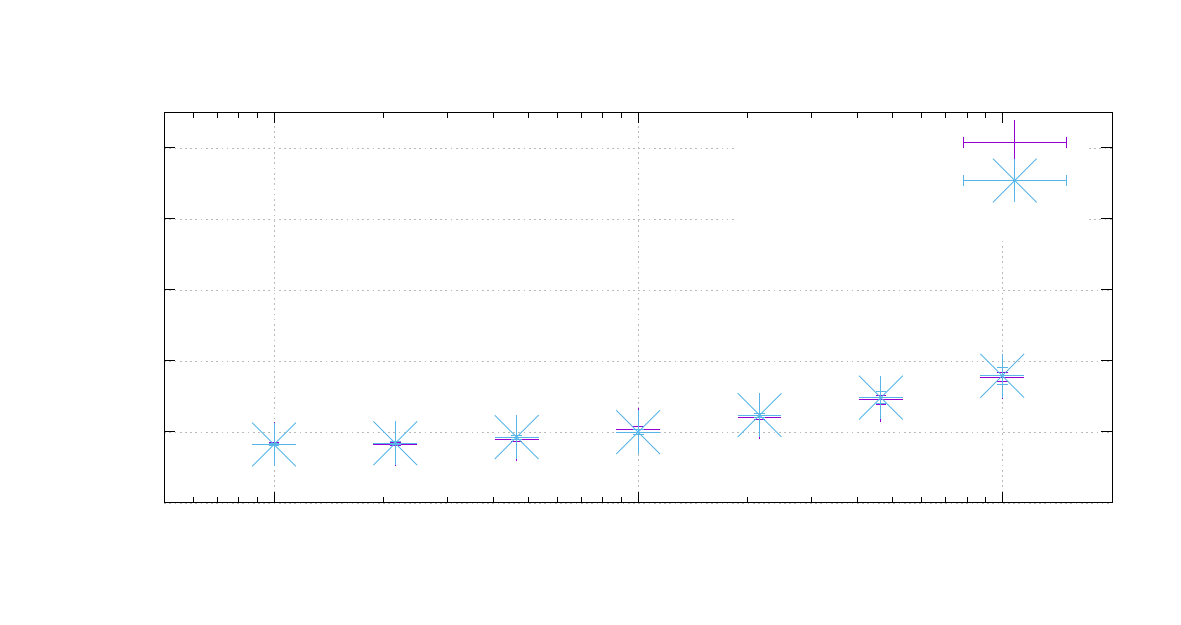}}%
    \gplfronttext
  \end{picture}%
\endgroup
}\hfill\scalebox{\threescale}{%
\begingroup
  \makeatletter
  \providecommand\color[2][]{%
    \GenericError{(gnuplot) \space\space\space\@spaces}{%
      Package color not loaded in conjunction with
      terminal option `colourtext'%
    }{See the gnuplot documentation for explanation.%
    }{Either use 'blacktext' in gnuplot or load the package
      color.sty in LaTeX.}%
    \renewcommand\color[2][]{}%
  }%
  \providecommand\includegraphics[2][]{%
    \GenericError{(gnuplot) \space\space\space\@spaces}{%
      Package graphicx or graphics not loaded%
    }{See the gnuplot documentation for explanation.%
    }{The gnuplot epslatex terminal needs graphicx.sty or graphics.sty.}%
    \renewcommand\includegraphics[2][]{}%
  }%
  \providecommand\rotatebox[2]{#2}%
  \@ifundefined{ifGPcolor}{%
    \newif\ifGPcolor
    \GPcolortrue
  }{}%
  \@ifundefined{ifGPblacktext}{%
    \newif\ifGPblacktext
    \GPblacktextfalse
  }{}%
  \let\gplgaddtomacro\g@addto@macro
  \gdef\gplbacktext{}%
  \gdef\gplfronttext{}%
  \makeatother
  \ifGPblacktext
    \def\colorrgb#1{}%
    \def\colorgray#1{}%
  \else
    \ifGPcolor
      \def\colorrgb#1{\color[rgb]{#1}}%
      \def\colorgray#1{\color[gray]{#1}}%
      \expandafter\def\csname LTw\endcsname{\color{white}}%
      \expandafter\def\csname LTb\endcsname{\color{black}}%
      \expandafter\def\csname LTa\endcsname{\color{black}}%
      \expandafter\def\csname LT0\endcsname{\color[rgb]{1,0,0}}%
      \expandafter\def\csname LT1\endcsname{\color[rgb]{0,1,0}}%
      \expandafter\def\csname LT2\endcsname{\color[rgb]{0,0,1}}%
      \expandafter\def\csname LT3\endcsname{\color[rgb]{1,0,1}}%
      \expandafter\def\csname LT4\endcsname{\color[rgb]{0,1,1}}%
      \expandafter\def\csname LT5\endcsname{\color[rgb]{1,1,0}}%
      \expandafter\def\csname LT6\endcsname{\color[rgb]{0,0,0}}%
      \expandafter\def\csname LT7\endcsname{\color[rgb]{1,0.3,0}}%
      \expandafter\def\csname LT8\endcsname{\color[rgb]{0.5,0.5,0.5}}%
    \else
      \def\colorrgb#1{\color{black}}%
      \def\colorgray#1{\color[gray]{#1}}%
      \expandafter\def\csname LTw\endcsname{\color{white}}%
      \expandafter\def\csname LTb\endcsname{\color{black}}%
      \expandafter\def\csname LTa\endcsname{\color{black}}%
      \expandafter\def\csname LT0\endcsname{\color{black}}%
      \expandafter\def\csname LT1\endcsname{\color{black}}%
      \expandafter\def\csname LT2\endcsname{\color{black}}%
      \expandafter\def\csname LT3\endcsname{\color{black}}%
      \expandafter\def\csname LT4\endcsname{\color{black}}%
      \expandafter\def\csname LT5\endcsname{\color{black}}%
      \expandafter\def\csname LT6\endcsname{\color{black}}%
      \expandafter\def\csname LT7\endcsname{\color{black}}%
      \expandafter\def\csname LT8\endcsname{\color{black}}%
    \fi
  \fi
    \setlength{\unitlength}{0.0500bp}%
    \ifx\gptboxheight\undefined%
      \newlength{\gptboxheight}%
      \newlength{\gptboxwidth}%
      \newsavebox{\gptboxtext}%
    \fi%
    \setlength{\fboxrule}{0.5pt}%
    \setlength{\fboxsep}{1pt}%
\begin{picture}(6802.00,3614.00)%
    \gplgaddtomacro\gplbacktext{%
      \csname LTb\endcsname%
      \put(814,704){\makebox(0,0)[r]{\strut{}$0$}}%
      \csname LTb\endcsname%
      \put(814,1113){\makebox(0,0)[r]{\strut{}$0.2$}}%
      \csname LTb\endcsname%
      \put(814,1522){\makebox(0,0)[r]{\strut{}$0.4$}}%
      \csname LTb\endcsname%
      \put(814,1931){\makebox(0,0)[r]{\strut{}$0.6$}}%
      \csname LTb\endcsname%
      \put(814,2340){\makebox(0,0)[r]{\strut{}$0.8$}}%
      \csname LTb\endcsname%
      \put(814,2749){\makebox(0,0)[r]{\strut{}$1$}}%
      \csname LTb\endcsname%
      \put(1578,484){\makebox(0,0){\strut{}$10^{3}$}}%
      \csname LTb\endcsname%
      \put(3676,484){\makebox(0,0){\strut{}$10^{4}$}}%
      \csname LTb\endcsname%
      \put(5773,484){\makebox(0,0){\strut{}$10^{5}$}}%
    }%
    \gplgaddtomacro\gplfronttext{%
      \csname LTb\endcsname%
      \put(176,1828){\rotatebox{-270}{\makebox(0,0){\strut{}Avg. Local Clustering Coeff.}}}%
      \put(3675,154){\makebox(0,0){\strut{}Number $n$ of nodes}}%
      \put(3675,3283){\makebox(0,0){\strut{}Mixing: $\mu = 0.4$}}%
      \csname LTb\endcsname%
      \put(5418,2780){\makebox(0,0)[r]{\strut{}Orig}}%
      \csname LTb\endcsname%
      \put(5418,2560){\makebox(0,0)[r]{\strut{}EM}}%
    }%
    \gplbacktext
    \put(0,0){\includegraphics{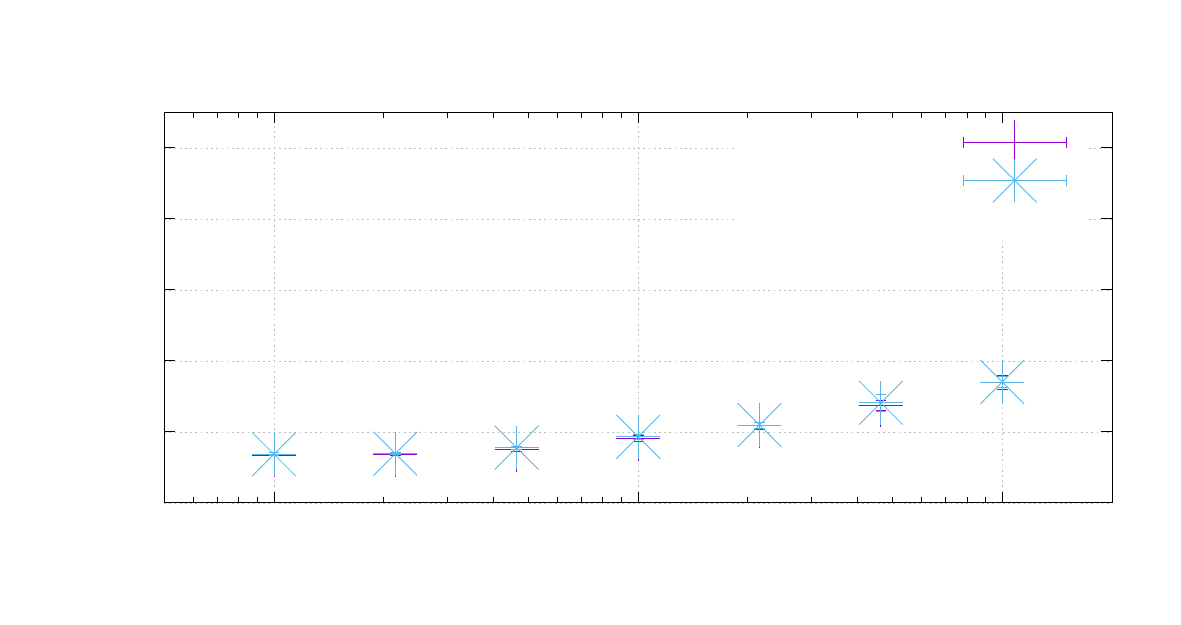}}%
    \gplfronttext
  \end{picture}%
\endgroup
}\hfill\scalebox{\threescale}{%
\begingroup
  \makeatletter
  \providecommand\color[2][]{%
    \GenericError{(gnuplot) \space\space\space\@spaces}{%
      Package color not loaded in conjunction with
      terminal option `colourtext'%
    }{See the gnuplot documentation for explanation.%
    }{Either use 'blacktext' in gnuplot or load the package
      color.sty in LaTeX.}%
    \renewcommand\color[2][]{}%
  }%
  \providecommand\includegraphics[2][]{%
    \GenericError{(gnuplot) \space\space\space\@spaces}{%
      Package graphicx or graphics not loaded%
    }{See the gnuplot documentation for explanation.%
    }{The gnuplot epslatex terminal needs graphicx.sty or graphics.sty.}%
    \renewcommand\includegraphics[2][]{}%
  }%
  \providecommand\rotatebox[2]{#2}%
  \@ifundefined{ifGPcolor}{%
    \newif\ifGPcolor
    \GPcolortrue
  }{}%
  \@ifundefined{ifGPblacktext}{%
    \newif\ifGPblacktext
    \GPblacktextfalse
  }{}%
  \let\gplgaddtomacro\g@addto@macro
  \gdef\gplbacktext{}%
  \gdef\gplfronttext{}%
  \makeatother
  \ifGPblacktext
    \def\colorrgb#1{}%
    \def\colorgray#1{}%
  \else
    \ifGPcolor
      \def\colorrgb#1{\color[rgb]{#1}}%
      \def\colorgray#1{\color[gray]{#1}}%
      \expandafter\def\csname LTw\endcsname{\color{white}}%
      \expandafter\def\csname LTb\endcsname{\color{black}}%
      \expandafter\def\csname LTa\endcsname{\color{black}}%
      \expandafter\def\csname LT0\endcsname{\color[rgb]{1,0,0}}%
      \expandafter\def\csname LT1\endcsname{\color[rgb]{0,1,0}}%
      \expandafter\def\csname LT2\endcsname{\color[rgb]{0,0,1}}%
      \expandafter\def\csname LT3\endcsname{\color[rgb]{1,0,1}}%
      \expandafter\def\csname LT4\endcsname{\color[rgb]{0,1,1}}%
      \expandafter\def\csname LT5\endcsname{\color[rgb]{1,1,0}}%
      \expandafter\def\csname LT6\endcsname{\color[rgb]{0,0,0}}%
      \expandafter\def\csname LT7\endcsname{\color[rgb]{1,0.3,0}}%
      \expandafter\def\csname LT8\endcsname{\color[rgb]{0.5,0.5,0.5}}%
    \else
      \def\colorrgb#1{\color{black}}%
      \def\colorgray#1{\color[gray]{#1}}%
      \expandafter\def\csname LTw\endcsname{\color{white}}%
      \expandafter\def\csname LTb\endcsname{\color{black}}%
      \expandafter\def\csname LTa\endcsname{\color{black}}%
      \expandafter\def\csname LT0\endcsname{\color{black}}%
      \expandafter\def\csname LT1\endcsname{\color{black}}%
      \expandafter\def\csname LT2\endcsname{\color{black}}%
      \expandafter\def\csname LT3\endcsname{\color{black}}%
      \expandafter\def\csname LT4\endcsname{\color{black}}%
      \expandafter\def\csname LT5\endcsname{\color{black}}%
      \expandafter\def\csname LT6\endcsname{\color{black}}%
      \expandafter\def\csname LT7\endcsname{\color{black}}%
      \expandafter\def\csname LT8\endcsname{\color{black}}%
    \fi
  \fi
    \setlength{\unitlength}{0.0500bp}%
    \ifx\gptboxheight\undefined%
      \newlength{\gptboxheight}%
      \newlength{\gptboxwidth}%
      \newsavebox{\gptboxtext}%
    \fi%
    \setlength{\fboxrule}{0.5pt}%
    \setlength{\fboxsep}{1pt}%
\begin{picture}(6802.00,3614.00)%
    \gplgaddtomacro\gplbacktext{%
      \csname LTb\endcsname%
      \put(814,704){\makebox(0,0)[r]{\strut{}$0$}}%
      \csname LTb\endcsname%
      \put(814,1113){\makebox(0,0)[r]{\strut{}$0.2$}}%
      \csname LTb\endcsname%
      \put(814,1522){\makebox(0,0)[r]{\strut{}$0.4$}}%
      \csname LTb\endcsname%
      \put(814,1931){\makebox(0,0)[r]{\strut{}$0.6$}}%
      \csname LTb\endcsname%
      \put(814,2340){\makebox(0,0)[r]{\strut{}$0.8$}}%
      \csname LTb\endcsname%
      \put(814,2749){\makebox(0,0)[r]{\strut{}$1$}}%
      \csname LTb\endcsname%
      \put(1578,484){\makebox(0,0){\strut{}$10^{3}$}}%
      \csname LTb\endcsname%
      \put(3676,484){\makebox(0,0){\strut{}$10^{4}$}}%
      \csname LTb\endcsname%
      \put(5773,484){\makebox(0,0){\strut{}$10^{5}$}}%
    }%
    \gplgaddtomacro\gplfronttext{%
      \csname LTb\endcsname%
      \put(176,1828){\rotatebox{-270}{\makebox(0,0){\strut{}Avg. Local Clustering Coeff.}}}%
      \put(3675,154){\makebox(0,0){\strut{}Number $n$ of nodes}}%
      \put(3675,3283){\makebox(0,0){\strut{}Mixing: $\mu = 0.6$}}%
      \csname LTb\endcsname%
      \put(5418,2780){\makebox(0,0)[r]{\strut{}Orig}}%
      \csname LTb\endcsname%
      \put(5418,2560){\makebox(0,0)[r]{\strut{}EM}}%
    }%
    \gplbacktext
    \put(0,0){\includegraphics{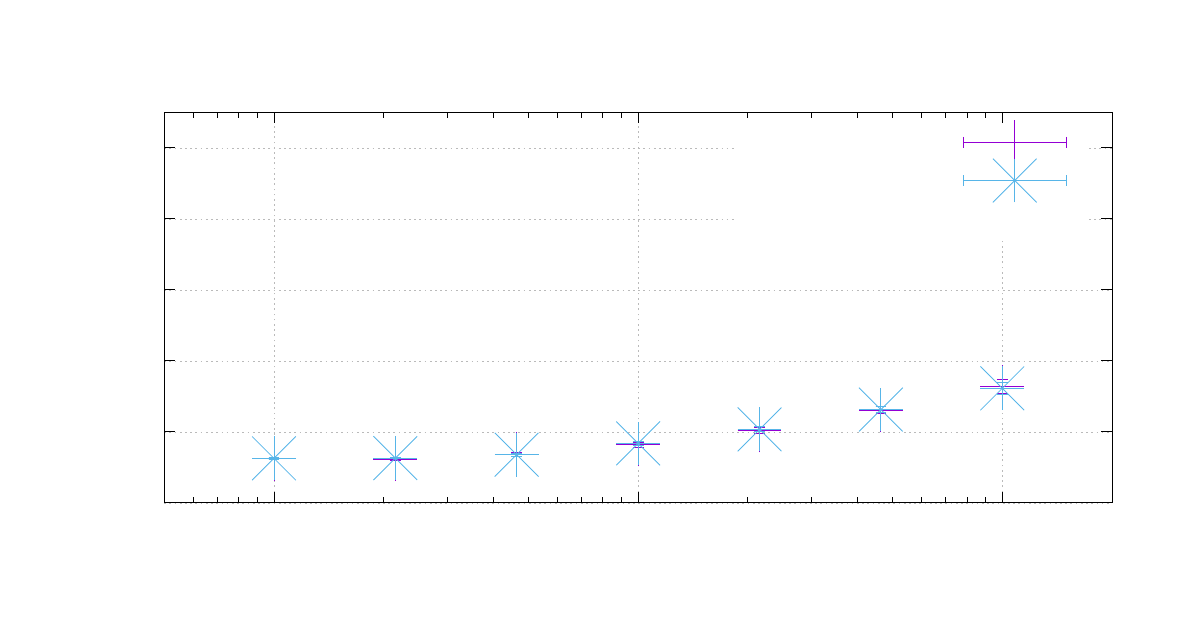}}%
    \gplfronttext
  \end{picture}%
\endgroup
}\hfill\\ %
\noindent\scalebox{\threescale}{%
\begingroup
  \makeatletter
  \providecommand\color[2][]{%
    \GenericError{(gnuplot) \space\space\space\@spaces}{%
      Package color not loaded in conjunction with
      terminal option `colourtext'%
    }{See the gnuplot documentation for explanation.%
    }{Either use 'blacktext' in gnuplot or load the package
      color.sty in LaTeX.}%
    \renewcommand\color[2][]{}%
  }%
  \providecommand\includegraphics[2][]{%
    \GenericError{(gnuplot) \space\space\space\@spaces}{%
      Package graphicx or graphics not loaded%
    }{See the gnuplot documentation for explanation.%
    }{The gnuplot epslatex terminal needs graphicx.sty or graphics.sty.}%
    \renewcommand\includegraphics[2][]{}%
  }%
  \providecommand\rotatebox[2]{#2}%
  \@ifundefined{ifGPcolor}{%
    \newif\ifGPcolor
    \GPcolortrue
  }{}%
  \@ifundefined{ifGPblacktext}{%
    \newif\ifGPblacktext
    \GPblacktextfalse
  }{}%
  \let\gplgaddtomacro\g@addto@macro
  \gdef\gplbacktext{}%
  \gdef\gplfronttext{}%
  \makeatother
  \ifGPblacktext
    \def\colorrgb#1{}%
    \def\colorgray#1{}%
  \else
    \ifGPcolor
      \def\colorrgb#1{\color[rgb]{#1}}%
      \def\colorgray#1{\color[gray]{#1}}%
      \expandafter\def\csname LTw\endcsname{\color{white}}%
      \expandafter\def\csname LTb\endcsname{\color{black}}%
      \expandafter\def\csname LTa\endcsname{\color{black}}%
      \expandafter\def\csname LT0\endcsname{\color[rgb]{1,0,0}}%
      \expandafter\def\csname LT1\endcsname{\color[rgb]{0,1,0}}%
      \expandafter\def\csname LT2\endcsname{\color[rgb]{0,0,1}}%
      \expandafter\def\csname LT3\endcsname{\color[rgb]{1,0,1}}%
      \expandafter\def\csname LT4\endcsname{\color[rgb]{0,1,1}}%
      \expandafter\def\csname LT5\endcsname{\color[rgb]{1,1,0}}%
      \expandafter\def\csname LT6\endcsname{\color[rgb]{0,0,0}}%
      \expandafter\def\csname LT7\endcsname{\color[rgb]{1,0.3,0}}%
      \expandafter\def\csname LT8\endcsname{\color[rgb]{0.5,0.5,0.5}}%
    \else
      \def\colorrgb#1{\color{black}}%
      \def\colorgray#1{\color[gray]{#1}}%
      \expandafter\def\csname LTw\endcsname{\color{white}}%
      \expandafter\def\csname LTb\endcsname{\color{black}}%
      \expandafter\def\csname LTa\endcsname{\color{black}}%
      \expandafter\def\csname LT0\endcsname{\color{black}}%
      \expandafter\def\csname LT1\endcsname{\color{black}}%
      \expandafter\def\csname LT2\endcsname{\color{black}}%
      \expandafter\def\csname LT3\endcsname{\color{black}}%
      \expandafter\def\csname LT4\endcsname{\color{black}}%
      \expandafter\def\csname LT5\endcsname{\color{black}}%
      \expandafter\def\csname LT6\endcsname{\color{black}}%
      \expandafter\def\csname LT7\endcsname{\color{black}}%
      \expandafter\def\csname LT8\endcsname{\color{black}}%
    \fi
  \fi
    \setlength{\unitlength}{0.0500bp}%
    \ifx\gptboxheight\undefined%
      \newlength{\gptboxheight}%
      \newlength{\gptboxwidth}%
      \newsavebox{\gptboxtext}%
    \fi%
    \setlength{\fboxrule}{0.5pt}%
    \setlength{\fboxsep}{1pt}%
\begin{picture}(6802.00,3614.00)%
    \gplgaddtomacro\gplbacktext{%
      \csname LTb\endcsname%
      \put(814,704){\makebox(0,0)[r]{\strut{}$0$}}%
      \csname LTb\endcsname%
      \put(814,1113){\makebox(0,0)[r]{\strut{}$0.2$}}%
      \csname LTb\endcsname%
      \put(814,1522){\makebox(0,0)[r]{\strut{}$0.4$}}%
      \csname LTb\endcsname%
      \put(814,1931){\makebox(0,0)[r]{\strut{}$0.6$}}%
      \csname LTb\endcsname%
      \put(814,2340){\makebox(0,0)[r]{\strut{}$0.8$}}%
      \csname LTb\endcsname%
      \put(814,2749){\makebox(0,0)[r]{\strut{}$1$}}%
      \csname LTb\endcsname%
      \put(1578,484){\makebox(0,0){\strut{}$10^{3}$}}%
      \csname LTb\endcsname%
      \put(3676,484){\makebox(0,0){\strut{}$10^{4}$}}%
      \csname LTb\endcsname%
      \put(5773,484){\makebox(0,0){\strut{}$10^{5}$}}%
    }%
    \gplgaddtomacro\gplfronttext{%
      \csname LTb\endcsname%
      \put(176,1828){\rotatebox{-270}{\makebox(0,0){\strut{}Degree Assortativity}}}%
      \put(3675,154){\makebox(0,0){\strut{}Number $n$ of nodes}}%
      \put(3675,3283){\makebox(0,0){\strut{}Mixing: $\mu = 0.2$, Degree Assortativity, Overlap: $\nu = 4$}}%
      \csname LTb\endcsname%
      \put(5418,2780){\makebox(0,0)[r]{\strut{}Orig}}%
      \csname LTb\endcsname%
      \put(5418,2560){\makebox(0,0)[r]{\strut{}EM}}%
    }%
    \gplbacktext
    \put(0,0){\includegraphics{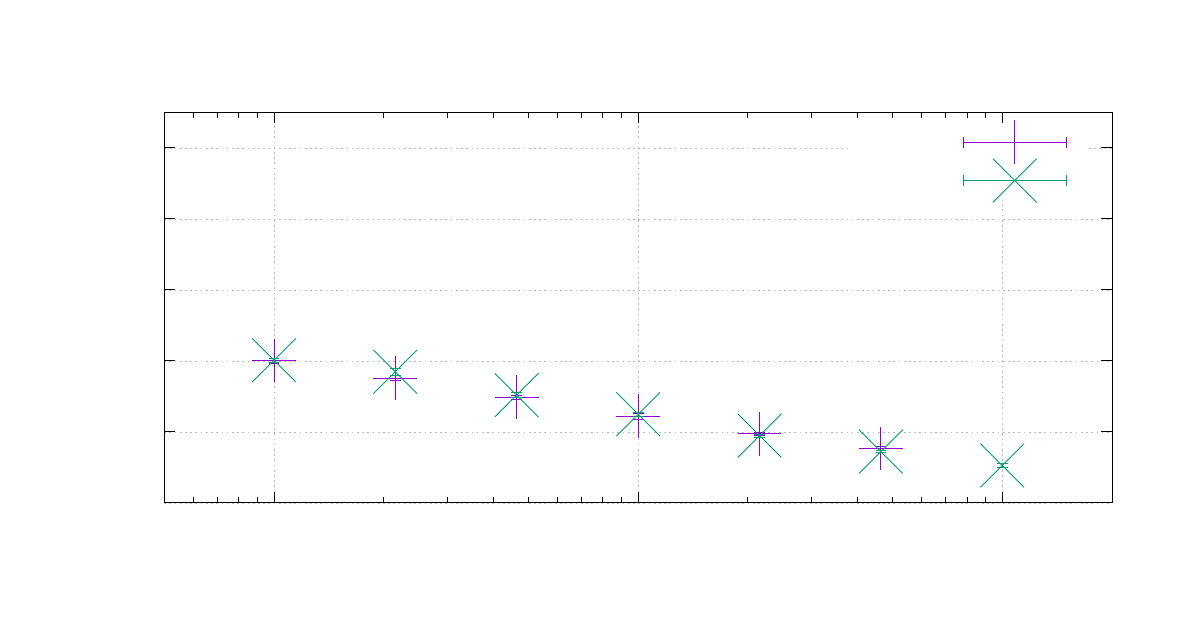}}%
    \gplfronttext
  \end{picture}%
\endgroup
}\hfill\scalebox{\threescale}{%
\begingroup
  \makeatletter
  \providecommand\color[2][]{%
    \GenericError{(gnuplot) \space\space\space\@spaces}{%
      Package color not loaded in conjunction with
      terminal option `colourtext'%
    }{See the gnuplot documentation for explanation.%
    }{Either use 'blacktext' in gnuplot or load the package
      color.sty in LaTeX.}%
    \renewcommand\color[2][]{}%
  }%
  \providecommand\includegraphics[2][]{%
    \GenericError{(gnuplot) \space\space\space\@spaces}{%
      Package graphicx or graphics not loaded%
    }{See the gnuplot documentation for explanation.%
    }{The gnuplot epslatex terminal needs graphicx.sty or graphics.sty.}%
    \renewcommand\includegraphics[2][]{}%
  }%
  \providecommand\rotatebox[2]{#2}%
  \@ifundefined{ifGPcolor}{%
    \newif\ifGPcolor
    \GPcolortrue
  }{}%
  \@ifundefined{ifGPblacktext}{%
    \newif\ifGPblacktext
    \GPblacktextfalse
  }{}%
  \let\gplgaddtomacro\g@addto@macro
  \gdef\gplbacktext{}%
  \gdef\gplfronttext{}%
  \makeatother
  \ifGPblacktext
    \def\colorrgb#1{}%
    \def\colorgray#1{}%
  \else
    \ifGPcolor
      \def\colorrgb#1{\color[rgb]{#1}}%
      \def\colorgray#1{\color[gray]{#1}}%
      \expandafter\def\csname LTw\endcsname{\color{white}}%
      \expandafter\def\csname LTb\endcsname{\color{black}}%
      \expandafter\def\csname LTa\endcsname{\color{black}}%
      \expandafter\def\csname LT0\endcsname{\color[rgb]{1,0,0}}%
      \expandafter\def\csname LT1\endcsname{\color[rgb]{0,1,0}}%
      \expandafter\def\csname LT2\endcsname{\color[rgb]{0,0,1}}%
      \expandafter\def\csname LT3\endcsname{\color[rgb]{1,0,1}}%
      \expandafter\def\csname LT4\endcsname{\color[rgb]{0,1,1}}%
      \expandafter\def\csname LT5\endcsname{\color[rgb]{1,1,0}}%
      \expandafter\def\csname LT6\endcsname{\color[rgb]{0,0,0}}%
      \expandafter\def\csname LT7\endcsname{\color[rgb]{1,0.3,0}}%
      \expandafter\def\csname LT8\endcsname{\color[rgb]{0.5,0.5,0.5}}%
    \else
      \def\colorrgb#1{\color{black}}%
      \def\colorgray#1{\color[gray]{#1}}%
      \expandafter\def\csname LTw\endcsname{\color{white}}%
      \expandafter\def\csname LTb\endcsname{\color{black}}%
      \expandafter\def\csname LTa\endcsname{\color{black}}%
      \expandafter\def\csname LT0\endcsname{\color{black}}%
      \expandafter\def\csname LT1\endcsname{\color{black}}%
      \expandafter\def\csname LT2\endcsname{\color{black}}%
      \expandafter\def\csname LT3\endcsname{\color{black}}%
      \expandafter\def\csname LT4\endcsname{\color{black}}%
      \expandafter\def\csname LT5\endcsname{\color{black}}%
      \expandafter\def\csname LT6\endcsname{\color{black}}%
      \expandafter\def\csname LT7\endcsname{\color{black}}%
      \expandafter\def\csname LT8\endcsname{\color{black}}%
    \fi
  \fi
    \setlength{\unitlength}{0.0500bp}%
    \ifx\gptboxheight\undefined%
      \newlength{\gptboxheight}%
      \newlength{\gptboxwidth}%
      \newsavebox{\gptboxtext}%
    \fi%
    \setlength{\fboxrule}{0.5pt}%
    \setlength{\fboxsep}{1pt}%
\begin{picture}(6802.00,3614.00)%
    \gplgaddtomacro\gplbacktext{%
      \csname LTb\endcsname%
      \put(814,704){\makebox(0,0)[r]{\strut{}$0$}}%
      \csname LTb\endcsname%
      \put(814,1113){\makebox(0,0)[r]{\strut{}$0.2$}}%
      \csname LTb\endcsname%
      \put(814,1522){\makebox(0,0)[r]{\strut{}$0.4$}}%
      \csname LTb\endcsname%
      \put(814,1931){\makebox(0,0)[r]{\strut{}$0.6$}}%
      \csname LTb\endcsname%
      \put(814,2340){\makebox(0,0)[r]{\strut{}$0.8$}}%
      \csname LTb\endcsname%
      \put(814,2749){\makebox(0,0)[r]{\strut{}$1$}}%
      \csname LTb\endcsname%
      \put(1578,484){\makebox(0,0){\strut{}$10^{3}$}}%
      \csname LTb\endcsname%
      \put(3676,484){\makebox(0,0){\strut{}$10^{4}$}}%
      \csname LTb\endcsname%
      \put(5773,484){\makebox(0,0){\strut{}$10^{5}$}}%
    }%
    \gplgaddtomacro\gplfronttext{%
      \csname LTb\endcsname%
      \put(176,1828){\rotatebox{-270}{\makebox(0,0){\strut{}Degree Assortativity}}}%
      \put(3675,154){\makebox(0,0){\strut{}Number $n$ of nodes}}%
      \put(3675,3283){\makebox(0,0){\strut{}Mixing: $\mu = 0.4$, Degree Assortativity, Overlap: $\nu = 4$}}%
      \csname LTb\endcsname%
      \put(5418,2780){\makebox(0,0)[r]{\strut{}Orig}}%
      \csname LTb\endcsname%
      \put(5418,2560){\makebox(0,0)[r]{\strut{}EM}}%
    }%
    \gplbacktext
    \put(0,0){\includegraphics{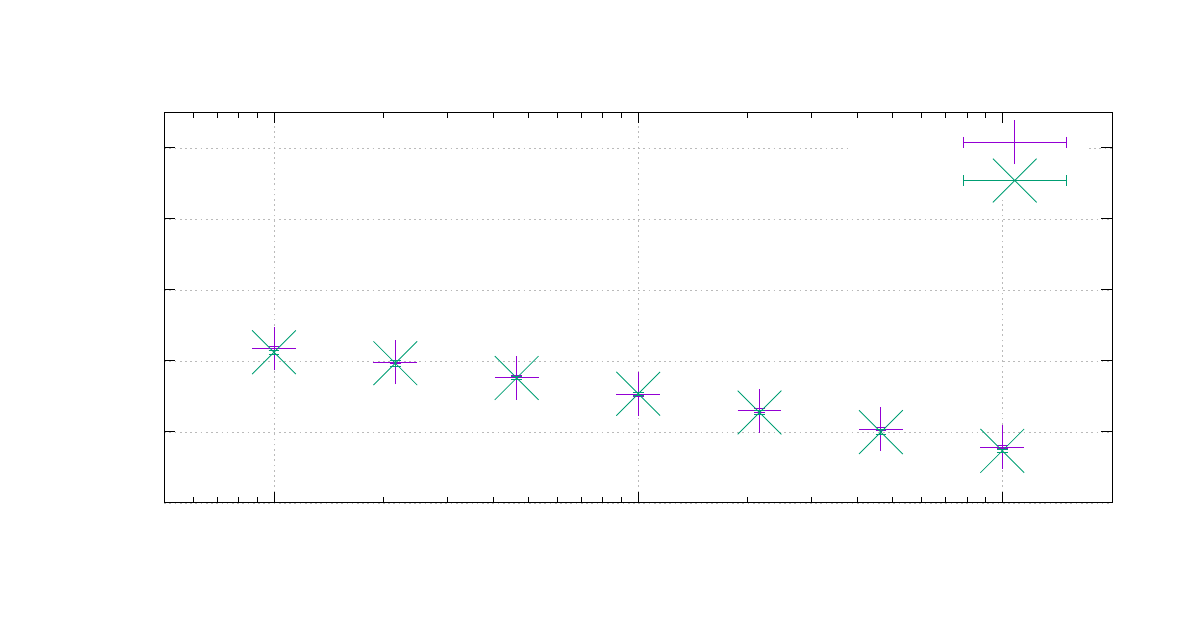}}%
    \gplfronttext
  \end{picture}%
\endgroup
}\hfill\scalebox{\threescale}{%
\begingroup
  \makeatletter
  \providecommand\color[2][]{%
    \GenericError{(gnuplot) \space\space\space\@spaces}{%
      Package color not loaded in conjunction with
      terminal option `colourtext'%
    }{See the gnuplot documentation for explanation.%
    }{Either use 'blacktext' in gnuplot or load the package
      color.sty in LaTeX.}%
    \renewcommand\color[2][]{}%
  }%
  \providecommand\includegraphics[2][]{%
    \GenericError{(gnuplot) \space\space\space\@spaces}{%
      Package graphicx or graphics not loaded%
    }{See the gnuplot documentation for explanation.%
    }{The gnuplot epslatex terminal needs graphicx.sty or graphics.sty.}%
    \renewcommand\includegraphics[2][]{}%
  }%
  \providecommand\rotatebox[2]{#2}%
  \@ifundefined{ifGPcolor}{%
    \newif\ifGPcolor
    \GPcolortrue
  }{}%
  \@ifundefined{ifGPblacktext}{%
    \newif\ifGPblacktext
    \GPblacktextfalse
  }{}%
  \let\gplgaddtomacro\g@addto@macro
  \gdef\gplbacktext{}%
  \gdef\gplfronttext{}%
  \makeatother
  \ifGPblacktext
    \def\colorrgb#1{}%
    \def\colorgray#1{}%
  \else
    \ifGPcolor
      \def\colorrgb#1{\color[rgb]{#1}}%
      \def\colorgray#1{\color[gray]{#1}}%
      \expandafter\def\csname LTw\endcsname{\color{white}}%
      \expandafter\def\csname LTb\endcsname{\color{black}}%
      \expandafter\def\csname LTa\endcsname{\color{black}}%
      \expandafter\def\csname LT0\endcsname{\color[rgb]{1,0,0}}%
      \expandafter\def\csname LT1\endcsname{\color[rgb]{0,1,0}}%
      \expandafter\def\csname LT2\endcsname{\color[rgb]{0,0,1}}%
      \expandafter\def\csname LT3\endcsname{\color[rgb]{1,0,1}}%
      \expandafter\def\csname LT4\endcsname{\color[rgb]{0,1,1}}%
      \expandafter\def\csname LT5\endcsname{\color[rgb]{1,1,0}}%
      \expandafter\def\csname LT6\endcsname{\color[rgb]{0,0,0}}%
      \expandafter\def\csname LT7\endcsname{\color[rgb]{1,0.3,0}}%
      \expandafter\def\csname LT8\endcsname{\color[rgb]{0.5,0.5,0.5}}%
    \else
      \def\colorrgb#1{\color{black}}%
      \def\colorgray#1{\color[gray]{#1}}%
      \expandafter\def\csname LTw\endcsname{\color{white}}%
      \expandafter\def\csname LTb\endcsname{\color{black}}%
      \expandafter\def\csname LTa\endcsname{\color{black}}%
      \expandafter\def\csname LT0\endcsname{\color{black}}%
      \expandafter\def\csname LT1\endcsname{\color{black}}%
      \expandafter\def\csname LT2\endcsname{\color{black}}%
      \expandafter\def\csname LT3\endcsname{\color{black}}%
      \expandafter\def\csname LT4\endcsname{\color{black}}%
      \expandafter\def\csname LT5\endcsname{\color{black}}%
      \expandafter\def\csname LT6\endcsname{\color{black}}%
      \expandafter\def\csname LT7\endcsname{\color{black}}%
      \expandafter\def\csname LT8\endcsname{\color{black}}%
    \fi
  \fi
    \setlength{\unitlength}{0.0500bp}%
    \ifx\gptboxheight\undefined%
      \newlength{\gptboxheight}%
      \newlength{\gptboxwidth}%
      \newsavebox{\gptboxtext}%
    \fi%
    \setlength{\fboxrule}{0.5pt}%
    \setlength{\fboxsep}{1pt}%
\begin{picture}(6802.00,3614.00)%
    \gplgaddtomacro\gplbacktext{%
      \csname LTb\endcsname%
      \put(814,704){\makebox(0,0)[r]{\strut{}$0$}}%
      \csname LTb\endcsname%
      \put(814,1113){\makebox(0,0)[r]{\strut{}$0.2$}}%
      \csname LTb\endcsname%
      \put(814,1522){\makebox(0,0)[r]{\strut{}$0.4$}}%
      \csname LTb\endcsname%
      \put(814,1931){\makebox(0,0)[r]{\strut{}$0.6$}}%
      \csname LTb\endcsname%
      \put(814,2340){\makebox(0,0)[r]{\strut{}$0.8$}}%
      \csname LTb\endcsname%
      \put(814,2749){\makebox(0,0)[r]{\strut{}$1$}}%
      \csname LTb\endcsname%
      \put(1578,484){\makebox(0,0){\strut{}$10^{3}$}}%
      \csname LTb\endcsname%
      \put(3676,484){\makebox(0,0){\strut{}$10^{4}$}}%
      \csname LTb\endcsname%
      \put(5773,484){\makebox(0,0){\strut{}$10^{5}$}}%
    }%
    \gplgaddtomacro\gplfronttext{%
      \csname LTb\endcsname%
      \put(176,1828){\rotatebox{-270}{\makebox(0,0){\strut{}Degree Assortativity}}}%
      \put(3675,154){\makebox(0,0){\strut{}Number $n$ of nodes}}%
      \put(3675,3283){\makebox(0,0){\strut{}Mixing: $\mu = 0.6$, Degree Assortativity, Overlap: $\nu = 4$}}%
      \csname LTb\endcsname%
      \put(5418,2780){\makebox(0,0)[r]{\strut{}Orig}}%
      \csname LTb\endcsname%
      \put(5418,2560){\makebox(0,0)[r]{\strut{}EM}}%
    }%
    \gplbacktext
    \put(0,0){\includegraphics{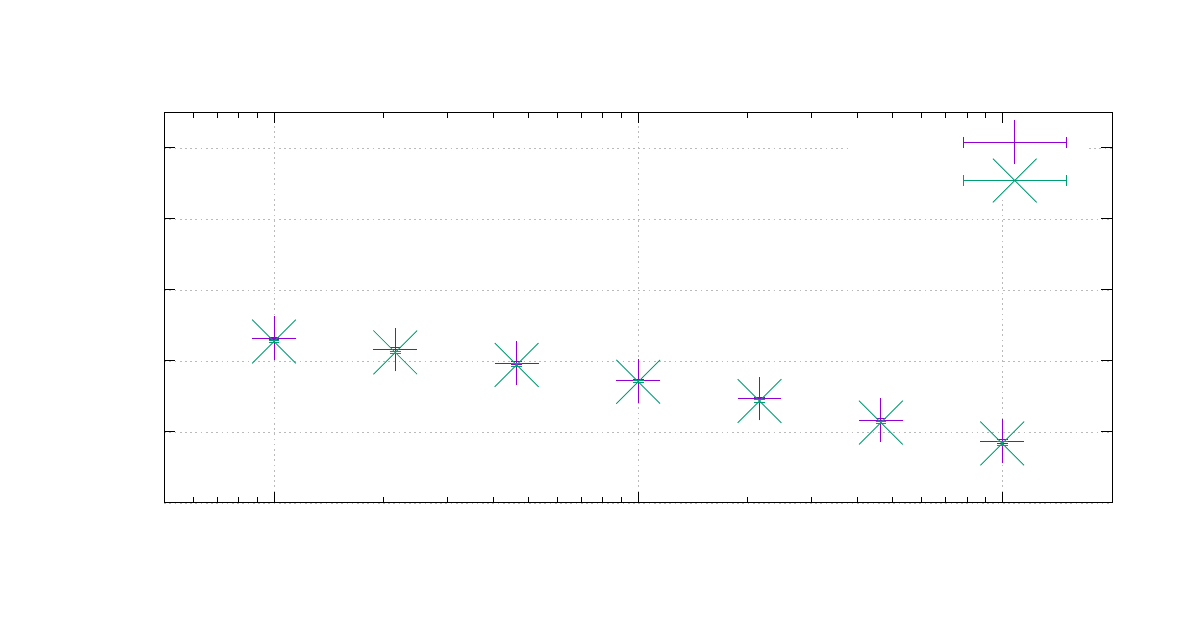}}%
    \gplfronttext
  \end{picture}%
\endgroup
}\hfill\\ %

Comparison of the original LFR implementation and our EM solution for values values of 
$10^3 \le n \le 10^6$, 
$\mu{\in}\{0.2, 0.4, 0.6\}$, 
$\nu{\in}\{2,3,4\}$,
$O=n$,
$\gamma{=}2$, $\beta{=}-1$
$d_\text{min}{=}10$, $d_\text{max}{=}n/20$,
$s_\text{min}{=}10\nu$, $s_\text{max}{=}\nu\cdot n/20$.
Clustering is performed using OSLOM and compared to the ground-truth emitted by the generator using a generalized Normalized Mutual Information (NMI); $S \ge 5$.
}
 \clearpage
\section{Comparing \emes{} and \emcmes{}}\label{sec:appendix-emcmes}
\noindent
\includegraphics[width=0.45\textwidth]{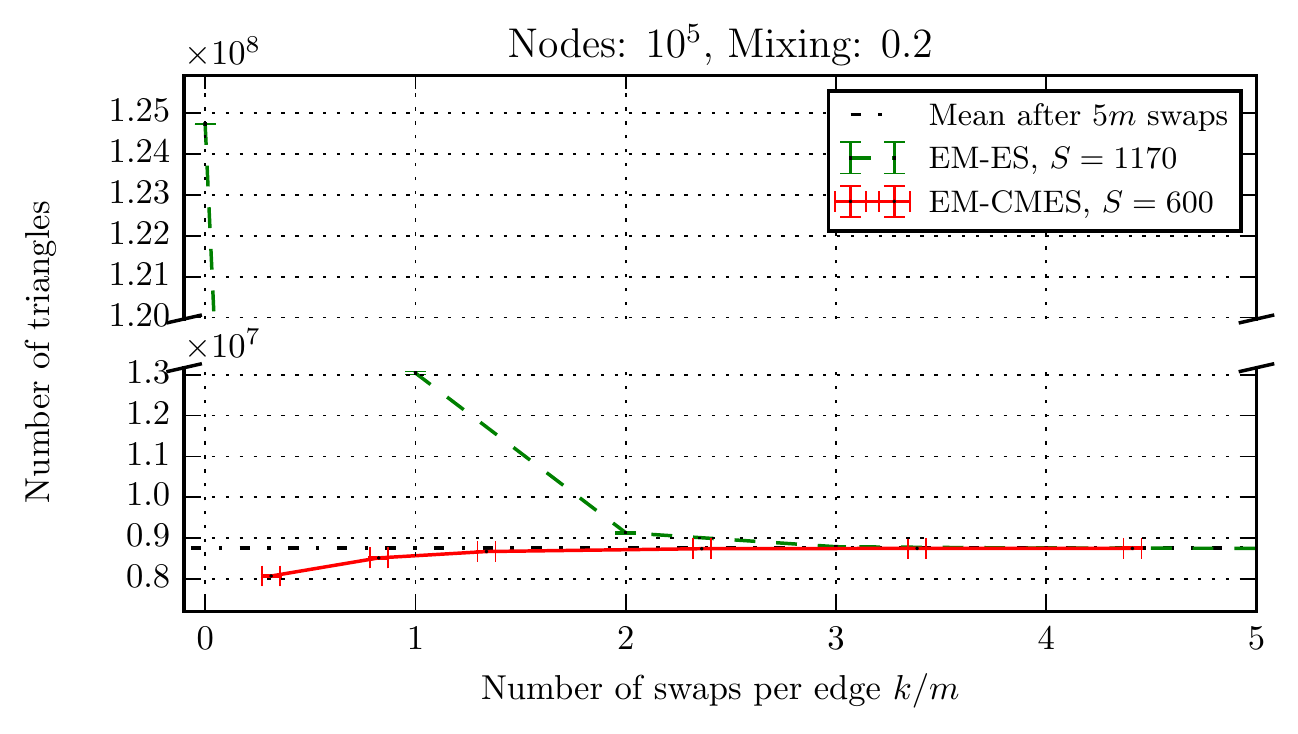}\hfill
\includegraphics[width=0.45\textwidth]{{edgeswap_bench_n100K_mu1.0_TRI}.pdf}\\

\noindent
\includegraphics[width=0.45\textwidth]{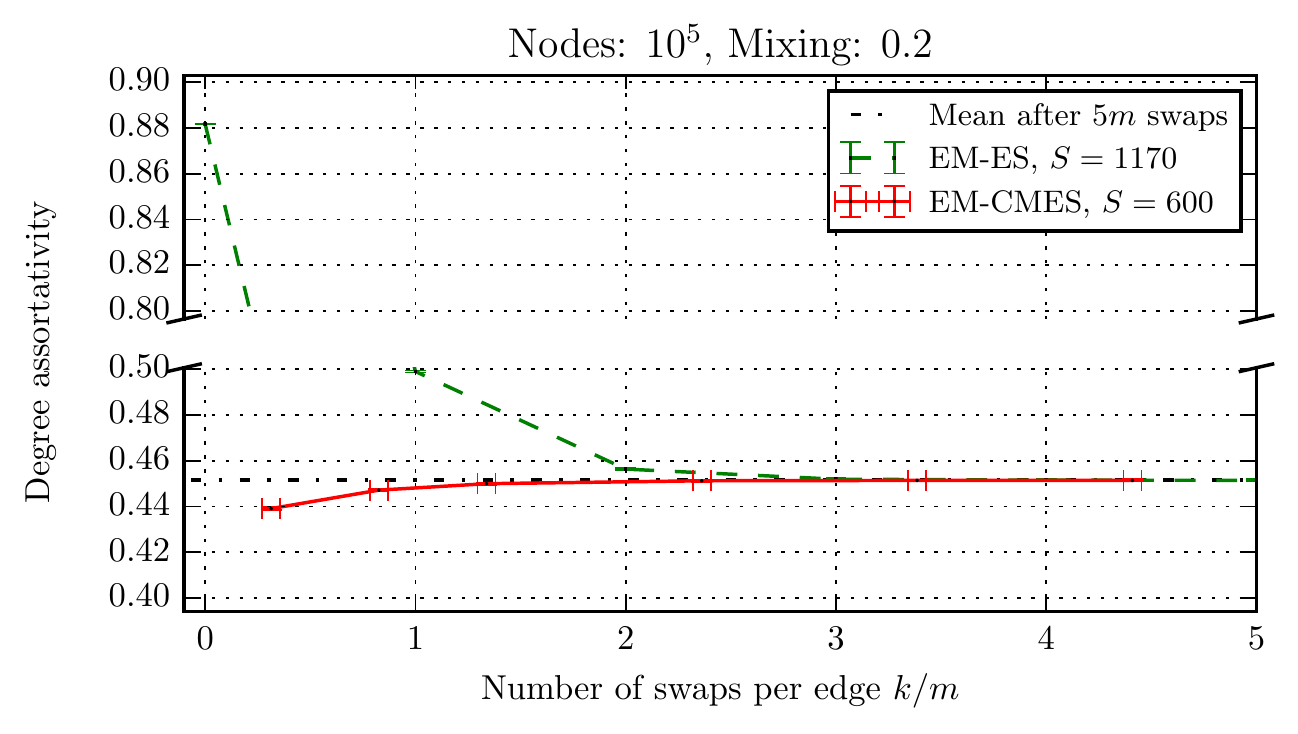}\hfill
\includegraphics[width=0.45\textwidth]{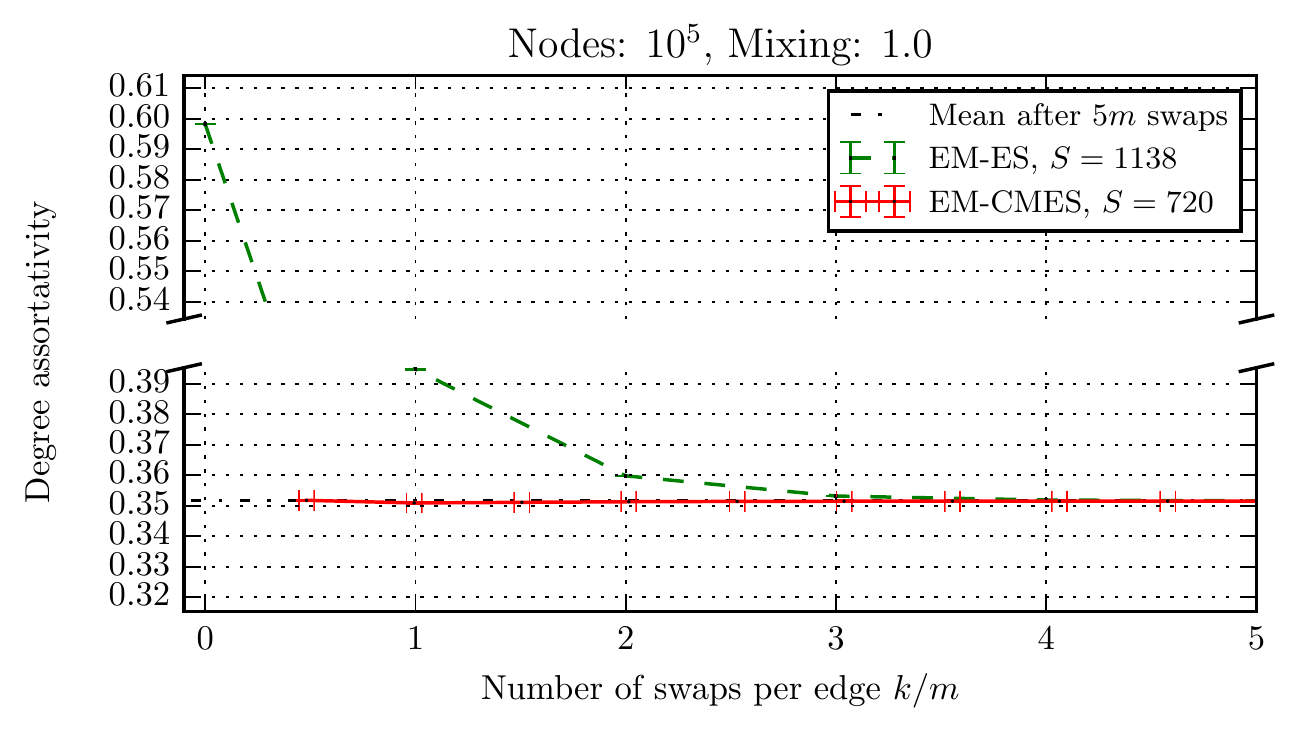}\\

\noindent
\includegraphics[width=0.45\textwidth]{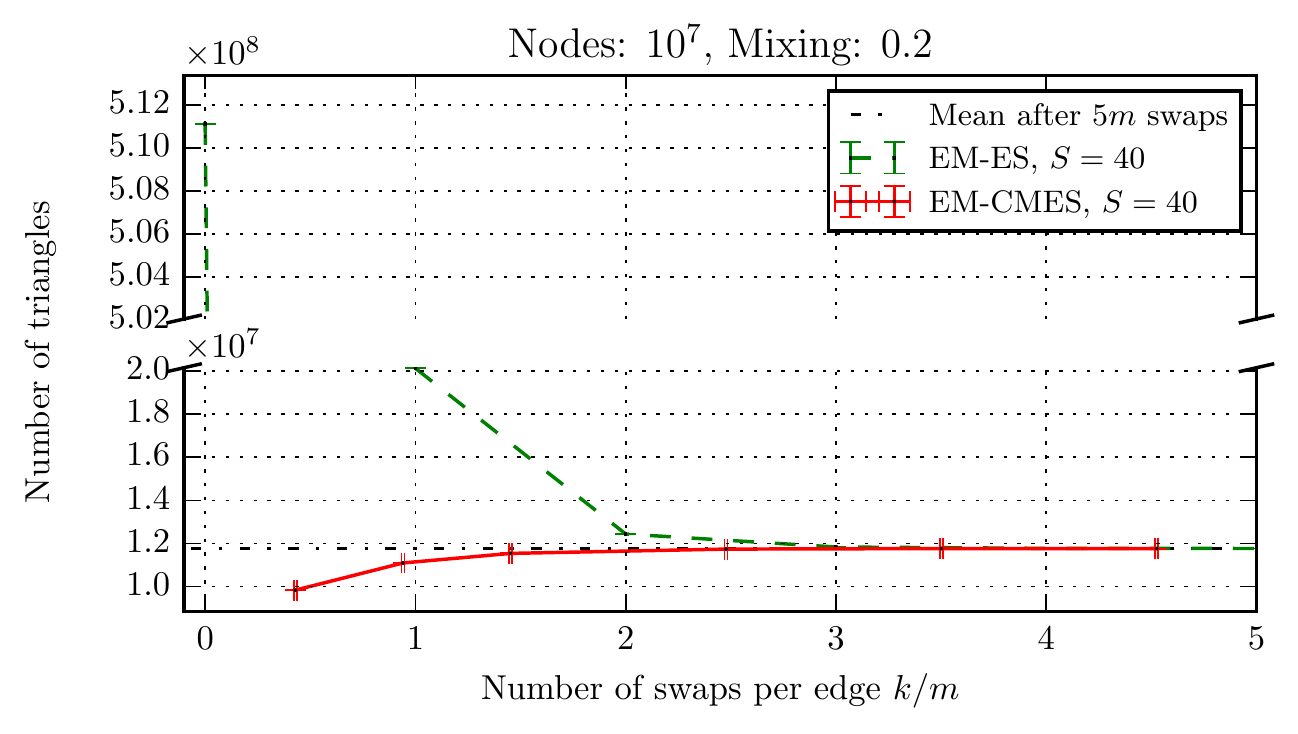}\hfill
\includegraphics[width=0.45\textwidth]{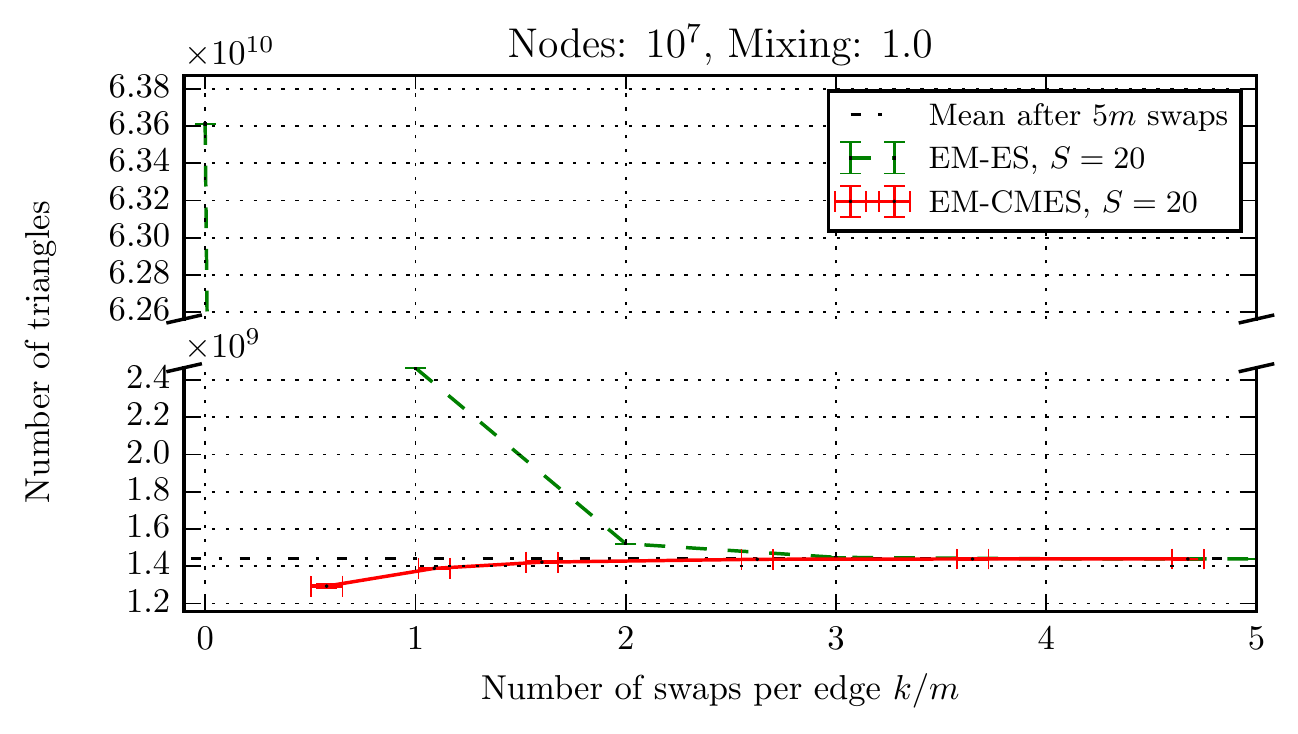}

\noindent
\includegraphics[width=0.45\textwidth]{{edgeswap_bench_n10M_mu0.2_DEGASS}.pdf}\hfill
\includegraphics[width=0.45\textwidth]{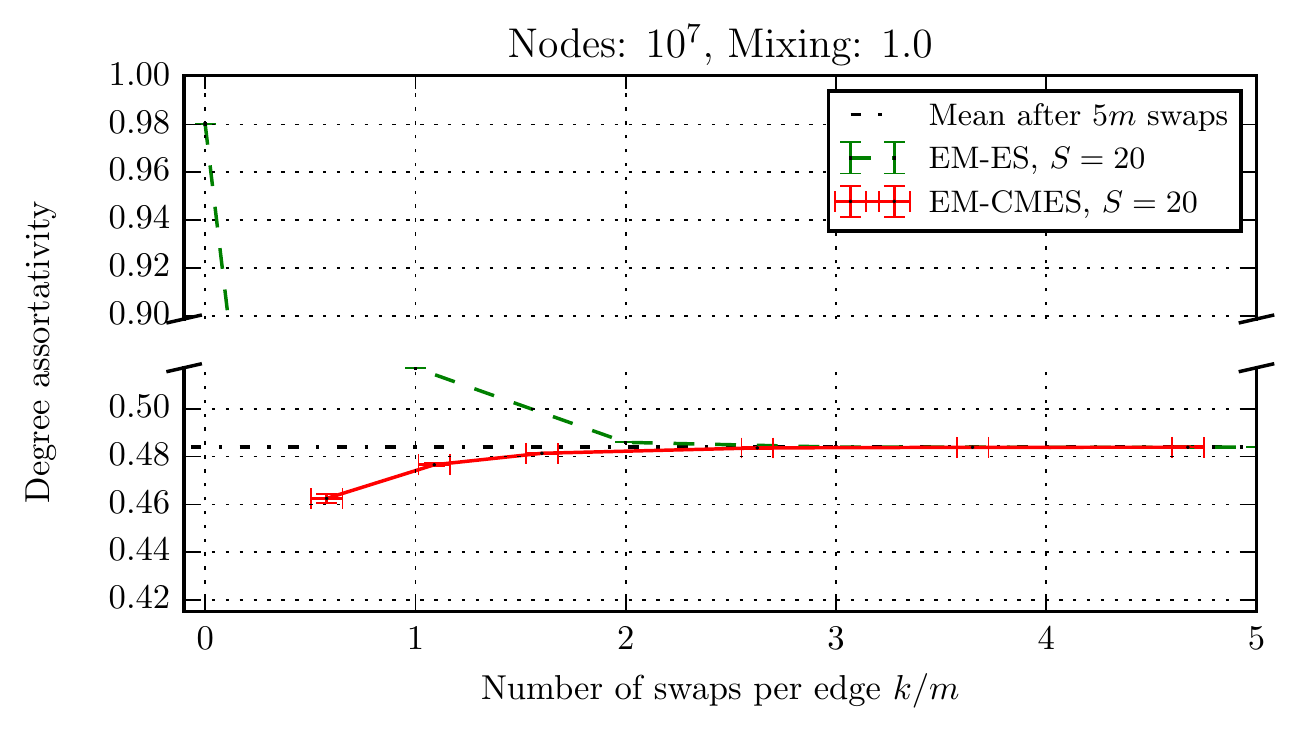}

Triangle count and degree assortativity of a graph ensemble obtained by applying random swaps/the Configuration Model to a common seed graph.
Refer to section \ref{subsec:emcmes-exp} for experimental details. \end{document}